%% file: Thesis.tex
\documentclass[twoside,a4paper,11pt]{book}
\usepackage{amsmath}
\usepackage{amsthm}
\usepackage{graphicx}
\usepackage[usenames,dvipsnames]{xcolor}
\usepackage{bbm}
\usepackage{fancyhdr}
\usepackage{parskip} 
\usepackage[english]{babel} 
\usepackage[titletoc]{appendix} 
\usepackage{rotating} 
\usepackage[justification=centering,labelfont=bf,font=small]{caption} 
\usepackage{subcaption} 

\usepackage{enumitem}

\newcommand{\Tr}{\mathrm{Tr}}

\newcommand{\bra}[1]{\langle #1 |}
\newcommand{\bbra}[1]{\langle\!\langle #1 |}
\newcommand{\ket}[1]{| #1 \rangle}
\newcommand{\kket}[1]{| #1 \rangle\!\rangle}
\newcommand{\bk}[2]{\langle #1 | #2 \rangle}
\newcommand{\ball}[1]{\mathcal{B}^{\epsilon} (#1)}
\newcommand{\supp}{\mathrm{supp}}
\newcommand{\kb}[2]{| #1 \rangle \langle #2 |}
\newcommand{\hilbert}{\mathcal{H}}

\DeclareMathOperator{\id}{id}
\DeclareMathOperator{\sech}{sech}

\newcommand{\identity}{\mathds{1}}

\usepackage{xspace} 
\newcommand{\stage}{stage\xspace}

\newcommand{\emptypage}{\newpage\null\thispagestyle{plain}\newpage}

\newtheoremstyle{mystyle}
  {\baselineskip}
  {\topsep}
  {\itshape}
  {0pt}
  {\bfseries}
  {.}
  {5pt plus 1pt minus 1pt}
  {}
  \theoremstyle{mystyle}
\newtheorem{thm}{Theorem}[section]
\newtheorem{lemma}[thm]{Lemma}
\newtheorem{defn}[thm]{Definition}

\usepackage{indentfirst}

\usepackage[T1]{fontenc}
\newcommand{\tfont}{\raggedright\usefont{T1}{qhv}{b}{n}\selectfont} 
\newcommand{\tofont}{\usefont{T1}{qhv}{m}{n}\selectfont} 

\usepackage[charter,expert]{mathdesign}
\usepackage{charter}
\usepackage{dsfont} 
\newcommand*{\quotefont}{\itshape} 

\newcommand{\tspace}{\hspace{0.6em}}

\usepackage{titlesec}
\titleformat{\chapter}[display]{\fontsize{22pt}{1em}\tfont}{{\chaptertitlename} \thechapter\vspace{0.4cm}}{0pt}{\fontsize{30pt}{1em}\tfont}
\titleformat{\section}[block]{\Large\tfont}{\thesection\tspace}{0pt}{\Large\tfont}
\titleformat{\subsection}[block]{\large\tfont}{\thesubsection\tspace}{0pt}{\large\tfont}
\titleformat{\subsubsection}[block]{\tfont}{\thesubsubsection\tspace}{0pt}{\tfont}

\usepackage{tocloft}
\tocloftpagestyle{plain}
\addto\captionsenglish{}

\usepackage[numbers]{natbib}

\definecolor{citec}{RGB}{190,40,5} 
\definecolor{linkc}{RGB}{5,160,40} 
\usepackage[colorlinks=true,linkcolor=linkc,citecolor=citec,urlcolor=blue,linktocpage,backref=page]{hyperref}

\usepackage{etoolbox}
\apptocmd{\sloppy}{\hbadness 10000\relax}{}{} 
\makeatletter
\patchcmd{\BR@backref}{\newblock}{\newblock[}{}{}
\patchcmd{\BR@backref}{\par}{]\par}{}{}
\makeatother

\begin{document}

\pagenumbering{Roman}
\input{Title.tex}
\pagenumbering{roman}

\emptypage

\input{Acknowledgements.tex}

\emptypage

\input{Abstract.tex}

\emptypage

\thispagestyle{empty}
\vspace*{4.5cm}
\begin{figure}[h]
  {\raggedright{}\quotefont
   In all science we have to distinguish two sorts of laws:\\
   \mbox{first, those that are empirically verifiable but probably only approximate;}\\
   secondly, those that are not verifiable, but may be exact.\par\bigskip
   }  
  \raggedleft{}\small{\scshape -- Bertrand Russell}, On the Notion of Cause (1913)\par
\end{figure}
\clearpage

\emptypage
\tableofcontents
\emptypage
\phantomsection
\addcontentsline{toc}{chapter}{Figures}
\listoffigures
\emptypage

\newpage

\pagenumbering{arabic}

\pagestyle{plain}

\input{Notation.tex}

\input{Introduction.tex}

\input{Preliminaries.tex}

\input{Security_Proofs.tex}

\input{Assumptions.tex}

\input{Contributions.tex}

\input{Conclusion.tex}

\begin{appendices}
\input{Appendix}
\end{appendices}

\cleardoublepage 
\phantomsection 
\addcontentsline{toc}{chapter}{Bibliography}
\bibliographystyle{alphaarxiv.bst}
\bibliography{library2}

\end{document}

%% file: Title.tex

\begin{titlepage}
\thispagestyle{empty}
\begin{center}

DISS. ETH NO. $22269$ \\[2.6cm]

{\Large \tfont{ASSUMPTIONS IN QUANTUM CRYPTOGRAPHY}} \\[2.6cm]

A thesis submitted to attain the degree of \\[0.2cm]
DOCTOR OF SCIENCES of ETH ZURICH \\[0.2cm]
(Dr. sc. ETH Zurich) \\[1.8cm]

presented by\\[0.6cm]
\textit{NORMAND JAMES BEAUDRY}\\[0.6cm]
Master of Science The University of Waterloo\\[0.2cm]
born on \textit{19.09.1984}\\[0.2cm]
citizen of Canada\\[1.8cm]

accepted on the recommendation of\\[0.6cm]
\textit{Prof.~Dr.~Renato Renner, examiner}\\
\textit{Prof.~Dr.~Christian Schaffner, coexaminer}\\
\textit{Prof.~Dr.~Norbert L{\"u}tkenhaus, coexaminer}

\vfill

2014

\end{center}
\end{titlepage}

%% file: Acknowledgements.tex

\thispagestyle{plain}
\section*{\centerline{Acknowledgements}}
First, thank you to my supervisor, Renato Renner, for the freedom to pursue the research that interested me and the trust you had in my work throughout my PhD. Your keen insight and positive outlook always seem to lead to new ways of tackling any problem.

Thank you to my doctoral committee, including Christian Schaffner and Norbert L{\"u}tkenhaus, for many helpful comments and suggestions for this thesis.

Thanks to the group members of the quantum information group at ETH for an enjoyable time during my PhD. Thank you to those who read early parts of this thesis: David Sutter, L{\'i}dia del Rio, Lea Kr{\"a}mer, Christopher Portmann, Phillipp Kammerlander, Omar Fawzi, Daniela Frauchiger, and Rotem Arnon-Friedman. You all provided very valuable feedback. Thank you to my office mates for helping me with my many questions: Fr{\'e}d{\'e}ric Dupuis, Marco Tomamichel, Rotem Arnon-Friedman, and Christopher Portmann. To Johan {\AA}berg, David Sutter, and Joe Renes, thanks for listening to my questions, despite not being in my office.

To my collaborators, thank you for the opportunity to work with you and our many fruitful discussions: Marco Lucamarini, Stefano Mancini, Nikola Ciagnovi{\'c}, Lana Sheridan, Adriana Marais, Johannes Wilms, and Omar Fawzi, as well as Oleg Gittsovich, Varun Narasimhachar, Ruben Alvarez, Tobias Moroder, and Norbert L{\"{u}}tkenhaus. Thank you to the students I supervised: Nikola Ciganovi{\'c}, Pascal Basler, Paul Erker, and David Reutter. It was a pleasure to work with you.

Thanks to Marco Tomamichel for some of the formatting of this thesis, most notably the excellent references with links. For the translation of the abstract into German, thanks to Volkher Scholz, Lea Kr{\"a}mer, and David Sutter.

Thank you to my family and Jan for all your love and support.

Lastly, special thanks to Fred, Marco, L{\'i}dia, and Johan for all your help and encouragement. Without you, this thesis would not be possible.

%% file: Abstract.tex

\thispagestyle{plain}
\section*{\centerline{Abstract}}

Quantum cryptography uses techniques and ideas from physics and computer science. The combination of these ideas makes the security proofs of quantum cryptography a complicated task.

To prove that a quantum-cryptography protocol is secure, assumptions are made about the protocol and its devices. If these assumptions are not justified in an implementation then an eavesdropper may break the security of the protocol. Therefore, security is crucially dependent on which assumptions are made and how justified the assumptions are in an implementation of the protocol.

This thesis analyzes and clarifies the connection between the security proofs of quantum-cryptography protocols and their experimental implementations. In particular, we focus on quantum key distribution: the task of distributing a secret random key between two parties.

We propose a framework that decomposes quantum-key-distribution protocols and their assumptions into several classes. Protocol classes can be used to clarify which proof techniques apply to which kinds of protocols. Assumption classes can be used to specify which assumptions are justified in implementations and which could be exploited by an eavesdropper.

We provide a comprehensive introduction to several concepts: quantum mechanics using the density operator formalism, quantum cryptography, and quantum key distribution. We define security for quantum key distribution and outline several mathematical techniques that can either be used to prove security or simplify security proofs. In addition, we analyze the assumptions made in quantum cryptography and how they may or may not be justified in implementations.

%% file: Notation.tex

\chapter*{Notation}
\addcontentsline{toc}{chapter}{Notation}
\pagestyle{plain}
\begin{table}[ht]
\begin{tabular}{c | p{9cm}}
\bf{Abbreviation} & \bf{Description} \\ \hline
CPTP & Completely positive and trace preserving \\
QKD & Quantum key distribution \\
i.i.d. & Independent and identically distributed \\
POVM & Positive operator valued measure\\
CQ & Classical-quantum  \\
CJ & Choi-Jamio{\l}kowski \\
P\&M & Prepare and Measure
\end{tabular}
\caption{List of common abbreviations.}
\end{table}

\begin{table}[ht]
\begin{tabular}{c | p{9cm}}
\bf{Term} & \bf{Description} \\ \hline
Bit & A binary digit that is either $0$ or $1$ \\
String & A list of bits (or other numbers) \\
Qubit & A quantum bit, i.e.~a two-level quantum system, typically represented with the basis $\{\ket{0},\ket{1}\}$ \\
Key & A string that is supposed to be secret \\
Seed & A short random string used as a catalyst to extract randomness from a system \\
Key rate & The ratio of secret key to number of signals in the limit as the number of signals goes to infinity \\
Error rate & The ratio of the number of errors in the key to the size of the key \\
Threshold & Maximum tolerable error rate \\
Active & A device that requires active control \\
Passive & A device that does not require active control \\
\end{tabular}
\caption{List of common terms in quantum key distribution and information theory.}
\end{table}

\begin{table}[ht]
\vspace{1cm}
\begin{tabular}{c | p{9.5cm}}
\bf{Symbol} & \bf{Description} \\ \hline
$A,B,C,...$ & Quantum systems \\
$W,X,Y,Z$ & Classical systems \\
$\hat{a}, \hat{D}$ & Quantum operators \\
$\hilbert_{A}$ & Hilbert space corresponding to the system $A$ \\
$\hilbert^{\otimes n}$ & $n$ copies of $\hilbert$: $\hilbert\otimes\hilbert\otimes\cdots\otimes\hilbert$ \\
$\Tr,\Tr_A$ & The trace and the partial trace of the system $A$ \\
$\mathcal{P}(\mathcal{H})$ & The set of positive semi-definite operators on $\mathcal{H}$ \\
$A\geq B$ & An operator inequality equivalent to $A-B\in \mathcal{P}(\mathcal{H})$ \\
$S_{=}(\mathcal{H})$ & The set of normalized quantum states \\
$S_{\leq}(\mathcal{H})$ & The set of sub-normalized quantum states \\
$\log \equiv \log_2$ & The logarithm with base 2 \\
$\ln$ & The natural logarithm \\
$\mathds{R},\mathds{C}$ & The real and complex numbers \\
$\identity$ & The identity operator \\
$\id$ & The identity superoperator \\
$\rho,\sigma,\tau$ & Quantum density operators \\
$D(\rho,\sigma)$ & Trace distance between $\rho$ and $\sigma$ \\
$F(\rho,\sigma)$ & The generalized fidelity between $\rho$ and $\sigma$ \\
$P(\rho,\sigma)$ & The purified distance between $\rho$ and $\sigma$ \\
$X^{\dag}$ & The Hermitian adjoint of operator $X$ \\
$X^{T}$ & The transpose of operator $X$ \\
$X^{-1}$ & The generalized inverse of operator $X$ \\
$\| X \|_{\infty}$ & The operator norm of operator $X$ \\
$\| X \|_{1}$ & The trace norm of operator $X$ \\
$[n]$ & The set of integers $\{1,2,\dots,n\}$
\end{tabular}
\caption{List of commonly used symbols and expressions.}
\end{table}

%% file: Introduction.tex

\chapter{Introduction} \label{chap:intro}
\pagestyle{fancy}

Physics aims to describe our physical reality so that we can make predictions about our universe. With mathematics as its backbone, physics has been the most successful way humanity has devised to describe the physical world, allowing us to reach the technological advancement we have today. Usually the fundamental theories of physics have a simple description. This fact is a remarkable feature of our universe! We can reduce the complicated phenomena we observe to mathematical models. However, this raises two questions. Firstly, do the models we use to describe reality are actually what we mean by `reality?' It is ambiguous what the difference is between the models we use to describe reality and what we mean by `reality.' Secondly, do our descriptions properly describe the way nature works or will we continually find that our models are never accurate enough? Maybe it is continually necessary to update our models as we do new and more accurate experiments that go beyond what we have done previously.

Consider Newton's law of gravity. If there are two point masses $m_1$ and $m_2$ with a distance $r$ between them, then the strength of the force that they exert on each other is
\begin{equation}\label{eq:gravity}
F=\frac{Gm_1m_2}{r^2},
\end{equation}
where $G$ is the gravitational constant.

How do we know that this is the way gravity works? First, you could imagine performing an experiment where you try different masses for $m_1$ and $m_2$ separated at different distances and measure the force between them. Then you could see that Eq.~\ref{eq:gravity} seems to describe the value of the force within a certain level of accuracy. As more and more precise experiments are performed, by more accurately measuring the masses and distances, a better estimate on the exact value of the gravitational constant $G$ could be obtained and it could be determined if Eq.~\ref{eq:gravity} holds. Not only can we perform more accurate experiments, but we can also push the boundaries of these parameters. We can try very large or very small masses, as well as very large and very small distances. In these two ways, we can test whether Eq.~\ref{eq:gravity} describes reality or not.

Sometimes laws like Eq.~\ref{eq:gravity} are interpreted as not just a model for reality but as reality itself. That reality \emph{is} the model. However, this equivalence is not true! Physics can only make models for physical reality; we never have direct access to reality itself.

There are two ways in which Eq.~\ref{eq:gravity} can fail. The first is that this model may be fundamentally wrong because there are ranges of parameters or a level of accuracy where the model no longer describes reality. For example, there could be a term we can add that just has a small influence on the force, such as
\begin{equation}
F=\frac{Gm_1m_2}{r^2} + \varepsilon \frac{Gm_1m_2}{r^3},
\end{equation}
for a small constant $\varepsilon$. Maybe we have not performed an experiment that is accurate enough to find this small deviation. Maybe one of the assumptions that is made about Newton's law of gravity, such as the uniformity of three-dimensional space, is wrong. Only by doing more experiments, trying to increase the ranges of the parameters, can we see in which situations our models are applicable. Indeed, we now know that Newton's law of gravity is actually a special case of general relativity. Many situations deviate from Eq.~\ref{eq:gravity}, such as the orbit of planet Mercury \cite{verrier59}.

The second way the model can break down is if the experimental conditions are not ideal. For example, in practice there are no point masses, so does Eq.~\ref{eq:gravity} still apply to reality? For many practical purposes, indeed it is applicable. If the masses are very far apart, then they can be treated approximately like point particles. By using approximations, simple mathematical models like Eq.~\ref{eq:gravity} can be very successful. They describe the way the world works with surprising accuracy and applicability in a variety of situations.

This thesis concerns itself with this second way that reality deviates from the models used to describe it: when the approximations and assumptions we make in order to apply a model to a physical situation are no longer true. Enter cryptography.

Cryptography is the field of study of tasks in the presence of an adversary. One general task in this field is to enable separated people to communicate without giving away any information to an eavesdropper who tries to figure out what they are communicating.

In contrast to physics, cryptography and its parent field, computer science, start with an idealized model, which is implemented using physical devices. This strategy makes the construction of protocols easier to work with, as they are precisely defined. To show that a cryptographic protocol is secure against an adversary can be (relatively) straightforward because a precise model is used that avoids the two types of deviations mentioned above. However, there may be imperfections with the physical devices used to implement the cryptographic protocol. The security may be compromised by imperfections, since these imperfections may leak information to an eavesdropper or decrease the efficiency so that the protocol no longer accomplishes the goal it was designed for. For example, the amount of power a computer uses may tell an adversary what calculation it is running. As another example, two people may want to communicate securely over the internet but imperfections may lead to a leak of their secure messages to an eavesdropper. This potential information leakage means that better cryptographic models are necessary in order to guarantee security in real implementations. It is not enough to prove that a protocol is secure in an idealized setting.

Information is inherently physical, since implementing cryptographic protocols requires the use of physical devices. This means that the disconnect between the models of cryptography and cryptographic implementations is actually the same problem as with the models of physics and physical reality. This relationship is especially apparent in quantum cryptography where quantum physics is used to perform cryptographic tasks. Usually these protocols are described in an idealized setting and then security is proved in these settings. While this idealization is useful, especially when showing that a certain protocol can be secure in principle, it does not say very much about whether any actual implementation is secure or not.

There is an additional challenge: how do we prove that an implementation of a protocol is secure? There have been several efforts to close the gap between the idealized models and their actual implementations. However, much work remains to ensure that the models are robust and realistic enough to be applicable with minimal assumptions. This thesis aims to clarify this connection.

A model can always be applied to an implementation if enough assumptions are made. Therefore, security of a cryptographic protocol is proven \emph{under} a set of assumptions. To apply this security proof to an implementation, the assumptions need to be justified (i.e.~devices need to behave as modelled). If they are not justified, then an adversary may break the security by exploiting this imperfection. It is therefore extremely important that the assumptions made are clearly presented and understood, so that cryptography can be implemented in a way that is as secure as possible. There are two kinds of assumptions: those that are fundamental (such as that quantum mechanics is correct) and those that are practical (such as the characterization of a device). These latter assumptions are the ones that adversaries can exploit and therefore need to be justified.

This thesis will focus primarily on quantum key distribution, but many of the implementations of protocols in quantum cryptography use the same physical devices and have similar assumptions. In the assumptions chapter (Chapter~\ref{chap:assumptions}) many of the issues discussed will be applicable to quantum cryptography in general.

The reader should have a basic understanding of quantum mechanics, including operators, the quantum harmonic oscillator, Dirac notation, and Hilbert spaces. In addition, the reader should have some mathematical knowledge of linear algebra and statistics.

We take an abstract approach to the field of quantum cryptography and in particular, quantum key distribution. This approach will give us the advantage of starting with simple quantum systems. Various protocols in quantum cryptography can then be defined without having to deal with the physical devices used in their implementations. This abstraction sets the foundation for the two goals of this thesis: how security can be proven for quantum key distribution and how these security proofs correspond to implementations. We will not present a complete security proof for a protocol, but instead describe several tools and outline how they are used to prove security. This framework has the advantage that we can separate the techniques and challenges of proving security for idealized models from the techniques for connecting these idealized proofs with implementations. Then we will explain how these protocols may be implemented such that a secure model applies to experiments. As we will see, there are many challenges to overcome to bridge the gap between the perfect models and the physical devices in quantum cryptography. 

In this introduction we will start with an overview of quantum cryptography and some of the protocols that are illustrative of what kinds of tasks are possible in this field. Then, quantum key distribution will be introduced. We start out with describing simple models for several protocols and introduce various abstract resources that are needed to perform quantum key distribution.

\section{Quantum Cryptography}

Quantum cryptography uses quantum states and quantum maps to perform communication or computational tasks in a secure way. There are several tasks and protocols that have been studied, each with a specific goal they try to accomplish. Many of these protocols share similar resources, so before describing particular protocols in quantum cryptography, we list a few resources which are often used.

The protocols used in quantum cryptography usually involve two parties called Alice and Bob. They are named in order to simplify discussions of the protocol. Also, there may be a malicious third party, Eve, who tries to stop Alice and Bob's cryptography protocol or try to learn information that is supposed to be hidden from her.

\subsection{Resources}

One of the basic resources for communication and cryptography are \emph{channels}. Channels allow communication between two or more parties and are usually specified by which kind of messages they allow to be transmitted. For example, a channel may only transmit classical messages or it may allow for quantum states. Also, the channel may be \emph{authenticated}, which means that if one party, Alice, sends a message to another party, Bob, then Bob knows that the messages he receives from this channel must have come from Alice and not from an eavesdropper, Eve. Eve will have access to the communication in an authenticated channel, but she will not be able to change it.

Channels have three eavesdropping models. \emph{Secure} channels only allow communication between the communicating parties and no eavesdropper can get any access to the communication. However, the eavesdropper may learn the length of the communication sent through the secure channel. \emph{Public} channels announce their messages to any eavesdroppers in addition to the communicating parties but the eavesdropper cannot interfere with the communication. Finally, a channel may be \emph{insecure}, which means that Eve can interfere with the signal sent through the channel as much as she likes. For example, for a quantum channel, Eve could apply any quantum map to the signals jointly with an ancillary system of her own.

A classically authenticated public channel between two parties can be constructed from an insecure channel and a shared secret \emph{key}.\footnote{A key is a string (e.g.~a list of numbers) in cryptography that is supposed to be unknown to an adversary.} The key does not need to be uniformly random, but may instead have a lower bound on its entropy \cite{renner03}.

Another resource is a \emph{source}. Sources are either classical or quantum, and produce either random variables with particular distributions (in the classical case) or quantum states. There are also \emph{measurements}. These take quantum states as input and have a classical output.

Lastly, there is \emph{randomness}. This is a string of bits\footnote{A string is a list of characters (but for our purposes, these characters will just be numbers), and bits are the binary numbers that are either $0$ or $1$. \label{foot:string_bits}} that is (preferably) uniformly random at a fixed length. For some applications it can be sufficient to have a non-uniform random string but there may be a guarantee of having a certain amount of randomness, such as a lower bound on the min-entropy with respect to an eavesdropper (see Defn.~\ref{defn:Hmin}).

\subsection{Quantum-Cryptography Protocols} \label{sec:qc_protocols}

Now that we have outlined typical resources, we describe some examples of protocols in quantum cryptography to give a brief overview of the field. 

Sometimes in the literature the term quantum cryptography is used synonymously with quantum key distribution though this is not correct. There are a wide variety of quantum-cryptography tasks.

Many of the protocols below have analogous protocols in a classical setting but using quantum states or quantum computers often have an advantage over what is possible classically.

\begin{itemize}
\item {\bf Secure quantum distributed computing}

Secure distributed computing can be related to many tasks where one party, Alice, wants an untrusted party, Bob, to implement a computation for her. One such protocol is quantum homomorphic encryption: Alice, who usually only has a simple quantum device, wants to get the result of a computation \cite{rohde12}. She then asks Bob (who has a quantum computer) to do this computation for her. However, Alice does not want Bob to find out what her data is. To accomplish this secure computation, Alice encodes her data and sends it to Bob, Bob applies the computation on the encoded data and sends the output to Alice who then decodes the output. Ideally, Bob's computation does not reveal any information about Alice's data to Bob and Alice's decoded output should correspond to the computation applied to her original unencrypted input. For homomorphic encryption Bob knows what computation is being performed.

Quantum homomorphic encryption has been shown to be possible with perfect security \cite{liang13}, and it is possible using boson sampling \cite{rohde12}.

Another distributed computing protocol is blind computation. It is the same protocol as homomorphic encryption, except it should be even more secure: Bob should not know what the computation is either. In this case, Alice sends Bob an encrypted description of the computation she wants to be performed on her encrypted data \cite{childs05}. Bob can input this encrypted description into his quantum computer to tell it what computation to perform. At any stage in the computation Bob should not be able to figure out what Alice's data is or what computation is being applied.

Blind quantum computation is both possible for arbitrary quantum computations \cite{broadbent09} and efficient in the amount of communication needed and the simplicity of the quantum device Alice needs to interact with Bob \cite{giovannetti13,mantri13}. Also, the situation where Alice can do quantum measurements has been considered \cite{morimae13}.

In general, distributed computing is secure, even under composition with other protocols (see Section~\ref{sec:composability}) \cite{dunjko13}. Classical distributed computing can also be enhanced by using quantum devices \cite{dunjko14b}. However, if Alice only has classical devices then she cannot perform quantum distributed computing securely \cite{morimae14}.

\item {\bf Quantum coin flipping}

Quantum coin flipping is designed to have two mutually distrustful parties, Alice and Bob, jointly flip a coin. Even if one of them tries to influence the coin flip, the flip should still be uniformly random \cite{bb84}. While the coin flip cannot be performed perfectly \cite{lo98,mayers97}, it can be performed with a bound of $1/\sqrt{2}$ on the probability a dishonest party gets the outcome they want \cite{chailloux09a}. While this scenario is called strong coin flipping, the task of weak coin flipping is where Alice wants to bias the coin to one result and Bob wants to bias it to the other. In this case the probability a dishonest party gets the outcome they want is $1/2$, which is the optimal achievable bound on the bias \cite{mochon07a,aharonov14}.

It is important to note that quantum coin-flipping protocols can always outperform classical ones \cite{aharonov00}. Coin flipping has been implemented \cite{berlin11} and has applications in other areas of quantum information \cite{damgaard09}.

\item {\bf Quantum zero-knowledge proofs}

Zero-knowledge proofs involve two parties, a prover and a verifier, where the prover tries to convince the verifier that a certain statement is true without revealing any information about the proof, only that the statement is indeed true. This task is usually done in a probabilistic way so that the verifier will be certain with high probability that the statement is true \cite{watrous02}.

As a classical example of a zero-knowledge proof, consider a colour-blind Bob who has two spheres that are identical, except one is red and the other is green. Bob cannot tell them apart, but Alice, who is not colour blind, can still prove to Bob that they are different. Bob takes one sphere in each hand, which Alice can see, and then secretly either leaves them that way or switches which hand holds which sphere. Alice can then tell whether Bob made a switch or not. If Alice can tell them apart, then after many repetitions of the game, Bob will be convinced the spheres are different. If Alice cannot tell them apart then Alice will not be able to guess what Bob did and will probably make a mistake in guessing whether Bob did a switch or not.

Another related task is zero-knowledge proof of knowledge, where the prover not only tries to prove that something is true but that they have access to the proof \cite{unruh12}. For example, not only that a signature from a trusted authority exists but that the prover has such a signature.

Zero-knowledge proofs can be used in cryptography to ensure that honest parties are indeed honest, without needing to reveal any other information and ensures that the quantum-cryptography protocol does not leak any additional information to an eavesdropper or a dishonest party. Zero-knowledge proofs also have applications to the hardness of determining whether the output state of a quantum circuit is entangled or separable (see Section~\ref{sec:operators}) \cite{hayden13}.

Some classical and quantum zero-knowledge proofs can be secure against verifiers who either try to get some information about the proof or provers who try to lie about knowing that the statement is true \cite{watrous06}. Some classical zero-knowledge proofs are not secure in the quantum setting \cite{ambainis14}. The connection between classical and quantum zero-knowledge proofs has been analyzed \cite{chailloux08}. Also, quantum zero-knowledge proofs can be constructed based on quantum bit commitment (see below)\cite{donascimento09}.

\item {\bf Random number generation}

Random numbers are useful for a variety of tasks, such as online gambling, computation, and cryptographic protocols. In classical computation, pseudorandom numbers are often used and are sufficient for many applications. Pseudorandom numbers are generated through a deterministic process but may appear under some statistical tests to be random. However, for cryptography, it can be completely insecure to use pseudorandom numbers in the place of truly generated random numbers.

Random numbers can only be produced from physical processes that are stochastic. Examples include atmospheric noise, thermal noise, or quantum processes. Quantum devices can produce randomness with relatively simple devices and rely on the randomness inherent in quantum mechanics, since measurement outcomes sample a probability distribution. It has recently been shown that quantum random numbers can even be extracted by using the camera in a mobile phone \cite{sanguinetti14}.

Random number generation is a cryptographic task because randomness is defined as having some information (such as a string of bits) that is independent of any adversary who tries to get information about the randomness during its generation.

There are related tasks, such as trying to amplify randomness: by starting from a small string of randomness (called a \emph{seed}) a larger string of randomness can be constructed \cite{colbeck12}. Some recent results show that any information that is not completely deterministic can be made completely random, even in the presence of noise \cite{gallego13,brandao13}. Also, randomness can be extracted from devices without making assumptions about the structure of the devices used \cite{colbeck06,pironio10a}.

\item {\bf Quantum oblivious transfer}

Oblivious transfer involves one party, Alice, who has a list of possible messages, and another party, Bob, who wants to learn one of the messages \cite{wiesner83}. However, Alice should not learn which message Bob asked for and Bob should not be able to get any information about any of Alice's other messages except for the message he requests. This protocol would ideally work even if Alice or Bob tries to behave adversarially. In general, oblivious transfer protocols are denoted as ``$k$-out-of-$n$'', meaning that Bob requests $k$ messages from the total, $n$. The simplest oblivious transfer protocol is then $1$-out-of-$2$.

As with coin flipping, this protocol cannot be implemented perfectly: Alice may learn which messages Bob requested and Bob may be able to get access to some of Alice's messages he did not request \cite{mayers97,lo98}. The minimum probability that Alice or Bob can cheat in this protocol and not be detected is $2/3$ \cite{chailloux13}. If additional assumptions are made, such as that the adversary has a limit on her computational power (in the classical case) or can only store a certain size of quantum system (in the quantum case) then oblivious transfer is possible \cite{damgard06}.

\item {\bf Quantum bit commitment}

Quantum bit commitment is closely related to coin flipping, zero-knowledge proofs, and oblivious transfer. It is the task of having one party, Alice, commit to a value that is hidden until a later point when she will reveal the value. To implement this protocol, Alice sends a quantum state to Bob that will contain an encrypted version of her committed value. At some later point she will reveal her value by telling Bob how to decode the encrypted value from the state he received. The protocol is secure against a cheating Alice if Alice cannot change the value after she has committed to it. The protocol is secure against a cheating Bob if Bob cannot learn the value before Alice chooses to reveal it. However, quantum bit commitment is not completely secure against a cheating Alice or a cheating Bob unless additional assumptions are made \cite{lo98,mayers97,brassard97}.

The optimal bound on the probability that Alice changes her commitment without being detected in this setting is $0.739$ \cite{chailloux11}. If special relativity is used with quantum mechanics, then bit commitment can be made secure against a cheating Alice \cite{kent11,kent12,kent12a,croke12}. Also, if additional assumptions are made about the capabilities of Alice and Bob, quantum bit commitment is possible and can be implemented with current technology \cite{loura14}.

\item {\bf Quantum key distribution}

The goal of quantum key distribution (QKD) is to distribute a secret random string of classical bits between two (or more) trusted parties. That is, they want to have a string of bits (see Footnote~\ref{foot:string_bits}) that are identical and unknown to an eavesdropper that has tried to figure out what the string is by listening to or modifying the communication between the parties. In order for the string to be secret, it should be random, which means that there is an equal probability of getting a $0$ or a $1$ at every position in the string, independent of any other bit in the string as well as any other information. The string in this context is referred to as a \emph{key}.

Secret random classical strings are useful for a variety of tasks in cryptography and computer science. One straightforward use is as a key for the one-time pad encryption (also called the Vernam cipher) \cite{miller1882,vernam19}. It is a protocol that allows for two parties to communicate privately (i.e.~to construct a private channel) by using a secret random classical string that they share and an authenticated public classical channel. Alice encodes her message by adding it bit-wise $\text{mod } 2$ to her key (see Fig.~\ref{fig:OTP}), which results in a string called a \emph{ciphertext}. Alice sends the ciphertext through an authenticated classical channel to Bob. Then, since Bob has the same key, if he adds the key to his received message, he gets Alice's original message.

\begin{figure} \centering
\includegraphics[width=\textwidth]{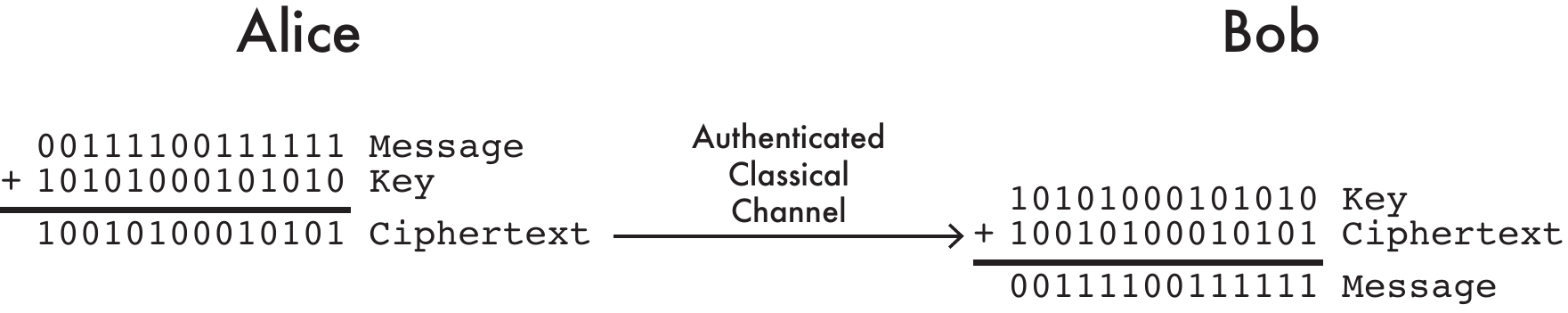}
\caption[The One-Time Pad]{The one-time pad protocol. Alice and Bob share a secret key that is at least the length of a message. Alice adds her key to her message (modulo $2$), which encodes the message as a \emph{ciphertext}. Alice sends the ciphertext to Bob through an authenticated classical channel. Bob can determine Alice's message by adding the key to the ciphertext (modulo $2$).}
\label{fig:OTP}
\end{figure}

In order to distribute a secret random string several resources are needed. Alice and Bob will use an insecure quantum channel to send quantum states to each other. Alice and Bob will also need to communicate classically, which they do through an authenticated classical channel. They will also need some randomness which may be used to choose measurement settings or for classical post-processing.

Since an authenticated channel is used for QKD, QKD has also been referred to as a quantum key growing or a quantum key extending protocol, since it often uses an authenticated channel constructed from a shared secret key that Alice and Bob share. QKD then extends, or grows, this key. Other authentication schemes can be used which do not require an initial shared secret key, such as in \cite{mosca13}.

\end{itemize}

While this thesis focuses on QKD, many implementations of quantum-cryptography protocols use similar devices, and therefore many of the issues discussed throughout this thesis will also apply to other quantum-cryptography protocols.

\section{Quantum Key Distribution}

The task of distributing random secret keys can be accomplished without performing quantum key distribution (QKD). Instead of going through the trouble of using quantum mechanics, keys could be distributed by using a source of randomness and then copying this randomness onto two hard drives, where Alice keeps one and gives the other to Bob. Also, why do we need quantum key distribution in combination with the one-time pad for secure communication if we can use current classical cryptography used for online security today?

There are several advantages that QKD provides over other alternatives \cite{scarani09,scarani09b}. Current classical cryptography is usually based on the assumption that a particular mathematical problem is hard, such as factoring large integers \cite{diffie76}. Using this kind of cryptography carries the risk that it may be broken if a classical algorithm is invented that is faster at factoring large numbers than what is currently known. Also, if and when quantum computers are built, they can factor large integers efficiently using an algorithm by Shor \cite{shor99}. Furthermore, even if a problem is hard to solve, it can still be solved! So if an eavesdropper has enough time and computing power they can always decode the secure communication. For information that must be secret for a long period of time, classical cryptography may not be sufficient.

In contrast to classical cryptography, QKD does not rely on the computational difficulty of a mathematical problem but instead it relies on \emph{information-theoretic} security, which means that the probability that an eavesdropper gets any information about the key can be made incredibly small, no matter what computational power an eavesdropper has at their disposal. There are other notions of security but these rely, for example, on computational hardness assumptions.

QKD also has an advantage over the distributed hard drive scenario above, since it can make arbitrary long keys from an initial seed. New hard drives would have to be distributed to extend the key in the other scenario.

QKD also has some disadvantages over classical cryptography. Due to losses and errors QKD cannot be done over distances longer than $\sim 200$-$300$km with current technology \cite{scarani09,xu13,korzh14}. However, it may be possible in the future to use satellites to extend this distance \cite{meyer-scott11,wang12,qi14,radchenko14,vallone14}. Also, the speed at which a secure key can be generated is typically much slower than what is possible with current classical cryptography. Lastly, some assumptions that are needed for a quantum-cryptography protocol to be secure are impractical or not yet possible with current technology. These are some of the challenges that QKD faces in order to become more widely used.

In this section, we discuss the structure that QKD protocols follow. Then we list several protocols and how they would be implemented in an ideal way. We categorize QKD protocols by whether they have discrete or continuous measurement outcomes, as the devices used in these two kinds of protocols are different. We also discuss device-independent protocols that do not make assumptions about the structure of the devices or the states used in the protocol.

\subsection{General QKD Structure}

Almost all QKD protocols follow the same general structure. We will focus on bipartite QKD, where there are two parties, Alice and Bob, who are trying to construct a shared secret random string. However, there are also schemes for multi-party QKD \cite{cabello00,lee04}.

First, there is a quantum \stage followed by a classical \stage. In the quantum \stage, Alice and Bob send quantum states to each other, or perhaps only Alice sends states to Bob, through an insecure quantum channel. These quantum states are associated with classical bits that Alice and Bob are trying to communicate to each other. The classical stage, usually called classical post-processing, is performed on their measurement outcomes to correct any errors due to noise in the quantum channel or in their devices. Also, an eavesdropper could have interfered with the signals, and they need to ensure that any knowledge an eavesdropper has gained is removed.

In classical post-processing there are usually at least three steps: parameter estimation, information reconciliation (also sometimes called error correction), and privacy amplification. Alice and Bob will need to communicate classically for the classical post-processing and they need to know that an eavesdropper does not interfere with this communication, so they use an authenticated classical channel.

Typically, there is an asymmetry in the quantum \stage of the protocol between Alice and Bob. For example, Alice may prepare quantum states that Bob measures. This implies that the classical data that Alice and Bob hold after the quantum \stage come from different sources. Alice may have prepared a uniformly random string to pick which quantum states she prepares, while Bob gets his classical data from the output of a quantum measurement. This creates an asymmetry in the classical post-processing, which can be performed in one of three ways. The first way, called \emph{direct reconciliation}, is if Alice only sends classical information about her string to Bob and Bob does not tell Alice anything about his string. If the roles of Alice and Bob are reversed, so that Bob only sends classical information about his string to Alice, then this is called \emph{reverse reconciliation}. Direct and reverse reconciliation are \emph{one-way} classical post-processing. Despite that the communication is one-way, the other party, such as Bob in direct reconciliation, may need to communicate some auxiliary information to Alice, such as whether they should abort or continue the protocol (see below for more information on aborting).

They can also implement the post-processing by using \emph{two-way} communication, where Alice and Bob send information to each other about their strings. Typically one-way communication is considered since it is usually easier to analyze and sufficient to perform the post-processing. Throughout this thesis we will assume that direct reconciliation is being performed.

The first step of classical post-processing is parameter estimation, where Alice and Bob can get some statistical knowledge about their strings in order to figure out how many errors they have and also how much information an eavesdropper may have on their strings. Then they use the information they learned from parameter estimation to perform an information-reconciliation step to correct any errors between their two strings. After this subprotocol they should have the same string (at least with very high probability). They finish with privacy amplification, in order to remove any information that an eavesdropper may have about their strings (at least with very high probability). In order to perform the classical post-processing, Alice and Bob need a source of randomness (see Section~\ref{sec:post_processing}).

Note that after parameter estimation they may see that their strings have a large fraction of errors between them. In this case they have to abort the protocol since an eavesdropper could have gained so much knowledge about their strings that no amount of privacy amplification would make their key secure. We call the number beyond which the error rate (or other statistical quantity) cannot exceed the \emph{threshold} of the protocol. To find this threshold, the parameters of the protocol need to be analyzed (see Section~\ref{sec:tuning}).

\subsubsection{Parameter Estimation}

Parameter estimation in QKD is the task of using statistics on a small sample of Alice's and Bob's strings to estimate a global property of those strings. For example, the number of errors between Alice's and Bob's strings can be estimated from a small sample by using Chernoff-Hoeffding type bounds \cite{chernoff52,hoeffding63,serfling74} (see Lemma~\ref{lemma:serfling}). These bounds are statistical inequalities that state that if a random subset of data is known, then a statistical property of the sample must be close to the statistical property of all of data. In the example of estimating the number of errors, Alice communicates to Bob a fraction of her string and Bob finds that they have an error fraction (also called an error ratio or error rate), say, of $5\%$. Then they know that the total error rate of their strings is (with high probability) close to $5\%$. The closeness is exponentially close in the size of the sample (see Lemma~\ref{lemma:serfling}).

Parameter estimation can be accomplished if Alice sends Bob a small sample of her string through the authenticated classical channel. Bob can then tell Alice what error rate he sees so that Alice also knows the error rate. If they see that their error rate is beyond the threshold allowed, they abort the protocol. Otherwise, they continue.

See Section~\ref{sec:parameter_estimation2} for the details of parameter estimation. After Alice and Bob have done the estimation, they are ready to correct the errors between their strings.

\subsubsection{Information Reconciliation}

In information reconciliation, Alice and Bob try to correct the errors between their strings which may have been caused by an eavesdropper or noise in the channel and devices they used. They want to communicate a minimal amount of relevant information to each other over the classically authenticated channel so that they can correct any errors. From parameter estimation they have an estimate on the number of errors between their strings, so they just need to figure out where their errors are \cite{wegman81}.

The information-reconciliation procedure may be probabilistic so that with high probability it succeeds in correcting all the errors and with a small probability it does not. Alice and Bob may have to check if error correction has succeeded or not. Therefore, they can communicate a small amount of information to ensure they have the same string after their error correction.

See Section~\ref{sec:error_correction2} for the details of how information reconciliation can be implemented.

\subsubsection{Privacy Amplification}

After information reconciliation, Alice and Bob have the same strings. Now they need to remove any information an eavesdropper may have learned about their shared string. Privacy amplification achieves this task at a cost of reducing the size of Alice and Bob's string \cite{bennett95a}. The shorter they make their shared string, the more secure their shared string will be.

Note that the eavesdropper gets information about Alice and Bob's string in one of two ways. One is through manipulating the quantum states during the quantum \stage of the protocol. The other is by using the information that is sent through the authenticated classical channel, which includes the communication used for parameter estimation, the communication used to correct the errors during information reconciliation, and the communication to make sure the error correction procedure has succeeded.

For the details of privacy amplification, see Section~\ref{sec:privacy_amplification2}.

We now list common QKD protocols in two categories that classify what kind of states are used (see Section~\ref{sec:classes} for the full classification of QKD protocols). First, there are discrete protocols that have measurements with discrete outcomes, and second, there are continuous-variable protocols that have measurements with continuous outcomes. We present the protocols here in their idealized form for clarity of exposition and leave the details of their implementations for later (Chapter~\ref{chap:assumptions}). We will also discuss the current status of the security of these protocols. In the assumptions chapter (Chapter \ref{chap:assumptions}), we will discuss how these protocols are actually implemented and how these implementations differ from their idealized form.

\subsection{Discrete Protocols} \label{sec:discrete_protocols}

Discrete protocols have at least one quantum measurement whose outcomes come from a (usually small) discrete set. Typically, they are modelled in an ideal setting by the encoding of classical bits in finite-dimensional quantum states.

First, we list several protocols that use qubits (i.e.~two-level quantum systems) as the quantum states that are sent through the quantum channel and then we will list some protocols that are still discrete but do not use qubits for their quantum states.

\subsubsection{BB84}

BB84 was the first QKD protocol, developed in 1984 by Bennett and Brassard (hence the name) \cite{bb84}. It is probably the most analyzed QKD protocol, not only due to it being the first, but also due to its simplicity and symmetry. The BB84 protocol has several security proofs that apply under various assumptions, for example \cite{lo99,mayers01,mayers96,biham97,shor00,koashi03,gottesman04,rennerphd,kraus05,renner05a,tomamichel12a,ferenczi12}.

The protocol is defined as follows. First, Alice prepares one of four qubit states
\begin{equation} \label{eq:bb84states}
\ket{0},\ket{1},\ket{+}:=\frac{\ket{0}+\ket{1}}{\sqrt{2}},\ket{-}:=\frac{\ket{0}-\ket{1}}{\sqrt{2}},
\end{equation}
and she sends them through an insecure quantum channel to Bob (see Fig.~\ref{fig:BB84}). Bob randomly chooses one of two bases ($\{\ket{0},\ket{1}\}$ or $\{\ket{+},\ket{-}\}$) uniformly at random to measure each signal he receives (see Defn.~\ref{defn:measurement}). These bases are often referred to as the $Z$ and $X$ basis respectively, since they are the set of eigenvectors of the Pauli matrices
\begin{equation}
\sigma_Z = \begin{pmatrix} 1 & 0 \\ 0 & -1\end{pmatrix}, \sigma_X = \begin{pmatrix} 0 & 1 \\ 1 & 0\end{pmatrix}.
\end{equation}
Whenever Alice or Bob send/measure the states $\ket{0}$ or $\ket{+}$ they store a $0$ in their classical computer and whenever they send/measure $\ket{1}$ or $\ket{-}$ they store a $1$. They now both have a string of bits.

\begin{figure} \centering
\includegraphics[width=\textwidth]{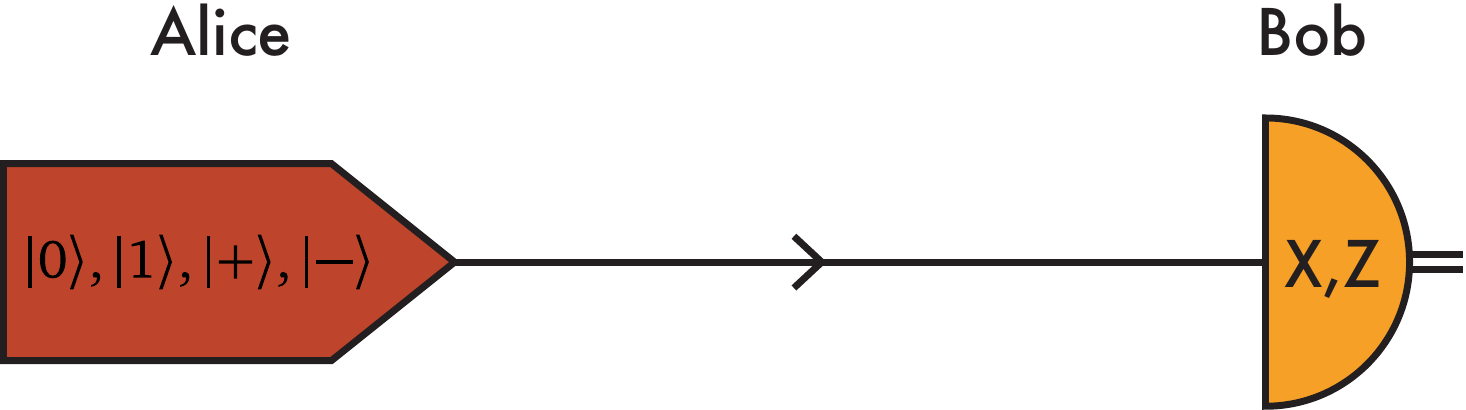}
\caption[The BB84 Protocol]{The BB84 protocol. Alice prepares one of the four states $\{\ket{0},\ket{1},\ket{+},\ket{-}\}$ with equal probability and Bob measures in the $X$ basis ($\{\ket{+},\ket{-}\}$) or $Z$ basis ($\{\ket{0},\ket{1}\}$) with equal probability.}
\label{fig:BB84}
\end{figure}

Alice then classically communicates which basis her states were in and Bob tells Alice which bases he measured in. Alice and Bob throw away the bits where Bob's measurement basis and Alice's signal do not match. This step of checking their bases and throwing away these bits is called \emph{basis sifting}. They continue on to the classical post-processing steps after basis sifting.

\subsubsection{Ekert91}

The Ekert91 protocol \cite{ekert91} is similar to the BB84 protocol and in an ideal setting is actually the same as the BB84 protocol \cite{bennett92b}. Here we present a slightly different version of what Ekert originally presented, in order to connect it with the BB84 protocol. Eve, or another untrusted source, prepares entangled bipartite qubit states (see Defn.~\ref{defn:separable}). Ideally this state has the form
\begin{equation}
\ket{\psi^+} = \frac{\ket{00}+\ket{11}}{\sqrt{2}},
\end{equation}
which is from the Bell basis (see Eq.~\ref{eq:bell_2}) \cite{nielsen00}. Alice gets one of the qubits, and Bob gets the other (see Fig.~\ref{fig:Ekert}). Uniformly at random they each choose a basis to measure in and do the same measurement as in the BB84 protocol ($\{\ket{0},\ket{1}\}$ or $\{\ket{+},\ket{-}\}$).

\begin{figure} \centering
\includegraphics[width=\textwidth]{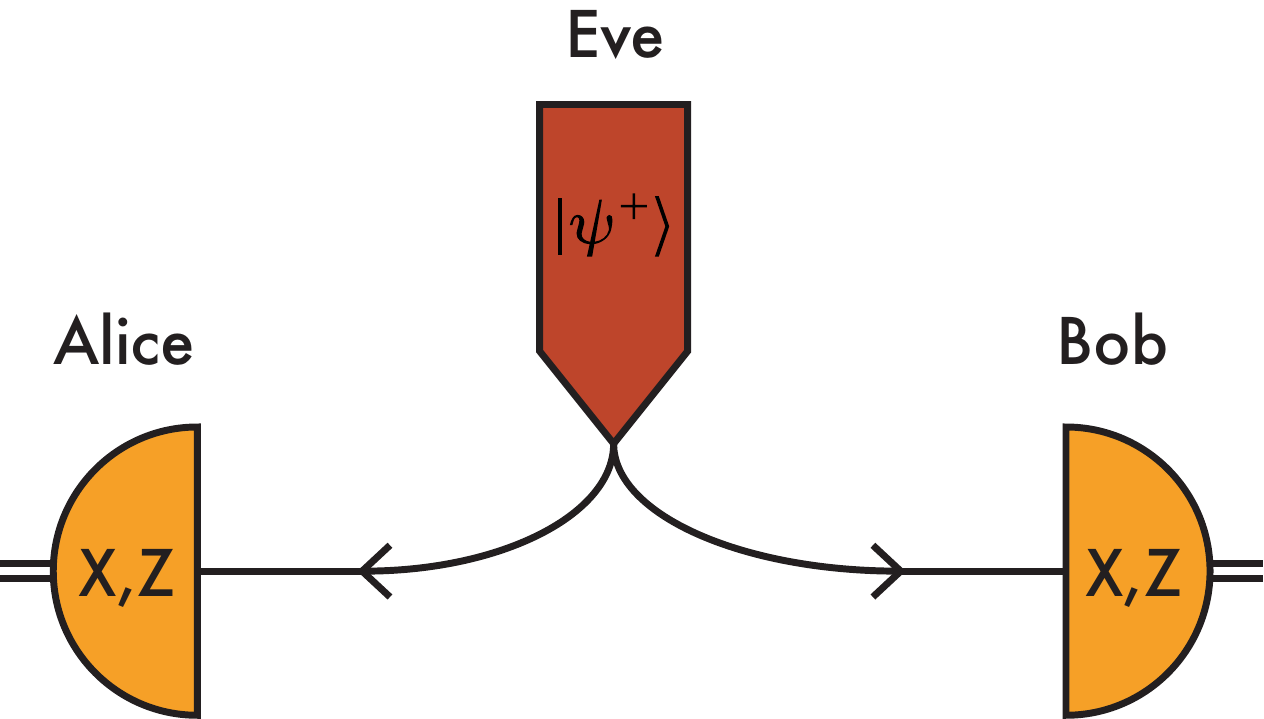}
\caption[The Ekert91 Protocol]{The Ekert91 protocol. Eve prepares a bipartite state that is ideally a maximally entangled two-qubit state. Alice and Bob uniformly at random measure in the $X$ or $Z$ basis.}
\label{fig:Ekert}
\end{figure}

Alice and Bob then do basis sifting, as in the BB84 protocol, followed by classical post-processing.

To see how the BB84 and Ekert protocols are equivalent, notice that the production of the entangled state $\ket{\psi^+}$ and a measurement on one of the qubits in one of the two bases $\{\ket{0},\ket{1}\}$ or $\{\ket{+},\ket{-}\}$ leaves the other qubit in one of the states from Eq.~\ref{eq:bb84states}. In the BB84 protocol Alice chooses one of four states to send, which she could choose by doing a measurement on a four-dimensional ancillary system consisting of the states $\ket{0},\ket{1},\ket{2},\ket{3}$. We write Alice's state in the BB84 protocol as:
\begin{align}
&\frac{1}{2}(\ket{0}\ket{0}+\ket{1}\ket{1}+\ket{2}\ket{+}+\ket{3}\ket{-}) \\
&= \frac{1}{2}\left(\ket{0}\ket{0}+\ket{1}\ket{1}+\ket{2}\left(\frac{\ket{0}+\ket{1}}{\sqrt{2}}\right)+\ket{3}\left(\frac{\ket{0}-\ket{1}}{\sqrt{2}}\right)\right) \\
&= \frac{1}{2}\left(\left(\ket{0}+\frac{\ket{2}+\ket{3}}{\sqrt{2}}\right)\ket{0}+\left(\ket{1}+\frac{\ket{2}-\ket{3}}{\sqrt{2}}\right)\ket{1}\right) \\
&= \frac{1}{\sqrt{2}}\left(\ket{\tilde{0}}\ket{0}+\ket{\tilde{1}}\ket{1}\right),
\end{align}
where $\ket{\tilde{0}}$ and $\ket{\tilde{1}}$ are orthonormal states that are linear combinations of the basis vectors in Alice's four-dimensional space. Therefore, if Alice prepares the entangled state and does a measurement on one half of it then it is the same as having a source that just prepares one of the states from Eq.~\ref{eq:bb84states}. Therefore, the BB84 and Ekert91 protocols are equivalent if Alice or Bob prepare the entangled state $\ket{\psi^+}$ and measure qubits.

If Eve is preparing the bipartite states in the Ekert91 protocol, then she will have more power than in the BB84 protocol, since in the BB84 protocol she can only modify the state sent from Alice to Bob.

In an experiment it is more difficult to connect the BB84 protocol to the Ekert91 protocol (see Section~\ref{sec:connection}). Also, the original protocol by Ekert was intended to be device-independent (see Section~\ref{sec:DIprotocols} for a description of device-independent protocols).

\subsubsection{BB84 Variants}

There are several variants of the BB84 protocol. Two notable examples are the six-state protocol \cite{bruss98} and SARG \cite{scarani04}.

The six-state protocol is an extension of the BB84 protocol from four states $\{\ket{0},\ket{1},\ket{+},\ket{-}\}$ to six states by adding $\ket{i}:=(\ket{0}+i\ket{1})/\sqrt{2},\ket{-i}:=(\ket{0}-i\ket{1})/\sqrt{2}$, called the Y basis, since it is the set of eigenvectors of the Pauli matrix
\begin{equation}
\sigma_Y=\begin{pmatrix}0 & -i \\ i & 0\end{pmatrix}.
\end{equation}
The six-state protocol is of interest because it was found to be more efficient than the BB84 protocol \cite{bruss98}. Also, the measurements are extended to include a third basis $\{\ket{i},\ket{-i}\}$. Bob then chooses one of the three bases uniformly at random. Alice and Bob do basis sifting afterwards, discarding any measurement/preparation pairs that are not in the same basis.

The SARG protocol was introduced as an alternative to the BB84 protocol to counteract an attack that Eve can apply to implementations of  BB84 \cite{brassard00,bennett92a,dusek99,huttner95}. It works the same as the BB84 protocol, except it reverses the role of the states and bases. If Alice sent a state in the $Z$ basis, she writes a $0$ and she sent a state in the $X$ basis, she writes a $1$. Bob's string is more complicated and will be explained below.

After the quantum \stage of the protocol, Alice communicates one of the following four sets that contains her sent state $\{\ket{0},\ket{+}\}, \{\ket{0},\ket{-}\},$ $\{\ket{1},\ket{+}\},$ $\{\ket{1},\ket{-}\}$. Since these sets have some states in common, Alice will uniformly at random choose a set that is compatible with the state she sent. Bob can then figure out which state Alice sent with probability $1/2$. For example, if Alice announces the set $\{\ket{0},\ket{+}\}$ and she sent the state $\ket{+}$, then if Bob measured in the $Z$ basis and gets outcome $\ket{1}$ he knows that Alice must have sent the state $\ket{+}$, and therefore writes down the bit $1$. Similarly, if Alice had sent the state $\ket{0}$ and announced the same set, and if Bob measured in the $X$ basis and got the outcome $\ket{-}$ he knows Alice must have sent $\ket{0}$ and he writes down a $0$.

Alice and Bob do basis sifting as in the BB84 protocol. If Bob gets a measurement outcome that is not in Alice's announced set or that is inconclusive (such as getting outcome $\ket{0}$ and the set announced is $\{\ket{0},\ket{+}\}$) then he tells Alice and they discard this measurement outcome.

Classical post-processing follows the six-state and SARG protocols after basis sifting.

\subsubsection{B92}

Another BB84 protocol variant is the B92 protocol \cite{bennett92a}. It differs from BB84 by only using two states: $\ket{0}$ and $\ket{+}$ (see Fig.~\ref{fig:B92}). Sometimes two non-orthogonal states are used other than $\ket{0}$ and $\ket{+}$, but here we use $\ket{0}$ and $\ket{+}$ for simplicity. Also, Bob only does a single measurement; he does not have a basis choice. This means that the basis sifting step is not necessary.

Bob's measurement is unambiguous state discrimination \cite{nielsen00}. For the states $\ket{0}$ and $\ket{+}$ Bob's measurement is described by the three positive operator valued measure (POVM) elements (see Defn.~\ref{defn:measurement})
\begin{equation}\label{eq:USD}
F_{0} = \frac{\sqrt{2}}{1+\sqrt{2}}\kb{-}{-}, \quad F_{1} = \frac{\sqrt{2}}{1+\sqrt{2}}\kb{1}{1}, \quad F_{?} = \identity - F_0 - F_1.
\end{equation}

With this measurement, Bob knows that when he gets outcome $0$ that he could not have had the state $\ket{+}$, since $\ket{+}$ and $\ket{-}$ are orthogonal ($\bk{+}{-}=0$). Similarly, when Bob gets outcome $1$, he could not have had the state $\ket{0}$. If he gets outcome `$?$' then he does not know which state he received. Bob will also keep track of the number of `$?$' measurement outcomes he gets. The `$?$' outcomes are important, since Eve could just do the same measurement as Bob before him, and always know what Bob's measurement outcomes would be. However, if Eve does the same measurement then Bob will see a higher number of `$?$' outcomes. Alice and Bob will abort the protocol if the number of `$?$' events is beyond a certain threshold.

Also, Bob reveals the positions in which he got outcome `$?$' so that Alice knows to throw that bit of her string away. Alice and Bob then continue to the classical post-processing steps.

\begin{figure} \centering
\includegraphics[width=\textwidth]{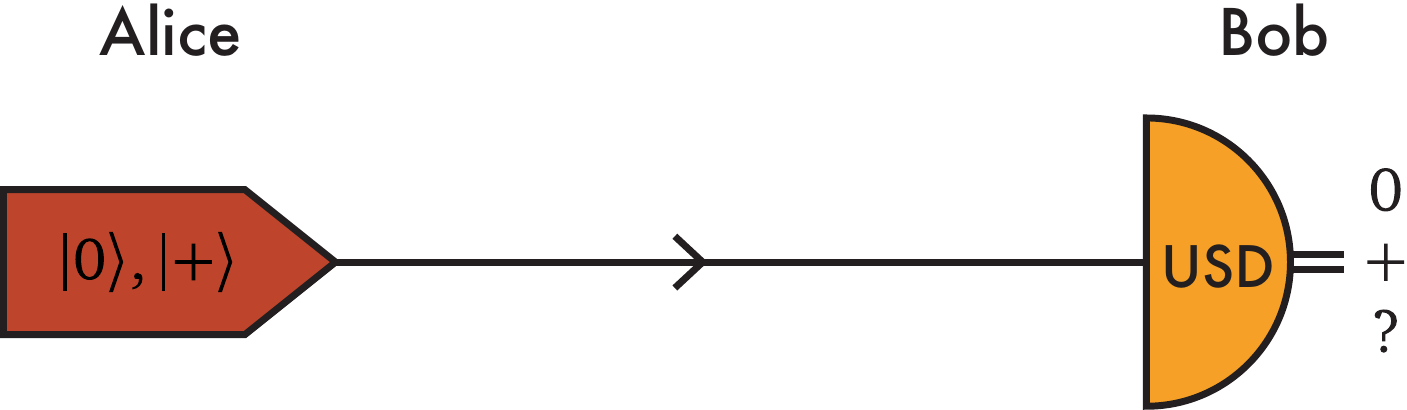}
\caption[The B92 Protocol]{The B92 protocol. Alice randomly prepares either $\ket{0}$ and $\ket{1}$ and Bob does unambiguous state discrimination between these states (Eq.~\ref{eq:USD}). Bob either gets outcome $0$ or $+$ to indicate which state he received or `$?$' when his measurement is inconclusive.}
\label{fig:B92}
\end{figure}

B92 has been proven secure for single photons \cite{tamaki03,tamaki04} as well as with more sophisticated models for the states used \cite{koashi04,tamaki09}.

When there is loss in the quantum channel, Eve can attack the B92 protocol by replacing the lossy channel with a lossless channel and by doing the same unambiguous state discrimination measurement as Bob. When she gets a definite outcome, she forwards the post-measurement state to Bob. If she gets the `$?$' outcome then she does not send Bob a state. If the loss is high enough in the channel, then Alice and Bob will not be able to tell this attack apart from loss, and Eve gets full information about the state that Alice sent whenever Bob gets a measurement outcome. To avoid this attack, some have proposed that Alice send a strong reference pulse with each quantum state \cite{koashi04,tamaki09}. The strong reference pulse is a laser pulse that has a huge number of photons and it can usually be considered to be a classical optical signal. Therefore, Bob is guaranteed to get a strong reference pulse, even if there is high loss in the channel between Alice and Bob, which makes it possible for them to detect when Eve is doing this attack.

Now we describe two discrete protocols that do not use qubits. These two protocols are \emph{distributed phase} protocols because they encode information in the relative phase between pulses of light.

\subsubsection{Differential Phase Shift (DPS)}

One of the problems with the above qubit protocols is that they often require a basis choice, which needs either active elements (i.e.~moving parts that require inputs) in the devices to choose the basis, or a device to do the basis choice in a passive way (without having to actively change the device, see Section~\ref{sec:measurements}). The DPS protocol was first proposed by \cite{inoue02,inoue03} as a protocol that can be implemented simply and in a passive way.

We present the simplified version of the protocol from \cite{inoue03} instead of how it was original proposed \cite{inoue02}. First, we introduce the notion of a coherent state, defined as
\begin{equation}\label{eq:coherent_state}
\ket{\alpha} := e^{-\frac{|\alpha|^2}{2}} \sum_{n=0}^{\infty} \frac{\alpha^n}{\sqrt{n!}}\ket{n}, \quad \alpha\in\mathds{C},
\end{equation}
where $\ket{n}$ is a Fock state\footnote{The Fock state, $\ket{n}$, is the energy eigenstate of the quantum harmonic oscillator with Hamiltonian $H=\hat{a}^{\dag} \hat{a}+\frac{1}{2}$, with creation and annihilation operators $\hat{a}^{\dag}$ and $\hat{a}$. This state represents the number of photons that are in a pulse from a laser. A coherent state is an eigenstate of the annihilation operator: $\hat{a}\ket{\alpha} = \alpha\ket{\alpha}$.}. Coherent states are a superposition of a Poisson distribution over the state for each number of photons. To see that this superposition follows a Poisson distribution, note that the probability of getting outcome $n$ when doing a projective measurement of the number of photons is
\begin{equation}\label{eq:Poisson}
\Pr\left[\text{n photons}\right] = \left|\bk{n}{\alpha}\right|^2 = e^{-|\alpha|^2} \frac{|\alpha|^{2n}}{n!},
\end{equation}
which means that the average number of photons is $|\alpha|^2$. Often in this context, instead of using the parameter $\alpha$, the average photon number $\mu:=|\alpha|^2$ is used instead (so a coherent state would be written as $\ket{\sqrt{\mu}}$).

In the protocol, Alice pulses her laser at fixed intervals to produce a train of pulses that each contain a coherent state (see Fig.~\ref{fig:DPS}). For each of the pulses she sends, she uses a secret random bit string, $S$, to choose if she will change the phase\footnote{By phase, we mean a factor $e^{i\varphi}$ in front a quantum state. Note that while global phases in quantum mechanics cannot be measured (and therefore descriptions of states with a global phase are all equivalent descriptions), relative phases can be measured. Also note that the phase $e^{i\varphi}$ is different from an optical phase (see Appendix~\ref{app:squeezed}). See Section~\ref{sec:phase_coherence} for more details on relative phases.} of the next pulse relative to the previous pulse. This phase encodes a classical string $S=S_1S_2\cdots S_n$, where $S_i\in\{0,1\}$ determines the relative phase between the pulses:
\begin{equation}
\ket{\Psi} = \ket{e^{i\phi_1}\alpha}\ket{e^{i\phi_2}\alpha} \cdots \ket{e^{i\phi_{n+1}}\alpha},
\end{equation}
where $\phi_i = \phi_{i-1} + \pi\cdot{S_{i-1}}, i\in\{2,3,\dots,n+1\}$ and $\ket{e^{i\phi_i}\alpha}$ is a coherent state. This leaves the global phase, $\phi_1$, as arbitrary.

Note that $\ket{\Psi}$ cannot be written as a tensor product state such that each individual state only depends on one bit of $S$:
\begin{equation}\label{eq:tensorDPS}
\ket{\Psi} \neq \ket{\psi(S_1)}\ket{\psi(S_2)} \cdots \ket{\psi(S_n)}.
\end{equation}

To measure this state, Bob uses a Mach-Zehnder interferometer (see Fig.~\ref{fig:DPS}). The input first goes into a 50:50 beamsplitter where each end has a different length (see Section~\ref{sec:devices}). The length difference is the distance between the pulses in Alice's state. These paths are recombined on two inputs of another beamsplitter so that these paths can interfere. The result is that the phases of neighbouring pulses will interfere and a detector can be placed at each end of the second beamsplitter. Depending on which one clicks, Bob will know the relative phase of Alice's pulses (either $\phi_i-\phi_{i-1}$ is $0$ or $\pi$). See Section~\ref{sec:MZ} to see how a Mach-Zehnder interferometer achieves this phase measurement.

\begin{figure} \centering
\includegraphics[width=\textwidth]{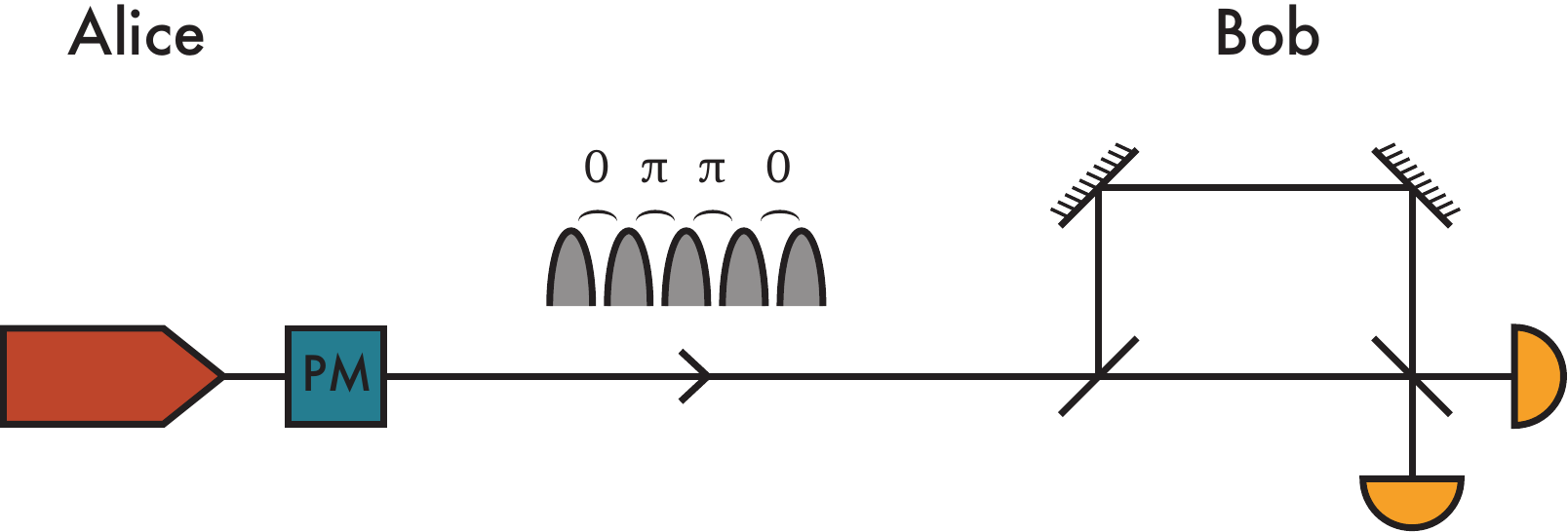}
\caption[The DPS Protocol]{The DPS protocol. Alice produces coherent states where she modulates their relative phase using a phase modulator (PM). She picks either $0$ or $\pi$ to be the phase angle between the pulses. Bob measures using a Mach-Zehnder interferometer that measures these relative phases.}
\label{fig:DPS}
\end{figure}

The security proof of this protocol is more challenging than for the qubit protocols listed earlier, since Alice's state cannot be broken down into the tensor product form of Eq.~\ref{eq:tensorDPS}. This means that there are less symmetries that can be exploited in order to use the same tools that work for qubit protocols. While there is no security proof for the way the protocol is described above, there is a security proof if a single photon is split up into $m$ pulses (called a \emph{block}) and then many of these independent blocks are used instead of using coherent states in a long chain of pulses \cite{wen09}. Attacks on the DPS protocol have also been analyzed \cite{curty08,gyongyosi12,moroder12}.

\subsubsection{Coherent One-Way (COW)}

Another protocol that does not use independent qubits, but is still discrete, is the coherent one-way (COW) protocol \cite{gisin04,stucki05}. Similarly to the DPS protocol, the COW protocol can be implemented with passive elements on Bob's side and the state that Alice sends cannot be decomposed into a tensor product of states that only depend individually on one of the bits Alice is trying to send to Bob.

The states that Alice prepares to send her uniformly random bit string, depending on if the bit in position $i\in\{1,2,\dots,n\}$ is $0$ or $1$, are
\begin{equation}\label{eq:COWstate1}
\ket{0_L}_i = \ket{\alpha}_{2i-1}\ket{0}_{2i}, \quad \ket{1_L}_i = \ket{0}_{2i-1}\ket{\alpha}_{2i},
\end{equation}
where $\ket{0}$ is the vacuum state\footnote{Sometimes $\ket{0}$ is used to denote the vacuum state, which will be used in some contexts, such as when we write coherent states (Eq.~\ref{eq:coherent_state}). When there is a conflict of notation between the bit values to correspond to the states (where here Alice wants to send the logical bit $0$) and the vacuum state, the logical bit will be written with the subscript $0_L$.}, $\ket{\alpha}$ is a coherent state, and $\ket{0_L}_i$ and $\ket{1_L}_i$ denote the logical bits Alice wants to send, $0$ and $1$, as the $i$th signal. Note that Alice will have two pulses per bit she would like to send (see Fig.~\ref{fig:COW}). Also, $\ket{0_L}_i$ and $\ket{1_L}_i$ in Eq.~\ref{eq:COWstate1} are not orthogonal, since the coherent state has a vacuum component.

In order to counteract an eavesdropper, Alice also has to send some other states that will not be used for Alice's and Bob's strings, but will only be used to detect an eavesdropper. Alice will, with probability $q$ prepare a decoy state that spans two time slots:
\begin{equation}
\ket{\text{decoy}}_{i} = \ket{\alpha}_{2i-1}\ket{\alpha}_{2i}.
\end{equation}
With probability $1-q$ she prepares her $\ket{0_L}$ or $\ket{1_L}$ state according to her starting string.

Bob's measurement is composed of two parts. With probability $p$ he will measure if there is at least one photon in each pulse. This measurement will tell him if Alice was trying to send a $0$ or a $1$. With probability $1-p$, he does a Mach-Zehnder interferometer measurement as in the DPS protocol. This interferometer can measure the relative phase between two sequence of states: between a neighbouring $\ket{1_L}$ followed by $\ket{0_L}$ that Alice sent, as well as the phase between the two pulses of a decoy state. The measurement can also measure the phase between a decoy state that is preceded by $\ket{1_L}$ or followed by $\ket{0_L}$. It turns out to be impossible for Eve to coherently measure both the $\ket{0_L}$ and $\ket{1_L}$ states as well as keep the phases undisturbed for the decoy states \cite{gisin04}. As such, the interferometer measurement outcomes will be used for parameter estimation to detect if there is an eavesdropper.

\begin{figure} \centering
\includegraphics[width=\textwidth]{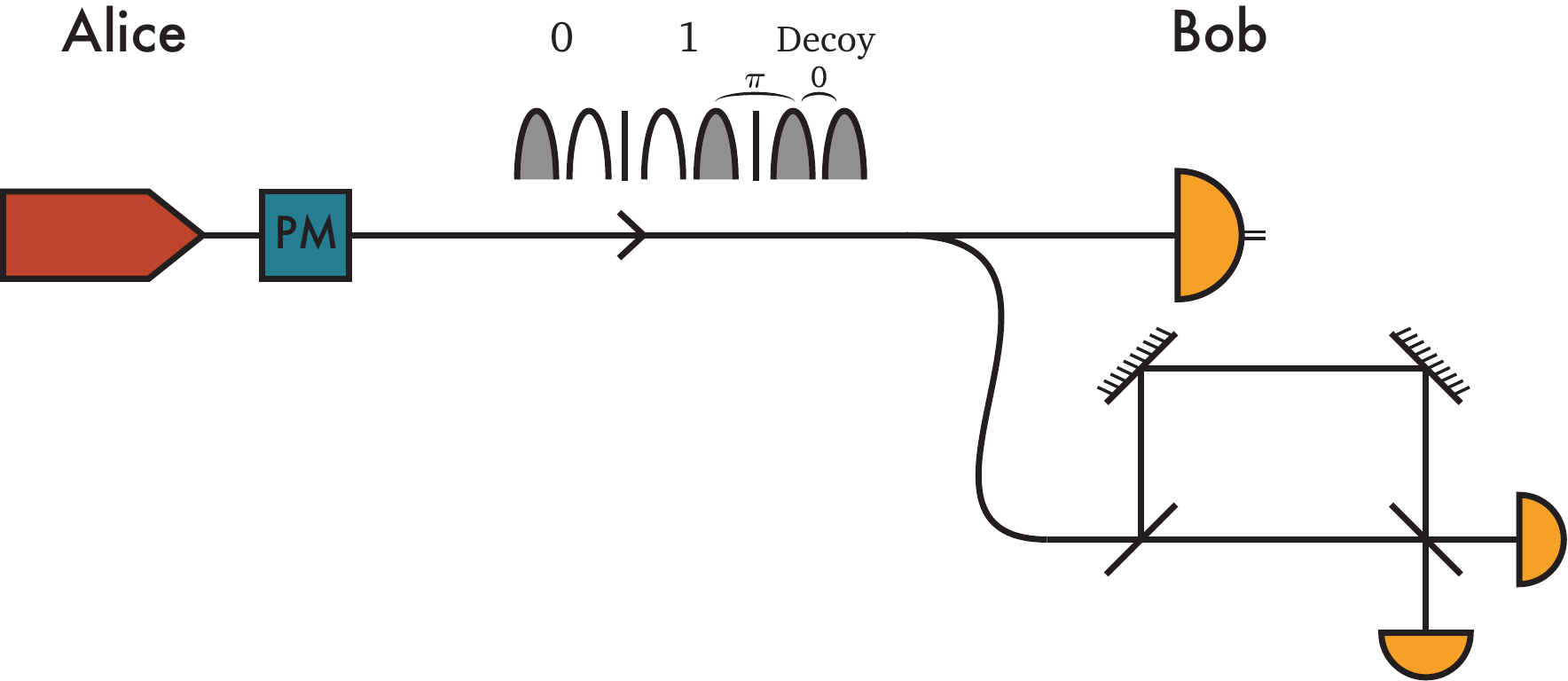}
\caption[The COW Protocol]{The COW protocol. Alice prepares one of three states: a coherent state followed by the vacuum ($0$), the vacuum followed by a coherent state ($1$), or a decoy state that is two coherent states, one after another. Alice modulates the relative phase between the coherent states by using her phase modulator (PM). Bob randomly chooses to either measure the timing of the incoming pulses or uses a Mach-Zehnder interferometer to measure the relative phase of the incoming pulses. He may measure the relative phase between two pulses of a decoy state, between a $1$ followed by a decoy state, between a decoy state followed by a $0$, or between a $1$ followed by a $0$.}
\label{fig:COW}
\end{figure}

Alice and Bob do a sifting step where Alice will tell Bob where she sent decoy states and he will throw away the measurement results when he measured those signals in his first measurement. Bob also tells Alice where he got measurement results in his interferometer, since these bits will be used for parameter estimation. They then continue with the classical post-processing steps on their classical strings.

The COW protocol does not have a full security proof, but a variant just like the DPS protocol variant \cite{wen09} that breaks up the protocol into blocks, with a single photon in each block, does have a security proof \cite{moroder12}. The COW protocol has been implemented experimentally \cite{walenta13}.

\subsection{Continuous-Variable Protocols} \label{sec:continuous_protocols}

Continuous-variable protocols typically use one of two kinds of states: coherent states, as in Eq.~\ref{eq:coherent_state}, and \emph{squeezed states}. Squeezed states are a more general state than coherent states. For a description of squeezed states and how they can be represented in phase space, see Appendix~\ref{app:squeezed}.

Usually, continuous-variable protocols are variations of the same protocol \cite{ralph99,reid00}. First, Alice prepares either a coherent state or a squeezed coherent state. If Alice prepares a squeezed coherent state, Bob does homodyne detection \cite{hillery00}. Homodyne detection is the measurement of the difference in the number of photons after interfering the input state and a local oscillator. The local oscillator for this measurement is a coherent state in phase with the input state. Homodyne detection actually measures either the $\hat{X}$ or $\hat{Y}$ quadrature operators (see Appendix~\ref{app:squeezed}), though the analysis to show this fact is beyond the scope of this thesis.

If Alice prepares coherent states then Bob does heterodyne detection. This detection can be thought of as measuring both $\hat{X}$ and $\hat{Y}$ simultaneously. Due to Heisenberg's uncertainty relation, there is some error inherent in this measurement, since both $\hat{X}$ and $\hat{Y}$ are non-commuting observables. Heterodyne detection is the same as homodyne detection, except instead of measuring photon numbers, the outputs of the beamsplitter are combined on a non-linear crystal (see Section~\ref{sec:devices}).

When Alice prepares coherent or squeezed states, she can choose different ways to vary her choice of state. Alice chooses her states from a finite discrete set in discrete protocols, and here she could do the same. She can also vary her states by choosing the parameters for the coherent or squeezed states according to a Gaussian sampling.

Continuous-variable protocols have been proven to be secure. For example, they are secure if coherent states are used and Gaussian variability is used to choose $\alpha$ \cite{leverrier13,leverrier14}.

\subsection{Device-Independent Protocols} \label{sec:DIprotocols}

Device-independent QKD was originally proposed by Ekert \cite{ekert91}. Unlike device-dependent protocols, no assumptions should be made about the devices used in the protocol. Instead, the idea is to verify that Alice and Bob share quantum states that have strong correlations (see Defn.~\ref{defn:separable} and Section~\ref{sec:CHSH}). If Alice and Bob have high correlations in their states then they can verify that Eve cannot have strong correlations with either Alice's or Bob's state.

For a history of device-independent QKD security, see \cite{vazirani12} and references therein. Security proofs of these kinds of protocols typically had to make unreasonable assumptions about the implementation, such as the need for a separate measurement device for each signal or that the protocol has no losses. However, the recent security proof of \cite{vazirani12} is the only proof to date that avoids these problems.

\subsection{Measurement-Device-Independent Protocols} \label{sec:MDI}

Measurement-device-independent (MDI) quantum key distribution is a hybrid of the device-independent and the device-dependent scenario. The advantage of using these protocols is that they are device-independent on the side of the measurement, which avoids many assumptions that are typically necessary to prove security (see Section~\ref{sec:measurements}). MDI QKD has another advantage over traditional QKD protocols, since it can be performed over longer distances than what is typically possible \cite{lo12,lim13a}.

There are two discrete-variable-type protocols, one which is an entanglement based version of the other (see Section~\ref{sec:classes}), in the same way that the BB84 protocol is equivalent to the Ekert91 protocol in their perfect descriptions.

The prepare and measure (P\&M, see Section~\ref{sec:classes}) protocol from \cite{lo12} starts with Alice and Bob uniformly at random preparing a state from the BB84 protocol (see Fig.~\ref{fig:MDI}). Alice and Bob send these states to Eve, who is untrusted. Ideally, Eve does a joint measurement of Alice and Bob's states in the Bell basis, a basis for two qubits $\{\ket{\psi^+},\ket{\psi^-},\ket{\phi^+},\ket{\phi^-}\}$, defined as
\begin{equation}\label{eq:bell_2}
\begin{aligned}
\ket{\psi^+}_{AB} &= \frac{\ket{00}_{AB}+\ket{11}_{AB}}{\sqrt{2}} \\
\ket{\psi^-}_{AB} &= \frac{\ket{00}_{AB}-\ket{11}_{AB}}{\sqrt{2}} \\
\ket{\phi^+}_{AB} &= \frac{\ket{01}_{AB}+\ket{10}_{AB}}{\sqrt{2}} \\
\ket{\phi^-}_{AB} &= \frac{\ket{01}_{AB}-\ket{10}_{AB}}{\sqrt{2}}.
\end{aligned}
\end{equation}
Eve publicly announces the measurement outcome she gets to Alice and Bob. Alice and Bob will also announce which basis they prepared their states in, followed by basis sifting to ignore measurement outcomes where their states were prepared in different bases. Alice and Bob correspond bit values to their prepared states the same way as in the BB84 protocol. Depending on the state that Eve announces, Alice may need to flip her bit value. For example, if Eve reveals $\phi^+$ and Alice and Bob prepared states in the $Z$ basis then Alice will flip her bit. Equivalently, Bob could flip his bit value instead. As another example, if Alice and Bob prepare states in the $Z$ basis and Eve reveals $\psi^+$ then Alice will not flip her bit value.

\begin{figure} \centering
\includegraphics[width=\textwidth]{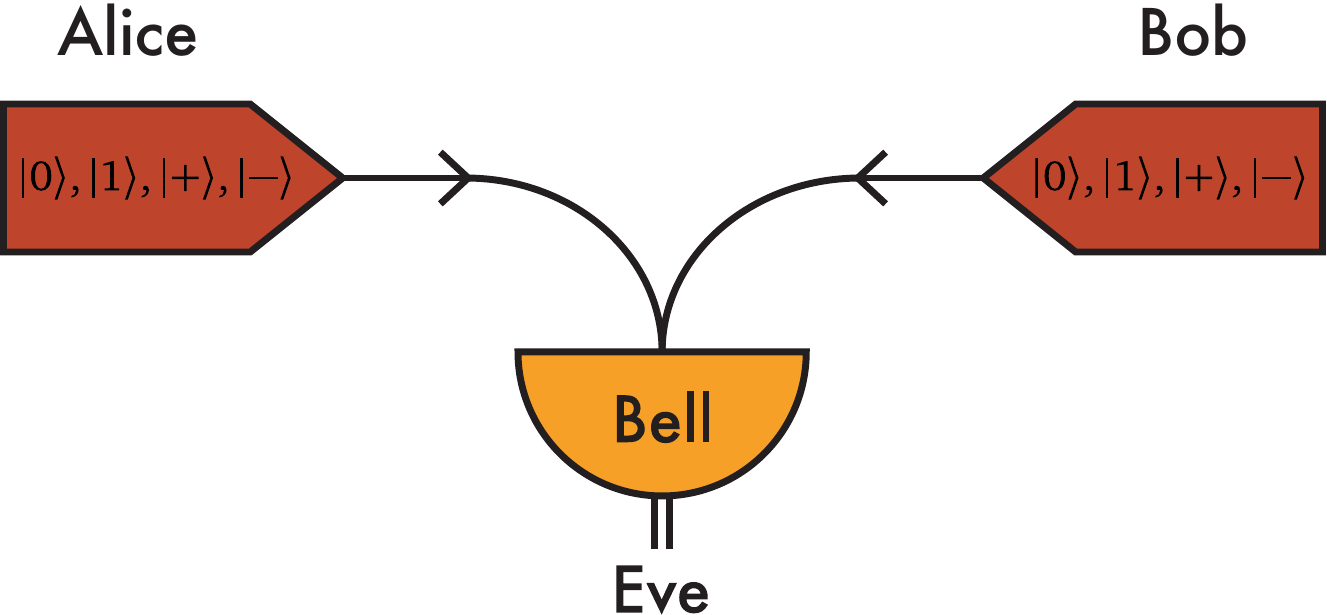}
\caption[The MDI protocol of \cite{lo12}]{MDI QKD \cite{lo12}. Alice and Bob randomly prepare one of the BB84 states ($\{\ket{0},\ket{1},\ket{+},\ket{-}\}$. Eve measures these states using a measurement in the Bell basis (Eq.~\ref{eq:bell_2}). Eve then communicates her measurement outcomes to Alice and Bob.}
\label{fig:MDI}
\end{figure}

Note that Eve cannot determine Alice's and Bob's bit values, since even if she knows the basis and the Bell measurement outcome, she only knows that Alice and Bob have the same bit value but not which bit value it is.

The entanglement based MDI protocol is the same as the above, except that Alice and Bob each prepare a copy of the state $\ket{\psi^+}$ \cite{lim13a}. Alice and Bob do a BB84 measurement on one half of this state and send the other half to Eve for her Bell measurement. The rest of the protocol follows the same steps as the P\&M version.

Both of these protocols are secure \cite{lo12,lim13a}. MDI QKD also has a continuous-variable version \cite{zhang14,li14,ma14a}, a uncharacterized qubit source version \cite{yin13,yin14}, a version that uses quantum repeaters to extend the maximum possible distance of the protocol \cite{piparo14,piparo14a,panayi14,azuma14}, and a version that uses Bell/CHSH inequalities \cite{yin14b,zhang14b}. Several experiments have now been performed \cite{tang14}.

\subsection{Counterfactual QKD}

Counterfactual (CF) QKD is where the quantum states used for the key are prepared and measured in Alice's lab. Bob infers the key from a setting of his device but he does not use measurement outcomes for his key. The states that are sent through and measured after going through the quantum channel from Alice to Bob are only used for parameter estimation. CF QKD was first introduced in \cite{noh09} and we describe this protocol here. It is related to two-way QKD protocols in its construction (see Section~\ref{sec:two_protocols}).

Alice prepares one of two qubits: $\ket{0_L}$ or $\ket{1_L}$ (see Fig.~\ref{fig:CF}). She inputs these states to a 50:50 beamsplitter (see Section~\ref{sec:BS}), which creates the state along two paths, $a$ and $b$:
\begin{equation} \label{eq:CF_state1}
\frac{\ket{0}_a\ket{\psi}_b+\ket{\psi}_a\ket{0}_b}{\sqrt{2}},
\end{equation}
where $\ket{0}$ is the vacuum state and $\ket{\psi}\in\{\ket{0_L},\ket{1_L}\}$. The state on path $a$ is kept by Alice while the state on path $b$ is sent to Bob. Bob uniformly at random uses a filtering switch that outputs $\ket{0_L}$ from $\ket{1_L}$ into different outputs. One output of this switch is the state Bob accepts and one is the state he rejects. The accepted state goes to a Z-basis measurement ($D_1$ in Fig.~\ref{fig:CF}). The rejected state is sent back to Alice.

Alice's state on path $a$ is put into a beamsplitter at the same time as the state from Bob is (potentially) returned on path $b$. If Bob rejected Alice's state then the states at the beamsplitter will interfere resulting in the output the state $\ket{\psi}$, which goes to the Z-basis measurement $D_2$ (see Fig.~\ref{fig:CF}). If Bob got a measurement result then Alice's detectors will not click. If Bob did not get a measurement result and the state and Alice's and Bob's state choices were the same then the state from Eq.~\ref{eq:CF_state1} collapses to
\begin{equation}
\ket{\psi}_a\ket{0}_b.
\end{equation}
Therefore, Alice may get a measurement in detector $D_1$ since there is no interference happening at the beamsplitter. If Alice's measurement outcome in $D_1$ is the same as her prepared state then she announces to Bob that she got a measurement outcome in $D_1$ but she does not reveal the outcome. Bob will know the outcome because it is the same as his choice of state.

Note that detector $D_1$ will click with probability $1/4$ since it clicks when Bob chooses his state to be the same as Alice's (which happens with probability $1/2$) and he does not get a measurement outcome in $D_3$ (which happens with probability $1/2$). This means that the fraction of measurement outcomes that can be used for the key is $1/4$.

\begin{figure} \centering
\includegraphics[width=\textwidth]{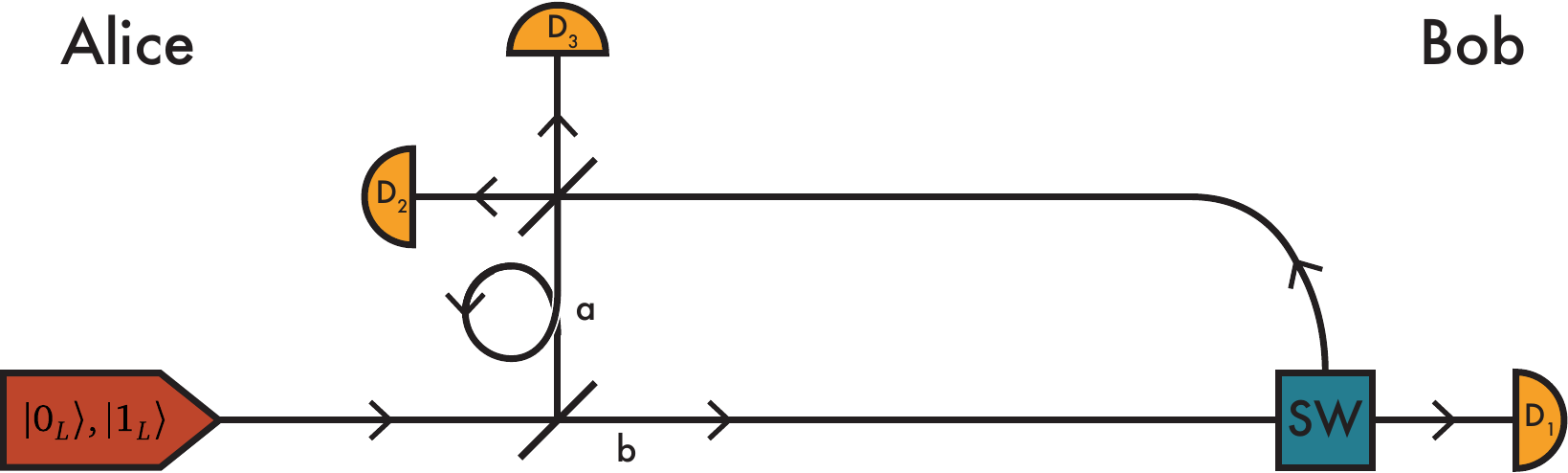}
\caption[The CF protocol of \cite{noh09}]{The counterfactual QKD protocol of \cite{noh09}. Alice prepares one of two orthogonal states: $\ket{0_L}$ or $\ket{1_L}$. After the first beamsplitter the state will be in a superposition of going to path $a$ or path $b$. On path $b$ Bob will choose a filtering switch (SW) that will select either $\ket{0_L}$ or $\ket{1_L}$. The state that Bob accepts goes to a measurement in the $\{\ket{0_L},\ket{1_L}\}$ basis at $D_1$. The state that he rejects goes back to Alice. The state along path $a$ is delayed so that it will arrive at the upper beamsplitter at the same time as Bob's rejected state. Alice then measures in the $\{\ket{0_L},\ket{1_L}\}$ basis at $D_2$ and $D_3$.}
\label{fig:CF}
\end{figure}

Alice and Bob reveal their measurement outcomes from detectors $D_2$ and $D_3$, as well as Alice's measurement outcomes when her outcome in $D_1$ did not match her prepared state. They use this information for parameter estimation.

The actual Noh09 protocol is more general and uses uneven beamsplitters, though we omit this generalization here (see Section~\ref{sec:BS}).

Another CF QKD protocol has been proposed as well \cite{salih13}.

The efficiency of the Noh09 protocol was improved in \cite{sun10}. It is not yet known if these protocols are secure, though the security of the Noh09 protocol has been analyzed in \cite{yin10,zhang12a,zhang12b}. Counterfactual QKD has also been implemented \cite{ren10,brida12}.

\section{Outline}

In Chapter~\ref{chap:preliminaries} we give an introduction to quantum mechanics using the density operator formalism as well as mathematical definitions and properties used throughout the thesis. Chapter~\ref{chap:security_proofs} discusses the security of QKD in detail and tools that can be used to prove security for a variety of protocols. Chapter~\ref{chap:assumptions} analyzes the different assumptions made in QKD and quantum cryptography and how these assumptions may be justified or may lead to insecurities. Chapter~\ref{chap:contributions} discusses two contributions of the author under the framework for security and assumptions developed in this thesis. Chapter~\ref{chap:conclusion} concludes with thoughts about the future of QKD and quantum cryptography.

Appendix~\ref{app:squeezed} presents squeezed states and phase space for continuous-variable QKD protocols. Appendix~\ref{app:math} outlines miscellaneous mathematical results used throughout this thesis.

%% file: Preliminaries.tex

\chapter{Preliminaries} \label{chap:preliminaries}
\section{Introduction}

In this chapter we outline several mathematical properties and tools that will be used in this thesis. We also give an introduction to the density operator formalism of quantum mechanics.

There are several resources available to learn quantum information and computation. There are lecture notes by John Watrous \cite{watrous13}, John Preskill \cite{preskill98}, and Renato Renner \cite{renner12a}. There are also several books, such as the most widespread quantum information and computation book \cite{nielsen00} and the recent book on quantum information theory by Mark Wilde \cite{wilde13}. The Preliminaries chapter of the PhD thesis of Marco Tomamichel \cite{tomamichelthesis} has a technical introduction to quantum information as well.

We start with an introduction to quantum mechanics using density operators instead of wave functions (Section~\ref{sec:qm}). Next we present various entropies (Section~\ref{sec:entropy}) and mutual information (Section~\ref{sec:mi}).

Further mathematical details can be found in Appendix~\ref{app:math}.

\section{Quantum Mechanics}\label{sec:qm}

Quantum mechanics is the physical model we use to characterize the quantum-cryptography protocols in this thesis. In order to prove that a quantum-cryptography protocol is secure, we need to be able to characterize what an eavesdropper or dishonest party is able to do to attack the protocol. For example, a very conservative assumption in quantum key distribution is that an eavesdropper can do anything to the states sent between Alice and Bob in the quantum channel that is allowed by quantum mechanics. Implicitly, by using quantum mechanics to characterize Eve's attack, it is assumed that Eve cannot get any further information about the quantum states sent between Alice and Bob than what quantum mechanics allows her to learn. This assumption is that quantum mechanics is complete, which will be discussed in further detail in Section~\ref{sec:foundational_assumptions}.

In addition, we assume that Alice's and Bob's devices are characterized by quantum mechanics. This limits what kind of states, measurements, and operations Alice and Bob can perform in quantum cryptography. Since these protocols are characterized by quantum mechanics, we provide descriptions of what states and transformations are permissible in this theory.

We assume that the reader understands the basics of quantum mechanics, which includes Dirac notation, Hamiltonians, and the Schr\"odinger equation. Mathematically, the reader should be familiar with the basics of linear algebra in finite dimensions such as vectors, matrices, and eigenvalues; as well as statistics such as random variables, expectation values, and probability distributions.

We introduce the density operator formalism for quantum mechanics, which is useful for treating quantum mechanics from a computer science and information theory perspective. It provides an equivalent formulation of quantum mechanics as the Schr\"odinger or Heisenberg picture using Hamiltonians, wavefunctions, and the Schr\"odinger equation.

Intuitively, the density operator formalism of quantum mechanics represents states and the transformations of states as operators and superoperators respectively. Instead of states as elements in a Hilbert space, they are operators that act on a Hilbert space. States can then be represented as matrices. The transformations allowed are no longer described by Hamiltonians (or equivalently, unitaries) but superoperators: linear maps from operators to operators. For the purposes of this thesis we will remove the time component of these superoperators and instead consider that a Hamiltonian has acted for a predetermined time. This complete transformation will then be a fixed map from operators acting on one Hilbert space to another set of operators acting on another Hilbert space.

The density operator formalism is powerful: it is a mathematically simple way (i.e.~it uses linear algebra) to represent quantum mechanics. This formalism also makes quantum mechanics easier to combine with computer science and information theory as it allows for the consideration of states that are not physical (i.e.~unnormalized states), which can be a helpful mathematical technique in quantum information theory. Unnormalized states are then related to physical states by a renormalization.

Ironically, it may also be useful to turn these matrices (and maps between matrices) that represent states and their transformations back into vectors and the matrices that act on them respectively.

\subsection{Operators and States} \label{sec:operators}

We begin by introducing operators, states, and quantum maps. This section is adapted from the more thorough exposition in \cite{tomamichelthesis}. First recall that a set of vectors in a Hilbert space $\ket{e_i}\in\hilbert$ is an orthonormal basis if $\bk{e_i}{e_j}=\delta_{ij}$ and $\text{span}{\{\ket{e_i}\}} = \hilbert$. Now we define linear operators.

\begin{defn}[Linear Operators] A linear operator $L$ is a linear map from Hilbert space $\hilbert_A$ to $\hilbert_B$ that takes elements of $\hilbert_A$, $\ket{\psi}_A \in\hilbert_A$ to $\hilbert_B$: $L\ket{\psi}_A \in \hilbert_B$. A linear operator can be represented as a matrix in a pair of orthonormal bases for $\hilbert_A$ and $\hilbert_B$, $\ket{e_i}_A$ and $\ket{f_j}_B$, respectively, for $i \in [d_A], j \in [d_B]$, where $d_A$ and $d_B$ are the dimensions of $\hilbert_A$ and $\hilbert_B$ and $[d_A]:=\{1,2,\dots,d_A\}$. The matrix representation for $L$ is then given by
\begin{equation} \label{eq:matrix_rep}
L = \sum_{i,j}\bra{f_j}L\ket{e_i} \kb{f_j}{e_i},
\end{equation}
so that the matrix element $L_{i,j}$ is given by $\bra{f_j}L\ket{e_i}$. We define the set of linear operators that map from $\hilbert_A$ to $\hilbert_B$ as $\mathcal{L}(\hilbert_A,\hilbert_B)$ and the linear operators that map from $\hilbert$ to $\hilbert$ (i.e.~endomorphisms) as $\mathcal{L}(\hilbert)$.
\end{defn}

In addition, the adjoint of an operator that maps from $\hilbert_A$ to $\hilbert_B$ is denoted as $L^{\dag}$ and is defined via 
\begin{equation}
\bra{\psi}L\ket{\phi} = \left(\bra{\phi}L^{\dag}\ket{\psi}\right)^{*} \quad \forall \ket{\phi}\in\hilbert_A, \ket{\psi}\in\hilbert_B,
\end{equation}
where $^{*}$ is the complex conjugate.

One special case of linear operators are \emph{projectors}. They are operators, $\Pi\in\mathcal{L}(\hilbert)$ that satisfy $\Pi^2=\Pi$. They can be written as $\sum_i \kb{\phi_i}{\phi_i}$ for a set of orthonormal states $\{\ket{\phi_i}\}$ that is not necessarily complete (i.e.~$\sum_i \kb{\phi_i}{\phi_i} \leq \identity$).

Another special case of linear operators are the valid physical states on Hilbert spaces: density operators. To define these, we define a few different kinds of operators and the trace of an operator.

An operator $L$ is \emph{Hermitian} if $L\in\mathcal{L}(\hilbert)$ and $L^{\dag}=L$. A \emph{positive-semidefinite} operator is a linear operator, $M$, that is Hermitian and that satisfies
\begin{equation}
\bra{\psi}M\ket{\psi} \geq 0, \quad\forall \ket{\psi} \in \hilbert .
\end{equation}
A positive semidefinite operator can be written as $M\geq 0$ and the set of all such states in a given Hilbert space is denoted as $\mathcal{P}(\hilbert)$. A \emph{unitary} operator, $U$, is a linear operator $U\in\mathcal{L}(\hilbert)$ that satisfies
\begin{equation}
UU^{\dag}=U^{\dag}U = \identity,
\end{equation}
where $\identity$ is the identity operator, which can be written as $\sum_i \kb{e_i}{e_i}$ for an orthonormal basis $\{\ket{e_i}\}$. A more general kind of operator than a unitary is an \emph{isometry}. An isometry satisfies $U\in\mathcal{L}(\hilbert_A,\hilbert_B)$ and $U^{\dag}U=\identity_A$, but $UU^{\dag} =\identity_B$ does not necessarily hold. This means that the operator $U$ maps from $\hilbert_A$ to a subspace of its full output space, $\hilbert_B$.

In addition to various kinds of operators, we also need the trace function.

\begin{defn}[Trace] Given an orthonormal basis $\{\ket{e_i}\}$ for a Hilbert space $\hilbert$ the trace of a Hermitian operator, $L$, is defined as
\begin{equation}
\Tr L := \sum_i \bra{e_i}L\ket{e_i}.
\end{equation}
\end{defn}
The trace is independent of the choice of orthonormal basis, since if the basis is chosen to be the eigenvectors of $L$ then $\Tr L$ is the sum of the eigenvalues of $L$. Specifically, if we write $L$ in its eigendecomposition (also called the \emph{spectral} decomposition) $L=\sum_i\lambda_i\kb{f_i}{f_i}$ (where $\lambda_i$ and $\ket{f_i}$ are the eigenvalues and eigenvectors of $L$ respectively) then for any unitary $U$ it holds that
\begin{equation}
U^{\dag}LU = \sum_i \lambda_i \kb{g_i}{g_i},
\end{equation}
where $\ket{g_i} = U\ket{f_i}$. Note that the set of states $\{\ket{g_i}\}$ are orthonormal ($\bk{g_i}{g_j} = \bra{f_i}U^{\dag}U\ket{f_j}=\bk{f_i}{f_j}=\delta_{ij}$) so $\lambda_i$ are the eigenvalues for $U^{\dag}LU$ as well as $L$. This means that for any basis $\{\ket{e_i}\}$ there exists a unitary $U$ such that $\ket{e_i}=U\ket{f_i}$ so that $\Tr L = \sum_i \bra{e_i}L\ket{e_i} = \sum_i \bra{f_i} U^{\dag}LU\ket{f_i} = \sum_i\lambda_i$, which does not depend on the basis $\ket{e_i}$ and therefore the trace does not depend on the basis $\{\ket{e_i}\}$ used to calculate the trace.

It is now straightforward to define quantum states in this formalism as density operators.

\begin{defn}[Density operators] A density operator, $\rho$, is defined as a Hermitian linear operator such that
\begin{equation}
\rho \in \mathcal{P}(\hilbert) \text{ and } \Tr\rho =1.
\end{equation}
The set of all density operators on a Hilbert space $\hilbert$ is written as $S_{=}(\hilbert)$.
\end{defn}

As an example, there are pure states $\ket{\psi}$ that have a corresponding density operator $\kb{\psi}{\psi}$, which can be represented as a rank-one matrix. Density operators that cannot be written as a rank-one matrix are called \emph{mixed}. Sometimes, for mathematical convenience, we will use unnormalized density operators, so that $\rho\in\mathcal{P}(\hilbert)$ and $\Tr\rho \leq 1$. These set of states on a Hilbert space $\hilbert$ is denoted as $S_{\leq}(\hilbert)$. While these states are not physical, they can be related to physical states by renormalization. If $\rho\in S_{\leq}(\hilbert)$ then $\rho/\Tr\rho\in S_{=}(\hilbert)$, which is physical.

To describe two separate systems as a single joint system the tensor product is used.

\begin{defn}[Tensor product] Given two Hilbert spaces, $\hilbert_A$ and $\hilbert_B$, the tensor product, denoted by $\hilbert_A\otimes\hilbert_B \equiv \hilbert_{AB}$ is the combination of these two spaces together. In particular, for two independent states $\rho_A\in S_{\leq}(\hilbert_A)$ and $\rho_B\in S_{\leq}(\hilbert_B)$ the global state state that describes the system is given by $\rho_A\otimes\rho_B$. If $\rho_A$ is written in an orthonormal bases for its space as $\rho_A = \sum_{ij}c_{ij}\ket{e_i}_{A}\bra{e_j}$, then the tensor product combines them in block matrix form
\begin{equation}\label{eq:tensor}
\rho_A\otimes\rho_B = \begin{pmatrix}
c_{1,1}\rho_B & c_{1,2}\rho_B & \cdots & c_{1,d_A}\rho_B \\
c_{2,1}\rho_B & c_{2,2}\rho_B & \cdots & c_{2,d_A}\rho_B \\
\vdots & \vdots & \ddots & \vdots \\
c_{d_A,1}\rho_B & c_{d_A,2}\rho_B & \cdots & c_{d_A,d_A}\rho_B
\end{pmatrix},
\end{equation}
where $d_A$ is the dimension of $\hilbert_A$. A constant times a matrix, $c\rho$, is the matrix $\rho$ with each of its elements multiplied by $c$.
\end{defn}

With composite systems, the trace may also be taken only over one of the systems.

\begin{defn}[Partial trace] Given a joint quantum state $\rho_{AB}\in S_{\leq}(\hilbert_{AB})$ and an orthonormal basis for $\hilbert_A$, $\{\ket{e_i}\}$, then the partial trace over $A$ is defined as
\begin{equation}
\Tr_A \rho_{AB} = \sum_{i} \left(\bra{e_i}_A\otimes\identity_B\right) \rho_{AB} \left(\ket{e_i}_A\otimes\identity_B\right).
\end{equation}
\end{defn}

We now define a state to be classical with respect to a quantum system if it can be written in the following form.

\begin{defn}[Classical-quantum (CQ) state] \label{defn:CQ} A state, $\rho\in S_{\leq}(\hilbert_{XB})$ is a CQ state if it can be decomposed as
\begin{equation}
\rho_{XB} = \sum_{i} p_i \ket{i}_{X}\bra{i} \otimes \rho_{B}^{i},
\end{equation}
for probabilities $p_i$, $\rho_{B}^{i}\in S_{\leq}(\hilbert_{B})\forall i$, $i\in [d_X]$, and $\ket{i}$ are orthogonal states in $\hilbert_X$.
\end{defn}

A very important set of quantum states are those that are \emph{entangled}. They represent states that have stronger correlations than what is possible with two quantum systems that are only correlated in a classical way (i.e.~are separable, see Defn.~\ref{defn:separable}). They are defined by those states that cannot be written in a \emph{separable} form. 

\begin{defn}[Separable and entangled states] \label{defn:separable} Let $\rho_{AB}\in S_{\leq}(\hilbert_{AB})$, then $\rho_{AB}$ is separable if it can be written in the form
\begin{equation}\label{eq:separable}
\rho_{AB} = \sum_{i} p_i \rho_{A}^i\otimes \rho_{B}^i,
\end{equation}
for some probabilities $p_i$ and states $\rho_{A}^i\in S_{\leq}(\hilbert_A)$ and $\rho_B^i\in S_{\leq}(\hilbert_B)$. A state that cannot be written as Eq.~\ref{eq:separable} is entangled. Also, a state is maximally entangled if it is a pure state $\sigma_{AB}\in S_{\leq}(\hilbert_{AB})$ such that the reduced density operators $\sigma_A := \Tr_{B}(\sigma_{AB})$ and $\sigma_{B} = \Tr_{A}(\sigma_{AB})$ are maximally mixed and equal to $\identity/d$, where $d$ is the dimension of $\hilbert_A$ or $\hilbert_B$ respectively.
\end{defn}

Lastly, an important equivalence between mixed states and pure states is \emph{purification}. Given a mixed state $\rho_{A}\in S_{\leq}(\hilbert_A)$ a purification of $\rho_A$ is a pure state $\ket{\Psi}_{AB}\in\hilbert_{AB}$ such that $\Tr_{B}(\kb{\Psi}{\Psi})=\rho_{A}$. In addition, for all $\rho_A$ there exists a system $B$ and a pure state $\ket{\psi}_{AB}$ such that the dimension of $B$ is at most the dimension of $A$ and $\ket{\psi}_{AB}$ is a purification of $\rho_A$. If the spectral decomposition of $\rho_A$ is written as $\sum_{i}\lambda_i \kb{i}{i}$ then one such purification can be written as
\begin{equation}
\sum_{i} \sqrt{\lambda_i}\ket{i}_A\ket{i}_B.
\end{equation}
All purifications of $\rho_A$ are equivalent up to an isometry on the purifying system, B.

\subsection{Quantum Maps} \label{sec:maps}

Now that we have defined states, we can also define the way in which states can be transformed. All possible quantum transformations are captured by completely-positive trace-preserving maps. 

\begin{defn}[Completely-positive trace-preserving (CPTP) maps]
A completely-positive trace-preserving (CPTP) map is a superoperator. Superoperators map linear operators in $\mathcal{L}(\hilbert_A)$ to linear operators in $\mathcal{L}(\hilbert_B)$. A superoperator, $\mathcal{E}$, is trace-preserving if
\begin{equation}
\Tr\mathcal{E}(L) = \Tr L, \quad\forall L \in \mathcal{L}(\hilbert_A).
\end{equation}
A super operator is completely positive if
\begin{equation}
\mathcal{E}\otimes \id (L) \geq 0, \quad\forall L\in\mathcal{P}(\hilbert_{AC}), \forall \hilbert_C
\end{equation}
where $\hilbert_C$ is an auxiliary Hilbert space and $\id$ is the identity superoperator
\begin{equation}
\id(M) = M, \quad\forall M\in\mathcal{L}(\hilbert_C).
\end{equation}
\end{defn}

In addition, a map is called \emph{positive} if its output is a positive operator.

In order to represent a CPTP map in a concrete way, there are several options. The typical one is the Kraus-operator representation. 

\begin{lemma}[Kraus-operator representation] Any CPTP map $\mathcal{E}$ can be represented as a set of linear operators $A_i$ that satisfy $\sum_{i}A^{\dag}_i A_i=\identity$ (called Kraus operators) so that $\mathcal{E}$ maps states $\rho_A\in S_{=}(\hilbert_A)$ to $S_{=}(\hilbert_B)$ by
\begin{equation}
\mathcal{E}(\rho_A) = \sum_i A_i\rho_A A_i^{\dag}.
\end{equation}
\end{lemma}

A particular kind of CPTP map is a measurement, where a quantum system is mapped to a classical one. Measurements can be put into two frameworks that are equivalent. These frameworks are projective measurements and positive operator valued measures (POVMs). These are equivalent because a POVM can be seen as projective measurement on a larger Hilbert space. POVMs will be the only framework for measurements we need for this thesis, so we introduce them here. For more information on the relation between projective measurements and POVMs, see \cite{nielsen00}.

Now we define POVM measurements. 

\begin{defn}[Quantum measurements] \label{defn:measurement} A POVM is a set of linear operators $\{F_i\}$ (each operator $F_i$ is called a POVM element) that are positive semidefinite $F_i\in\mathcal{P}(\hilbert)$ that satisfy $\sum_i F_{i} = \identity$. A measurement is defined with a POVM, where the measurement has classical outcomes $i$. Given a state $\rho\in S_{=}(\hilbert)$ that is measured using the POVM $\{F_i\}$ the probability of getting outcome $i$ is $\Tr(F_i\rho)$. The post-measurement state for an input $\rho\in S_{\leq}(\hilbert)$ is given by $\sum_i \Tr(F_i\rho)\kb{i}{i}$. Measuring in a basis $\{\ket{\psi_i}\}$ corresponds to measuring the POVM $\{\kb{\psi_i}{\psi_i}\}$.
\end{defn}

There are two properties of quantum maps that are both conceptually striking and incredibly useful. The first, the Stinespring dilation, is a correspondence between CPTP maps and unitaries. Essentially, any CPTP map can be considered as a unitary on a higher dimensional space. The second, the Choi-Jamio{\l}kowski isomorphism, is a mapping from CPTP maps to quantum states.

One problem with the Kraus representation of CPTP maps is that the set of operators $\{A_i\}$ that describe it are not unique! Kraus operators are not unique because if a set of operators $\{A_i\}$ represent a CPTP map, then so do $C_i:=UA_i$, where $U$ is a unitary. To see that this is the case, note that $\sum_{i} C_i^{\dag}C_i = \sum_{i} A_i^{\dag}U^{\dag}UA_i = \sum_i A_i^{\dag}A_i = \identity$ and 
\begin{equation}
\sum_i C_i \rho_A C_i^{\dag} = \sum_i U A_i \rho_A A_i^{\dag} U^{\dag} = U \mathcal{E}(\rho) U^{\dag}.
\end{equation} 
The unitary applied to the last term is just a change of basis for the system in $S_{=}(\hilbert_B)$ and therefore does not change the outcomes of the map $\mathcal{E}$. 

The lack of uniqueness for the Kraus operator representation makes it unideal for the analysis of some quantum information tasks (see Section~\ref{sec:squashing}) and so we use the Choi-Jamio{\l}kowski representation. The Choi-Jamio{\l}kowski (CJ) representation for quantum maps can be constructed from the Choi-Jamio{\l}kowski isomorphism, which is a linear transformation that is an isomorphism (i.e.~a transformation with an inverse) from CPTP maps to quantum states. The CJ isomorphism as presented here is not as general as it can be, since it can also apply to linear maps in general and not just ones that map positive semidefinite operators. However, here we state the CJ isomorphism only for the correspondence between quantum maps and quantum states.

\begin{thm}[Choi-Jamio{\l}kowski (CJ) isomorphism \cite{jamiolkowski72,choi75}] \label{thm:CJ} Given a CPTP map $\mathcal{E}$ that maps states in $S_{\leq}(\hilbert_A)$ to states in $S_{\leq}(\hilbert_B)$, where $\hilbert_A$ and $\hilbert_B$ have dimensions $d_A$ and $d_B$ respectively, then the CJ map is given by
\begin{equation}
\tau: \mathcal{E} \mapsto \Xi = \mathcal{E}\otimes\id (\kb{\Psi}{\Psi}) ,
\end{equation}
where $\ket{\Psi}= 1/d_A\sum_{i=1}^{d_A^2} \ket{e_i}_A\otimes\ket{e_i}_{A'}$, $A'$ is an auxiliary space that has the same dimension as $A$, and $\{\ket{e_i}\}$ is an orthonormal basis for $\hilbert_A$ and $\hilbert_{A'}$. $\Xi$ is called the Choi-Jamio{\l}kowski (CJ) matrix or CJ state.
\end{thm}

The CJ matrix $\Xi$ is therefore a $d_Ad_B \times d_Ad_B$ sized matrix. Note that since $\mathcal{E}$ is completely positive, it is clear that $\Xi\geq 0$. The way that the output of the map can be represented using $\Xi$ is by
\begin{equation} \label{eq:CJ2}
\mathcal{E}(\rho_A) = \Tr_{A'} \left( \identity_{B} \otimes \rho_{A'}^{T} \;\Xi\right),
\end{equation}
where $A'$ is a system of the same dimension as $A$, $\rho_{A'}$ is the same as $\rho_{A}$ but is in $S_{\leq}(\hilbert_{A'})$ instead of $S_{\leq}(\hilbert_{A})$, and $\rho_A^{T}$ is the transpose of $\rho_A$ with respect to an orthonormal basis $\ket{e_i}$, defined here. 

\begin{defn}[Transpose] Given a state $\rho_A\in S_{\leq}(\hilbert_A)$ and an orthonormal basis $\{\ket{e_i}\}$ for $\hilbert_A$ then the transpose with respect to this basis is defined as
\begin{equation}
\rho_A^{T} := \sum_{i,j} \bra{e_j}\rho_A \ket{e_i}\ket{e_i}\bra{e_j}.
\end{equation}
\end{defn}

We can use Eq.~\ref{eq:CJ2} to see what the trace-preserving property of $\mathcal{E}$ implies for the CJ matrix:
\begin{equation} \label{eq:tpp}
\Tr(\rho_A) = \Tr(\mathcal{E}(\rho_A)) = \Tr_A (\rho_{A}^{T} \Tr_B(\Xi)).
\end{equation}
Since Eq.~\ref{eq:tpp} has to hold for all possible $\rho_A\in S_{\leq}(\hilbert_A)$ then it holds that
\begin{equation}
\Tr_B(\Xi) = \identity_A.
\end{equation}

For more information about the CJ isomorphism, see \cite{fletcher07}, Exercise 8 at \cite{renner12a}, and the lecture notes mentioned at the beginning of this chapter. The CJ map has a concrete connection to the Kraus-operator representation. To define this connection, we introduce a notation found in \cite{fletcher07} as a representation of operators as vectors. 

\begin{defn}[Vector representation \cite{fletcher07}] \label{defn:vector} Given a linear operator $L \in L(\hilbert_A,\hilbert_B)$ that has a matrix representation from Eq.~\ref{eq:matrix_rep} where we define $c_{ij}=\bra{f_j}L\ket{e_i}$, then the vector representation of $L$ is defined as
\begin{equation}
\kket{L} := \sum_{ij} c_{ij} \ket{f_j}\ket{e_i}.
\end{equation}
\end{defn}
A ket is used here to show that $L$ is represented as a vector but the double bracket is included to show that $L$ is an operator.

Using this notation, we can represent the CJ matrix $\Xi$ in terms of the Kraus operators $A_i$ as \cite{fletcher07}
\begin{equation}\label{eq:CJ_decomp}
\Xi = \sum_i \kket{A_i}\bbra{A_i}.
\end{equation}
This means that the eigenvectors of the CJ matrix are the Kraus operators represented as vectors! Given a matrix, $\Xi$, its decomposition into a set of vectors $\kket{A_i}$ in Eq.~\ref{eq:CJ_decomp} is not necessarily unique. The decomposition, Eq.~\ref{eq:CJ_decomp}, therefore implies that the Kraus operators are not unique. In addition, this relation is a way to find one representation from the other. Given the Kraus operators and by turning them into vectors, the CJ matrix can be found. If the CJ matrix is known, find its eigenvectors, and a set of Kraus operators to represent the same map can be found as well.

In addition to the CJ isomorphism, there is another representation that is closely related, which explicitly shows the linear nature of CPTP maps. 

\begin{defn}[Normal representation \cite{watrous13}] \label{defn:normal_rep} Given a CJ matrix representation, $\Xi$, of a CPTP map, $\mathcal{E}$, and orthonormal bases for the input and output Hilbert spaces of $\mathcal{E}$, $\{\ket{e_i}\}$ and $\{\ket{f_j}\}$, then the Normal representation is defined as the matrix
\begin{equation}
\Xi^{R} = \sum_{ijkl} \bra{e_i}\bra{f_j}\Xi \ket{e_k}\ket{f_l} \; \ket{e_i}\ket{e_k}\bra{f_j}\bra{f_l}.
\end{equation}
\end{defn}

This representation is useful because of the way it acts on states. Instead of as in the CJ representation, Eq.~\ref{eq:CJ2}, a CPTP map acts as
\begin{equation}
\kket{\mathcal{E}(\rho_A)} = \Xi^R \kket{\rho_A}.
\end{equation}
This makes the linearity of CPTP maps clear: it is a matrix acting on an input vector. For complete positivity, it is easier to use the CJ representation, i.e.~$\Xi\geq 0$. The trace-preserving property, however, can be written as $\bbra{\identity}\Xi^R\kket{\identity} = 1$.

As is known from traditional quantum mechanics, all quantum maps can be represented as unitaries. In the CPTP map framework, this unitary representation comes from the Stinespring dilation. 

\begin{thm}[Stinespring dilation] \label{thm:stinespring} Given a CPTP map $\mathcal{E}$ from $S_{=}(\hilbert_A)$ to $S_{=}(\hilbert_B)$, this map can be represented as an isometry, $U_{\text{iso}}$, from $\hilbert_A$ to $\hilbert_{BR}$ followed by a partial trace over an ancillary system, $R$,
\begin{equation}
\mathcal{E}(\rho_A) = \Tr_{R} (U_{\text{iso}} \rho_{A} U^{\dag}_{\text{iso}}).
\end{equation}
Moreover, if the input space is extended to include another input system $A'$ in a fixed state $\rho_0$, then the CPTP map can be represented as a unitary, $U$, mapping $AA'$ to $BR$:
\begin{equation}
(\mathcal{E}\otimes\id)(\rho_{A}\otimes \rho_{0}) = \Tr_{R} (U \rho_{A}\otimes \rho_{0} U^{\dag}).
\end{equation}
\end{thm}

This relation, along with the representations above allow us to go between various forms of CPTP maps. They can be represented as isometries, unitaries, matrices, or a set of (Kraus) operators. Some have important advantages, such as that CJ matrices are unique, the Natural representation matrix can be applied in a simple way to states, and unitaries and isometries have particular properties (such as invertibility) that can be exploited.

\section{Entropies} \label{sec:entropy}

Entropy is a mathematical tool to quantify an amount of uncertainty. Conversely, entropy can also be used to quantify the amount of information contained in a physical system. The quantity typically used for this purpose is the Shannon or von Neumann entropy. The former applies to classical systems while the latter applies to quantum systems. Historically, the idea of entropy originated from thermodynamics and then later entropy was defined for information theory.

The Shannon and von Neumann entropy have been used in many areas of science. They apply to the situation where a process is repeated many times in exactly the same way. This is called the \emph{independent and identically distributed} (i.i.d.) scenario. Because of this repetition and independence, the Shannon and von Neumann entropies actually characterize the average uncertainty in the system over these repetitions.

It can be useful to characterize uncertainty for a single process without any repetitions. This is called the \emph{one-shot} scenario. In this case, there are classical and quantum generalizations of the Shannon and von Neumann entropies, which we call one-shot entropies. Before discussing these one-shot entropies we derive the Shannon entropy \cite{shannon48} from some basic axioms and define the von Neumann entropy \cite{vonneumann55}.

\subsection{I.I.D.~Entropy} \label{sec:iid_entropy}

Uncertainty is such a useful concept, and so widely used, that we derive entropy from a set of intuitive axioms here to give a motivation for the definition that is used. We would like any good quantifier of uncertainty to satisfy the following intuitive properties.
\begin{enumerate}
\item Uncertainty should only depend on the probabilities of a random variable, not its values. \label{ax:probs}
\item Uncertainty should increase monotonically in the number of outcomes of an experiment if all of the outcomes are equiprobable. \label{ax:mono}
\item Uncertainty is additive. If two systems are independent then the uncertainty of both systems together should be the sum of the uncertainties of each system by itself. \label{ax:add}
\item Uncertainty should be a continuous function of the probabilities of a random variable. \label{ax:smooth} 
\end{enumerate}
The first property means that, for example, the only thing uncertainty should depend on should be the probability that it rains, and not the fact that the value associated to that probability is ``raining.''

The second property means that uncertainty should increase if there are more possible outcomes. For example, an equally balanced six-sided die has less uncertainty than an equally-balanced ten-sided die, just from the fact that there are more possible outcomes for the latter die roll.

The third property means that, for example, the uncertainty about the weather tomorrow and the uncertainty about the outcome of rolling a six-sided die should just be the addition of their individual uncertainties. This property could be changed to use another ``combining'' operation instead of addition (such as multiplication), though this choice can lead to strange behaviour of the uncertainty. For example, if it is very likely that it rains tomorrow and very unlikely that the number on a die rolls a $6$ then the total uncertainty for both events, using multiplication, may be small, while the additive uncertainty would be large.

The fourth property means that if the probability of an event changes slightly, the difference in the uncertainty should be bounded by a small constant dependent on the change in the probability. This property avoids strange behaviour of the uncertainty as a function of the probabilities.

If we take these properties to be axioms for our quantity, then we necessarily reach the following unique definition (up to a constant factor). The following proof is based on the original by Shannon \cite{shannon48} and Exercise 11.2 in \cite{nielsen00}. Another proof can be found in \cite{preskill98}. 

\begin{thm}[Shannon Entropy] A measure of uncertainty, called entropy, of a random variable $X$ with values $x_i, i=\{1,2,\dots,n\}$ and probabilities $p_i$ that satisfy the above axioms must necessarily have the form
\begin{equation}
-c \sum_{i=1}^{n} p_i\log p_i,
\end{equation}
where $c$ is a positive constant.
\end{thm}
\begin{proof}
Let $A(k)$ be a function of uncertainty of a random variable $X$, where all of the probabilities are equal: $p_i=1/k$. By axiom (\ref{ax:probs}) we know that the function $A$ can only depend on $k$.

Now consider a random variable $Y$ with $s^m$ outcomes, where $s$ and $m$ are integers, and all probabilities are equal to $1/s^m$. We can also construct a similar random variable $Z$ with $t^n$ outcomes, where $t$ and $n$ are integers, and all probabilities are equal to $1/t^n$. Then we can always find an $n$ and $m$ such that\footnote{Note that the difference $s^{m+1}-s^m$ can be made arbitrarily large by increasing $m$ since $\frac{\mathrm{d}}{\mathrm{d}m}\left(s^{m+1}-s^m\right) = ms^{m-1}(s-1) \geq 0$. Another way to see this increasing difference is that $\frac{s^{m+1}(s-1)}{s^m(s-1)}=s$, so the gap between $s^m$ and $s^{m+1}$ grows by a factor of $s$ by increasing $m$ by $1$.}
\begin{equation} \label{eq:sandt}
s^m \leq t^n \leq s^{m+1}.
\end{equation}
Taking the logarithm and dividing by $n\log s$ gives
\begin{equation}\label{eq:smoothA}
\frac{m}{n} \leq \frac{\log t}{\log s} \leq \frac{m}{n} + \frac{1}{n} \implies \left| \frac{m}{n} - \frac{\log t}{\log s} \right| \leq \frac{1}{n}.
\end{equation}
Note that $n$ and $m$ can be chosen arbitrarily large and the equation is still satisfied. From axiom (\ref{ax:mono}) we can apply the function $A$ to Eq.~\ref{eq:sandt}:
\begin{equation}\label{eq:Afn}
A(s^m) \leq A(t^n) \leq A(s^{m+1}).
\end{equation}
Note that the random variable $Y$ is equivalent to considering $m$ different choices and then $s$ choices with equal probability (and similarly for $Z$). By the additivity axiom (\ref{ax:add}), this means that we can write Eq.~\ref{eq:Afn} as
\begin{equation}
m A(s) \leq n A(t) \leq (m+1)A(s).
\end{equation}
Dividing by $nA(s)$ and using Eq.~\ref{eq:smoothA} gives
\begin{align}
\frac{m}{n} \leq \frac{A(t)}{A(s)} \leq \frac{m}{n}+\frac{1}{n} \implies \left|\frac{m}{n}-\frac{A(t)}{A(s)}\right| \leq \frac{1}{n} \\
\left|\frac{A(t)}{A(s)}-\frac{\log t}{\log s}\right| \leq \frac{2}{n}.
\end{align}
Since $n$ can be made arbitrarily large, it implies that $A(t)=c\log t$, where $c$ is a constant. By the monotonicity axiom (\ref{ax:mono}), the constant $c$ must be positive.

Now consider a random variable $W$ with $n$ outcomes and probabilities $p_i=N_i/\sum_i N_i$, where $N_i$ are integers (see Fig.~\ref{fig:W}). Let each of the $N_i$ be associated with $N_i$ objects. Assume we do an experiment whose outcomes are described by $W$. We can consider getting outcome $i$ with probability $p_i$ and then uniformly at random picking one of the $N_i$ objects. The uncertainty about which object we get is then given by
\begin{equation}
H(p_1,\dots,p_n) + \sum_i p_i c \log (N_i),
\end{equation}
where $H(p_1,\dots,p_n)$ is the uncertainty in getting outcome $i$ from $W$ and the second term is the uncertainty of uniformly picking the $N_i$ objects.

Equivalently, we can consider getting one of the $\sum_{i=1}^{n}N_i$ objects with equal probability. The uncertainty in this case is $c\log (\sum_i N_i)$. From axiom (\ref{ax:add}) these uncertainties should be the same:
\begin{align}
c\log\left(\sum_i N_i\right) &\overset{(3)}{=} \sum_i p_i c\log (N_i) + H(p_1,\dots,p_n) \\
H(p_1,\dots,p_n) &= c\left( \sum_i p_i \log (N_i) - \sum_i p_i \log\left(\sum_i N_i\right)\right) \nonumber \\
&= - c \sum_i p_i \log p_i, \label{eq:entropy_end}
\end{align}
where we use axiom (\ref{ax:probs}) to write $H$ as a function of just the probabilities. Axiom~(\ref{ax:smooth}) implies that Eq.~\ref{eq:entropy_end} holds even for probability distributions different than $W$ but that are close to $W$. A similar argument can be made to argue that Eq.~\ref{eq:entropy_end} holds for all random variables \cite{shannon48,preskill98,nielsen00}. The constant $c$ is taken to be $1$ for convenience.
\end{proof}

\begin{figure} \centering
\includegraphics[width=0.5\textwidth]{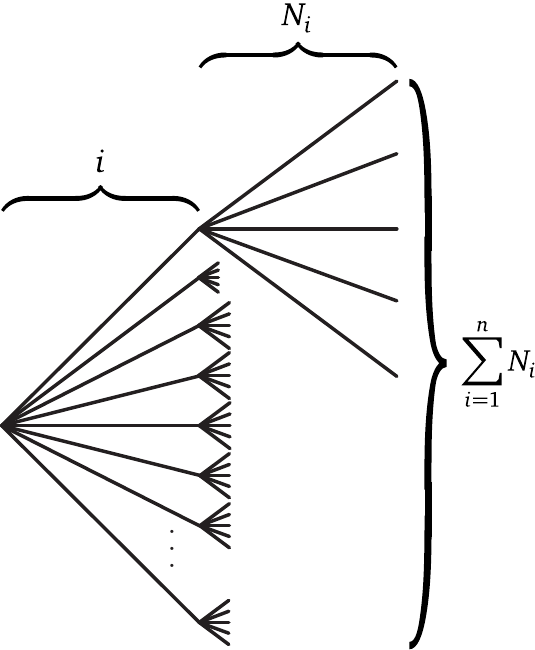}
\caption[A Random Variable, $W$]{The random variable $W$. Either one of the $\sum_{i=1}^n N_i$ items is chosen uniformly at random and associated to its group $N_i$ or item $i$ is chosen with probability $p_i=N_i/\sum_iN_i$.}
\label{fig:W}
\end{figure}

Note that in this derivation we took $n$ to be very large (i.e.~in the limit as $n$ goes to infinity). $n$ characterized the number of independent repetitions of the random variable $Z$. Therefore, the Shannon entropy only applies to the identical and independent distribution (i.i.d.)~limit.

A specific case of the Shannon entropy is for a single bit. 

\begin{defn}[Binary entropy] \label{defn:binary_entropy} Given a random variable $X$ for a single bit, with probability $p = \Pr[X=0]$, then
\begin{equation}
H(X) \equiv h(p) = -p\log p -(1-p) \log (1-p).
\end{equation}
\end{defn}

The quantum analogue of the Shannon entropy, called the von Neumann entropy, takes the eigenvalues of a density operator as probabilities and inputs them into the Shannon entropy. The von Neumann entropy can be thought of as the uncertainty in the outcomes from measuring a quantum state in its eigenbasis. 

\begin{defn}[von Neumann Entropy] \label{defn:vN} Let $\rho_{A}\in S_{\leq}(\hilbert_{A})$ then the von Neumann entropy is defined as
\begin{equation}
H(A)_{\rho} := -\Tr (\rho\log\rho).
\end{equation}
A function acting on a state is defined as the function acting on the state's eigenvalues in the state's spectral decomposition. For example if $\rho$ has spectral decomposition $\rho = \sum_i\lambda_i\kb{i}{i}$ then $\log\rho = \sum_i \log(\lambda_i) \kb{i}{i}$.
The von Neumann entropy of $\rho$ can then be written as
\begin{equation}
H(A)_{\rho} = - \sum_i \lambda_i \log\lambda_i .
\end{equation}
\end{defn}

\begin{defn}[Conditional von Neumann Entropy] \label{defn:cvN} Let $\rho_{AB}\in S_{\leq}(\hilbert_{AB})$. Then the conditional von Neumann entropy is defined as
\begin{equation}
H(A|B)_{\rho} := H(AB)_{\rho} - H(B)_{\rho}.
\end{equation}
\end{defn}

The conditional Shannon entropy can be defined in the same way as the von Neumann entropy. The conditional entropy can be interpreted as the amount that the uncertainty changes for the system $A$ upon learning $B$.

The subscript on the entropy will be dropped if it is clear from the context which state the entropy refers to (i.e.~$H(A)=H(A)_{\rho}$).

A fundamental property of the von Neumann entropy is the data-processing inequality.

\begin{thm}[Data-Processing Inequality (DPI)] Let $\rho_{ABC}\in S_{\leq}(\hilbert_{ABC})$. Then
\begin{equation}\label{eq:DPI}
H(A|BC)_{\rho} \leq H(A|B)_{\rho}. 
\end{equation}
\end{thm}
The data-processing inequality means that the uncertainty about a system $A$ cannot decrease if another system $C$ is lost. This inequality actually implies something stronger: that the uncertainty of $A$ cannot decrease under any CPTP map acting on the conditioning system. Since the Stinespring dilation (Thereom~\ref{thm:stinespring}) can represent any CPTP map as a unitary followed by a partial trace. Since the entropy is invariant under unitaries (since unitaries do not change the eigenvalues of a state) and the DPI shows that the uncertainty does not decrease under a partial trace, then for any CPTP map from a system $B$ to $D$ the uncertainty on $A$ cannot decrease:
\begin{equation}
H(A|B) \leq H(A|D).
\end{equation}

The proof of the data-processing inequality is surprisingly non-trivial and it will be discussed in Section~\ref{sec:DPI}. However, if one-shot entropies are considered instead (Section~\ref{sec:OSE}) then the data-processing inequality is straightforwardly proven (Theorem~\ref{thm:DPI_min}).

The data-processing inequality is related to another property called strong subadditivity. Given a state $\rho_{ABC}\in S_{\leq}(\hilbert_{ABC})$, then strong subadditivity is
\begin{equation}\label{eq:SSA}
H(ABC) +H(B) \leq H(AB) +H(BC).
\end{equation}
It is clear from the definition of the conditional von Neumann entropy and Shannon entropy that Eq.~\ref{eq:SSA} is equivalent to Eq.~\ref{eq:DPI} for the Shannon and von Neumann entropies.

Note that all good entropy measures should satisfy the DPI, otherwise they may decrease under CPTP maps (meaning arbitrary information may be gained by just applying maps to an isolated system). However, the same is not true for strong subadditivity. The min- and max-entropy in the next section are examples of entropies that satisfy the DPI but do not satisfy strong subadditivity.

Another important property of the von Neumann entropy is for pure states $\rho_{AB}\in S_{=}(\hilbert_{AB})$: $H(A)=H(B)$. To prove this property, we use the \emph{Schmidt} decomposition, which enables any pure state to be written as $\ket{\Psi}_{AB}=\sum_i \alpha_i\ket{\psi_i}\ket{\phi_i}$, where $\{\ket{\psi_i}\}$ and $\{\ket{\phi_i}\}$ are orthonormal bases for $\hilbert_{A}$ and $\hilbert_{B}$ respectively. The reduced states on $A$ and $B$ are then $\rho_A=\Tr_B\kb{\Psi}{\Psi}=\sum_i |\alpha_i|^2 \kb{\psi_i}{\psi_i}$ and $\rho_B = \Tr_A \kb{\Psi}{\Psi}=\sum_i |\alpha_i|^2 \kb{\phi_i}{\phi_i}$, which means that $\rho_A$ and $\rho_B$ have the same eigenvalues. Since the entropy is only a function of the eigenvalues of the state, then clearly $H(A)=H(B)$.

\subsection{One-Shot Entropies} \label{sec:OSE}

As mentioned previously, the Shannon and von Neumann entropies apply in the i.i.d.~scenario where an experiment is repeated independently and infinitely many times. For the one-shot scenario there are two important entropies, the min- and max-entropy, which we call \emph{one-shot} entropies. They come from a family of entropies called R\'enyi entropies \cite{renyi61}. It turns out that all of R\'enyi entropies are approximately equivalent to the (smooth) min- and max-entropy, so that they characterize all of the R\'enyi entropies \cite{tomamichelthesis}. We will not use these entropies in this thesis, and therefore we only discuss these two representative ones. Also, we only include their quantum definitions; their classical counterparts are defined similarly. For an in-depth discussion of one-shot entropies, see \cite{tomamichelthesis}.

\begin{defn}[Min-Entropy] \label{defn:Hmin} Let $\rho_{AB}\in S_{\leq}(\hilbert_{AB})$ then the conditional min-entropy is defined as
\begin{equation}
H_{\min}(A|B) := \max_{\sigma_B \in S_{\leq}(\hilbert_B)} \sup_{\lambda} \{ \lambda \in \mathbb{R} : \rho_{AB} \leq 2^{-\lambda} \mathbbm{1}_A \otimes \sigma_B\}.
\end{equation}
\end{defn}

The min-entropy of a classical-quantum (CQ) state $\rho_{XB}$ can be interpreted as the amount of independent number of bits that can be distilled from $X$ so that the quantum system $B$ does not have any information about the system $X$ \cite{rennerphd,koenig08}. This is the task of randomness extraction. For more details on how this task can be used in cryptography, see Section~\ref{sec:privacy_amplification2}. 

Another interpretation of the conditional min-entropy of a CQ state $\rho_{XB}$ is as a guessing probability \cite{koenig08}. If the quantum system $B$ undergoes the optimal measurement to try to predict the value of $X$ given access to the system $B$, then the probability of guessing $X$ correctly is given by $2^{-H_{\min}(X|B)}$. 

\begin{defn}[Max-Entropy] \label{defn:Hmax} Let $\rho_{AB}\in S_{\leq}(\hilbert_{AB})$ then the conditional max-entropy is defined as
\begin{equation}
H_{\max}(A|B) := \max_{\sigma_B \in S_{\leq}(\hilbert_B)} \log \left\| \sqrt{\rho_{AB}}\sqrt{\mathbbm{1}_A\otimes \sigma_B} \right\|_1^2.
\end{equation}
\end{defn}

The max-entropy characterizes the amount of entanglement required for a task called state merging \cite{berta09a}. State merging is when there is a tripartite pure state $\rho_{ABC}\in S_{=}(\hilbert_{ABC})$, where Alice and Bob hold systems $A$ and $B$ respectively, and Alice wants to send her state to Bob by only using classical communication (see Fig.~\ref{fig:state_merge}). If Alice and Bob share certain entangled states they can use a protocol called \emph{teleportation} that transfers a quantum state by only using entangled states and classical communication \cite{bennett93}. The amount of entanglement required for this task can then be quantified by the max-entropy.

\begin{figure} \centering
\includegraphics[width=0.6\textwidth]{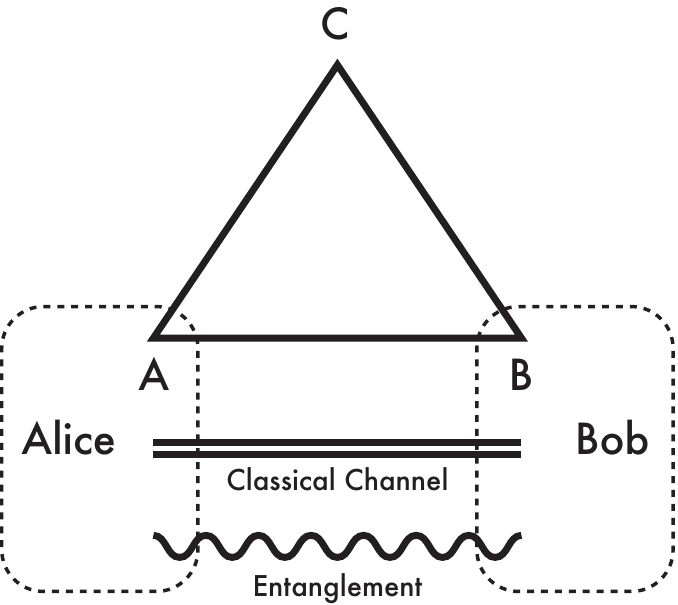}
\caption[State Merging]{State merging. Alice and Bob share a state that is purified with system $C$. Alice wants to send $\rho_A$ to Bob by communicating through the classical channel and by using entanglement shared with him.}
\label{fig:state_merge}
\end{figure}

Given a CQ state, $\rho_{XB}$, another interpretation of the max-entropy is the size of the system that $X$ can be compressed to, such that given access to the quantum system $B$, $X$ can be recovered \cite{renes12}.

The min- and max-entropy also characterize other protocols such as channel coding: the task of trying to reliably send messages through a noisy channel \cite{koenig08,tomamichelthesis}.

One problem with the above definitions is that they do not tolerate any errors in the tasks they characterize. To allow for an error probability, we define \emph{smooth} versions of these quantities. These smooth definitions will also be continuous in the quantum state, while the non-smooth definitions are not continuous \cite{tomamichelthesis}. We use the purified distance (Defn.~\ref{defn:pdist}) for our sense of closeness for the definition of the smooth min- and max-entropy. To specify a region of close states around a fixed state, we define a ball. 

\begin{defn}[$\varepsilon$-Ball] \label{defn:ball} Let $\rho\in S_{\leq}(\hilbert)$ then an $\varepsilon$-Ball around the state $\rho$ is defined as the set
\begin{equation}
\mathcal{B}^{\varepsilon}(\rho) := \{\rho' : \rho'\in S_{\leq}(\hilbert), P(\rho,\rho')\leq\varepsilon\}.
\end{equation}
\end{defn}

We can now define smooth entropies as optimizing the min- and max-entropy over a ball of states that are close to the state of interest. 

\begin{defn}[Smooth Entropies] \label{defn:sHminmax} Let $\rho_{AB}\in S_{\leq}(\hilbert_{AB})$ then the smooth conditional min- and max-entropy are defined as
\begin{align}
H^{\varepsilon}_{\min}(A|B) &:= \max_{\rho' \in \mathcal{B}^{\varepsilon}(\rho)} H_{\min}(A|B)_{\rho'}\\
H^{\varepsilon}_{\max}(A|B) &:= \min_{\rho' \in \mathcal{B}^{\varepsilon}(\rho)} H_{\max}(A|B)_{\rho'}.
\end{align}
\end{defn}

There are many properties of the min- and max-entropy which may be useful \cite{tomamichelthesis}, however for this thesis we will only need a duality of these entropies \cite{tomamichel10c}, an uncertainty relation they obey \cite{tomamichel11a}, and a special case that relates these entropies to the von Neumann entropy \cite{tomamichel08}. 

\begin{thm}[Duality of min- and max-entropy \cite{tomamichel10c}] Given a pure state $\rho_{ABC}\in S_{\leq}(\hilbert_{ABC})$ and $\epsilon\geq 0$ then
\begin{equation}
H_{\min}^{\epsilon}(A|B) = - H_{\max}^{\epsilon}(A|C).
\end{equation}
\end{thm}

\begin{thm}[Uncertainty relation for min- and max-entropy \cite{tomamichel11a}] \label{thm:ur} Let $\rho_{ABC}\in S_{\leq}(\hilbert_{ABC})$, $\epsilon\geq 0$, and define two POVMs $F$ and $G$ described by POVM elements $\{F_{x}\}$ and $\{G_{z}\}$ acting on system $A$ giving outcomes $X$ and $Z$, then
\begin{equation}
H_{\min}^{\epsilon}(X|C) + H_{\max}^{\epsilon}(Z|B) \geq \log \frac{1}{c},
\end{equation}
where $c=\max_{x,z}\|\sqrt{F_{x}}\sqrt{G_{z}}\|_{\infty}^2$ is the overlap between the measurements $F$ and $G$.
\end{thm}

This uncertainty relation can be used for cryptography, since it puts a lower bound on the entropy of Alice's measured state $X$ conditioned on an adversary's quantum system. We would like this entropy to be high, which happens when the entropy of Alice's other measurement outcome $Z$ conditioned on another system that Bob controls, $B$, is low. See Section~\ref{sec:current_methods} for how this uncertainty relation can be related to cryptography. 

\begin{thm}[Quantum Asymptotic Equipartition Property \cite{tomamichel08}] \label{thm:qaep}
Let $\rho_{AB}\in S_{=}(\hilbert_{AB})$. Then 
\begin{equation}
\lim_{\epsilon\to 0} \lim_{n\to\infty}\frac{1}{n}H^{\epsilon}_{\min / \max}(A^n|B^n)_{\rho^{\otimes n}} = H(A|B)_{\rho}.
\end{equation}
\end{thm}

This means that in the limit of having an i.i.d.~quantum state the min- and max-entropy approach the von Neumann entropy. Therefore, the min- and max-entropy are generalizations of the von Neumann entropy to the one-shot scenario.

\section{Mutual Information} \label{sec:mi}

The mutual information quantifies the amount of correlations between two systems. Like entropy, it is a useful quantity in various contexts. We define mutual information using entropy. 

\begin{defn}[Mutual Information] \label{defn:MI} Let $\rho_{AB}\in S_{\leq}(\hilbert_{AB})$ then the mutual information is defined as
\begin{equation}
I(A:B) := H(A) - H(A|B) = H(B) - H(B|A).
\end{equation}
For classical systems the Shannon entropy can be used in the definition instead.
\end{defn}

The classical mutual information quantifies the amount of information that can be sent through a channel per bit (called the channel capacity) \cite{shannon48}. In general, the mutual information quantifies the correlations between the systems $A$ and $B$.

There is also a conditional mutual information, defined similarly to the conditional von Neumann entropy. 

\begin{defn}[Conditional Mutual Information] Let $\rho_{ABC}\in S_{\leq}(\hilbert_{ABC})$ then the conditional mutual information is defined as
\begin{equation}
I(A:B|C) := H(A|C) - H(A|BC).
\end{equation}
\end{defn}

The conditional and non-conditional mutual information also apply to the i.i.d.~setting and recent efforts have tried to generalize these quantities to the one-shot scenario \cite{ciganovic14,berta14}. It is not yet clear if these definitions are good generalizations since they have limited operational meaning. However, they satisfy many mathematical properties that are required of generalizations, such as the QAEP, DPI, and generalizations of properties of the von Neumann entropy.

%% file: Security_Proofs.tex

\chapter{Security Proofs} \label{chap:security_proofs}

\section{Introduction}

Security in quantum cryptography uses several ideas from physics, information theory, and computer science. Here we deconstruct the notion of security for quantum key distribution (QKD) into its component parts and detail the steps required to make a proof. We discuss general methods that can be used to prove security. Other quantum-cryptography security proofs also use some of the same tools presented here.

For some QKD protocols, security can be thought of as stemming from the fact that non-orthogonal quantum states cannot be perfectly distinguished, such as $\ket{0}$ and $\ket{+}$ from the BB84 protocol. This means that if an eavesdropper, Eve, tries to distinguish them, she will introduce errors that Alice and Bob can detect. Either Alice and Bob can see that Eve has tampered with the quantum states and abort the protocol or Eve's interference is low enough that Alice and Bob can both correct any errors they have and remove any possible information Eve may have about their strings.

For QKD protocols that use entanglement, security can be thought of as coming from the monogamy of entanglement: If Alice and Bob share a maximally entangled state, then necessarily Eve cannot have any correlations with Alice or Bob. As long as they can verify that they indeed share highly entangled states (i.e.~states that are close to maximally entangled under some measure) then they can also correct errors and remove any information that Eve has about their measurements or abort if they see that they do not have enough entanglement.

Yet another way to see how QKD could be secure is via the no-cloning principle. Given an unknown quantum state $\rho$ there is no CPTP map that copies $\rho$: $\mathcal{E}(\rho)=\rho\otimes\rho$. To see how cloning quantum states is not possible, assume that there did exist such a map. Consider the input states $\ket{0}$ and $\ket{1}$. These get turned into $\ket{00}$ and $\ket{11}$ respectively by the cloning map $\mathcal{E}$. By linearity, this implies that $(\ket{0}+\ket{1})/\sqrt{2}$ should be mapped to $(\ket{00}+\ket{11})/\sqrt{2}$. However, if we apply the map directly to $(\ket{0}+\ket{1})/\sqrt{2}$ we get $(\ket{0}+\ket{1})(\ket{0}+\ket{1})/2\neq(\ket{00}+\ket{11})/\sqrt{2}$, which is a contradiction with our assumption that such a map existed. So as long as there is some uncertainty in what the state is (from Eve's perspective) then she cannot make (perfect) copies of the states sent through the quantum channel.

As discussed in Chapter~\ref{chap:intro}, security of a QKD protocol is proven for a \emph{model}. A model is a description of the protocol that includes a series of instructions for Alice and Bob to perform the protocol. Models include a characterization of and assumptions about the devices used in the protocol, such as sources, measurements, and Eve's attack. While in Chapter~\ref{chap:intro} the security of various protocols was mentioned, these were only meant as a statement of whether these protocols are secure in principle, i.e.~for at least one model of the protocol. It is an entirely different challenge to prove that a practical model of a protocol is secure.

In this chapter we will discuss general tools without going into the details of how to prove security for practical models. The connection to security for practical models will be discussed in Chapter~\ref{chap:assumptions}. Also, we will focus on finite-dimensional Hilbert spaces and therefore discrete QKD protocols. We will also discuss continuous-variable and device-independent security, but to a lesser extent. This focus is mainly due to the fact that there are some general tools for discrete protocols that work for a variety of protocols, while the current proofs for non-discrete protocols are usually more specialized.

First, we define security in a precise way (Section~\ref{sec:security_definition}). Then we discuss the classical post-processing steps used in QKD and how the results from these other fields can be used to help reduce the security definition to a different kind of problem (Section~\ref{sec:post_processing}). Lastly, we show several methods that are used to prove security by using several reduction techniques (Section~\ref{sec:methods}).

\section{Security Definition} \label{sec:security_definition}

Before we describe how to prove security, it is important to define what we mean by security so we know what we actually want to prove in the first place! Intuitively we want to make sure that Alice and Bob share a key that no eavesdropper has any information about. This definition is too strong as we can only achieve approximate security, but approximate security is adequate for practical purposes. More precisely, we want that an eavesdropper knows nothing about the key Alice and Bob have with very high probability (secrecy). Also, we need to be sure that the protocol generates the same strings for Alice and Bob in the presence of an adversary (correctness). Lastly, we need to ensure that the protocol succeeds with high enough probability when there is no eavesdropper but some noise is present (robustness).\footnote{In computer science the definitions of correctness and robustness are typically different than what is presented here.}

Note that we need all three of these conditions (secrecy, correctness, and robustness) to hold simultaneously, otherwise some protocols would be considered secure that are either not useful or do not fit with our intuitive notion of security. Consider the following three examples.

If a protocol is secret and correct then we consider it secure but it is not robust. In this case we would allow a protocol that always aborts to be considered secure. While this situation may fit in with the notion of security, these kinds of protocols are not useful, so we will also require a certain level of robustness.

If a protocol is correct and robust but not secret, then Eve may have some information about Alice and Bob's key. This protocol is clearly not secure!

Finally, if a protocol is secret and robust but not correct, then Alice and Bob may have secret keys but they are not the same, which defeats the purpose of what QKD is trying to achieve.

In addition to secrecy, correctness, and robustness, we also want to make sure that we can compose a QKD protocol with other protocols. For example, we could use a key from QKD for the one-time pad encryption to send a secure classical message. Then we want to make sure that even if Eve keeps whatever knowledge she has from the QKD protocol in a quantum memory she cannot find out any information about the key no matter what other protocols come afterwards. This notion is called \emph{composable security}. We will discuss how our definition of security ensures that QKD is composable (Section~\ref{sec:composability}).

We now discuss three models of what Eve can and cannot do, since security can be proven under each of these models.

\subsection{Eve's Attacks} \label{sec:eve_attack}

There are three different classes of attacks for Eve that are considered in the literature. In increasing order of power given to Eve, they are: \emph{individual} attacks, \emph{collective} attacks, and \emph{coherent} attacks. The first two attack strategies are considered in order to simplify the analysis, while the third strategy is the most general attack allowed by quantum mechanics. When facing the daunting task of proving security for a QKD protocol a first attempt may be made to prove security against individual attacks before moving on to proving full security under coherent attacks. Also, when a new QKD protocol is proposed it can be helpful to consider some simple individual attacks to see if the protocol is secure at all or if Eve can learn information without introducing a disturbance in the quantum states sent between Alice and Bob.

Individual attacks are the least powerful for Eve: Eve attacks each signal as it is sent from Alice to Bob in the same way (i.e.~individually). Her attack consists of a quantum operation on each signal with some CPTP map jointly with some systems of her own. After her CPTP map, Eve is required to measure her systems, but there is a discrepancy in the literature about which point Eve has to perform this measurement \cite{scarani09}. Some say that this measurement happens after each signal is sent, which corresponds to the situation where Eve does not have a quantum memory, while others say that Eve measures at the time after all the classical post-processing is finished except for privacy amplification.

Collective attacks are when Eve attacks the signals independently as with individual attacks but keeps her own systems in a quantum state and does not have to measure them.

Coherent attacks are the most general: Eve is allowed to do any attack allowed by quantum mechanics to the quantum systems sent between Alice and Bob.

Many security proof methods only prove security against collective attacks. However, there are mathematical tools that can be used to generalize these proofs to prove that a protocol is secure against coherent attacks such as the de Finetti theorem for quantum states or the post-selection technique (see Section \ref{sec:reductions}).

Before giving a definition of security, we have a historical note on what security used to mean in the QKD community.

\subsection{Historical Definition} \label{sec:historical}

Intuitively, security in the context of QKD is to ensure that Eve only has a negligible amount of information about Alice and Bob's key after the protocol. One measure of information used to quantify Eve's information was the \emph{accessible information}. If Alice and Bob share a key, $K$, after the QKD protocol and $Y$ is a random variable that describes the outcome of a measurement Eve applies to her system after the protocol, then the accessible information is defined as the mutual information $I(K:Y)$. Then security was defined as
\begin{equation}\label{eq:accinfo}
I(K:Y) \approx 0,
\end{equation}
for all possible strategies Eve can use to attack the protocol and measurements she can perform on her system. Since mutual information is a measure of correlations between the random variables ($K$ and $Y$ in this case) and the operational interpretation of the mutual information as a quantification of the correlations between two systems (Section~\ref{sec:mi}), it was thought that this definition captures the intuitive meaning of security. 

However, the accessible information was shown to not be secure. Using the accessible information assumes that Eve does a measurement after the QKD protocol that is independent of any other information she could learn through a future protocol that uses the key. Eve could do a measurement that does depend on new information she learns during such future protocols. Indeed, an example was presented in \cite{konig07} that shows that Eq.~\ref{eq:accinfo} can be satisfied and Eve can still gain information about the key. If the key is split into two parts $K=K_1K_2$ and Eve delays the measurement of her system until she finds out the first part of the key $K_1$, then it is possible that $I(K_2:Y')\gg 0$, where $Y'$ is obtained from Eve measuring her system using her knowledge of $K_1$.

This kind of security loophole is a lack of composability (see Section~\ref{sec:composability}), since a part of the key is not secure when composed with the public revealing of another part of the key. We therefore want a security definition that can be composed with arbitrary other protocols and whatever part of the key Alice and Bob keep secret should still remain secure.

Since the discovery of the lack of security of the accessible information \cite{konig07} a new definition has been proposed, which we use here \cite{konig07,rennerphd,portmann14}. The definition that we use has both an operational interpretation that agrees with the intuition we have for security (Eve has negligible information about Alice and Bob's shared key) and is also composable. We first introduce the greater framework in which cryptographic security can be defined in general and then state our definition of security for QKD.

\subsection{Abstract Cryptography} \label{sec:AC}

Throughout this chapter we will consider QKD in the cryptographic framework known as \emph{Abstract Cryptography} (AC) \cite{maurer11}. This framework takes a top-down approach to cryptography, where protocols are abstract black boxes that perform pre-defined actions by taking inputs from, and giving outputs to, various parties, some who are honest and some who are adversarial. Other approaches build up a framework in a bottom-up way by starting to define a computation or communication model \cite{pfitzmann00,canetti01}, but we want to avoid the details here of the individual components of protocols by using the AC framework instead.

The AC framework is helpful to define security in an abstract and precise way. While we will not define the AC framework explicitly here, we will introduce some notions that help to define and understand security. For more details on this framework, see \cite{portmann14,maurer11}.

Two kinds of protocols in AC are the ideal protocol and the real protocol.\footnote{In the AC framework these are usually called real and ideal systems. However, to avoid confusion with quantum systems, we call these entities protocols.} For QKD, the ideal protocol runs a simulation of the real protocol and if the simulated protocol succeeds then Alice and Bob are given newly constructed identical secret keys and Eve gets no information about these new keys (Fig.~\ref{fig:ideal}). If the simulated protocol fails then Alice, Bob, and Eve are notified that the protocol failed. Note that Eve learns whether the protocol succeeded or failed but she never learns anything else.

The real protocol is the model of what actually happens, where Eve is allowed to attack quantum communication between Alice and Bob and can get information about the key Alice and Bob are trying to construct (Fig.~\ref{fig:real}). Note that this is a very general model that encompasses any possible quantum channel Alice and Bob use and any attack strategy by Eve that is allowed by quantum mechanics.

\begin{figure} \centering
\begin{subfigure}[b]{0.6\textwidth}
\includegraphics[width=\textwidth]{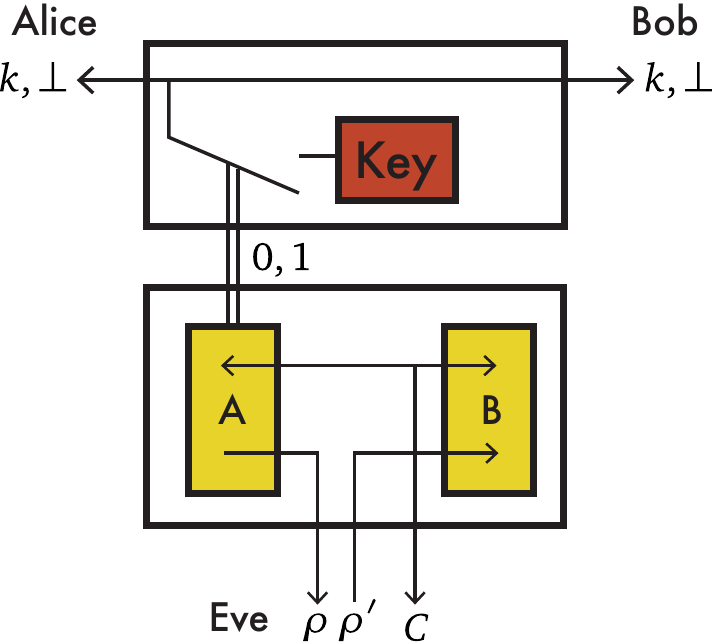}
\caption{The ideal QKD protocol. A simulation of the real protocol is performed. If the simulation succeeds then Alice and Bob get access to a shared secret key. If the simulation fails then Alice and Bob get symbol $\bot$ that indicates a failure.}
\label{fig:ideal}
\end{subfigure} \\\vspace{1cm}
\begin{subfigure}[b]{0.8\textwidth}
\includegraphics[width=\textwidth]{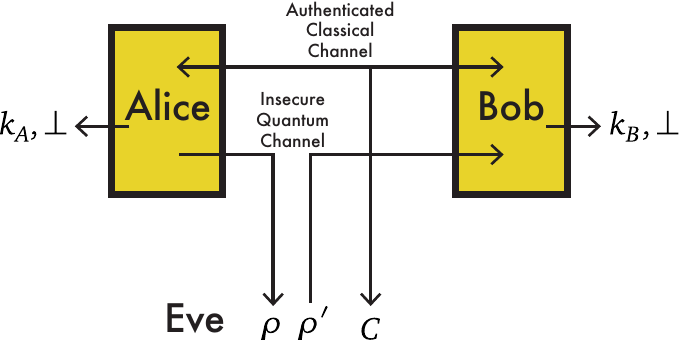}
\caption{The real QKD protocol. Alice and Bob have protocols they perform by interacting with an authenticated classical channel and an insecure quantum channel that Eve can attack. Alice tries to send state $\rho$ to Bob, which Eve may interfere with and send another state $\rho'$ to Bob instead. Eve also gets a copy of the classical communication $C$ sent through the authenticated classical channel. At the end of the protocol Alice and Bob have $k_A$ and $k_B$ respectively or the protocol aborts and they get the symbol $\bot$.}
\label{fig:real}
\end{subfigure}
\caption[The Real and Ideal QKD Protocols]{The real and ideal QKD protocols \cite{portmann14}.}
\end{figure}

Security is defined as the distance between the states shared by Alice, Bob, and Eve from the ideal protocol and the real protocol. To define this distance, we use the notion of a \emph{distinguisher}. A distinguisher in QKD is an agent who has complete control of all inputs and outputs of Alice and Bob in a QKD protocol. The distinguisher may use any strategy (i.e.~choices of inputs and interactions with outputs) to try to distinguish the real protocol from the ideal protocol.

The distinguisher has a \emph{distinguishing advantage} $\varepsilon=2p-1$ if the distinguisher can distinguish between the real and ideal protocol with probability $p$. Note that the distance measure that describes the distinguishing advantage is the trace distance (Defn.~\ref{defn:tdist}) due to its operational interpretation. If two states $\rho$ and $\sigma$ are given to a distinguisher that has to distinguish which state they have, the probability of guessing correctly is given by $1/2+1/2D(\rho,\sigma)$. The amount by which the distinguisher can do better than randomly guessing is the advantage, given by $1/2D(\rho,\sigma)$.

The distinguishing advantage is used as a definition for security since the distinguishing advantage implies that AC protocols can be composed with other protocols and they still remain secure.

\subsection{Composability}\label{sec:composability}

It is important that protocols can be composed with other protocols to form new protocols and the security should not be compromised by this composition. If a protocol can be composed in any way with any other protocol, and the statement of its security still holds, then it is called \emph{universally composable}.

For example, part of the key from QKD could be used to form an authenticated classical channel. It is crucial that the rest of the key that is not used is still secure, even if Eve has kept her state from QKD in a quantum memory and then measures her state using new information she gains from the authentication protocol. It is important that the part of the key that is used for authentication can be used as if it were a secure key, even though it is only approximately secure.

The distinguishing advantage implies that the protocol is universally composable. Formal proofs of the composability of protocols whose security is characterized by the trace distance can be found in \cite{benor05,muller-quade09,maurer11,maurer12,portmann14}.

Composition can be broken up into two scenarios: sequential and parallel composition. Sequential composition is where a protocol uses outputs of a first protocol as inputs to a second protocol (such as the example described above). Parallel composition is where two protocols are run simultaneously and are combined to be considered as one protocol.

Sequential composition can be proven by using the triangle inequality for the distinguishing advantage. If one protocol is secure except with probability $\varepsilon$ then we call it $\varepsilon$-secure. If one protocol is $\varepsilon_1$-secure and another is $\varepsilon_2$-secure then these two protocols together are $(\varepsilon_1+\varepsilon_2)$-secure. Parallel composition comes from a similar argument. For further details, see \cite{maurer11,maurer12,portmann14}.

We can now use the trace distance to define security that is composable. We decompose security into two separate notions: secrecy and correctness \cite{rennerphd,tomamichel12a,hanggi10b}. These simplify the process of proving security by reducing it to concrete statements about Alice's and Bob's strings and Eve's quantum state as opposed to having to deal with the abstract ideal and real protocols.

\subsection{Secrecy}

Secrecy for QKD is the notion that Eve does not have any information about Alice's key. Secrecy is defined as the distance between the shared state of Alice and Eve in the real protocol and ideal protocol (see Section~\ref{sec:AC}).

It is helpful to consider the distance between the states in the ideal protocol and real protocol as being decomposed into two scenarios: one where the protocol aborts and one where the protocol does not abort. Note that when the protocol aborts, Alice's key is trivial, which means that the distance between the real and ideal state in the two protocols is zero, since Eve's state is the same in both protocols. In the ideal protocol, Eve has no information about Alice's state and Alice's state is uniformly random whenever the protocol does not abort: $\tilde{\rho}^{\text{pass}}_{AE}:=\identity_{A}/d_A\otimes\rho_{E}$. This means that the distance between the ideal protocol's state $\tilde{\rho}_{AE}$ and the real protocol's state $\rho_{AE}$ is
\begin{equation}
\left\|\rho_{AE}-\tilde{\rho}_{AE}\right\|_1 \leq p_{\text{abort}} \cdot 0 + (1-p_{\text{abort}})\left\|\rho_{AE}^{\text{pass}}-\tilde{\rho}_{AE}^{\text{pass}}\right\|_1,
\end{equation}
where the latter states are conditioned on not aborting. This leads to the formal definition of secrecy \cite{benor05,konig07,rennerphd}. 

\begin{defn}[$\varepsilon$-secrecy] \label{defn:secrecy}
A protocol is $\varepsilon$-secret if for any state $\rho_{AE}^{\text{pass}}\in S(\hilbert_{AE})$, the state of the shared system between Alice and Eve after a QKD protocol (conditioned on not aborting) satisfies
\begin{equation} \label{eq:secrecy}
(1-p_{\text{abort}}) \; D\left(\rho_{AE}^{\text{pass}},\frac{\identity_{A}}{d_A}\otimes\rho_{E}^{\text{pass}}\right) \leq \varepsilon.
\end{equation}
where $p_{\text{abort}}$ is the probability of aborting the protocol and $d_A$ is the dimension of $\hilbert_A$.
\end{defn}

This definition means that the state after the real protocol is close to the ideal protocol (Fig.~\ref{fig:ideal}), i.e.~the real protocol's state is close to the situation where Alice's string is uniformly random and independent of Eve. Also, since we performed the same protocol inside the ideal protocol, Eve has the same state in the ideal protocol: $\rho_E^{\text{pass}} = \Tr_A(\rho_{AE}^{\text{pass}})$.

Importantly, the secrecy definition does not specify anything about the state $\rho_{AE}^{\text{pass}}$. Therefore, the real or ideal protocols are secret regardless of how Eve tries to attack them. We want to be sure that our security definition ensures that Eve does not have any useful information about the key. It turns out that Defn.~\ref{defn:secrecy} implies that Eve does not get any information (with high probability), which is another reason why we use the trace distance for our security definition \cite{konig07}.

Another way of interpreting the security definition other than the difference between the ideal and real protocol's states is given by the following lemma. If the distinguisher has a distinguishing advantage of $\varepsilon$ then the keys are exactly the same as the ideal keys, except with probability $\varepsilon$.

\begin{lemma}[Lemma 1 of \cite{renner05}, Prop. 2.1.1 in \cite{rennerphd}, Corr. A.7 in \cite{portmann14}]\label{lemma:sec_interpret}
Given two probability distributions $P_X$ and $P'_{X'}$ over the same alphabet, there exists a joint distribution $P_{XX'}$ such that $P_X$ and $P'_{X'}$ are the marginals of $P_{XX'}$ such that
\begin{equation}\label{eq:interp}
\Pr_{(x,x')}[x\neq x'] \leq D(P_X,P'_{X'}).
\end{equation}
\end{lemma}

A proof of this lemma can be found in \cite{portmann14}.

To see how this lemma implies the second interpretation of the security definition above, consider the following scenario.\footnote{This interpretation is from Christopher Portmann.} Let us assume Eve uses a strategy to measure her system to try to find out $A$ and gets a classical outcome $W$. Let $P_{AW}$ be the distribution of Alice's key and Eve's outcome in the ideal scenario, and $P_{\tilde{A}\tilde{W}}$ be the distribution in the real protocol. Then Lemma~\ref{lemma:sec_interpret} says that we can construct a joint distribution $P_{AW\tilde{A}\tilde{W}}$ with the property of Eq.~\ref{eq:interp}. Now we define the event
\begin{equation}
\Omega := [A=\tilde{A}\text{ and } W=\tilde{W}],
\end{equation}
where the ideal and real protocols have the same output. Lemma~\ref{lemma:sec_interpret} implies that the probability that $\Omega$ does not occur is
\begin{equation}
\Pr[\neg\Omega] \leq D(P_{AW},P_{\tilde{A}\tilde{W}}). 
\end{equation}
Since the ideal case is secure, if the event $\Omega$ happens on a run of the protocol, then the real protocol is also secure. Lemma~\ref{lemma:sec_interpret} together with the definition of security and the fact that the trace distance only decreases under CPTP maps (Lemma.~\ref{lemma:CPTP_dist}, where here the CPTP map is Eve's measurement to obtain $W$ from $E$) implies that
\begin{equation}
\Pr [\Omega] \geq 1-D(P_{AW},P_{\tilde{A}\tilde{W}}) \geq 1 - D\left(\rho_{AE}^{\text{pass}},\frac{\identity_{A}}{d_A}\otimes\rho_{E}\right) \geq 1-\varepsilon.
\end{equation}
This means that the real protocol is completely secure except with probability $\varepsilon$. This gives an operational interpretation to $\varepsilon$: it is the probability of failure for the protocol conditioned on not aborting.

\subsection{Correctness}

Next we have the definition of approximate correctness of a QKD protocol. This definition is straightforwardly motivated since we want to be sure that Alice's and Bob's keys are almost always the same. We just require that the probability of their keys being different is low.

To make this definition we first define Alice's and Bob's keys at the end of the protocol as $K_A$ and $K_B$ respectively. If the protocol succeeds then these keys will represent the strings that Alice and Bob have. If the protocol aborts, then we will write $K_A=\bot$ and $K_B=\bot$ to denote that Alice and Bob know that the protocol aborted.

\begin{defn}[$\varepsilon$-correctness] Let $K_A$ and $K_B$ be the random variables for the strings that Alice and Bob have at the end of the QKD protocol respectively. Then the protocol is $\varepsilon$-correct if
\begin{equation}
\Pr\left[K_A\neq K_B\right] \leq \varepsilon.
\end{equation}
\end{defn}

We can now combine secrecy and correctness to define security. Since we want the protocol to be indistinguishable from a secret and correct protocol, we can combine these two properties in the following way.

\subsection{Security: Combining Secrecy and Correctness}

Security is defined as a protocol that is both correct and secret (see the above sections). The precise definition of security can be somewhat confusing in that it is defined differently throughout the literature. Sometimes security is defined just as secrecy; or sometimes as secrecy, correctness, and robustness. Here, we clearly state security as an operational combination of secrecy and correctness. Robustness will be added as a separate criteria, and while robustness is considered an essential property of a protocol, it is not included in the security definition itself. 

\begin{defn}[$\varepsilon$-security] \label{defn:security}
Let $\rho_{ABE}^{\text{pass}}\in S(\hilbert_{ABE})$ be the state of the shared system between Alice, Bob, and Eve after a QKD protocol, conditioned on not aborting. Then the protocol is $\varepsilon$-secure if under any attack strategy by Eve:
\begin{equation} \label{eq:security}
(1-p_{\text{abort}}) D\left(\rho_{ABE}^{\text{pass}},\rho_{AB}^{\text{sec}}\otimes\rho_{E}^{\text{pass}}\right) \leq \varepsilon,
\end{equation}
where $p_{\text{abort}}$ is the probability of aborting (which is the same for the real and ideal protocols) and $\rho_{AB}^{sec}:=1/2^{|K|}\sum_{k} \kb{k,k}{k,k}$.
\end{defn}

Note that we do not need to define security conditioned on not aborting but instead we can define security as the trace distance between the real protocol's state $\rho_{ABE}$ and $\rho_{AB}^{\text{sec}_2}\otimes\rho_E$, where $\rho_{AB}^{\text{sec}_2}:=(1-p_{\text{abort}})\rho_{AB}^{\text{sec}}+p_{\text{abort}}\kb{\bot,\bot}{\bot,\bot}$. However, as with the definition of secrecy, this definition is equivalent to Eq.~\ref{eq:security} since Alice's and Bob's strings are trivially the same in the real and ideal protocols when they abort (since they get the symbol $\bot$ when the protocol aborts) and Eve's state is also the same in both protocols conditioned on aborting (since she only knows that the protocol has aborted and nothing else).

In addition, the definition of security does not make any assumptions about the state shared by Alice, Bob, and Eve. This means that Eve can do any attack allowed in the ideal and real protocols.

Now we can show the relationship between secrecy, correctness and our definition of security. 

\begin{thm}[$\varepsilon$-security] If a protocol is $\varepsilon_{\text{sec}}$-secret and $\varepsilon_{\text{cor}}$-correct then the protocol is $\varepsilon$-secure, where $\varepsilon=\varepsilon_{\text{sec}}+\varepsilon_{\text{cor}}$.
\end{thm}
We include a proof of this theorem here, since this theorem is essential to define security from the definitions of secrecy and correctness for QKD. To see why the sum of the parameters for secrecy and correctness can be used for secrecy, we use the following proof from \cite{portmann14}. The proof follows from the triangle inequality for the trace distance.
\begin{proof}
First, we define $p_{k_A,k_B}$ to be the probability that Alice and Bob get keys $k_A$ and $k_B$ conditioned on the protocol not aborting. Also, we define the quantum state Alice, Bob, and Eve share after the real protocol, $\rho_{ABE}$, which can be written as a CQ state (Defn.~\ref{defn:CQ}):
\begin{equation}\label{eq:ABE_state}
\rho_{ABE} := p_{\text{abort}}\kb{\bot,\bot}{\bot,\bot}\otimes\rho_E^{\bot} + \sum_{k_A,k_B} p_{k_A,k_B}\kb{k_A,k_B}{k_A,k_B}\otimes \rho_{E}^{k_A,k_B}.
\end{equation}
If we define the state
\begin{equation}
\gamma_{ABE} := \frac{1}{1-p_{\text{abort}}}\sum_{k_A,k_B} p_{k_A,k_B}\kb{k_A,k_A}{k_A,k_A}\otimes \rho_{E}^{k_A,k_B},
\end{equation}
where Alice and Bob share the same key and Eve is independent of their keys, then by using the triangle inequality we get
\begin{equation} \label{eq:tri_security}
D\left(\rho_{ABE}^{\text{pass}},\rho_{AB}^{\text{sec}}\otimes\rho_{E}^{\text{pass}}\right) \leq D\left(\rho_{ABE}^{\text{pass}},\gamma_{ABE}\right) +D\left(\gamma_{ABE},\rho_{AB}^{\text{sec}}\otimes\rho_{E}^{\text{pass}}\right).
\end{equation}
Note that we can write $\rho_{ABE}^{\text{pass}}$ using Eq.~\ref{eq:ABE_state} as
\begin{equation}
\rho_{ABE}^{\text{pass}} =\frac{1}{1-p_{\text{abort}}}\sum_{k_A,k_B} p_{k_A,k_B}\kb{k_A,k_B}{k_A,k_B}\otimes \rho_{E}^{k_A,k_B},
\end{equation}
and therefore, by using the strong convexity of the trace distance (Theorem~\ref{thm:trace_convex})
\begin{align}
&D\left(\rho_{ABE}^{\text{pass}},\gamma_{ABE}\right) \\ &\leq \sum_{k_A,k_B}\frac{p_{k_A,k_B}}{1-p_{\text{abort}}}D\left(\kb{k_A,k_B}{k_A,k_B}\otimes \rho_{E}^{k_A,k_B},\kb{k_A,k_A}{k_A,k_A}\otimes \rho_{E}^{k_A,k_B}\right) \\
&=\sum_{k_A\neq k_B} \frac{p_{k_A,k_B}}{1-p_{\text{abort}}} = \frac{1}{1-p_{\text{abort}}}\Pr[K_A\neq K_B]. \label{eq:security_1}
\end{align}
For the other term in Eq.~\ref{eq:tri_security} note that $\gamma_{ABE}$ and $\rho_{AB}^{\text{sec}}\otimes\rho_{E}^{\text{pass}}$ both have the $B$ system as a copy of the $A$ system. Also, we know that $\Tr_B\gamma_{ABE}=\Tr_B\rho_{ABE}^{\text{pass}}$. Using these facts and that the trace distance does not increase under CPTP maps (Lemma~\ref{lemma:CPTP_dist}, in this case the map is the trace over the $B$ system) we get
\begin{align}
D\left(\gamma_{ABE},\rho_{AB}^{\text{sec}}\otimes\rho_{E}^{\text{pass}}\right) &= D\left(\gamma_{AE},\frac{\identity_{A}}{d_{A}}\otimes\rho_{E}^{\text{pass}}\right) \\ &= D\left(\rho_{AE}^{\text{pass}},\frac{\identity_{A}}{d_{A}}\otimes\rho_{E}^{\text{pass}}\right).\label{eq:security_2}
\end{align}
Combining Eq.~\ref{eq:security_1} and Eq.~\ref{eq:security_2} gives us
\begin{align}
&(1-p_{\text{abort}})D\left(\rho_{ABE}^{\text{pass}},\rho_{AB}^{\text{sec}}\otimes\rho_{E}^{\text{pass}}\right) \label{eq:security_eq} \\
&\leq \Pr[K_A\neq K_B] + (1-p_{\text{abort}})D\left(\rho_{AE}^{\text{pass}},\frac{\identity_A}{d_A}\otimes\rho_E^{\text{pass}}\right) \\
&\leq \varepsilon_{\text{cor}}+\varepsilon_{\text{sec}},
\end{align}
which implies that security (i.e.~both secrecy and correctness at the same time, Eq.~\ref{eq:security_eq}) is bounded by $\varepsilon_{\text{sec}}+\varepsilon_{\text{cor}}$.
\end{proof}

Note that through Lemma~\ref{lemma:sec_interpret} we can interpret security in a similar way to secrecy. This means that the security definition Eq.~\ref{eq:security} can be interpreted as Alice's and Bob's keys are the same and independent of Eve, except with probability $\varepsilon$.

Also, sometimes the definition of security is defined as
\begin{equation} \label{eq:sec_wrong}
(1-p_{\text{abort}}) \cdot \min_{\sigma_E} D\left(\rho_{ABE}^{\text{pass}},\frac{\identity_{AB}}{d_{AB}}\otimes\sigma_{E}\right) \leq \varepsilon,
\end{equation}
such as in the published version of \cite{tomamichel12a} or \cite{tomamichelthesis,furrer14b}. However, this definition is only known to be composable in parallel with an extra factor of $2$ (see \cite{portmann14}). Therefore, it is important to use the definition stated above, Defn.~\ref{defn:secrecy}.

\subsection{Robustness} \label{sec:robustness}

As mentioned in the introduction to this section, security is not sufficient for a QKD protocol, since a trivial protocol that outputs empty strings for Alice and Bob is secure. Therefore, we also need robustness to make sure that any protocol we consider is not only secure, but outputs keys of non-trivial size. 

\begin{defn}[$\varepsilon$-robustness \cite{portmann14}] A QKD protocol is $\varepsilon$-robust if the probability of aborting the real protocol when Eve does not attack the protocol is $p_{\text{abort}}^{\text{no Eve}}=\varepsilon$ .
\end{defn}
Note that to determine $p_{\text{abort}}^{\text{no Eve}}$ when Eve does not attack the protocol, a model of the quantum channel between Alice and Bob is required. If they know this model then they can calculate the probability that they will abort by estimating an error rate that is beyond the threshold allowed by the protocol.

Now that we have defined robustness, we discuss the classical post-processing that is performed after the quantum \stage of the QKD protocol in order to use some classical results to simplify the problem of proving security.

\section{Classical Post-Processing} \label{sec:post_processing}

Technically, proving security just entails showing that Eq.~\ref{eq:security} holds. While there may be many ways to do so, we use some standard  techniques that allow the reduction of the problem to one that is more easily proved. For example, security can be reduced to the problem of proving a lower bound on the entropy of Alice conditioned on Eve. These techniques come from the analysis of the classical post-processing performed after the quantum \stage of QKD. These are broken down into (in reverse chronological order): privacy amplification, information reconciliation, and parameter estimation.

In this thesis we focus on discrete variable protocols, where finite-dimensional Hilbert spaces are used, though some of these results apply just to classical strings and therefore are protocol independent and can be applied to continuous-variable protocols as well.

\subsection{Privacy Amplification} \label{sec:privacy_amplification2}

Privacy amplification is the process of removing any residual information that Eve may have about the key after all the other steps in the QKD protocol. This subprotocol can be achieved by using randomness extractors. Randomness extractors are functions that take a source of randomness as input, e.g.~a string with a lower bound on its entropy, as well as a small uniformly random string called a \emph{seed}, and output an almost uniformly random output that is longer than the seed. We are interested in not just extracting randomness but extracting randomness with respect to a quantum adversary. We are also interested in an extractor that is \emph{strong}, where the seed and output string are independent of each other. Together, we want a strong randomness extractor against quantum adversaries, defined here. 

\begin{defn}[Quantum-Proof Strong Randomness Extractor, Defn.~3.2 in \cite{de12}] A $(k,\varepsilon)$-strong quantum-proof randomness extractor, $\text{Ext}$, is a function from $\{0,1\}^n\times\{0,1\}^d$ to $\{0,1\}^m$ if for all CQ states $\rho_{XE}$ with a classical $X\in\{0,1\}^n$ with min-entropy $H_{\min}(X|E)_{\rho}\geq k$ and a uniform seed $Y \in\{0,1\}^d$ we have
\begin{equation}
D\left(\rho_{\text{Ext}(X,Y)YE},\frac{\identity}{2^m}\otimes \rho_{Y}\otimes \rho_{E}\right) \leq \varepsilon.\footnote{Note that $\{0,1\}^n$ is the set of bit strings of length $n$.}
\end{equation}
\end{defn}

There are two main randomness extractors used for privacy amplification in QKD: the leftover hashing lemma \cite{mcinnes87,impagliazzo89a,impagliazzo89b,rennerphd} and Trevisan's extractor \cite{de12,mauerer12}.

\subsubsection{Leftover Hashing}

Informally, the leftover hashing lemma shows how much randomness can be extracted from a classical source that has at least a certain amount of min-entropy. This lemma has also been generalized to the case where there is a quantum system that has correlations with the classical source \cite{tomamichel10a,tomamichelthesis}. For QKD this means we can prove a lower bound on the min-entropy of Alice's string given Eve's quantum system. The set of functions that achieve this randomness extraction is called a two-universal family of hash functions. 

\begin{defn}[$\delta$-almost Two-Universal Hash Functions \cite{carter79}] \label{defn:hash}
Let $\delta>0$ and let $f$ be a function in a family (i.e.~a set) $\mathcal{F}$ with input space $\mathcal{X}$ and output space $\mathcal{Y}$. Then $\mathcal{F}$ is a $\delta$-almost two-universal family of hash functions if
\begin{equation}
\Pr_{f\in \mathcal{F}} \left[ f(x) = f(x') \right] \leq \delta,
\end{equation}
for any $x\neq x' \in \mathcal{X}$. $\mathcal{F}$ is a two-universal family of hash functions if $\delta=1/|\mathcal{Y}|$.
\end{defn}

In addition, a family of two-universal hash functions always exists from $\{0,1\}^n$ (i.e.~the set of strings of $n$ bits) to $\{0,1\}^{\ell}$ for all integers $n$ and $\ell$ \cite{carter79,wegman81,rennerphd}. A family of $\delta$-almost two-universal hash functions always exists from $\mathcal{F}^r$ to $\mathcal{F}$ for $\delta=(r-1)/|\mathcal{F}|$, where $r$ is an integer and $\mathcal{F}$ is a field \cite{tomamichel10a}.

An example of a family of two-universal hash functions is the set $\mathcal{F} = \{f_\alpha\}_{\alpha\in\{0,1\}^n}$ with functions mapping from strings of bits $\{0,1\}^n$ to $\{0,1\}^{\ell}$ by
\begin{equation}
f_\alpha(x) = x\cdot \alpha \mod 2^{\ell},
\end{equation}
where $x\cdot\alpha$ is multiplication in the field $GF(2^n)$ \cite{carter79,tomamichel10a} (see Section~\ref{sec:fields}). To see why this family is two-universal, notice that
\begin{equation} \label{eq:pr_field}
\begin{aligned}
&\Pr_\alpha \left[x\cdot\alpha\mod 2^{\ell} = x'\cdot\alpha\mod 2^{\ell}\right] \\
&= \Pr_\alpha \left[(x-x')\cdot\alpha\mod 2^{\ell} = 0\right].
\end{aligned}
\end{equation}
To interpret this probability we will use the isomorphism between strings and elements of a finite group. A string of bits can be represented as members of $GL(2^{\ell})$ by representing the string modulo $2^{\ell}$. Eq.~\ref{eq:pr_field} implicitly contains the isomorphism from strings of length $n$ to the field $GL(2^n)$ in order to perform the multiplication $(x-x')\cdot\alpha$.

Let us now consider the outcome of the multiplication $(x-x')\cdot\alpha$ for all possible values of $\alpha$. If we let $a$ be a non-zero element of $GL(2^n)$ then we can either write $\alpha=a^j$ or $\alpha=0$ and write $(x-x')=a^k$ for $j,k\in\{0,\dots,2^n-2\}$. Then by varying $j$ from $0$ to $2^n-2$ and including $\alpha=0$ the set of results of the multiplication $(x-x')\cdot\alpha$ is $\{a^k,a^{k+1},\dots,a^{2^n-2},1,a,\dots,a^{k-1},0\}$, which is just a permutation of the elements of the field. This fact shows that the mapping $\alpha\mapsto (x-x')\cdot\alpha$ is bijective.\footnote{A bijection is a function where every input of the function has a unique output and every output of the function has a unique input.} Now to take $(x-x')\cdot \alpha$ modulo $2^{\ell}$ we apply an isomorphism from $GL(2^n)$ to the integer set $\{0,1,\cdots,2^n\}$. Since the mapping on $\alpha$ is bijective, each of these integers appears once, and therefore taking them modulo $2^{\ell}$ will mean each value $\{0,\cdots,2^{\ell}-1\}$ will appear with equal probability since $2^{\ell}$ divides $2^n$:
\begin{equation}
\Pr_\alpha \left[(x-x')\cdot\alpha\mod 2^{\ell} = 0\right] = \frac{1}{2^{\ell}},
\end{equation}
and hence this family is two-universal.

An example of a family of $\delta$-almost two-universal hash functions is the set $\mathcal{F}=\{f_\alpha\}_{\alpha\in\mathds{F}}$ for any field $\mathds{F}$ (see Section~\ref{sec:fields}), where $f_{\alpha}$ maps from $\mathds{F}^r$ to $\mathds{F}$ by
\begin{equation}
f_\alpha(x_1,\dots,x_r) = \sum_{i=1}^rx_i\alpha^{r-i},
\end{equation}
where $x=(x_1,\dots,x_r)$ \cite{tomamichel10a}. This family is $\delta$-almost two-universal for $\delta=(r-1)/|\mathds{F}|$ since
\begin{align}
\Pr_\alpha\left[\sum_{i=1}^r x_i\alpha^{r-i} = \sum_{i=1}^rx'_i\alpha^{r-i}\right] &= \Pr_\alpha\left[\sum_{i=1}^r(x_i-x'_i)\alpha^{r-1}=0\right] \\
&\leq\frac{r-1}{|\mathds{F}|},
\end{align}
where the last step comes from the fact that a polynomial (in this case, in $\alpha$) of order $r-1$ has at most $r-1$ roots, and $\alpha$ is chosen uniformly at random out of the elements of $\mathds{F}$.

The motivation for considering $\delta$-almost two-universal functions as well as two-universal ones is so that we can minimize the amount of randomness necessary for Alice and Bob to implement the hashing. Alice needs to have enough uniform randomness to pick the hash function from the family that she then applies to her string. This randomness can be difficult to obtain and therefore we want to minimize the amount of uniform randomness needed in the protocol.

The amount of randomness required to choose the function $f$ from a family of two-universal hash functions (if $n$ is the length of the input) is $O(n)$ \cite{carter79}, while for $\delta$-almost families the amount of randomness is $O(\ell)$ \cite{stinson94}. While the analysis can be more complicated with $\delta$-almost families, they can reduce the amount of randomness (and communication) needed in a run of a QKD protocol. Also, we will use hash functions in the information reconciliation step (see Section~\ref{sec:error_correction2}), which will minimize the amount of randomness and communication required there as well.

With the definition of a family of two-universal hash functions we can now present the leftover hashing lemma. 

\begin{lemma}[Leftover Hashing, Corollary 5.6.1 in \cite{rennerphd}] \label{lemma:loh}
Let $K\in\mathcal{K}$ be a random variable, $E$ be a quantum system, and let $\mathcal{F}$ be a two-universal family of hash functions from $\mathcal{K}\in\{0,1\}^n$ to $\mathcal{K'}\in\{0,1\}^{\ell}$. If we define the states
\begin{equation}
\rho_{K'}^{f} = \kb{f(K)}{f(K)}, \quad \rho_{K'EF} := \sum_{f\in\mathcal{F}} p(f) \rho_{K'E}^{f} \otimes \ket{f}_F\bra{f},
\end{equation}
then
\begin{equation} \label{eq:loh}
D\left( \rho_{K'EF}, \frac{\identity_{K'}}{d_{K'}}\otimes\rho_{EF} \right) \leq \varepsilon  + 2^{-\frac{1}{2}(H_{\min}^{\varepsilon}(K|E)-\ell)-1}.
\end{equation}
\end{lemma}

This lemma can also be stated in a similar form using $\delta$-almost universal hashing functions \cite{tomamichel10a}.

To understand how this lemma is useful, consider that Eve gets access to the function $f$ that Alice and Bob use for hashing since they communicate $f$ through the authenticated classical channel. This means that Eve's state is her quantum system from before privacy amplification plus a description of the function $f$. If we compare Eq.~\ref{eq:loh} to Eq.~\ref{eq:secrecy}, we see that Eve's system in the definition of secrecy $E$ is the system $EF$ for leftover hashing. Also, Alice's system here is $K'$, while in the secrecy Alice's system was written as $A$.

We will only need the leftover hashing lemma when the protocol does not abort and therefore we are implicitly conditioning the states in Lemma~\ref{lemma:loh} on the event that the protocol has not aborted. This means that instead of trying to bound the trace distance to prove secrecy, we can now try to bound the min-entropy $H_{\min}^{\varepsilon}(K|E)$. If we take the $\log$ of Eq.~\ref{eq:loh} then we can rewrite it as
\begin{equation}
-2\log \left(D\left( \rho_{K'EF}, \frac{\identity_{K'}}{d_{K'}}\otimes\rho_{EF} \right)\right) + 2\log \varepsilon -2 -2\ell \geq H_{\min}^{\varepsilon}(K|E).
\end{equation}
This means that to upper bound the trace distance, we can instead try to lower bound the smooth min-entropy of Alice's state conditioned on Eve's state.

The leftover hashing lemma is also optimal, in the sense that very little randomness and communication is necessary and it gives an exponentially tight bound on the trace distance for secrecy by the min-entropy \cite{rennerphd}. This exponential bound is the kind of scaling that is necessary for efficient QKD. See Section~\ref{sec:tuning} for more details.

\subsubsection{Trevisan's Extractor}

Another way of relating the trace distance to the min-entropy is by using Trevisan's extractor. This extractor achieves the same goal as what the leftover hashing lemma accomplished: by using a small amount of randomness to choose a function from a family of two-universal hash functions the secrecy trace distance could be upper bounded.

Trevisan's extractor is a classical randomness extractor \cite{trevisan01a} that is also a quantum-proof randomness extractor \cite{de12,mauerer12}. Similarly to $\delta$-almost universal hashing, this extractor requires $O(\ell)$ bits of communication (see Defn.~\ref{defn:bigO}). However, it requires a seed of size $O(\log^2(n/\varepsilon)\log\ell)$ as apposed to a seed of size $O(\ell)$ as in leftover hashing. Therefore Trevisan's extractor is more efficient in the amount of randomness necessary compared to leftover hashing.

The details of the function used to implement this extraction can be found in \cite{trevisan01a,de12}. Trevisan's extractor is particularly useful for proving security when assuming that Eve has a limited memory at her disposal (called the \emph{bounded storage model}) \cite{maurer92,de10}.

\subsection{Information Reconciliation}\label{sec:error_correction2}

By using the leftover hashing lemma (Lemma~\ref{lemma:loh}) or Trevisan's extractor the problem of proving a QKD protocol is secret (Defn.~\ref{defn:secrecy}) has been turned into the problem of lower bounding the conditional smooth min-entropy, $H_{\min}^{\varepsilon}(K|E)$ of a classical string, $K$, conditioned on Eve's quantum state, $\rho_E$. But we also need to be able to make sure that the protocol is correct, which can be accomplished by using an error correcting code to correct any errors between Alice's and Bob's strings. These errors can be due to Eve, noise in the quantum channel, and/or devices used in the protocol.

The task of classical error correction is to correct errors in a string (for example, the communication from a noisy channel), while classical information reconciliation is to turn two strings with correlations into two strings that are the same by possibly changing both of them. However, it is usually easier to consider information reconciliation in the special case of error correction, where Bob corrects his string to make it the same as Alice's (which is called direct reconciliation). Alice and Bob can also do reverse reconciliation where Alice corrects her string to be the same as Bob's. We consider direct reconciliation here for the simplicity of the presentation.

Consider the following scenario at this point in the protocol. Alice has a string $K_A$ and Bob has a string $K_B$ that may be different from $K_A$, while Eve has a quantum state $\rho_E$ that may have correlations with $K_A$ and $K_B$.\footnote{In the previous sections we have used $K_A$ and $K_B$ to denote the keys of Alice and Bob that may include the aborting outcome $\bot$. However, in this section we will consider $K_A$ and $K_B$ to be conditioned on not aborting given Alice and Bob's results in parameter estimation.} Alice wants to send some function of her key to Bob so that Bob can use this information and $K_B$ to reconstruct $K_A$.

What is known from parameter estimation is an estimate of the error rate and an upper bound on the smooth max-entropy of Alice's string conditioned on Bob's (see Section~\ref{sec:parameter_estimation2}). If these things are known then the only thing that Bob does not know is where his errors are in his string. Explicit error correcting codes define what communication is necessary so that Bob can find out where his errors are and correct them.

Two examples of explicit error correcting codes are low-density parity-check (LDPC) codes \cite{gallager63} and polar codes \cite{arikan08}. These codes provide an important advantage over other codes in that they are computationally efficient, achieving speeds that can be orders of magnitude faster compared to other codes. In certain cases, polar codes perform better than LDPC codes \cite{jouguet12}. Also, both codes only require communication in one direction and with one message, while other codes can require communication back and forth over many rounds. There are explicit codes in the notes \cite{steane06} or the books \cite{hamming80,macwilliams77,jones79,hill86}.

Both LDPC codes and Polar codes are linear block codes, which mean that the message that Alice needs to send to Bob in order for him to correct his errors is given by the multiplication of Alice's string with a matrix.

LDPC codes use the parities of small sets of bits. Alice can compute the parities of small subsets of her string and send them to Bob. There are several algorithms available for Bob to use these parities with his string to find out where his errors are.

Polar codes use a particular matrix to be applied to Alice's string that can be made in a recursive way. For example, if Alice's string has a length that is a power of $2$, her matrix is constructed by using
\begin{equation}
F = \begin{pmatrix}
1 & 0 \\
1 & 1
\end{pmatrix},
\end{equation}
to get the matrix $F^{\otimes n} = \overbrace{F\otimes F \otimes \cdots \otimes F}^{n \text{ times}}$. Other matrices can be similarly constructed if Alice's string is not a power of $2$.

Alternatively, it is not necessary for Alice and Bob to estimate the number of errors in parameter estimation for the error correction procedure. They can do their estimation before the QKD protocol by running a short version of the quantum \stage of the protocol. Bob can communicate his measurement results to Alice through a classical channel that is not necessarily authenticated. Alice and Bob can then estimate how many errors they will have when they run the actual QKD protocol. This method has the advantage that Alice and Bob can choose which error correcting code they will use for the protocol that is optimized for the number of errors they expect.

In the case where there is no eavesdropper and Alice and Bob estimate their errors before the QKD protocol, a good estimate can be found for the number of errors Alice and Bob will have when they run the QKD protocol due to noise in the quantum channel and their devices. If there is an eavesdropper then their estimated error rate may be wrong and therefore Alice and Bob will need to check to see if their error correction succeeds or fails during the QKD protocol. Note that Bob does not have access to Alice's system, so neither Alice or Bob know if error correction succeeded or not. We can use two-universal hash functions again (Defn.~\ref{defn:hash}) for this checking procedure.

Alice can (uniformly at random) choose a two-universal hash function from a family of such functions and apply it to her key. Alice then sends the function $f_{\text{cor}}$ and the evaluation of the function $f_{\text{cor}}(K_A)$ to Bob, who computes the function on his key $f_{\text{cor}}(K_B)$. If the hash values are equal, then with high probability Alice's and Bob's keys are the same. Due to the defining property (Defn.~\ref{defn:hash}) of families of hash functions, it is clear that the QKD protocol is $\varepsilon_{\text{cor}}$-correct if two-universal hash functions are used with an output space of $2^{-\lceil \log(1/\varepsilon_{\text{cor}}) \rceil }$, since
\begin{equation}
\Pr\left[ f_{\text{cor}}(K_A)= f_{\text{cor}}(K_B) \mid K_A \neq K_B\right] \leq 2^{-\lceil \log(1/\varepsilon_{\text{cor}}) \rceil } \leq \varepsilon_{\text{cor}}
\end{equation}
which implies that
\begin{align}
&\underbrace{\Pr\left[f_{\text{cor}}(K_A)= f_{\text{cor}}(K_B) \mid K_A \neq K_B\right]}_{\leq \varepsilon_{\text{cor}}} \underbrace{\Pr\left[K_A \neq K_B\right]}_{\leq 1} \\
&= \Pr\left[ K_A \neq K_B \mid f_{\text{cor}}(K_A)= f_{\text{cor}}(K_B)\right] \underbrace{\Pr\left[f_{\text{cor}}(K_A)= f_{\text{cor}}(K_B)\right]}_{=1},
\end{align}
where we use the fact that the protocol aborts when $ f_{\text{cor}}(K_A) \neq f_{\text{cor}}(K_B)$. Therefore we have
\begin{equation}
\Pr\left[ K_A \neq K_B \mid  f_{\text{cor}}(K_A)= f_{\text{cor}}(K_B)\right] \leq \varepsilon_{\text{cor}},
\end{equation}
which means that Alice's and Bob's strings are the same after error correction if their hash values agree, except with probability $\varepsilon_{\text{cor}}$.

For security we need that the keys that are put through the hash function in privacy amplification are correct. If the keys $K_A$ and $K_B$ after information reconciliation are the same (which happens with probability at least $1-\varepsilon_{\text{cor}}$) then their hashes are guaranteed to be the same, which implies that the protocol is $\varepsilon_{\text{cor}}$ even after privacy amplification:
\begin{equation}
\Pr\left[ f_{\text{pa}}(K_A) \neq f_{\text{pa}}(K_B) \right] \leq \Pr [K_A \neq K_B] \leq \varepsilon_{\text{cor}},
\end{equation}
where $f_{\text{pa}}$ is the hash function applied in privacy amplification.

Note that this checking procedure guarantees that the protocol is $\varepsilon_{\text{cor}}$-correct without needing to make any assumptions about the error rate or the error correcting code. Alice and Bob can therefore employ any error correcting code and can check their errors before the protocol, even without the use of an authenticated channel.

It is important to know how much information has been leaked to Eve during the error correcting code. Typically, all of the bits of communication sent from Alice to Bob in the error correction protocol are considered to be leaked bits of information to Eve. The amount of communication will depend on the particular error correcting code used. The fundamental limit on the minimal amount of communication necessary for finite-key QKD was recently analyzed in \cite{tomamichel14a}. There are also upper bounds on the amount of leaked information to Eve under various assumptions in \cite{renes12,rennerphd,renner05b,scarani08a,scarani08b}.

The communication that leaks information to Eve can be accounted for with privacy amplification by removing the classical information from Eve's system before error correction by using a chain rule for the min-entropy. If $C$ is the classical communication about the key that Eve learns from error correction, then \cite{tomamichel12a}
\begin{equation} \label{eq:classical_chain}
H_{\min}^{\varepsilon}(K_A|EC) \geq H_{\min}^{\varepsilon}(K_A|E)-\log|C|,
\end{equation}
where $|C|$ is the number of strings that are the same length as $C$. This means that if a lower bound on $H_{\min}^{\varepsilon}(K_A|E)$ can be shown then $H_{\min}^{\varepsilon}(K_A|EC)$ will also be lower bounded and therefore the protocol can be proven secure.

We have reduced proving correctness to estimating the number of errors, either through parameter estimation in the QKD protocol or by doing an estimation procedure before the protocol. It still remains to show that $H_{\min}^{\varepsilon}(K_A|E)$ is lower bounded so that the protocol is approximately secret. For example, this proof can be done by bounding the max-entropy (see Section~\ref{sec:device_dependent}). The max-entropy can be estimated from the number of errors between Alice's and Bob's string, which is one of the possible goals for parameter estimation.

\subsection{Parameter Estimation} \label{sec:parameter_estimation2}

After the quantum \stage of the QKD protocol, Alice and Bob have to estimate the error rate between their strings $K_A$ and $K_B$. This rate will upper bound the smooth max-entropy. If the error correcting code is checked by using hashing, then it is not necessary to estimate the error rate or max-entropy during parameter estimation for the information reconciliation step. However, as we will see in Section~\ref{sec:methods}, an estimate of the max-entropy of Alice's string conditioned on Bob's string can be used to prove a lower bound on the min-entropy of Alice's string conditioned on Eve's state, which proves that the protocol is secret (see Section~\ref{sec:privacy_amplification2}).

Parameter estimation can be dependent on what kinds of assumptions are made in the model of the protocol. These assumptions will be discussed in Chapter~\ref{chap:assumptions}. However, parameter estimation can be performed for many protocols independently of these assumptions. We break down its discussion into two scenarios: the finite-key and infinite-key scenario. The infinite-key scenario is just the limit as the number of signals goes to infinity (see Section~\ref{sec:classes}).

One way to perform parameter estimation is for Alice to send a uniformly random subset of her string to Bob along with the positions that describe her subset. Bob will compare this subset with the same subset of his string and announce the ratio of the number of errors between the subsets and the length of the subset. If this ratio is above a threshold, $\lambda_{\max}$, they will abort the protocol and otherwise they will continue.

Another way to perform parameter estimation is to do it simultaneously with information reconciliation. One such protocol is the cascade protocol \cite{brassard94}. This protocol compares the parity of small sets of bits to see if they are the same or not. If the parities are different then Alice and Bob will do an error-correcting procedure on this set of bits. Alice and Bob repeat the checking of several parities for different randomly-chosen sets of bits to correct their errors. By checking these parities, Alice and Bob can also estimate the number of errors between their strings. The cascade protocol is less efficient than the information reconciliation protocols from Section~\ref{sec:error_correction2}, so we do not consider it here. Instead, we focus on parameter estimation that is done completely prior to information reconciliation.

\subsubsection{Finite-Key Parameter Estimation}

Bob can apply one of several bounds to estimate the total error between Alice's and Bob's strings using the subset that Alice communicates. The tightest of these for our purposes is due to Serfling \cite{serfling74}. Serfling's inequality is an improvement on a bound by Hoeffding \cite{hoeffding63}, which is related to bounds by Chernoff \cite{chernoff52}. For our purposes Serfling's bound can be stated as follows. 

\begin{lemma}[Serfling's Inequality \cite{serfling74}] \label{lemma:serfling}
Given a set of random variables $K_i$ with values $k_i\in\{0,1\}$, where $i\in[N]$, we define the average as $K:=1/N\sum_{i=1}^N K_i$. If a sample (without replacement) of size $n$ out of $\{K_i\}_i$ is taken with values $x_j$, where $j\in[n]$, then its average is defined as $X:=1/n\sum_{j=1}^n X_j$. Let $k=N-n$ and $0 \leq \beta \leq 1$. Then
\begin{equation}\label{eq:serfling}
\Pr\left[X \geq K + \beta \right] \leq e^{-\frac{2\beta^2 nN}{k+1}}.
\end{equation}
\end{lemma}

This inequality means that the probability that the sample average is bigger than the total average is exponentially small in the sample size. The weaker bound by Hoeffding \cite{hoeffding63} is sometimes used for simplicity, which changes the upper bound in Eq.~\ref{eq:serfling} to $e^{-2\beta^2n}$.

Now we want to use this bound to show how a sample of size $k$ communicated from Alice to Bob can put a bound on the probability that the error ratio in the remaining $n$ bits ($\Lambda_n$) is larger than the observed error ratio in the sampled $k$ bits ($\Lambda_k$). This probability is conditioned on the error ratio being lower than a certain threshold. Formally, we want an upper bound to:
\begin{equation}
\Pr\left[\Lambda_{n} \geq \Lambda_k + \gamma \mid \Lambda_k \leq \lambda_{\max} \right],
\end{equation}
where $\gamma$ is a small constant. Formally these error ratios are defined as $\Lambda_n:=\frac{1}{n}|K_A^n\oplus K_B^n|$ and $\Lambda_k:=\frac{1}{k}|K_A^k\oplus K_B^k|$, where  Alice's key is split into the set of $k$ bits and $n$ bits $K_A= K_A^k K_A^n$; and $|K_A^n\oplus K_B^n|$ is the Hamming weight of the string $K_A^n\oplus K_B^n$.\footnote{The Hamming weight of a binary string $X=(X_1,X_2,\dots,X_n)$ is defined as $|X|:=\bigoplus_i X_i$, i.e.~the number of $1$'s in $X$.\label{footnote:hamming}} Bob's key is divided along the same partition of $k$ and $n$ bits.

Note that since the $k$ bits will be communicated they should be sampled without replacement, which is in accordance with Lemma~\ref{lemma:serfling}. The following bound on this probability is from \cite{tomamichel12a}.

First, from Bayes' theorem we can write
\begin{equation}\label{eq:bayes_prob}
\Pr\left[\Lambda_{n} \geq \Lambda_k + \gamma \mid \Lambda_k \leq \lambda_{\max} \right] \leq \frac{\Pr \left[ \Lambda_{n} \geq \Lambda_k + \gamma \right]}{\Pr\left[\Lambda_{k}\leq\lambda_{\max}\right]}.
\end{equation}
If we define the ratio $\nu=k/N$ then we can write the total error rate as:
\begin{equation}
\Lambda = \nu \Lambda_k + (1-\nu) \Lambda_n,
\end{equation}
where $\Lambda:=\frac{1}{N}|K_A\oplus K_B|$ is the error ratio between Alice's and Bob's complete strings. Now we can bound
\begin{align}
\Pr\left[ \Lambda_n \geq \Lambda_k + \gamma \right] &= \Pr\left[ \nu \Lambda_n \geq \nu \Lambda_k + \nu\gamma \right] \\
&= \Pr\left[ \Lambda_n \geq \nu\Lambda_k + (1-\nu)\Lambda_n +\nu\gamma \right] \\
&= \Pr\left[ \Lambda_n \geq \Lambda + \nu\gamma \right] \\
&\leq e^{-2\frac{k^2n}{(k+1)N}\gamma^2},
\end{align}
where in the last line we apply Serfling's inequality (Lemma~\ref{lemma:serfling}) and we use the definition that $\nu=k/N$. Eq.~\ref{eq:bayes_prob} can be written as
\begin{equation} \label{eq:midstat}
\Pr\left[\Lambda_{n} \geq \Lambda_k + \gamma \mid \Lambda_k \leq \lambda_{\max} \right] \leq \frac{e^{-2\frac{k^2n}{(k+1)N}\gamma^2}}{\Pr\left[\Lambda_{k}\leq\lambda_{\max}\right]}.
\end{equation}
This inequality means that the probability that the error ratio on the rest of the key $K^n$ is larger than the error ratio on the smaller sample $K^k$ plus a small amount $\gamma$, given that the protocol has an upper bound on the error rate on the sample $k$. However, what we really want is to upper bound the max-entropy to show that the protocol is secret, as we will show later (Section~\ref{sec:current_methods}).

We can use Eq.~\ref{eq:midstat} to show an upper bound on the max-entropy \cite{tomamichel12a}, since from the definition of the max-entropy (Defn.~\ref{defn:Hmax}) for classical random variables, the max-entropy is just the size of the support of the random variable (see Eq.~\ref{eq:Hmax_classical} below). The problem we have at this point is that we only have a probabilistic bound on the number of errors, Eq.~\ref{eq:midstat}, and we need instead a fixed upper bound.

To get to a fixed bound on the number of errors, consider the probability distribution
\begin{equation}
P_{K_AK_B\Lambda_k}(k_A,k_B,\lambda_k):=\Pr[K_A=k_a,K_B=k_B,\Lambda_k=\lambda_k | \Lambda_k\leq\lambda_{\max}].
\end{equation}
We can define another probability distribution
\begin{equation}
Q_{K_AK_B\Lambda_k}(k_a,k_b,\lambda_k):= \begin{cases}
\frac{P_{K_AK_B\Lambda_k}(k_A,k_B,\lambda_k)}{\Pr[\Lambda_n < \Lambda_k+\gamma | \Lambda_k\leq\lambda_{\max}]} & \text{if } \lambda_n < \lambda_k + \gamma \\
0 & \text{otherwise}
\end{cases}.
\end{equation}
We construct this distribution because under the distribution $Q$ we know that $\Lambda_n < \Lambda_k +\gamma \leq \lambda_{\max}+\gamma$ with probability $1$. This means that the number of errors on the $n$ key bits, $W:=n\Lambda_n$, satisfies
\begin{equation}
W \leq \lfloor n(\lambda_{\max}+\gamma) \rfloor.
\end{equation}
To bound the max-entropy, we need that $P$ and $Q$ are close with respect to the purified distance (Defn.~\ref{defn:pdist}), which is true since the fidelity is bounded using Eq.~\ref{eq:midstat}:
\begin{align}
F(P,Q) &= \sum_{k_A,k_B,\lambda_k}\sqrt{P(k_A,k_B,\lambda_k)Q(k_A,k_B,\lambda_k)} \\
&= \sum_{\substack{k_A,k_B,\lambda_k\\ \lambda_n < \lambda_k + \gamma}}\frac{P(k_A,k_B,\lambda_k)}{\sqrt{\Pr\left[\Lambda_n < \Lambda_k+\gamma | \Lambda_k\leq\lambda_{\max}\right]}} \\
&= \sqrt{\Pr\left[\Lambda_n < \Lambda_k+\gamma | \Lambda_k\leq\lambda_{\max}\right]}.
\end{align}
Now we can use the definition of the conditional max-entropy for classical probability distributions:
\begin{equation}\label{eq:Hmax_classical}
H_{\max}(X|Y)_P = \max_{y\in\mathcal{Y}} \log \left|\supp P_{X|Y=y}\right|,
\end{equation}
where $\mathcal{Y}$ is the set of possible values for the distribution $P_Y$ and $P_{XY}$ is a probability distribution with marginal distribution $P_Y$. This implies that
\begin{equation}\label{eq:Hmax_bound1}
H_{\max}^{\varepsilon}(K_A|K_B)_{P} \leq H_{\max}(K_A|K_B)_{Q} \leq \log \sum_{w=0}^{\lfloor n(\lambda_{\max}+\gamma)\rfloor} \binom{n}{w},
\end{equation}
where $\varepsilon:=e^{-\frac{k^2n}{(k+1)N}\gamma^2}/\sqrt{\Pr[\Lambda_k\leq \lambda_{\max}]}$. In the first inequality we used the definition of the smooth max-entropy (Defn.~\ref{defn:sHminmax}). In the second inequality we used the definition of the max-entropy for classical distributions, Eq.~\ref{eq:Hmax_classical}. Since the distribution $Q$ only has support for strings with $\lambda_n < \lambda_k+\gamma \leq \lambda_{\max}+\gamma$ we just count how many strings of length $n$ that have less than $\lambda_{\max}+\gamma$ errors.

We can end by using a technical result from Theorem 1.4.5 of \cite{lint99} which gives the upper bound
\begin{equation}\label{eq:Hmax_bound2}
 \log \sum_{w=0}^{\lfloor n(\lambda_{\max}+\gamma)\rfloor} \binom{n}{w} \leq n h(\lambda_{\max}+\gamma),
\end{equation}
where $h(\cdot)$ is the binary entropy function (Defn.~\ref{defn:binary_entropy}). By combining Eq.~\ref{eq:Hmax_bound1} and Eq.~\ref{eq:Hmax_bound2} we get an upper bound for the max-entropy:
\begin{equation} \label{eq:Hmax_pe}
H_{\max}^{\varepsilon}(K_A|K_B)_{P} \leq  n h(\lambda_{\max}+\gamma).
\end{equation}

Note that the number of random bits needed to choose $k$ elements from $N$ elements is given by $\left\lceil\log \binom{N}{k}\right\rceil$ since there are $\binom{N}{k}$ numbers of ways to do this. Therefore a string of $\left\lceil\log\binom{N}{k}\right\rceil$ bits of uniform randomness is needed to choose the set of $k$ measurement outcomes (or basis-sifted measurement outcomes) that should be communicated for parameter estimation.

If the size of the subset does not need to be fixed, then by picking each measurement outcome with probability $k/N$, the subset will approximately (and on expectation) be of size $k$. The number of bits of uniform randomness that are required in this case are $\left\lceil N h(k/N)\right\rceil$, where $h(\cdot)$ is the binary entropy function. Note that $\log\binom{N}{k} \leq N h(k/N)$ (which can be proved using Sterling's approximation) and so less randomness is needed by using the previous picking method. However, the difference between these methods is negligible for large $N$, which is a consequence of the method of types described in the next section.

Now we have shown an upper bound to the max-entropy, Eq.~\ref{eq:Hmax_pe}, which can be used to show that the QKD protocol is correct (see Section~\ref{sec:error_correction2}). What still remains is to lower bound the min-entropy in order to show that the protocol is secret (see Section~\ref{sec:privacy_amplification2}). The bound on the min-entropy is more dependent on the type of protocol than the bound on the max-entropy. Therefore, we discuss how this bound can be done in various scenarios in Section~\ref{sec:methods}. One of these methods (see Section~\ref{sec:device_dependent}) will relate the max-entropy to the min-entropy to show that the protocol is secret.

\subsubsection{Infinite-Key Parameter Estimation} \label{sec:ik_pe}

The finite-key parameter estimation estimation inequality (Eq.~\ref{eq:Hmax_pe}) can be taken in the limit of an infinite number of signals sent between Alice and Bob. In this limit, we can use the Quantum Asymptotic Equipartition Property (QAEP) (Theorem~\ref{thm:qaep}) to transform Eq.~\ref{eq:Hmax_pe} into Eq.~\ref{eq:H_pe}. However, we can also prove this result directly, without the need of the max-entropy or the QAEP. We include this proof in this section for completeness.

In the limit as the key has infinite length, the max-entropy approaches the von Neumann entropy, due to the QAEP. So in this case we only need to show an upper bound on $H(K_A|K_B)$. Since $K_A$ and $K_B$ are strings, $H(K_A|K_B)$ is the Shannon entropy. This entropy can be bounded by using the method of types \cite{csiszar98}. The method of types puts an upper bound on the entropy of $H(K_A|K_B)$ by the binary entropy function (Defn.~\ref{defn:binary_entropy}). 

\begin{lemma}[Error rate and entropy] \label{lemma:mot} Given two classical strings $K_A$ and $K_B$ then
\begin{equation}\label{eq:H_pe}
h(q) \geq H(K_A|K_B),
\end{equation}
where $q$ is the error rate between $K_A$ and $K_B$ in the limit as the size, $n$, of the strings goes to infinity. The error rate is defined as
\begin{equation}
q:=\lim_{n\to\infty}q_n := \lim_{n\to\infty}\frac{|K_A^n\oplus K_B^n|}{n},
\end{equation}
where $K_A^n$ and $K_B^n$ are the first $n$ bits of $K_A$ and $K_B$ respectively, and $|\cdot |$ denotes the Hamming weight (see Footnote~\ref{footnote:hamming}).
\end{lemma}
\begin{proof}
First, we prove that $H(K_A|K_B) \leq H(K_A\oplus K_B)$ from the definition of the conditional entropy:
\begin{align}
H(K_A\oplus K_B) &\geq H(K_A\oplus K_B|K_B)\\
&=\sum_{k_B} p(k_B) H(K_A\oplus k_B|K_B=k_B) \\
&= \sum_{k_B} p(k_B) H(K_A|K_B=k_b) \\
&= H(K_A|K_B), \label{eq:plus_cond}
\end{align}
where the first line comes from the data-processing inequality, and the third line comes from the fact that $K_A\oplus k_B$ has the same uncertainty as $K_A$ if $k_B$ is known.

Next, the method of types \cite{csiszar98} gives the following upper and lower bounds to the number of strings of length $n$ with error rate $q_n$, denoted as $T_{q_n}^n$:
\begin{equation}\label{eq:tbound}
\frac{2^{nh(q_{n})}}{n+1}\leq |T^{n}_{q_{n}}|\leq 2^{nh(q_{n})}.
\end{equation}
Taking the $\log\equiv\log_2$ of both sides and dividing by $n$, then taking the limit as $n\to\infty$ for the LHS gives:
\begin{align}
h(q) &= \lim_{n\to\infty} h(q_{n}) \leq  \lim_{n\to\infty} \left( \frac{\log(n+1)}{n} + \frac{\log|T^{n}_{q_{n}}|}{n} \right) \\
&= \lim_{n\to\infty}\frac{\log|T^{n}_{q_{n}}|}{n}.
\end{align}
For the RHS we get:
\begin{equation}
\lim_{n\to\infty} \frac{\log|T^{n}_{q_{n}}|}{n} \leq h(q).
\end{equation}
Combining the two bounds, we have
\begin{equation}
h(q) = \lim_{n\to\infty} \frac{\log|T^{n}_{q_{n}}|}{n}.
\end{equation}
Note that a uniform distribution $U$ over a set with $n$ elements has entropy
\begin{equation}
H(U) = -\sum_i \frac{1}{n}\log\frac{1}{n} = -\frac{n}{n}\log\frac{1}{n}=\log n.
\end{equation}
Now recall that $|T^{n}_{q_{n}}|$ is the size of the set of the number of strings of length $n$ with error rate $q_n$. This means that $\log |T^n_{q_n}|=H(U_{q_n}^n)$, where $H(U_{q_n}^n)$ is the entropy of a uniform distribution on the support over all strings that have length $n$ and error rate $q_n$. Therefore, we have that $H(K_A^n\oplus K_B^n) \leq H(U_{q_n}^n)$, since the maximum entropy occurs for a uniform distribution. Dividing this inequality by $n$ and taking the limit as $n\to\infty$, using $\log |T^n_{q_n}|=H(U_{q_n}^n)$ and Eq.~\ref{eq:plus_cond} gives the result:
\begin{equation}
H(K_A|K_B) \leq \lim_{n\to\infty}\frac{H(K_A^n\oplus K_B^n)}{n} \leq \lim_{n\to\infty} \frac{\log |T^n_{q_n}|}{n} = h(q),
\end{equation}
where we define the entropy in the asymptotic limit as $H(K_A \oplus K_B) := \lim_{n\to\infty}1/n \; H(K_A^n\oplus K_B^n)$.
\end{proof}

Now all that is left is to estimate the error rate $q$. This estimation can be done perfectly in the infinite-key limit, since Alice can tell Bob a small fraction of her infinitely-long string, which will also be infinitely-long. Bob then checks to see what their error rate is. Since their keys are infinitely long, they can get a perfect estimate on their error rate from Serfling's inequality (Lemma~\ref{lemma:serfling}). Alice and Bob can also estimate any other statistical quantity of their strings in this scenario since they have infinitely longs strings.

\subsection{Tuning Parameters}\label{sec:tuning}

In this chapter so far we have defined security and used the classical post-processing steps to reduce the problem of proving security via the trace distance between the states in the ideal protocol and the real protocol to a lower bound on the min-entropy of Alice's string conditioned on Eve's state. We have also found we can upper bound the max-entropy of Alice's string conditioned on Bob's string using the number of errors of a random subset of their strings, which can be used for information reconciliation and can also put a bound on the min-entropy (see Section~\ref{sec:device_dependent} below). In the infinite-key limit these entropies are the von Neumann entropy of Alice's string conditioned on Eve's state and the Shannon entropy of Alice's string conditioned on Bob's string respectively. For each of the post-processing steps there are several parameters that can be varied.

In privacy amplification using the leftover hashing lemma there is the size of the string output from the hash function $\ell$ and the failure probability $\varepsilon_{\text{pa}}$. In information reconciliation there is the failure probability of correcting the errors $\varepsilon_{\text{ir}}$. In parameter estimation there is the size of the sample $k$, the number of bits of Alice's and Bob's strings $N$, and the parameter $\gamma$. Depending on which family of hash functions are used; the explicit protocols used for privacy amplification and information reconciliation; and the parameters in parameter estimation, different bounds can be achieved for the security of the protocol.

One of the challenges of proving security for a QKD protocol is to analyze exactly what the bound is for the security and robustness. Since these bounds correspond to the failure probably of the protocol to be secure and robust, it is important to make sure that these are small enough. Typically, these should be small enough to be comparable to the failure probability of the devices used in the protocol, for example, of the order $10^{-20}$ \cite{renner12}. Other security proofs use less stringent security parameters, such as $10^{-10}$ or $10^{-14}$ \cite{tomamichel12a}.\footnote{For comparison, the probability that a person is struck by lightning is of the order of $10^{-6}$ \cite{BBC}, and the probability of winning the top prize of the EuroMillions lottery is of the order of $10^{-9}$ \cite{euromillions}.}

If the security parameter scales exponentially in terms of the number of signals sent (i.e.~it is of the form $2^{-cn}$ for a constant $c$) then numbers of the order of $10^{-6}-10^{-14}$ can be achieved. This scaling makes QKD efficiently scalable, so in order to increase the security parameter by an order of magnitude it only requires a linear increase in the number of signals sent.

In addition to tuning the security parameters the error threshold must also be decided. Recall that if Alice and Bob see an error ratio or error rate beyond a certain threshold they should abort the protocol. This value is calculated as the highest error rate such that there is still a positive lower bound to the number of bits of key that can be extracted using privacy amplification. The calculation of the error threshold is dependent on the particular protocol and its security proof.

\section{Security Proof Methods} \label{sec:methods}

There are many different ways to prove secrecy in QKD. While methods started as specific techniques that were restricted to specific protocols, more general techniques exist today. However, the various techniques of proving secrecy in QKD are still highly dependant on the structure of the protocol and what kind of assumptions are made. The resulting security is then dependent on these assumptions. That is, if an experimentalist would like to use a security proof for a given experimental setup, they should be able to justify the assumptions that are made in the security proof. If they cannot be justified, then it leaves a security loophole: an attacker may exploit the devices or sub-protocols that do not behave according to the assumption made and break the security of the protocol. These kinds of attacks are called \emph{side-channel attacks}. We will examine these in Chapter~\ref{chap:assumptions}.

Therefore, it is important to keep in mind that security is proved under certain assumptions. These assumptions can be grouped into what we call a \emph{model} for the protocol. Many of the techniques for proving security apply to various models and so we list various classes that help identify which techniques apply to which models (Section~\ref{sec:classes}). Note that almost all of the the classes of protocols listed below can use the classical post-processing steps outlined above in order to prove security because the classical post-processing usually does not require any information about where the classical data comes from.

When security proofs are presented in the literature, often there is a plot of a lower bound on the key rate that accompanies the proof. The key rate is the ratio of the number of bits of secure key that are extracted per signal sent. Plots are usually of the log of the key rate versus the error rate since the log of the key rate typically follows a linear dependence followed by an exponential drop off as the error rate increases.

In the finite-key regime the number of bits of secure key is plotted against the number of signals sent with a fixed error rate instead. This key rate asymptotically approaches the infinite-key regime's key rate as the number of signals becomes very large.

The lower bound on the key rate is a measure of how good a protocol is compared to others and ideally this bound is made as high as possible. The maximum for discrete protocols is upper bounded by the maximum amount of information that can be measured from the sent quantum states, called the Holevo bound \cite{holevo73}. There have also been investigations into the upper bounds of various key rates by analyzing particular attacks on protocols that Eve could do and plotting the resulting key rate as a function of the error ratio due to the attack. For example, there are upper bounds to the DPS \cite{gomez-sousa09} and COW protocols \cite{branciard08a}, as well as BB84 with different kinds of assumptions \cite{moroder06,moroder06a,curty09a}.

Assumptions are also important for the interpretation of the upper and lower bounds on key rates. While comparing different security proofs it can be misleading to only compare their rates, as there may be a tradeoff between how many assumptions are made and the key rate. If many assumptions are made, then the key rate may be high but if less assumptions are made, the key rate may be lower.

In order to clarify which assumptions are being made, we first list various properties of protocols, which we call protocol classes, in order to distinguish which proof techniques apply to which scenarios.

\subsection{QKD Protocol Classes} \label{sec:classes}

One model class is whether security is proven in the \emph{device-dependent} scenario or the \emph{device-independent} scenario. The device-dependent scenario assumes that devices are characterized. For example, a measurement device may be described by a known set of POVM elements, or a source may output states of a particular form. On the other hand, the device-independent scenario does not make assumptions about the structure of the measurement devices. There is even another regime in-between these two in which some devices are characterized and some are not characterized. We call this scenario the \emph{partially-device-independent} scenario.

Another class distinction is whether the protocol is run to produce an \emph{infinite key} or a \emph{finite key}. Sometimes a QKD protocol may be considered in the asymptotic case, where the protocol is run for an infinite time in order to produce an infinitely-long key. While this is not a practical assumption, it is helpful to consider it for several reasons. First, the asymptotic scenario usually simplifies the analysis, which makes it easier to show a protocol is at least secure in principle. Second, it can be helpful to compare the asymptotic behaviour of various protocols to one another to see which is most efficient in the error rate. However, for a protocol to be secure for practical purposes it is important to consider the finite-key regime.

Yet another distinction are the kinds of states which are used in the protocol, such as qubits, distributed phases, and continuous variables. The first two are described in finite-dimensional Hilbert spaces, while the third uses infinite-dimensional Hilbert spaces. For distributed phase protocols, a large global state that cannot be decomposed into qubits is sent from Alice to Bob. For example, information that Alice is trying to send to Bob can be encoded in the relative phase between a sequence of pulses. Continuous-variable protocols use squeezed or coherent states of light. Note that this distinction will be used to classify how the protocol, in principle, should be implemented and not whether the states are actually assumed to be implemented as intended. This assumption will be further discussed in Chapter~\ref{chap:assumptions}.

QKD protocols can also be broken down into protocols with a basis choice and those without one. A basis choice refers to whether the measurements and prepared states are decomposed into different bases or not. For example, a measurement device may not be passive, but it requires a random input to pick a basis for each measurement it performs (i.e.~it is active).

Protocols may have one of two structures: \emph{entanglement based} or \emph{prepare and measure} (P\&M). Entanglement based protocols involve the preparation of entangled states usually by an untrusted source, such as Eve, and Alice and Bob both do measurements on that state. A P\&M protocol is one where states are prepared by Alice, she sends them through an insecure quantum channel, and the state is measured by Bob. There are other protocols that do not follow this structure, though we do not consider them here. See Section~\ref{sec:two_protocols} for two examples.

Finally, Eve may attack the protocol either individually, collectively, or coherently (see Section~\ref{sec:eve_attack}).

In summary there are seven classes we consider: the device class (dependent, independent, or partially independent), the key class (infinite or finite), the state class (qubits (or another finite-dimensional Hilbert space), distributed phase, and continuous variable), the basis class (basis choice or no basis choice), the measurement class (active or passive), the type class (entanglement based or P\&M), and the attack class (individual, collective, or coherent). Note that three of the classes (device, key, and attack) are dependent on the assumptions made, while the other four (state, basis, measurement, and type) refer to a protocol's structure.

We now divide the proof methods into the device-dependent and device-independent scenarios. The partially device-dependent scenario will be discussed with the device-dependent scenario.

\subsection{The Device-Dependent Scenario} \label{sec:device_dependent}

There are many different techniques used to prove security. Some only apply to a specific protocol, while other techniques are more generic. The first security proofs of QKD proved that the accessible information between Alice's key and Eve's information was small (see Section~\ref{sec:historical}). However, since this is not a definition that is composable, we will only glance over the historical techniques that have been used to prove this kind of security.

\subsubsection{Historical Methods}

Many of the first proofs of QKD, which were for the BB84 protocol, exploited the specific structure of the states used in the protocol \cite{lo99,shor00,mayers01}. The idea behind the proof of \cite{lo99} was to use quantum error correcting codes on the states sent from Alice to Bob. This proof was simplified in \cite{shor00} to show that the quantum error correcting code does not need to be implemented, since the error correcting code commutes with Bob's measurement. Instead, he can use a classical error correcting code after his measurement. The proof of \cite{mayers01} is quite involved, so we omit a discussion of its method here.

These early proofs assumed the exact structure of the states and measurements performed (or quantum error correction, in the case of \cite{lo99}). With more recent techniques, we can prove universally composable security and also relax the kinds of strict assumptions that were made in these early proofs.

\subsubsection{Current Methods} \label{sec:current_methods}

A more recent proof technique is due to Devetak and Winter \cite{devetak05}. This proof technique applies to the infinite-key regime and for the case where Eve is restricted to collective attacks. The Devetak-Winter technique gives an explicit expression for a lower bound on the key rate, $r$, given outcomes from Alice's and Bob's raw keys $K_A$ and $K_B$, and Eve's system before measuring, $E$. The bound on the rate is usually written as
\begin{equation} \label{eq:DW1}
r \geq I(K_A:K_B) - \chi(K_A:E),
\end{equation}
where $\chi(K_A:E):=H(E) - \sum_{k_A}p(k_A)H(E|K_A=k_A)$ (with $p(k_A)$ is the probability of Alice getting key $k_A$) is the Holevo quantity. The Holevo quantity is really just the mutual information of the CQ state shared between Alice and Eve, since $\sum_{k_A}p(k_A)H(E|K_A=k_A)=H(E|K_A)$ and $H(E)-H(E|K_A) = I(K_A:E)$.

One way to prove security has been to exploit the explicit form of the protocol. For example, the entropy involving Eve's system in the Holevo quantity can be reduced to quantities that only contain Alice and Bob's quantum states or their measurement outcomes.

If the state shared between Alice, Bob, and Eve is pure, then Eve has more power than if their state was mixed. This fact is due to the data-processing inequality. Since the partial trace is a CPTP map, Eve has more information if her system is a purification of Alice's and Bob's systems instead of an extension of their state that is not pure (see \cite{tomamichelthesis} for the formal definition of an extension).

Therefore, without loss of generality we can say that the shared state before measuring is $\rho_{ABE}=\kb{\Psi}{\Psi}$, which implies that $H(AB)=H(E)$ (see Section~\ref{sec:iid_entropy}). The second term in $\chi$ can be estimated in a similar way if Alice's measurement is a rank one POVM. If this is the case, then the state between Bob and Eve conditioned on Alice's measurement outcome but before Bob and Eve measure is
\begin{equation}
\rho_{BE}^{k_A} = \frac{1}{\Pr[K_A]} \Tr_{A}(F_{A}^{k_A}\kb{\Psi}{\Psi}),
\end{equation}
which is pure. To see that this state is pure, first note that if a normalized state, $\sigma=\kb{\phi}{\phi}$ is pure then $\Tr(\sigma^2)=|\bk{\phi}{\phi}|^2=1$. Since $F_A^{k_A}$ is rank one, we can write it as $\kb{\phi^{k_A}}{\phi^{k_A}}$. Using the cyclicity of the trace \cite{nielsen00}, we get
\begin{align}
\Tr((\rho_{BE}^{k_A})^2) &= \frac{1}{\Pr[K_A]^2} \Tr \left(\Tr_{A} \left(\kb{\phi^{k_A}}{\phi^{k_A}}\kb{\Psi}{\Psi}\right)^2\right) \\
&= \frac{1}{\Pr[K_A]^2} \Tr \left((\bk{\Psi}{\phi^{k_A}}\bk{\phi^{k_A}}{\Psi})^2\right) \\
&= \frac{1}{\Pr[K_A]^2} (\bk{\Psi}{\phi^{k_A}}\bk{\phi^{k_A}}{\Psi})^2 = 1.
\end{align}
Since the state is pure, we can use the same trick as with the first term of $\chi$ to get $H(E|K_A)=H(B|K_A).$ Now the bound on the key rate can be written entirely with entropies involving Alice's and Bob's systems.

Another way to use the Devetak-Winter rate, Eq.~\ref{eq:DW1}, is to not write it in term of a difference of mutual informations, but instead write it as
\begin{equation}
r \geq H(K_A|E) -H(K_A|K_B).
\end{equation}
In this form, the bound on the rate has an intuitive interpretation: the amount of key Alice and Bob can get is just the difference between the amount of uncertainty that Eve has about Alice's key and the amount of uncertainty Bob has about Alice's key. If Eve has more uncertainty than Bob then the rate may be positive, but if Eve has more information than Bob then the rate cannot be positive.

Using the method of types (see Section~\ref{sec:ik_pe}) we can upper bound $H(K_A|K_B)$ using the binary entropy function of the error rate, $h(q)$.

Now we need to lower bound $H(K_A|E)$, which can be accomplished in a number of ways. If the state structure is assumed (e.g.~if qubits are assumed to be used) then the symmetry in the given protocol can be exploited to bound $H(K_A|E)$. See \cite{ferenczi13,rennerphd} for a detailed description of how symmetry can be used to prove security. If the dimensions of the states are assumed to be low then a brute-force search could be done through the Hilbert space to see which state gives Eve the most information that is compatible with a given error rate.

If there is no assumption made about the structure of the states used in the protocol, then there is another technique: the uncertainty relation for entropies (Theorem~\ref{thm:ur}). This uncertainty relation applies to the min- and max-entropy and therefore is relevant for the finite-key scenario. Using the QAEP this uncertainty relation can be used for the infinite-key scenario as well \cite{berta09}. This uncertainty relation is restricted to the case of entanglement-based protocols with two basis choices where one basis is used for the key, $X$, and one is used for parameter estimation, $Z$. This uncertainty relation puts a lower bound on the min-entropy of Alice's string conditioned on Eve's state:
\begin{equation}
H_{\min}^{\epsilon}(K_A^X|E) \geq \log \frac{n}{c} - H_{\max}^{\epsilon}(K_A^Z|B),
\end{equation}
where $c=\max_{x,z}\|\sqrt{F_{x}}\sqrt{G_{z}}\|_{\infty}^2$ is the \emph{overlap} between two measurements $F$ and $G$ that Alice could perform on her system, $n$ is the number of signals sent and measured by Bob, $B$ is Bob's system before he measures, $E$ is Eve's system, and $K_A^X$ and $K_A^Z$ are Alice's outcomes to these measurements. The lower bound can be simplified with the data-processing inequality by using the fact that Bob's measurement is in the same basis as Alice's: $H_{\max}^{\epsilon}(K_A^Z|B) \leq H_{\max}^{\epsilon}(K_A^Z|K_B^Z)$. Then Alice and Bob need to estimate this quantity in parameter estimation.

In order to use this uncertainty relation they need to have some assumptions about the measurements used in the protocol, namely that the overlap is known and each measurement is done independently (see Section~\ref{sec:measurements}).

To see how the uncertainty relation can be used to prove security, it is useful to consider two thought experiments (sometimes also called by the German term \emph{gedankenexperiment}). The actual experiment has one basis chosen with probability $p_x$ and the other with probability $p_z=1-p_x$. The thought experiments are the same as the actual protocol, but while choosing the bases in the same way, it turns out that all of the measurements happen to be in the $Z$ basis or all in the $X$ basis. We call these thought experiments the $Z$-basis thought experiment and the $X$-basis thought experiment respectively. Let Alice's and Bob's strings from the $Z$-basis thought experiment be $K_A^Z$ and $K_B^Z$ respectively, while in the $X$-basis thought experiment they are $K_A^X$ and $K_B^X$.

Recall that in parameter estimation Alice and Bob will communicate a subset of their strings (denoted with the size of this subset, $k$), from which they can estimate the max-entropy of their complete strings (denoted with $N=n+k$, where $n$ is the size of the string that is not communicated). Alice and Bob can estimate the max-entropy of the $Z$-basis thought experiment, $H_{\max}^{\epsilon}(K_A^Z|K_B^Z)$, using their communication of the subset $k$ of their strings from the actual experiment (as in Section~\ref{sec:parameter_estimation2}) since these signals were measured in the same basis. Then the uncertainty relation using this max-entropy puts a bound on $H_{\min}^{\epsilon}(K_A^X|E)$ for the $X$-basis thought experiment.

In the actual experiment Alice has used a fraction $\nu=k/N$ of her string for parameter estimation so she really wants a lower bound on $H_{\min}^{\epsilon}((K_A^X)_n |E)$ from the actual experiment for the $n$ bits she has kept to construct her key. There is a generalization of the data processing inequality that relates this min-entropy to the min-entropy of the second thought experiment (Theorem~5.7 in \cite{tomamichelthesis}) that gives us
\begin{equation}
H_{\min}^{\epsilon}((K_A^X)_n |E) \geq H_{\min}^{\epsilon}(K_A^X|E).
\end{equation}
This means that in the actual protocol $H_{\min}^{\epsilon}((K_A^X)_n|E)$ is lower bounded.

Note that Eve's system also contains the classical information that is communicated through the authenticated classical channel, which also needs to be taken into account in the security proof (see \cite{tomamichel12a} for an example).

In the infinite-key case Alice and Bob can get perfect statistics about their strings, and can therefore estimate $H(K_A^Z|K_B^Z)$ perfectly using the method of types (see Lemma~\ref{lemma:mot}).

The uncertainty relation has been used for security proofs of the BB84 protocol and two two-way protocols (see Section~\ref{sec:two_protocols}) \cite{tomamichel12a,beaudry13}. The uncertainty relation also has a continuous-variable version that can be used to prove security for CV QKD protocols \cite{furrer11,furrer12,berta13,furrer14a,furrer14b}. In addition, the uncertainty relation can be applied to P\&M protocols as well as entanglement-based protocols by showing an equivalence between them (see Section~\ref{sec:connection} and Section~\ref{sec:two_protocols}).

Most other techniques used to prove security of QKD to date exploit the structure of the states and/or measurements used in the protocol. As examples of security proofs that use these descriptions, the proofs of the B92 protocol \cite{tamaki03,tamaki04,koashi04,tamaki09}, many of the early proofs of the BB84 protocol \cite{lo99,mayers96,mayers01,shor00,koashi03,gottesman04,renner05,renner05a,kraus05}, and the single-photon security proofs of the DPS and COW protocol \cite{wen09,walenta13}. This assumption about the state structure makes it difficult to discuss a general strategy and so we omit the discussion of these kinds of techniques.

\subsection{Reductions} \label{sec:reductions}

Reductions in QKD protocols simplify the problem of proving security against any possible attack by an eavesdropper (i.e.~coherent attacks) to a reduced class of attacks, e.g.~collective attacks. These reductions require assumptions about the structure of the protocol.

There are two known reductions that reduce coherent attacks to collective attacks: the exponential de Finetti theorem of Renner \cite{renner07,rennerphd} and the post-selection technique \cite{christandl09a,renner10}. However, the exponential de Finetti theorem is less efficient than the post-selection technique except for the infinite-key regime, where they produce the same results. Therefore, we will focus primarily on the post-selection technique.

These reductions apply to entanglement-based protocols. They assume that the quantum states in the protocol act on a fixed Hilbert space, $\hilbert_Q$, and that the protocol is permutation invariant.\footnote{While the quantum \stage of the protocol and parameter estimation need to be permutation invariant to use these reductions, we will show that information reconciliation and privacy amplification do not need to be permutation invariant.} The first assumption means that each signal sent from Alice to Bob acts on $\hilbert_Q$ so that the total Hilbert space for the whole run of the protocol with $n$ signals is $\hilbert_Q^{\otimes n}$. This means that we are also assuming that Eve is restricted to sending Alice and Bob joint states in $\hilbert_Q$ for each signal. The second assumption, that the protocol is permutation invariant, means that for any permutation $\Pi$ of the input states of the protocol there exists a CPTP map $G_{\Pi}$ such that $G_{\Pi}\circ \mathcal{E} \circ \Pi = \mathcal{E}$ for the CPTP map $\mathcal{E}$ that represents the QKD protocol. A permutation on $\hilbert_Q^{\otimes n}$ is defined by its action on pure tensor product states:
\begin{equation}
\Pi\; \ket{\phi_1}\otimes\ket{\phi_2}\otimes\cdots\otimes\ket{\phi_n} = \ket{\phi_{\Pi^{-1}(1)}}\otimes\ket{\phi_{\Pi^{-1}(2)}}\otimes\cdots\otimes\ket{\phi_{\Pi^{-1}(n)}},
\end{equation}
where $\Pi^{-1}$ is the inverse of the permutation. The map $G_\Pi$ can be thought of as undoing the permutation on the output of the protocol in order to make sure that the outputs of $\mathcal{E}$ and $\mathcal{E}\circ \Pi$ are the same.

The de Finetti theorem \cite{renner07,rennerphd} relates states to approximate de Finetti states. de Finetti states are convex combinations of product states $\sigma_Q^{\otimes n}$ defined as:
\begin{equation}
\int \sigma_Q^{\otimes n} d\sigma_Q,
\end{equation}
where $d\sigma_Q$ is a measure over the set of density operators on $\hilbert_Q$. This measure can be thought of as a probability distribution over quantum states. The de Finetti state can be interpreted as the situation of picking a state according to the measure $d\sigma_Q$ and then the probability of getting a state in an $\varepsilon$-Ball defined by a distance measure between quantum states (Defn.~\ref{defn:ball}) is the same for all such balls with the same radius \cite{renner10}. The norm used to define distance in this case is the Hilbert-Schmidt norm (Defn.~\ref{defn:HSnorm}). We will now focus on the post-selection technique, instead of the de Finetti theorem. To see how the de Finetti theorem can be used for QKD, see \cite{renner07,rennerphd}.

\subsubsection{The Post-Selection Technique}

The post-selection technique is so named because a permutation invariant state can be extracted from a fixed state by post-selecting on a particular measurement \cite{christandl09a}. This situation is used in the proof technique but we do not discuss the proof of the post-selection technique here. In this section we will outline what the post-selection technique is and how it can be used in quantum cryptography.

Note that we can write the security criterion of a QKD protocol as a map acting on the initial shared state between Alice, Bob, and Eve. If we combine Alice's and Bob's systems into $AB=Q$ and have $\mathcal{E}$ and $\mathcal{F}$ be the maps representing the real protocol and the ideal protocol respectively, then the security definition (Defn.~\ref{defn:security}) can be written as
\begin{equation}
\Delta(\mathcal{E},\mathcal{F})_{\rho_{Q^n}} := \| \mathcal{E}\otimes\id(\rho_{Q^nE})-\mathcal{F}\otimes\id (\rho_{Q^nE})\|_{1} \leq \varepsilon,
\end{equation}
where $\rho_{Q^n}$ is the state of the protocol, $n$ is the number of signals sent in the protocol, $\rho_{Q^nE}$ is a purification of $\rho_{Q^n}$, and $E$ is Eve's system before the classical post-processing.

Now we can state the post-selection theorem as it applies to QKD. 

\begin{thm}[Post-selection theorem for QKD, Lemma~4 in \cite{renner10}] \label{thm:post_selection} Let $\mathcal{E}$ and $\mathcal{F}$ be any permutation invariant CPTP maps. Then for any $\rho=\rho_{Q^n}$
\begin{equation} \label{eq:PS}
\Delta(\mathcal{E},\mathcal{F})_{\rho} \leq (n+1)^{d_Q^2-1}\Delta(\mathcal{E},\mathcal{F})_{\tau},
\end{equation}
where $d_Q$ is the dimension of $\hilbert_Q$ and $\tau\equiv\tau_{Q^n}\in S_{=}(\hilbert_Q^{\otimes n})$ is the de Finetti state for $\hilbert_Q^{\otimes n}$.
\end{thm}

This theorem implies that instead of considering general states in the protocol, $\rho$, we can consider the de Finetti state $\tau$. Note the state $\tau$ is a fixed state. Using this theorem adds a factor of $(n+1)^{d_Q^2-1}$ to the security parameter. However, the security parameter is usually exponentially dependent on the number of signals (i.e.~$\varepsilon\sim 2^{-cn}$ for a constant $c$, see Section~\ref{sec:tuning}). This means that the polynomial factor does not change the security by much, since a logarithmic decrease (in the number of signals, $n$) in the final key length during privacy amplification can restore the same level of security as what would be possible without using this technique.

The post-selection theorem can be shown to imply that Eve gets virtually no advantage to attacking permutation invariant protocols using coherent attacks instead of collective attacks. It can be much easier to prove security of a QKD protocol by assuming that an i.i.d.~state is used (of the form $\sigma^{\otimes n}$), which is the case for collective attacks. In particular, security can usually be proved \emph{for all} i.i.d.~states, of which each state is in a fixed Hilbert space, $\hilbert_Q$. This kind of proof implies that any convex combination of i.i.d.~states must also be secure and therefore the de Finetti state must be secure.

Note that for product states, Eve will hold a purification of each subsystem independently. However, the post-selection theorem applies to the purification of the de Finetti state $\tau$, not to the purification of each subsystem independently. In \cite{christandl09a} the authors show that the purifying system of the de Finetti state $\tau$ has a dimension that is polynomial in $n$ (specifically, $(n+1)^{d_Q^2-1}$), which means that by doing polynomially more privacy amplification this extra information may be removed from Eve. Therefore to apply the post-selection technique together with the removal of the information Eve gets from her purification of the de Finetti state, $2(d_Q^2-1)\log(n+1)$ bits need to be removed in privacy amplification.

The post-selection technique can be used in continuous-variable QKD as well \cite{leverrier13}, though an analysis of this application is beyond the scope of this thesis.

\subsubsection{Post-Selection Example}

As an example of an application of the post-selection technique, consider the BB84 protocol in its entanglement-based implementation (see Section~\ref{sec:connection}). First, we decompose the protocol into two parts. The first part of the protocol needs to be permutation invariant, while we show that the second part of the protocol does not necessarily need to be permutation invariant. Consider the quantum \stage, sifting, and parameter estimation together as the first half of the protocol, $\mathcal{E}_1:=\text{PE}\circ\text{Sift}\circ F$, where $F$ is the quantum measurement. Then information reconciliation and privacy amplification will be the sub-protocol $\mathcal{E}_2:=\text{PA}\circ\text{IR}$ that follows $\mathcal{E}_1$.

To show that the BB84 protocol is permutation invariant, first consider $\mathcal{E}_1$. We need to show that there exists a CPTP map $G_\Pi$ such that $G_\Pi\circ\mathcal{E}_1\circ\Pi=\mathcal{E}_1$ for any permutation $\Pi$. Consider $\mathcal{E}_1\circ\Pi$ for a fixed permutation $\Pi$. Assume that Alice and Bob apply the permutation $\Pi$ to their systems and then measure their states in this permuted order. The permutation $\Pi$ will not need to be applied in an implementation of the protocol. We will only assume that Alice and Bob apply this permutation to argue that $\mathcal{E}_1$ is permutation invariant.

If we assume that the measurements on the quantum states of Alice and Bob are memoryless and identical (therefore their POVM elements are of the form $F^{\otimes M}$, where $F$ is a measurement on an individual signal and $M$ is the number of signals sent) then the permutation of the quantum states commutes with their measurements.

The sifting step removes bits from Alice's and Bob's strings where they measured in different bases. The sifting also removes bits of Alice's string where Bob did not get a measurement outcome. The sifting commutes with the permutation since it removes bits independent of their position in Alice's and Bob's strings.

We now know that
\begin{equation}\label{eq:permute}
\text{PE}\circ\text{Sift}\circ F^{\otimes n}\circ \Pi \equiv \text{PE}\circ \Pi\circ \text{Sift}\circ F^{\otimes n},
\end{equation}
for any permutation $\Pi$. We now need to argue that the permutation commutes with parameter estimation.

Note that parameter estimation is just the choice of a random subset of Alice's and Bob's strings that are communicated through the authenticated classical channel and removed from their strings (as well as an estimation procedure based on this communication). This means we can decompose parameter estimation into three parts: a choice of a random subset, the removal of the subset, and the estimation. Formally, we have the decomposition
\begin{equation}
\text{PE} \equiv \text{Estimation} \circ \text{Removal} \circ \text{Subset}.
\end{equation}

The choice of a random subset of Alice's and Bob's strings is equivalent to first applying a random permutation to Alice's and Bob's strings followed by the choice of the first $k$ bits of the string for the sample and then the inverse of the permutation. However, the communication of the positions of Alice's string will be different in $\text{PE}$ compared to $\text{PE}\circ\Pi$ since the positions of the bits are permuted. However, a classical transformation can be applied that undoes the permutation on the positions that are communicated. Formally, if the positions communicated in $\text{PE}$ are elements of a set $\{v_1,v_2,\dots,v_k\}$ then in $\text{PE}\circ\Pi$ the positions communicated are $\{\Pi(v_1),\Pi(v_2),\dots,\Pi(v_k)\}$. By applying the inverse permutation to each position, the original communication of $\text{PE}$ can be recovered. Therefore, the choice of random subset is permutation invariant.

The removal procedure is the removal of the randomly chosen subset from Alice's and Bob's strings, which is accomplished by communication of the subset from Alice to Bob (or vice versa). The bits that are removed are the same whether a permutation would be applied to Alice's and Bob's strings or not. Therefore, the removal procedure is permutation invariant.

The estimation procedure uses the communicated subset to do estimation. The estimation is also independent of the ordering of Alice's and Bob's strings, and therefore is permutation invariant.

In summary, the parameter estimation step is permutation invariant: $G_{\Pi}\circ\text{PE}\circ \Pi=\text{PE}$, where $G_{\Pi}$ is the inverse of the permutation $\Pi$ applied to Alice's and Bob's strings as well as the inverse permutation applied to each position communicated. Combining this fact with Eq.~\ref{eq:permute} means that the first half to the protocol, $\mathcal{E}_1$, is permutation invariant, under the assumption that the measurements are of the form $F^{\otimes M}$ and that parameter estimation chooses a random subset of Alice's and Bob's strings.

The above argument gives some insight as to why a random subset is chosen for parameter estimation; the random subset makes the protocol permutation invariant. If instead of a random subset a fixed subset was chosen, then $\mathcal{E}_1$ would not be permutation invariant.

We now focus on showing under which privacy amplification and information reconciliation protocols the post-selection theorem applies. Assume that information reconciliation and privacy amplification are permutation invariant, which defines a sub-protocol $\mathcal{E}'_2$. If we assume that the output state shared by Alice and Eve after $\mathcal{E}'_2$ is invariant under a permutation of the states input to $\mathcal{E}_1$ then $\mathcal{E}'_2\circ\mathcal{E}_1$ is permutation invariant because the map $G_{\Pi}$ that changes the communicated positions in parameter estimation commutes with $\mathcal{E}'_2$. Then we can apply the post-selection theorem to $\mathcal{E}'_2\circ\mathcal{E}_1$. If the protocol is secure, then Eq.~\ref{eq:PS} holds for the protocol $\mathcal{E}=\mathcal{E}'_2\circ\mathcal{E}_1$.

We want to show that we can replace permutation invariant information reconciliation and privacy amplification with non-permutation invariant information reconciliation and privacy amplification with a small cost to the security parameter of the protocol.

An example of a permutation invariant privacy amplification protocol is the one using hash functions described in Section~\ref{sec:privacy_amplification2} that goes with the leftover hashing lemma, Lemma~\ref{lemma:loh}. Recall that in the privacy amplification procedure a random hash function from a family of hash functions is selected by Alice which is then communicated. Alice and Bob then apply the hash function to their strings. If the family of hash functions $\mathcal{F}$ is taken to be the set of all linear functions from $\{0,1\}^n$ to $\{0,1\}^\ell$, then for every permutation $\Pi$ and string $K_A\in \{0,1\}^n$, there exists a unique pairing of every function $\tilde{f}\in\mathcal{F}$ to a function $f\in\mathcal{F}$ such that $\tilde{f}(K_A) = K'$ and $f(\Pi K_A \Pi) = K'$. Since each function is chosen with equal probability, the state shared by Alice and Eve with a permutation is the same as if a permutation was not applied.

The communication of the hash function in $\text{PA}\circ\Pi$ can be made the same as $\text{PA}$ by relabeling the hash functions. Since a pairing exists between the functions $f$ of $\text{PA}$ and the functions $\tilde{f}$ of $\text{PA}\circ\Pi$, a map can be applied to $\text{PA}\circ\Pi$ that relabels the function $\tilde{f}$ as $f$ if $\tilde{f}$ is communicated by Alice to Bob. After the relabeling of the communication and since the state shared by Alice and Eve is the same if the permutation was applied or not, there exists a permutation invariant privacy amplification protocol that commutes with the communication relabelling of $G_{\Pi}$.

Since we know that the protocol $\mathcal{E}'_2\circ\mathcal{E}_1$ is secure, this implies that the min-entropy must be at least a certain amount, otherwise there is no privacy amplification protocol that could succeed with at least probability $\varepsilon$. Theorem~8.2 in \cite{tomamichelthesis} says that if there is security at least $\varepsilon$ then the min-entropy before privacy amplification should be at least the size of the output string, $\ell'$:
\begin{equation} \label{eq:post_select_min}
\ell' \leq H_{\min}^{\sqrt{2\varepsilon-\varepsilon^2}}(K_A|E)_{\rho}.
\end{equation}
Since this bound is guaranteed, then we can apply another privacy amplification procedure (such as leftover hashing) that is not permutation invariant using the fact that this min-entropy is at least $\ell'$. By applying the leftover hashing lemma (Lemma~\ref{lemma:loh}) using Eq.~\ref{eq:post_select_min} the security statement is now
\begin{equation}\label{eq:post_security1}
D\left( \rho_{K'EF}, \frac{\identity_{K'}}{d_{K'}}\otimes\rho_{EF} \right) \leq \sqrt{2\varepsilon-\varepsilon^2}  + 2^{-\frac{1}{2}(\ell'-\ell)-1},
\end{equation}
where $\ell$ is the size of the output length of the string from the privacy amplification hash function.

A similar argument can be used for information reconciliation as with privacy amplification to show that we do not need a permutation invariant information reconciliation protocol. As in Section~\ref{sec:error_correction2}, there are two possible non-permutation invariant information reconciliation protocols that we can use.

The first of the two information reconciliation protocols must be universal so that it corrects errors for almost all strings that Alice and Bob could have. For example, an error correcting code exists in \cite{renes12} that corrects all errors with probability at least $1-\varepsilon_{\text{c}}$, so the amount of communication necessary for Alice to send Bob to achieve this probability is at least
\begin{equation}
\mathcal{C}^{\varepsilon_{\text{c}}} \geq H_{\max}^{\sqrt{2\varepsilon_{\text{c}}-\varepsilon_{\text{c}}}}(K_A|K_B),
\end{equation}
\cite{renes12} (Theorem 8.1 of \cite{tomamichelthesis}, also see \cite{renner05b}). With this bound on the max-entropy, we can apply another error correcting code instead that is not permutation invariant, such as the same one we have already used. This gives a bound on the amount of communication of
\begin{equation}
\mathcal{C}^{\varepsilon'} \leq \mathcal{C}^{\varepsilon_{\text{c}}} +2\log \frac{1}{\varepsilon_2}+4,
\end{equation}
where $\varepsilon' = \sqrt{2\varepsilon_{\text{c}}-\varepsilon_{\text{c}}} + \varepsilon_2$ is the upper bound on the failure probability of the error correcting code \cite{renes12}. The specific security statement can be calculated using the amount of communication $\mathcal{C}^{\varepsilon'}$ and the probability that the error correction succeeds, $\varepsilon'$, by combining the chain rule Eq.~\ref{eq:classical_chain} and Eq.~\ref{eq:post_security1}:
\begin{equation}
\Delta(\mathcal{E},\mathcal{F})_{\tau} \leq \sqrt{2\varepsilon-\varepsilon^2}  + 2^{-\frac{1}{2}(\ell' - \mathcal{C}^{\varepsilon'}-\ell)-1}+\varepsilon'.
\end{equation}
Using the post-selection theorem, the security parameter for all states $\rho$ is
\begin{equation}
\Delta (\mathcal{E},\mathcal{F})_{\rho} \leq (n+1)^{d_Q^2-1}\left(\sqrt{2\varepsilon-\varepsilon^2}  + 2^{-\frac{1}{2}(\ell' - \mathcal{C}^{\varepsilon'}-\ell)-1}+\varepsilon'\right).
\end{equation}

The second type of non-permutation invariant information reconciliation protocol is an error correcting code followed by a checking procedure as explained in Section~\ref{sec:error_correction2}. The checking procedure guarantees that we have corrected all of the errors with probability $1-\varepsilon_{\text{cor}}$ with $\lceil\log(1/\varepsilon_{\text{cor}})\rceil$ bits of communication and may increase the probability that the protocol aborts. The aborting probability, and hence the robustness, will depend on the particular choice of error correcting code.

The length of the communication in error correction, $\lceil\log(1/\varepsilon_{\text{cor}})\rceil$, should be taken into account in the privacy amplification analysis by using the chain rule Eq.~\ref{eq:classical_chain}. Using Eq.~\ref{eq:post_security1} the security parameter is
\begin{equation}
\Delta(\mathcal{E},\mathcal{F})_{\tau} \leq \sqrt{2\varepsilon-\varepsilon^2}  + 2^{-\frac{1}{2}(\ell' - \lceil\log(1/\varepsilon_{\text{cor}})\rceil -\ell)-1}+\varepsilon_{\text{cor}}.
\end{equation}
Combining this with the post-selection theorem, the security parameter for all states $\rho$ is
\begin{equation}
\Delta (\mathcal{E},\mathcal{F})_{\rho} \leq (n+1)^{d_Q^2-1}\left(\sqrt{2\varepsilon-\varepsilon^2}  + 2^{-\frac{1}{2}(\ell' - \lceil\log(1/\varepsilon_{\text{cor}})\rceil -\ell)-1}+\varepsilon_{\text{cor}}\right).
\end{equation}

For another example of using the post-selection technique to prove security of a QKD protocol, see \cite{sheridan10}.

A further reduction may be applied to a security proof that assumes a product state $\sigma^{\otimes n}$ by using representation theory and symmetries in the protocol. For example, the BB84 protocol is invariant under permutations of the states $\{\ket{0},\ket{1},\ket{+},\ket{-}\}$ to $\{\ket{+},\ket{-},\ket{1},\ket{0}\}$.\footnote{For those familiar with the Bloch sphere, this symmetry is just a rotation by $\pi/2$ in the $X-Z$ plane \cite{nielsen00}.} These kinds of symmetries imply that $\sigma_Q$ should be of a simple form that can either be completely fixed by the parameters in the protocol (such as the error rate) or only depend upon a few free parameters \cite{rennerphd}. If there are free parameters in $\sigma_Q$ then a minimization over the free parameters of the key rate can then be performed.

\subsection{Entanglement-Based and P\&M Connection}\label{sec:connection}

It can be useful to connect a P\&M protocol with an entanglement based one, since some proof techniques require an entanglement based protocol (such as the uncertainty relation in Section~\ref{sec:device_dependent}). The connection works by transforming the P\&M protocol to an entanglement-based protocol that gives more power to Eve.

In Section~\ref{sec:discrete_protocols} it was shown that the BB84 P\&M protocol can be related to the Ekert entanglement-based protocol, but this connection was under the assumption that the protocols were ideal. However, this assumption can be relaxed to the assumption that the preparation of states are qubits (see Section~\ref{sec:sources}). No assumptions need to be made about the measurements or other components in the protocol to make this connection.

Alternatively, Alice can just prepare a bipartite state (which ideally would be maximally entangled) and measure half of it. Depending on her measurement outcome, she will infer which quantum state she is sending to Bob. The protocol is then clearly entanglement-based, except Alice is preparing the state instead of Eve. Some proof techniques that characterize the states or dimensions of the protocol can assume that Alice's prepared state is known, which may aid in proving security. If this assumption is not made then this protocol is more pessimistic if it assumes that the bipartite state is prepared by Eve instead of Alice. If security is proved in the scenario where Eve prepares the state then it implies security for the protocol where Alice prepares the state.

For different protocols it may be necessary to have more assumptions about the P\&M protocol to transform it into an entanglement-based one, though this transformation will depend upon the security proof technique and structure of the protocol. See \cite{beaudry13} for an example.

\subsection{The Device-Independent Scenario} \label{sec:DI}

The proof methods used in the device-independent scenario are different than those used in the device-dependent scenario. These kinds of protocols do not rely on the structure of the states or devices, they just try to establish that Alice and Bob have strong correlations between their states. Proving that strong correlations exist is a more challenging task since there is no symmetry that can be exploited in the protocol's states, sources, or measurements. Intuitively, if these correlations are strong enough, by the monogamy of entanglement Eve cannot have strong correlations with either Alice or Bob. The strength of the correlations is usually measured using the CHSH inequality \cite{clauser69}, though other inequalities have also been considered \cite{H13}. For more information about entangled states and strongly correlated quantum systems, see the recent review \cite{brunner14}.

\subsubsection{The CHSH Inequality} \label{sec:CHSH}

One way of determining if strong correlations are shared between Alice and Bob is to use the Clauser-Horne-Shimony-Holt (CHSH) inequality \cite{clauser69}. This inequality is a particular example of a Bell inequality \cite{bell64}. The CHSH inequality, when violated (i.e.~when the inequality is false), indicates that the bipartite states involved must be correlated in a way that cannot be explained by using what is called a \emph{local hidden variable} theory. This is a theory where there is a variable that describes properties of each particle locally. While classical systems can be described using a local hidden variable theory, there exist quantum states that cannot be described in this way.

Also, there is a maximum violation that the inequality can reach by quantum states. In particular, the higher the violation is, the more correlated the states are. Just like we used an error rate in the device-dependent scenario, we can use an estimate of the amount of violation to quantify how much privacy amplification and information reconciliation is necessary. Since these steps only depend on classical strings, proving security in the device-independent setting can also be reduced to putting bounds on the relevant min- and max-entropies.

The experiment in which the CHSH inequality applies involves two space-like separated measurement devices\footnote{Two devices are space-like separated if they are outside each other's light cones, so that performing a measurement in each device cannot send signals to the other.} with two binary inputs and two binary outputs (see Fig.~\ref{fig:CHSH_experiment}). It has been shown that quantum states can violate the CHSH inequality \cite{freedman72,aspect81a,aspect82,aspect82a,weihs98,tittel98,rowe01a,pomarico11a,stuart12a,giustina13,christensen13}.

The CHSH experiment is usually presented either through expectation values of observables or as a game \cite{cleve04}. While these are both equivalent presentations, they may be helpful to understand how the CHSH inequality works depending whether one approaches the problem from a physics or computer science/mathematical point of view.

\begin{itemize}
\item CHSH: Expectation Values.

Alice and Bob each have a measurement device and are allowed to input bits $x$ and $y$ respectively to get outcomes $a$ and $b$. We do not need to characterize the states that they share and input into the measurement devices, we will only care about the expectation values for the two possible measurements they perform. If we define the observables that Alice and Bob measure as $F_x$ and $F_y$ respectively with eigenvalues $\{1,-1\}$ then we can define the product of the expectation values as
\begin{equation}
E(x,y):=\left<F_x\cdot F_y\right>.
\end{equation}
This notation allows us to state the CHSH inequality as
\begin{equation}\label{eq:CHSH_exp}
\left|E(0,0)+E(0,1)+E(1,0)-E(1,1))\right| \leq 2,
\end{equation}
where the upper bound of $2$ refers to what is possible by local hidden variable theories. The maximum allowable quantum bound is $2\sqrt{2}$ \cite{tsirelson80}.

\begin{figure} \centering
\includegraphics[width=\textwidth]{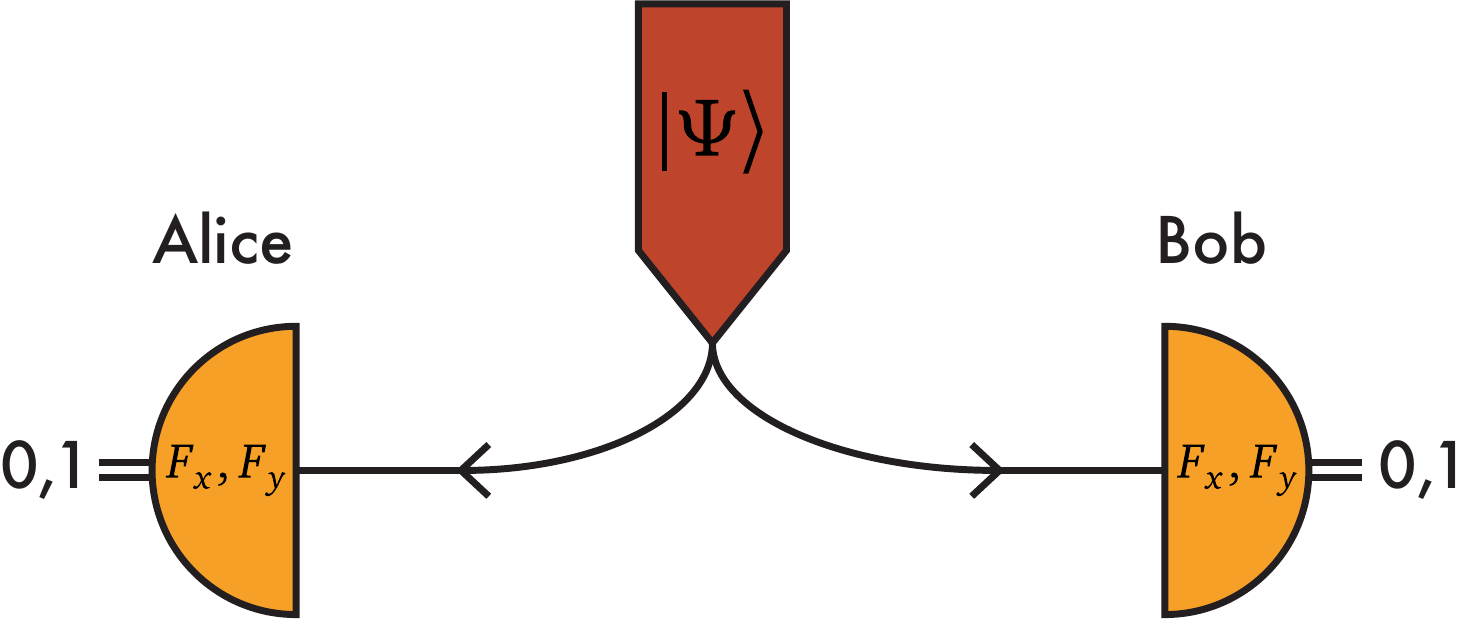}
\caption[The CHSH Experiment]{The CHSH experiment. Alice and Bob measure a bipartite state $\ket{\Psi}$ by choosing the set of POVM elements $\{F_x\}$ or $\{F_y\}$ uniformly at random. Alice and Bob get outcome $0$ or outcome $1$, which they can use to check the CHSH inequality, Eq.~\ref{eq:CHSH_exp}.}
\label{fig:CHSH_experiment}
\end{figure}

\item CHSH: Game.

Alice and Bob each receive uniformly random binary inputs from a referee and have to send binary outputs back to the referee (see Fig.~\ref{fig:CHSH_game}). Alice's input is labelled as $x$ and Bob's input is labelled as $y$, while their outputs are labelled as $a$ and $b$ respectively. Alice and Bob can discuss a strategy before starting the game but then they are separated and they cannot communicate during the game. The goal for Alice and Bob is to have $a\oplus b = x \wedge y$, that is, the binary sum of the inputs should equal the logical AND of their outputs. If their strategy is to share a joint physical state that has correlations that will give rise to the conditional probability distribution $P_{AB|XY}(ab|xy)$ then their probability of winning can be stated as 
\begin{equation}
\Pr [\text{win}] = \sum_{\substack{ xyab \\ a\oplus b = x \wedge y}} P_{XY}(xy)P_{AB|XY}(ab|xy).
\end{equation}
The maximum achievable success probability for this game where Alice and Bob only use classical states is $P_{\text{win}}\leq 3/4=0.75$. It was shown by Tsirelson \cite{tsirelson80} that the maximum success probability where Alice and Bob use quantum states is $\Pr[\text{win}] \leq \cos^2(\pi/8) \approx 0.85$.

\begin{figure} \centering
\includegraphics[width=0.8\textwidth]{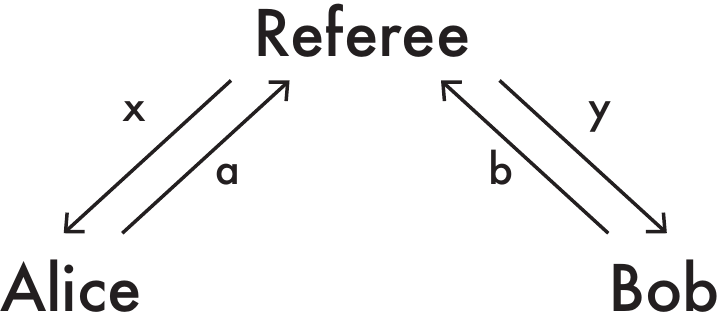}
\caption[The CHSH Game]{The CHSH game. Alice and Bob get two uniformly random bits ($x$ and $y$ respectively) from a referee. Alice and Bob then have to send bits $a$ and $b$ back to the referee such that $a\oplus b = x \wedge y$.}
\label{fig:CHSH_game}
\end{figure}

For further information about the CHSH game and other related games, see the review \cite{buhrman10}.
\end{itemize}

In a device-independent QKD protocol Alice and Bob will have measurement devices that take inputs (which may, for example, ideally pick a basis). Typically Alice and Bob will input uniformly random binary inputs into their measurements. Then Alice and Bob can estimate the number of outputs that satisfy the CHSH condition $a\oplus b = x\wedge y$, called the CHSH violation. They share a subset of their outcomes and can use Serfling's inequality (Lemma~\ref{lemma:serfling}) to bound the total CHSH violation over the remainder of their measurement outcomes.

In order to use the CHSH value to prove security, the following inequality was used in a security proof for a device-independent QKD protocol \cite{vazirani12}. 

\begin{lemma}[CHSH inequality, Eq.~A.10 in \cite{pironio10a}] \label{lemma:CHSH_ineq} Given a conditional probability distribution $q(a,b|x,y)$, a CHSH violation $I$, define $c_{xy}:=\{-1 \text{ if } (x,y)=(1,1),  1 \text{ otherwise}\}$, $d_{ab}:=\{1 \text{ if }  a=b, -1 \text{ if } a\neq b \}$, and the distribution
\begin{equation} \label{eq:max_prob}
q^*(a,b|x,y) = \underset{\overbrace{\begin{array}{c}
\sum_{a,b,x,y}d_{ab}c_{xy}q(a,b|x,y)=I \\
q(a,b|x,y)\geq 0 \\
\sum_{a,b}q(a,b|x,y)=1 \\
\sum_a q(a,b|x,y) = q(b,y) \\
\sum_b q(a,b|x,y) = q(a,x)
\end{array}}}{\max_{a,b}} q(a,b|x,y),
\end{equation}
then the following inequality holds
\begin{equation}\label{eq:CHSH_bound}
q^*(a,b|x,y) \leq \frac{3}{2}-\frac{I}{4}.
\end{equation}
\end{lemma}

Note that the maximum in Eq.~\ref{eq:max_prob} is over probability distributions that result in the CHSH violation observed and that are \emph{non-signalling}. Technically, the non-signalling condition is just the last two conditions in Eq.~\ref{eq:max_prob}, and it means that Alice's choice of input should not influence Bob's measurement outcome and vice-versa. This fact is due to relativity: if Alice's and Bob's measurement devices are space-like separated, then they cannot influence each other. The maximum is taken over these probability distributions in order to show that the upper bound in Eq.~\ref{eq:CHSH_bound} applies regardless of which probability distribution Alice and Bob actually have from their quantum states.

Also note that Eq.~\ref{eq:max_prob} maximizes over probability distributions that are not necessarily allowed by quantum mechanics. For example, they may come from distributions that satisfy the no-signalling conditions and have CHSH value $I$, but can win the CHSH game by more than $85\%$. Therefore, the bound Eq.~\ref{eq:CHSH_bound} may be too pessimistic, since it allows for distributions that may never occur from Alice's and Bob's measurements. However, Eq.~\ref{eq:CHSH_bound} does show a bound even for distributions that come from measurements on quantum systems. It is not trivial to relate the optimal probability distribution $q^*$ to a distribution that comes from quantum systems.

In addition, Eq.~\ref{eq:max_prob} needs to be related to the min-entropy to prove secrecy, which is beyond the scope of this thesis but more information can be found in \cite{pironio10a,vazirani12}.

There are several different ideas that come together to prove security in this setting \cite{acin07,vazirani12}. However, \cite{vazirani12} is currently the only protocol that is robust against noise and losses, and therefore there are no general techniques to date.

\section{Summary}

To conclude this chapter, we review each of the techniques discussed and under which QKD classes (Section~\ref{sec:classes}) they apply to. We also list any classical uniform randomness and/or any communication necessary for the classical post-processing steps (Fig.~\ref{fig:table}).

The privacy amplification step of the classical post-processing reduces the problem of proving secrecy of a QKD protocol to bounding the min-entropy of Alice's string conditioned on Eve's state. Information reconciliation reduces proving correctness of the protocol to performing an error correcting code followed by hashing to check that Alice's and Bob's strings are the same with high probability. Parameter estimation gives a way to estimate the number of errors between Alice's and Bob's strings as well as the max-entropy of Alice's string conditioned on Bob's string using the error rate (or CHSH violation) by having Alice and Bob communicate a small fraction of their strings.

We discussed several reductions, such as reducing the Hilbert space of the quantum signals to a small fixed Hilbert space (the post-selection technique) and relating P\&M protocols to entanglement based ones. Two methods of proving security are the Devetak-Winter rate in the infinite-key limit and an uncertainty relation, which bounds the min-entropy by the max-entropy.

The details of the assumptions needed to apply security proofs to implementations will be analyzed in detail in Chapter~\ref{chap:assumptions}.

\begin{sidewaysfigure}
    \begin{tabular}{c}
    \begin{tabular}{ | c | c | c |} \hline
    Task & Randomness & Communication \\ \hline
    Leftover Hashing & $O(n)$ or $O(l)$ & $O(n)$ or $O(l)$ \\ \hline
    Trevisan's Extractor & $O(\log^2(n/\varepsilon)\log(l))$ & $O(l)$ \\ \hline
    Information Reconciliation & $\text{rand}+\lceil \log (1/\varepsilon_{\text{cor}}) \rceil$\footnote{$\text{rand}$ is the amount of randomness that is communicated during the error correcting procedure, which does not give any information to Eve about the key.} & $\text{leak}+\text{rand}+2\lceil \log(1/\varepsilon_{\text{cor}}) \rceil$\footnote{$\text{leak}$ is any communication Alice sends to Bob that is correlated to her key.} \\ \hline
    Parameter Estimation & $O\left(\log\binom{N}{k}\right)$ & $O(k)$ \\ \hline
    \end{tabular} \vspace{0.5cm} \\ 
    \begin{tabular}{ | c  | c | c |} \hline
    Task & Classes & Assumptions \\ \hline
    Devetak-Winter Rate & \begin{tabular}{c} Infinite Key \\ Basis Choice \end{tabular} & - \\ \hline
    Uncertainty Relation & \begin{tabular}{c} Basis Choice \\ Entanglement Based \end{tabular} & Overlap $c$ \\ \hline
    Post-Selection Technique & Entanglement Based & \begin{tabular}{c} Hilbert Space Known \\ Permutation Invariance \end{tabular} \\ \hline
    Entanglement-Based and P\&M Equivalence & - & Qubits or Entangled States Prepared \\ \hline 
    \end{tabular}
    \end{tabular}
   \caption[QKD Subprotocols and Techniques Summary]{List of subprotocols and techniques used to prove security for QKD protocols. The length of Alice's and Bob's strings at each stage of the protocol is $N$ after any basis sifting, $n$ after parameter estimation, and $l$ after privacy amplification. The randomness and communication needed in the above tasks are the number of bits required. The assumptions are those that are specifically required to use the technique. All methods in the second table are used for device-dependent QKD security proofs.} \label{fig:table} 
  \end{sidewaysfigure}

%% file: Assumptions.tex

\chapter{Assumptions} \label{chap:assumptions}

\section{Introduction} \label{sec:intro_assumptions}

In this chapter we propose a framework that can be used to state assumptions in quantum key distribution and quantum cryptography in a clear and concise way. We provide a detailed list of the kinds of assumptions that are made in order to prove the security of QKD protocols and to connect the models under which security is proven with their implementations. Many of the assumptions in this chapter were previously mentioned in \cite{scarani09b}.

Recall that in Chapter~\ref{chap:intro} we introduced several descriptions of protocols that were implemented in an idealized setting, which we called perfect models (Sections~\ref{sec:discrete_protocols},\ref{sec:continuous_protocols}, and \ref{sec:DIprotocols}). However, there are several gaps between the perfect models of QKD protocols and their actual experimental realizations. Any deviation from the model under which security is proven may leak information to Eve or allow Eve to exploit the devices to gain information. This leakage of information compromises the security of the protocol and could even make the protocol entirely insecure! Therefore, the assumptions made are crucially important to the security of the protocol. It is not enough to prove security for an idealized model; the model must also accurately describe an implementation, otherwise the implementation may not be secure.

Whenever a model and the implementation disagree then Eve may employ \emph{side-channel attacks}: Eve may attack the implementation in a way that is not accounted for in the model.

We begin this chapter with a summary of the ways in which perfect models differ from implementations of QKD protocols.

\begin{itemize}
\item {\bf Lab Isolation.} The models assume that Alice's and Bob's devices are completely isolated so that Eve cannot interact with them in any way. However, since Alice and Bob need to input states and receive states from the quantum channel, they need to have some interface with the channel. If Eve can probe Alice's and Bob's devices through this interface then she may learn something about the measurement outcomes or prepared states.

\item {\bf Source states.} The perfect models in Chapter~\ref{chap:intro} assume the preparation of an exact state. In practice, however, states can only be prepared approximately. The actual prepared state may differ from the intended state in two ways. Either the prepared state is in the same Hilbert space but is not the intended state or the prepared state may be a superposition or mixture of the intended state with other states in other Hilbert spaces. The prepared state may also have a combination of these kinds of errors.

\item {\bf Measurements.} Similarly to source states, when measurements are performed, they may not perform the exact POVM elements that are intended. In addition, they may also measure states outside of the Hilbert space the protocol is designed to measure in. Since the measurements may react to states outside of the intended Hilbert space, Eve can modify the states in the quantum channel to exploit the full Hilbert space available to her.

Measurements may also give outcomes that are non-existent in the perfect protocol. For example, a measurement result could be output when there is no received signal. Conversely, there may be losses in the measurement device or in the quantum channel that result in no measurement outcome when a state was prepared.

\item {\bf Device calibration.} Something that is not considered in the perfect models is that the physical devices may need to be calibrated between Alice and Bob. For example, Alice and Bob may need to establish a shared reference frame before the QKD protocol, which may give Eve information about Alice's and Bob's devices.

\item {\bf Timing.} In addition to the device calibration, Alice and Bob also need to agree on the timing of signals. For example, in a P\&M protocol when Alice sends states to Bob through the quantum channel, Bob needs to know which measurement results correspond to which sent states. Therefore, Alice and Bob also need to fix a timing so that sent states are associated with the correct measurement outcomes. In addition, Bob's measurements are not performed instantaneously. His measurement has a finite measurement time, which Eve may exploit. Also, some measurements have a ``dead time'' where the measurement device will not respond to incoming signals (see Section~\ref{sec:threshold}).

\item {\bf Classical post-processing.} In the classical post-processing steps the estimation of the amount of information that Eve has from the quantum \stage of the protocol from the communication sent in the authenticated public classical channel should be quantified. If this estimation does not incorporate deviations from the model used to prove security then the estimation may be inaccurate, resulting in leaking more information to Eve than what the security proof accounts for. Also, randomness is used for many parts of the protocol. This randomness should be true randomness (see Section~\ref{sec:qc_protocols}), otherwise Eve may be able to make predictions about certain parts of the protocol.
\end{itemize}

To begin this chapter, we discuss the use of the term ``unconditional'' security (Section~\ref{sec:unconditional_security}). Then we classify assumptions into four categories (Section~\ref{sec:assumption_classification}).

After these preliminary sections, we discuss assumptions in quantum cryptography and quantum key distribution that are general (Section~\ref{sec:universal_assumptions}), which includes the foundations of physics (Section~\ref{sec:foundational_assumptions}), the isolation of Alice's and Bob's labs (Section~\ref{sec:isolation}), and the calibration of Alice's and Bob's devices (Section~\ref{sec:calibration}). Next, we introduce several physical devices and how they ideally behave (Section~\ref{sec:devices}). As an example of implementations we consider two implementations of the BB84 protocol (Section~\ref{sec:BB84_implementations}).

Lastly, we discuss assumptions about sources (Section~\ref{sec:sources}), measurements (Section~\ref{sec:measurements}), and classical post-processing (Section~\ref{sec:post-processing}).

\section{``Unconditional'' Security} \label{sec:unconditional_security}

Before discussing the assumptions made in QKD and quantum cryptography, we discuss the term ``unconditional security,'' which is used in the literature to imply that a protocol is secure against general (coherent) attacks by Eve (see Section~\ref{sec:classes}) \cite{scarani09b}. However, the term ``unconditional'' implies that the security is not conditioned on any assumptions or only relies on the fundamental assumption that quantum mechanics is complete (see Section~\ref{sec:foundational_assumptions}). Clearly protocols are not ``unconditionally'' secure: there are many assumptions made about each of the components used in the protocol. As was pointed out in \cite{scarani09b} the only part that has no conditions is what we assume about Eve's attack structure; we do make assumptions about Alice's and Bob's devices and subprotocols.

There are other terms that do not have this confusion about assumptions. One is just ``security,'' since security proofs always come with assumptions. Another term is ``information-theoretic security,'' which implies that security is proven using information theory, such as the security criteria in Section~\ref{sec:security_definition}.

Note that there are efforts to prove security under certain assumptions about Eve's attack, such as if Eve's memory is bounded \cite{damgard05}. There are also security proofs that try to prove that a quantum protocol is secure against adversaries that can do attacks in a theory more general than quantum mechanics \cite{barrett05b,acin06b,masanes09a,hanggi10a,hanggi10b,hanggi10c}.

\section{Assumption Classes}\label{sec:assumption_classification}

We decompose assumptions into four classes. The classification of assumptions we present can be used to discern how justified assumptions are and whether Eve can get an advantage from such assumptions.

First, an assumption may be \emph{fundamental}, which means that the assumption is assumed without any experimental verification. This assumption can be justified if it depends on foundational principles that are supported by our current understanding of physics, such as that information cannot travel faster than the speed of light or that quantum mechanics is a correct theory (see Section~\ref{sec:foundational_assumptions}). Fundamental assumptions may be unjustified if they are not even approximately correct. For example, it may be assumed that qubits are measured by Bob. If Eve is assumed to be able to do anything allowed by quantum mechanics and Bob does not check to see if he is getting qubits, then the assumption that Bob receives a qubit from her is unjustified and fundamental.

Second, there are \emph{calibrated} assumptions, which are approximately correct but cannot be guaranteed by an experiment. A device may approximate a model for the device, which can be checked with experiments but the experiments do not guarantee that this model will hold exactly in an implementation. For example, a measurement device may be constructed to approximate a particular POVM. The device may be tested to check that it approximately implements the desired POVM. However, if it is assumed that the device implements the model POVM then Eve may get an advantage from the deviation of the model from the implementation, even if the model is approximately correct.

Third, an assumption may be \emph{verifiable}, which means that the assumption can be verified experimentally or a theoretical analysis implies that Eve cannot gain any advantage (or the amount of the advantage is known) due to the model deviating from the implementation. For example, it may be assumed that measurements satisfy a particular property that can be experimentally verified before the protocol begins. Verifiable assumptions may also be about the structure of the protocol. For example, a measurement may be assumed to have two basis choices. If the protocol is implemented with this construction then this assumption is justified.

Fourth, there may be assumptions that can be justified by changing the implementation, such as adding a quantum device or modifying the classical post-processing, which we call \emph{satisfiable}. The modification of the implementation may lead to the need for more assumptions about additional devices or modifications. For example, it may be assumed that Eve does not send light into Bob's measurement device that is beyond a certain intensity. This can be a satisfiable assumption if Bob monitors the intensity of the incoming light, which requires the addition of an intensity monitor. Further assumptions may be necessary about the intensity monitor, which may not be justified.

The completely justifiable assumption classes that can be justified are verifiable assumptions, satisfiable assumptions that either require no further assumptions or assumptions that are justified, and some fundamental assumptions. Fundamental assumptions are either justified by the underlying physical theory or they are completely unjustified and are only made so that the model where the security proof applies is the same as the implementation, regardless of whether the implementation satisfies the assumption or not. Calibrated assumptions may be approximately justified, since the devices are approximately the same as their intended model. However, since Eve can exploit any deviation of Alice's and Bob's devices, it is not clear \emph{a priori} how much of an advantage Eve gets from a calibrated device that deviates from its model. This ambiguity makes the justification of calibrated assumptions unclear and the assumptions deserve further analysis to determine the extent of Eve's advantage. Satisfiable assumptions are justified by a modification of the protocol. However, the addition of other devices or modifications of the existing protocol usually requires further assumptions. Note that assumption classes other than satisfiable assumptions may be justified but they do not require a modification of the protocol.

We will use the four classes (fundamental, calibrated, verifiable, and satisfiable) to classify the assumptions in this chapter. We begin the detailed discussion of assumptions with universal assumptions that are applicable to almost all quantum-cryptography and quantum-key-distribution protocols.

\section{Universal Assumptions} \label{sec:universal_assumptions}

There are several basic assumptions that are made for almost all quantum-cryptography protocols. Here we outline foundational assumptions about the underlying physical theory used to define models of the protocols, the isolation of Alice's and Bob's devices from any eavesdropper, and the calibration of Alice's and Bob's devices before performing a protocol.

\subsection{Foundational Assumptions} \label{sec:foundational_assumptions}

Security of a quantum-cryptography protocol is usually proven with an adversary or dishonest party who is able to use any possible attack allowed by quantum physics. However, this assumption implicitly assumes that quantum physics is \emph{complete}.

A complete theory is one in which the predictions it makes about what is observable are the most accurate predictions possible by an experiment. Therefore, quantum mechanics is complete if it can make the best predictions about all possible measurement outcomes. This assumption implies that an adversary cannot get any more information about Alice's and Bob's keys in a QKD protocol than what is possible by quantum mechanics.

It was shown that instead of directly assuming that quantum mechanics is complete, two other assumptions can be made: that the theory is \emph{correct} and that \emph{free randomness} exists\cite{colbeck11,colbeck12a,colbeck12b}.

A correct theory is one that makes accurate predictions about what is observable. Quantum mechanics is correct if the predictions it makes about measurement outcomes are accurate. The assumption that free randomness exists is that measurement choices (such as a basis choice) can be chosen independently of the measurement device itself.

Therefore, a fundamental assumption we make for the security of QKD is that quantum mechanics is correct and free randomness exists, since these imply that quantum mechanics is complete.

There are other models for the underlying physical theory that are used instead of quantum mechanics, for example, that a generalized probabilistic theory describes physical reality \cite{barrett05b,acin06b,masanes09a,hanggi10a,hanggi10b,hanggi10c}.

\subsection{Isolation of Labs} \label{sec:isolation}

Alice's and Bob's devices should be completely isolated from Eve. If Eve is able to get information from their devices directly then the protocol may be completely compromised. For example, in a P\&M protocol, if Eve learns all of the measurement outcomes from Bob or knows what states were prepared by Alice in a P\&M protocol then the protocol is completely insecure.

There are a few known attacks of this type. For example, if Alice and Bob do the phase implementation of the BB84 protocol (see Section~\ref{sec:phase_BB84}) then Eve can send states into Alice's source via the quantum channel and learn the setting of Alice's phase modulator \cite{ribordy98,scarani09b}. Therefore, for this attack on the BB84 protocol, the assumption that Alice's lab is isolated is a satisfiable assumption, since Alice can monitor the intensity of incoming light from the quantum channel. If Alice detects incoming light then Alice and Bob would abort the protocol.

Another example of an attack against lab isolation is in any protocol that uses threshold detectors for a measurement (see Section~\ref{sec:threshold}). When threshold detectors recover after a detection they can emit light which can leak out into the quantum channel. Eve can then collect this light and potentially learn which threshold detector clicked \cite{kurtsiefer01,scarani09b}. In this case, the isolation of Bob's lab is a satisfiable assumption, since Bob can put a barrier between his measurement device and the quantum channel while his threshold detectors are recovering, so that any light would be blocked from leaking outside of his lab during his detector's recovery.

Yet another attack that violates lab isolation is for two-way QKD protocols, where the two quantum channels (as in Figs.~\ref{fig:LM05} and \ref{fig:SDC}) are actually the same quantum channel used in two directions. In this case, Alice is both sending and receiving states from the same quantum channel and therefore requires an open interface with the quantum channel. This interface allows Eve to send states into Alice's lab to potentially determine how Alice prepared her states or what her measurement basis choice is.

In general, the assumption that Alice's and Bob's labs are isolated is a fundamental assumption, because we assume that Eve cannot break into Alice or Bob's lab and steal their measurement outcomes.\footnote{This \href{http://xkcd.com/538/}{comic} (http://xkcd.com/538/) captures this idea.}

\subsection{Device Calibration} \label{sec:calibration}

There are two kinds of calibration that Alice and Bob can do before a quantum-cryptography protocol. First, Alice and Bob can calibrate their own devices so that they are working as they are intended. Second, Alice and Bob may need to perform a joint calibration that requires classical or quantum communication. The first kind of calibration can be done inside Alice's and Bob's isolated labs, and therefore under the assumption that their labs are isolated, no further assumptions are necessary about the calibration procedure. However, the second kind of calibration requires an interaction between Alice and Bob that Eve may interfere with. The calibration may leak information to Eve through Alice and Bob's communication and further assumptions may be necessary.

As an example of the second kind of calibration in P\&M protocols, Alice would like to prepare states such that Bob's measurement can distinguish them. Before the protocol starts, it is important that Alice and Bob calibrate their devices to optimize the correlations between Alice's sent states and Bob's measurement outcomes. In an entanglement based protocol, it is also important to calibrate both measurement devices so that Alice's and Bob's measurement results are as correlated as possible. For the polarization implementation of BB84 (see Section~\ref{sec:polarization_BB84}), what is defined as horizontal polarization for Alice is relative to a particular reference frame. Therefore, Bob needs to calibrate his measurement so that he shares the same reference frame as Alice.

The reference frame calibration procedure can be done before the QKD protocol. Alice can continually rotate her reference frame while sending many states to Bob and classically communicate through an authenticated channel which states she is sending. Bob can communicate his measurement outcomes to Alice. If Alice and Bob repeat this procedure for different angles then they can share approximately the same reference frame that will maximize their correlations for the run of the QKD protocol.

In addition to calibrating their reference frames, Alice and Bob need to agree on a timing of their signals so that Bob knows which states sent from Alice correspond to which measurement results. Note that Alice and Bob cannot just infer this correspondence from the order of the measurement outcomes and sent states during a run of the QKD protocol since some states may be lost between Alice and Bob due to losses (or Eve). Since the signals will be sent in rapid succession, it is important that Alice and Bob have accurate clocks so that they know which sent states correspond to which measurement outcomes during the protocol. Alice and Bob can synchronize their clocks by using a trusted third party. Alternatively, there are classical protocols that can be used to synchronize clocks without the need of a third party. Once their clocks are synchronized, Alice and Bob can also test to see how long it takes for Alice's states to reach Bob. Then, throughout the QKD protocol Alice can communicate through the classical authenticated channel to Bob when she sent her states so that Bob knows which measurement outcomes correspond to which of Alice's states.

The assumptions required for the model to match the implementation are dependent on the calibration method. For example, it is important that Alice communicates to Bob during the protocol only through the authenticated classical channel for timing calibration, otherwise Eve may send incorrect timing information to Bob, which could give her an advantage \cite{jain11}. It may also be necessary to make a fundamental assumption that a third party is trustworthy to synchronize their clocks. 

Reference frame calibration may be avoided if a QKD protocol is used that does not need this calibration \cite{sheridan10,liang14}.

\section{Devices for Quantum-Cryptography Implementations} \label{sec:devices}

The universal assumptions that apply to most quantum-cryptography and QKD protocols have now been discussed. Now we go into the details of specific devices used in QKD. Afterward, we present two examples of implementations of the BB84 protocol that use these devices (Section~\ref{sec:BB84_implementations}), followed by assumptions about the devices used in QKD and quantum-cryptography protocols.

We will not go into the full details of the physics that describe the devices used for quantum cryptography, though this is an interesting endeavour in its own right. Instead we will describe these devices with their ideal descriptions and how they can be modelled. In later sections we will describe how they may deviate from these models, which has consequences for the assumptions made in QKD protocols. Further details on how these optical devices work can be found in a quantum optics book, such as \cite{loudon00}.

We focus on devices used in discrete-variable and device-independent protocols, such as attenuated lasers (Section~\ref{sec:weak_laser}), parametric down-conversion (Section~\ref{sec:PDC}), beamsplitters (Section~\ref{sec:BS}), threshold detectors (Section~\ref{sec:threshold}), and Mach-Zehnder interferometers (Section~\ref{sec:MZ}).

\subsection{Weak Laser} \label{sec:weak_laser}

Ideally we would like a source of single particles to encode the states used in discrete P\&M QKD protocols. Typically photons are used since they can be easily transmitted either through fibre-optic cables or through free space (e.g.~the atmosphere or space). However, current technology does not allow single photons to be produced on demand. Usually coherent states are used instead (see Eq.~\ref{eq:coherent_state}).

One source of photons is a laser that produces coherent states. Coherent states are an approximation of the state a laser produces. This approximation requires the power given to the laser to be well over a certain threshold and requires a laser designed to produce single modes (i.e.~a single frequency of light) \cite{loudon00}. The phase of the produced states may also give information to Eve and should be taken into account (see Section~\ref{sec:phase_coherence}).

A laser can be given power for a short time to produce coherent states localized in a small spatial region followed by an attenuator (i.e.~a device that reduces the light's power). After the attenuator the state will be a coherent state with a low average photon number and a short spatial (or equivalently, temporal) width \cite{rosenberg07}. The spacial width is the wave function's spatial degree of freedom. The probability of measuring the photon at a particular time after its production is approximately distributed according to a Gaussian distribution \cite{loudon00}.

The values of the average photon number used for QKD are typically less than one photon per pulse \cite{scarani09}.

\subsection{Parametric Down-Conversion} \label{sec:PDC}

Another way to produce photons is to use a process called \emph{parametric down-conversion} (PDC). This process is performed by shining a laser continuously at a particular type of non-linear crystal. This crystal takes one state of light in a single mode (i.e.~a single frequency) and decomposes it into two states, each with half the frequency of the initial state. They also spread out in two spatial directions such that momentum is conserved. To conserve photon number, the average photon number of the initial pulse will be split such that the sum of the average photon numbers of each output pulse is equal to the average photon number of the initial laser light. While most of the laser light goes through the crystal without interacting with it, sometimes the state will be split into these two pulses. The two outputs from the crystal are called the \emph{signal} and the \emph{idler}.

These two output pulses can be used as a source for entangled photons. The two pulses will have orthogonal polarization (which we denote with $H$ and $V$, see Section~\ref{sec:polarization_BB84}) and are spatially distributed in two intersecting circles. At the intersection of these two circles the polarization of the output is ambiguous. Along these spacial modes the output is the maximally entangled pure state
\begin{equation}
\ket{\psi}=\frac{\ket{\alpha}\ket{-\alpha}+\ket{-\alpha}\ket{\alpha}}{\sqrt{2}},
\end{equation}
where $\ket{\psi}$ is in the Hilbert space of the polarization in the two spatial modes. Therefore, PDC can be used as a source of entangled bipartite states.

To produce single states from this process, a measurement device can be placed before the spatial location of the idler and whenever the measurement reveals that there is a signal then it is known that a signal state is present (see Section~\ref{sec:threshold} for the details of this measurement device). This kind of source, where a measurement indicates when a state is prepared, is called a \emph{heralded} source.

There are other sources other than weak laser pulses and parametric down conversion, such as Nitrogen vacancies in diamond and quantum dots (see \cite{albrecht14} and \cite{he13} for recent experiments that use these sources).

\subsection{Beamsplitters}\label{sec:BS}

A beamsplitter is a simple optical device that takes two input modes and has two output modes (see Fig.~\ref{fig:beamsplitter}). A beamsplitter can be modelled as a matrix acting on the creation operators for the two input modes:
\begin{equation} \label{eq:BS}
\begin{pmatrix}
T & R \\
R & T
\end{pmatrix}
\begin{pmatrix}
\hat{a}_1^{\dag} \\
\hat{a}_2^{\dag}
\end{pmatrix} =
\begin{pmatrix}
\hat{a}_3^{\dag} \\
\hat{a}_4^{\dag}
\end{pmatrix},
\end{equation}
where $T$ and $R$ are the transmissivity and reflectivity of the beamsplitter respectively. They satisfy $|R|^2+|T|^2=1$ and $RT^*+TR^*=0$ \cite{loudon00}.

\begin{figure} \centering
\includegraphics[width=0.4\textwidth]{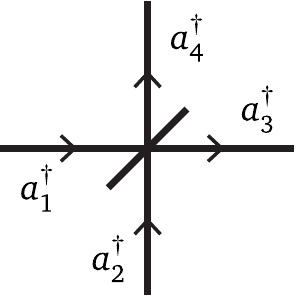}
\caption[A Beamsplitter]{A beamsplitter. It takes input modes $a_1^{\dag}$ and $a_2^{\dag}$ to output modes $a_3^{\dag}$ and $a_4^{\dag}$ according to Eq.~\ref{eq:BS}.}
\label{fig:beamsplitter}
\end{figure}

A particular example of a beamsplitter is a 50:50 beamsplitter, where $R=i/\sqrt{2}$ and $T=1/\sqrt{2}$.

Consider a single photon in a single mode input into arm $1$ to a 50:50 beamsplitter while the input into arm $2$ is the vacuum. In this case the output is given by $(\ket{1}_3+i\ket{1}_4)/\sqrt{2}$. This means that the photon is in a superposition of being transmitted through the beamsplitter or was reflected by the beamsplitter and acquiring a phase of $\pi/2$ (since $e^{i\pi/2}=i$).

Another example is if a coherent state, $\ket{\alpha}$, is input to a 50:50 beamsplitter in arm $1$ and the vacuum is input to arm $2$. In this case the output is
\begin{align}
e^{-\frac{|\alpha|^2}{2}}\sum_{n=0}^{\infty}\frac{\alpha^n}{\sqrt{n!}}&\left(\frac{\hat{a}_3^{\dag}+i \hat{a}_4^{\dag}}{\sqrt{2}}\right)^n\ket{0} \\
&= e^{-\frac{|\alpha|^2}{2}}\sum_{n=0}^{\infty}\frac{\alpha^n}{\sqrt{n!2^n}}\sum_{k=0}^n \binom{n}{k}\hat{a}_3^{\dag k}(i \hat{a}_4^{\dag})^{n-k}\ket{0} \\
&= e^{-\frac{|\alpha|^2}{2}}\sum_{n=0}^{\infty}\frac{\alpha^n}{\sqrt{n!2^n}}\sum_{k=0}^n \binom{n}{k}i^{n-k}\ket{k,n-k}_{3,4}.
\end{align}
This state is a coherent state that is distributed over the two output modes with a superposition of different possible photons in each output mode.

Another kind of beamsplitter is a polarizing beamsplitter (see Section~\ref{sec:polarization_BB84} for the details of polarization). Polarizing beamsplitters can separate two orthogonal polarization states into its two outputs. For example, if horizontally polarized light is sent into one arm then it is always transmitted and if vertically polarized light is sent into the same arm then it is always reflected.

\subsection{Threshold Detectors}\label{sec:threshold}

A threshold detector is a measurement that gives an output, \emph{click}, when it measures one or more photons and otherwise it outputs \emph{no click}. Formally, its POVM elements are the projection onto $\kb{0}{0}$ (the vacuum) and $\sum_{n=1}^{\infty}\kb{n}{n}$ (one or more photons). This kind of measurement can be implemented in various ways. Here we present an implementation of threshold detectors: avalanche photodiodes (see \cite{dixon09,yuan10,patel12,kalliakos14} for experiments characterizing these devices).

Avalanche photodiodes are made out of a semiconductor material (such as indium gallium arsenide, InGaAs) that has an electric field applied to it \cite{ramaswami02}. These detectors exploit the photoelectric effect so that an incident photon excites an electron in the semiconductor. Since an electric field is applied, the electron has enough energy to excite one or more electrons, which can excite further electrons, which go on to excite more electrons, leading to an avalanche of excited electrons. If many electrons are excited then a current can be measured, indicating that at least one photon hit the detector.

The avalanche is a random process that depends on the strength of the electric field. However, if an electron absorbs energy from the semiconductor (i.e.~a phonon) then an avalanche can occur without any incident photons. These events are called \emph{dark counts}. The stronger the electric field, the more likely it is that dark counts will occur. 

Conversely, the photon may excite an electron, but if too few electrons are excited then there is no avalanche, so no current will be registered. Therefore, the weaker the electric field, the more likely that a photon will not induce an avalanche, resulting in loss. Therefore, by changing the strength of the electric field there is a tradeoff between the probability of dark counts and the efficiency of the detector. In addition, the photon may not be absorbed by the material but may be reflected or pass through the material, which also results in loss.

After the avalanche, the semiconductor needs to have all of its electrons return to their unexcited state by turning the electric field off. The time it takes for the electrons to return to their unexcited state is called the \emph{recovery time} or \emph{dead time} (since the threshold detector cannot make a measurement when it is recovering). Sometimes there will also be \emph{after pulses}, when the energy from a relaxing electron causes the excitation of further electrons, resulting in a second avalanche causing a second click event (see Section~\ref{sec:time_resolution}).

Avalanche photodiodes also have a limited temporal resolution; from the time when a photon first starts an avalanche to the recovery of the detector, more photons may hit the material. This means that the detector is an integrated measurement: if a current is detected then one or more photons hit the detector in the time period designated for each detection event.

The quality of an avalanche photodiode primarily consists of three things: the probability of dark counts in a given time frame, the recovery time, and the efficiency. Typically, the efficiency is quite low for avalanche photodiodes at room temperature \cite{comandar14}. However, by decreasing the temperature the efficiency can increase dramatically, but the recovery time is extended due to the increased probability of after pulses, since it requires less energy to cause an after pulse at lower temperatures.

There are also efforts to use different kinds of single photon detectors at low temperatures, called superconducting single-photon detectors, as an alternative to avalanche photodiodes. Superconducting single-photon detectors have a higher efficiency and lower dark count rates due to less energy available to excite the system (see \cite{covi14,shibata14} as examples of recent experiments). For a comparison of the performance of these detectors compared to avalanche photodiodes, see \cite{scarani09}

\subsection{Mach-Zehnder Interferometers} \label{sec:MZ}

A Mach-Zehnder interferometer is a particular arrangement of beamsplitters and mirrors, which can be followed by threshold detectors (see Fig.~\ref{fig:MZI} and see \cite{micuda14} for a recent experiment). Here we describe an unbalanced interferometer that can be used to measure the relative phase between two pulses, as is necessary for the COW and DPS protocols (see Section~\ref{sec:discrete_protocols}), as well as one implementation of the BB84 protocol (Section~\ref{sec:phase_BB84}). The mirrors in the Mach-Zehnder interferometer can be thought of as a beamsplitter with reflectivity $i$ and transmissivity $0$.

As an example of the use of a Mach-Zehnder interferometer, consider a single photon distributed over two pulses separated by a distance equal to the relative distance between two arms of the Mach-Zehnder interferometer. If the two pulses have a relative phase of $\phi\in[0,2\pi)$ then the state before the Mach-Zehnder interferometer is
\begin{equation}
\frac{\ket{t}+e^{i\phi}\ket{t+1}}{\sqrt{2}},
\end{equation}
where $\ket{t}$ denotes a photon at time slot $t$ of the first pulse and $\ket{t+1}$ denotes a photon at the time slot of the second pulse. After the first beamsplitter, which has vacuum as the second input, we can use the relation Eq.~\ref{eq:BS} for a 50:50 beamsplitter to find that the state is 
\begin{equation}
\frac{\ket{t,S}-i\ket{t,L}+e^{i\phi}\ket{t+1,S}-ie^{i\phi}\ket{t+1,L}}{2},
\end{equation}
where $\ket{t,S}$ denotes a photon at time slot $t$ in the short arm of the interferometer and $\ket{t,L}$ is a photon at time slot $t$ in the long arm. After the delay in the long arm and the reflections on the two mirrors, but before the second beamsplitter, the state is
\begin{equation}
\frac{i\ket{t-1,L}+\ket{t}\left(\ket{S}+ie^{i\phi}\ket{L}\right)+e^{i\phi}\ket{t+1,S}}{2}.
\end{equation}
After the second beamsplitter the state becomes
\begin{align}
\frac{1}{2}\left(i\ket{t-1}\left(\frac{\ket{1}-i\ket{0}}{\sqrt{2}}\right)+\ket{t}\left(\frac{(1+e^{i\phi})\ket{0}-i(1-e^{i\phi})\ket{1}}{\sqrt{2}}\right)\right.\nonumber\\
\left.+e^{i\phi}\ket{t+1}\left(\frac{\ket{0}-i\ket{1}}{\sqrt{2}}\right)\right),
\end{align}
where $\ket{0}$ is a photon at threshold detector $D_0$ and $\ket{1}$ is a photon at threshold detector $D_1$ (see Fig.~\ref{fig:MZI}). If we condition on getting an outcome at time slot $t$ and if $\phi=0$ then only detector $D_0$ can click. If $\phi=\pi$ then only detector $D_1$ can click. This means that the relative phase between two pulses can be measured with certainty if $\phi\in\{0,\pi\}$. However, with probability $1/2$, either detector can click at time slot $t-1$ or $t+1$, where either $D_0$ or $D_1$ will click with equal probability.

\begin{figure} \centering
\includegraphics[width=0.6\textwidth]{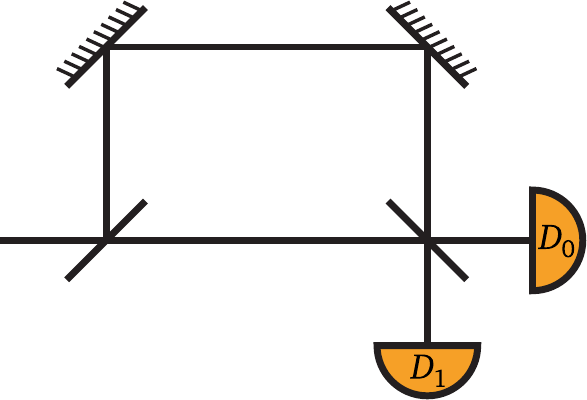}
\caption[An Unbalanced Mach-Zehnder Interferometer Measurement]{An unbalanced Mach-Zehnder interferometer measurement. It is composed of two 50:50 beamsplitters and two mirrors on the long arm followed by two threshold detectors. Two pulses separated by a distance equal to the length difference between the two paths in the interferometer will have their relative phase measured with probability $1/2$. At time slot $t$, if threshold detector $D_0$ clicks then the relative phase was at an angle of $0$ and if threshold detector $D_1$ clicks then the relative phase was at an angle of $\pi$. If a measurement result occurs at time slot $t-1$ or $t+1$ then the relative phase is unknown.}
\label{fig:MZI}
\end{figure}

The POVM elements that describe a perfect Mach-Zehnder interferometer with a detection at time slot $t$ are the projections onto the states
\begin{equation}
\frac{\ket{t}+\ket{t-1}}{\sqrt{2}} , \frac{\ket{t}-\ket{t-1}}{\sqrt{2}},
\end{equation}
for $D_0$ and $D_1$ respectively.

A phase modulator can be added on one arm of the Mach-Zehnder interferometer so that it can distinguish the relative phase between different phases other than $0$ and $\pi$. The next section describes a phase modulator.

\subsection{Other Devices} \label{sec:other_devices}

There are some other devices that are used in quantum-cryptography and QKD implementations, such as polarizers and phase modulators.

A polarizer is a filter that only allows output light to be of a particular fixed polarization. Polarization is a degree of freedom of a photon that represents the relative phase between the oscillating magnetic and electric fields of the photon \cite{saleh91}. If they are in phase then the photon may be linearly polarized either horizontally or vertically, relative to some reference frame (see Section~\ref{sec:polarization_BB84}). If they are out of phase (e.g.~there is a relative phase of $\pi/2$) then the polarization can be either left or right circularly polarized. Polarization can be thought of as the orientation of the combination of the waves while looking in the plane perpendicular to the direction of movement of the photon. From this view, the linear polarization is a line, while the circular polarization is a rotation around a circle (either clockwise or anti-clockwise). There is also elliptically polarized light, which is a superposition of circular and linear polarization.

There are two polarizers of interest: linear polarizers, which output linearly polarized light, and circular polarizers, which output circularly polarized light. They are constructed from materials which are \emph{birefringent}, which means that light has a different speed of travel depending on the its polarization. The result is that the light that is transmitted through the material has the desired polarization.

Two polarizers of interest can be constructed from half- and quarter-wave plates. Wave plates are birefringent materials that are chosen to have a thickness that induces a desired polarization. Half-wave plates induce a relative phase of $\pi$, while quarter-wave plates induce a relative phase of $\pi/2$. Also, there are materials which change their birefringence depending on an electric field that is applied to the material. This process is called the \emph{Kerr effect}, which can be used to change the polarization filter on demand \cite{saleh91}.

Phase modulators manipulate the relative phase between two pulses. For most QKD purposes, this phase modulation should only induce a short delay in the propagating light that is of the order of the light's wavelength.

Both polarizers and phase modulators have a loss associated with them, which should be taken into account in implementations of quantum-cryptography protocols (for example, see \cite{ferenczi12}).

\subsection{Channel Models}

Recall that for the robustness of the protocol the probability that the protocol aborts when there is no eavesdropper should be known (see Section~\ref{sec:robustness}). This probability is typically found by assuming a model for the quantum channel between Alice and Bob, as well as a model for Alice's and Bob's devices. Given these models, the probability that the protocol aborts can be calculated.

A typical model for the quantum channel is a \emph{depolarizing} channel. It can be described as
\begin{equation}
\rho \mapsto p\rho+(1-p)\frac{\identity}{d},
\end{equation}
for $\rho\in\hilbert$, $d$ is the dimension of $\hilbert$, and $p$ is a probability. Usually in QKD Bob does a measurement to try to distinguish two or more quantum states, $\rho_i$. With probability $p$ he will get $\rho_i$ so he can distinguish these states (conditioned on him measuring in the correct basis, for protocols with a basis choice) and with probability $1-p$ he will get a maximally mixed state, so he gets each of his measurement outcomes with probability $\Tr(F_i/d)$, for a POVM with POVM elements $\{F_i\}$. For many of the protocols in Section~\ref{sec:discrete_protocols}, $\Tr(F_i)$ is the same for all $i$ and therefore the probability of getting an error is the same for all measurement outcomes.

If errors are seen to be equally likely, regardless of the measurement outcome, then the channel can be modelled as a depolarizing channel. Therefore, this is a calibrated assumption that the channel is depolarizing, since the errors are usually only approximately equally likely.

Channels also have losses. The loss is characterized in units of $dB/km$, which is the log of the ratio of the power (of a classical optical signal) between the input and output signals, times ten, per kilometre. The lowest loss fibre-optic cables possible with current technology have a loss of $0.17$ $dB/km$, which means that the power decreases to $\sim 96\%$ of the input power over one kilometre of fibre \cite{saleh91}.

The loss in a quantum channel can be modelled as a beamsplitter that takes the input state as one input and the vacuum as its other input. The output of the channel is the transmission output of the beamsplitter and the other output of the beamsplitter is lost to the environment (which we can assume Eve gets). The loss can be modelled this way because, typically, the losses do not depend on the state of the system and are just probabilistic: a photon is transmitted through the channel or lost to the environment regardless of the photon's state \cite{scarani09}. However, it should be taken into account that Eve can control when losses occur and she may perform attacks where the loss may depend upon the state sent through the channel.

\section{Implementations of BB84} \label{sec:BB84_implementations}

Two practical examples we use to discuss assumptions about the devices used in QKD protocols are implementations of the BB84 protocol. We will describe the kinds of devices that are used in these implementations. For the perfect description of the BB84 protocol, see Section~\ref{sec:discrete_protocols}.

One implementation uses free space (e.g.~the atmosphere or space) to transmit photons that encode the qubits that Alice wants to send Bob in the polarization degree of freedom of the photon (see \cite{erven08} for an experiment). Since polarization of an individual photon is a two level system, polarization is an ideal property to use for the qubit of the BB84 protocol.

In fibre optics, a photon with a particular polarization undergoes polarization drift. Due to imperfections in the cable and in the environment (such as temperature differences) the polarization can be transformed as the photon goes through the fibre due to birefringence (see Section~\ref{sec:other_devices}). Over time scales smaller than the time it takes to perform the quantum \stage of the protocol this polarization drift can change, which makes it difficult to use polarization to transmit quantum data. In this case another implementation of BB84 can be used that encodes information in the relative phase of two pulses made from one photon. This encoding is similar to the distributed phase protocols DPS and COW (see Section~\ref{sec:discrete_protocols}).

\subsection{Polarization BB84}\label{sec:polarization_BB84}

The polarization of photons can be used to store quantum information (see Section~\ref{sec:other_devices}). Linearly and circularly polarized light form three bases of a qubit Hilbert space, and the polarization implementation of BB84 uses two of these three qubit spaces. The three possible bases are horizontal and vertical polarization $\{H,V\}$, diagonal linear polarization $\{D_+,D_-\}$ (where $D_+=(H+V)/\sqrt{2}$ and $D_-=(H-V)/\sqrt{2}$), and circular polarization $\{R,L\}$ (where $R=(H+iV)/\sqrt{2}$ and $L=(H-iV)/\sqrt{2}$). We can correspond the states in the BB84, SARG04, and six-state protocol with the polarization space, where $H=\ket{0}$, $V=\ket{1}$, $D_+=\ket{+}$, $D_-=\ket{-}$, $R=\ket{i}$, and $L=\ket{-i}$. Here we consider the implementation of the BB84 protocol that uses the $\{H,V\}$ basis and the $\{D_+,D_-\}$ basis.

Now we can implement the BB84 protocol as follows \cite{bennett92c} (see Fig.~\ref{fig:polarization}). The preparation of photons can be done by using a weak coherent laser pulse followed by a polarizer to set the polarization of the photons. For the measurement, Bob can use a polarizing beamsplitter, which separates two orthogonal polarization states. For example, one output of the polarizing beamsplitter could be $H$ and the other would then be $V$. Bob can place threshold detectors after each output of the polarizing beamsplitter to measure whether his state was $H$ or $V$. To measure in the other basis, he may actively control a polarization rotator before his polarizing beamsplitter so that he can measure $D_+$ and $D_-$ instead (see Fig.~\ref{fig:pol_active}).

Bob can also do his measurement in a passive way, so he does not have to control the orientation of his polarizing beamsplitter (see Fig.~\ref{fig:pol_passive}). First, he can put a 50:50 beamsplitter, which at one output has a polarizing beamsplitter and threshold detectors to measure in the $H/V$ basis, while the other output of the 50:50 beamsplitter has a polarizing beamsplitter and threshold detectors to measure in the $D_+/D_-$ basis. The 50:50 beamsplitter simulates the basis choice and Bob does not need to actively control his measurement device.

\begin{figure} \centering
\begin{subfigure}[b]{\textwidth}
\includegraphics[width=\textwidth]{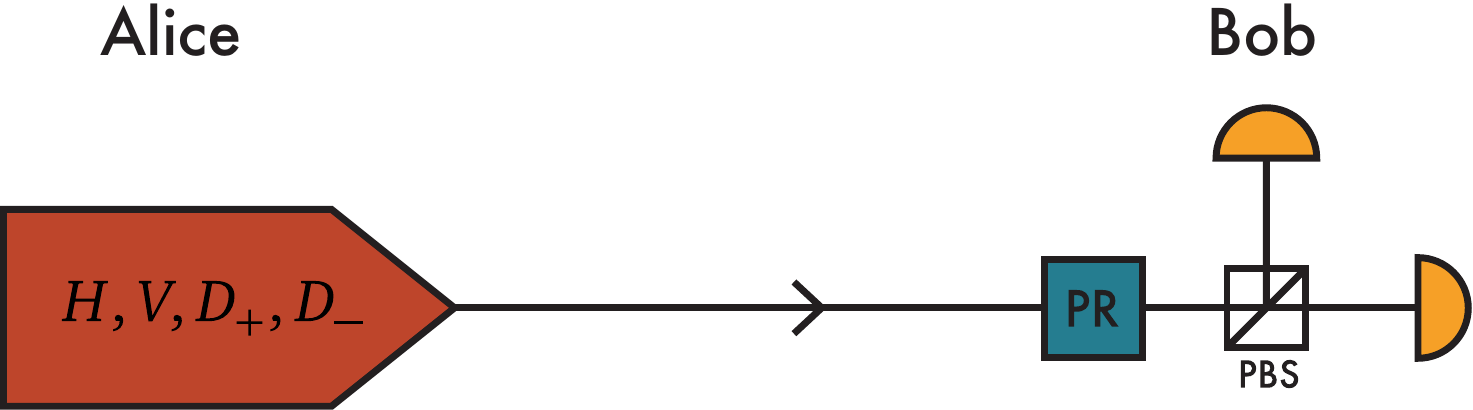}
\caption{The polarization implementation of the BB84 protocol with an active basis choice. Alice prepares states using a laser source of coherent states that go into a polarizer to produce $H,V,D_+,$ or $D_-$. Bob measures $H$ and $V$ by leaving the polarization the same by separating them using a polarizing beamsplitter (PBS) followed by two threshold detectors. Bob can measure $D_+$ and $D_-$ by applying a polarization rotator (PR) before his polarizing beamsplitter.}
\label{fig:pol_active}
\end{subfigure} \\ \vspace{1cm}
\begin{subfigure}[b]{\textwidth}
\includegraphics[width=\textwidth]{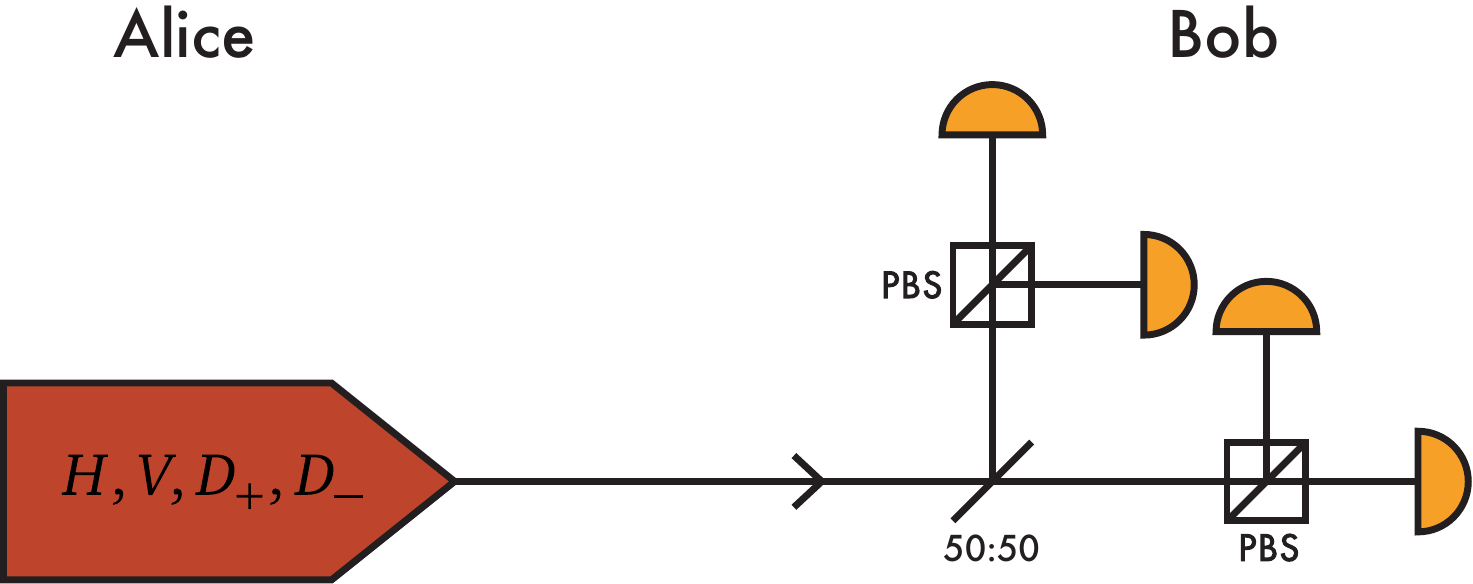}
\caption{The polarization implementation of the BB84 protocol with a passive basis choice. Instead of the polarization rotator Bob uses a 50:50 beamsplitter so that a single photon randomly goes to a measurement in the $H,V$ basis that uses one polarizing beamsplitter or the $D_+,D_-$ basis that uses another polarizing beamsplitter.}
\label{fig:pol_passive}
\end{subfigure}
\caption[The Polarization Implementation of the BB84 Protocol]{Two implementations of the BB84 protocol using polarized photons.}
\label{fig:polarization}
\end{figure}

\subsection{Phase BB84} \label{sec:phase_BB84}

In the phase implementation of BB84, the states $\{\ket{0},\ket{1},\ket{+},\ket{-}\}$ are represented as the relative phase between two pulses from a single photon \cite{ekert92} (see Fig.~\ref{fig:phase}). Alice can prepare a weak coherent laser pulse and input it into a 50:50 beamsplitter. Each arm of the beamsplitter has a different length and they will be recombined. On one arm of the beamsplitter Alice applies a phase modulation to change the relative phase between its output and the output of the other arm of the beamsplitter. Alice changes the relative phase to angles of $0$, $\pi$, $3\pi/2$, or $\pi/2$ (these phases are factors of $e^{i\phi}$ where $\phi$ is the phase resulting in $1$, $-1$, $i$, and $-i$). If we denote the two spatial modes as $s_1$ and $s_2$, we can write the states required for the BB84 protocol as
\begin{equation}
\ket{0} = \frac{\ket{s_1}+\ket{s_2}}{\sqrt{2}}, \ket{1} = \frac{\ket{s_1}-\ket{s_2}}{\sqrt{2}}, \ket{+}=\frac{\ket{s_1}+i\ket{s_2}}{\sqrt{2}}, \ket{-}=\frac{\ket{s_1}-i\ket{s_2}}{\sqrt{2}}.
\end{equation}
These are valid representations of the states for BB84 since they have the same overlaps, $\bk{\psi}{\varphi}$ for $\ket{\psi},\ket{\varphi}\in\{\ket{0},\ket{1},\ket{+},\ket{-}\}$.

On Bob's side, he will use a Mach-Zehnder interferometer followed by threshold detectors to measure this relative phase (see Section~\ref{sec:MZ}). A phase modulator is placed on the long arm in order to choose between measuring in the $\{0,\pi\}$ basis or the $\{\pi/2,3\pi/2\}$ basis. With probability $1/2$ he will get an outcome that tells him the phase and otherwise he gets an outcome that does not tell him what the relative phase was. Bob will then communicate to Alice when he gets a bad outcome and when he was able to discern the overlap depending on the timing of his measurement outcomes.

\begin{figure} \centering
\includegraphics[width=\textwidth]{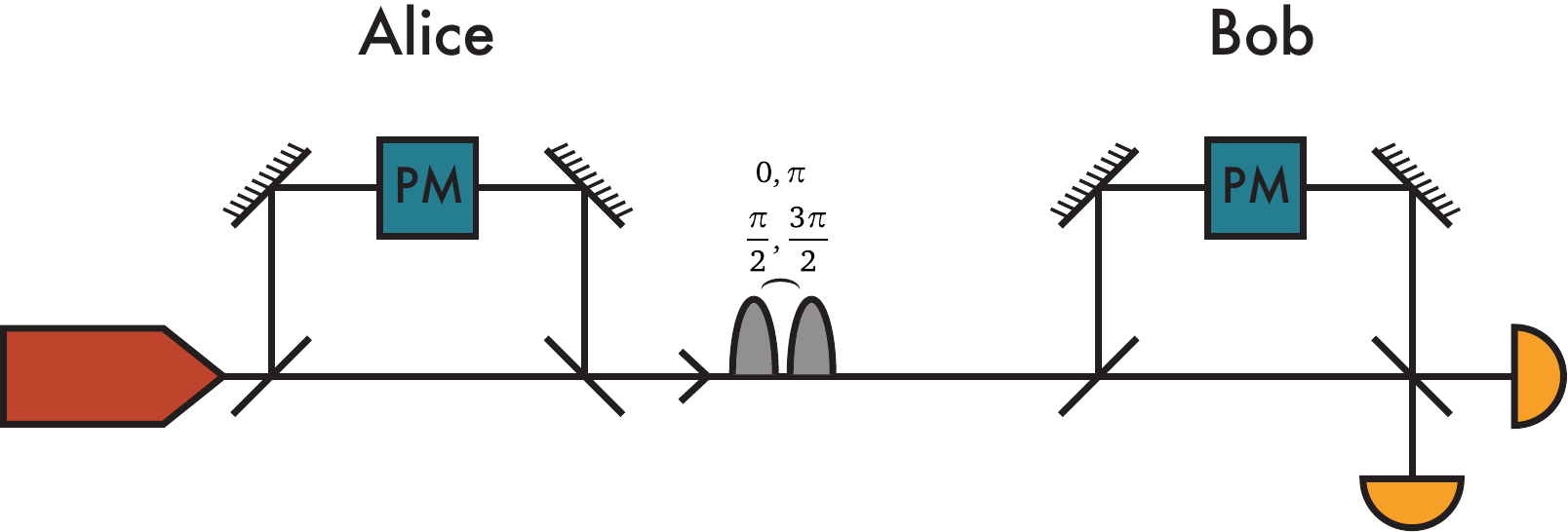}
\caption[The Phase Implementation of the BB84 Protocol]{The implementation of the BB84 protocol using the relative phase between two pulses. Alice prepares a coherent state from a laser, followed by a separation of this state into two pulses. She chooses the relative phase between them using a phase modulator (PM) resulting in a phase of $0,\pi/2,\pi$ or $3\pi/2$. Bob measures the relative phase by using an unbalanced Mach-Zehnder interferometer followed by two threshold detectors. His phase modulator chooses whether to measure $0$ and $\pi$; or $\pi/2$ and $3\pi/2$.}
\label{fig:phase}
\end{figure}

Now that implementations of the BB84 protocol have been introduced, we discuss the assumptions made about sources (Section~\ref{sec:sources}), measurements (Section~\ref{sec:measurements}), and the classical parts of the protocol (Section~\ref{sec:post-processing}).

\section{Source Imperfections and Assumptions} \label{sec:sources}

There are several ways that sources do not produce the idealized states required for a given protocol. There may be correlations between the state in its desired degree of freedom (such as polarization or relative phase) with other degrees of freedom, such as frequency or time. Subsequent states may not be independent, so an eavesdropper can get information from these correlations. We now list several of these assumptions, what class of assumption they are (fundamental, calibrated, verifiable, or satisfiable, see Section~\ref{sec:assumption_classification}), as well as any techniques used to justify these assumptions.

\subsection{Phase Coherence} \label{sec:phase_coherence}

For QKD protocols that do not use the phase as the degree of freedom for encoding often assume that the phase of each pulse is completely unknown to Eve. If this is the case, for example in the polarization implementation of the BB84 protocol, then the prepared state from the laser before encoding by Alice is a coherent state $\ket{\alpha}$ where $\alpha=re^{i\theta}$ and $r\in\mathds{R}^+, \theta\in [0,2\pi)$ \cite{lo07}:
\begin{equation}
\int_{0}^{2\pi} \frac{1}{2\pi}\kb{re^{i\theta}}{re^{i\theta}} d\theta = e^{-r^2}\sum_{n=0}^{\infty} \frac{r^{2n}}{n!}\kb{n}{n},
\end{equation}
which is a mixed state of a Poisson distribution of number states. However, if Eve has some information about the phase, then this is not an accurate description of the prepared state from Eve's perspective. Eve can get information about the phase in protocols that have a strong reference pulse, such as some forms of the B92 protocol \cite{tamaki03,tamaki04} and the Plug \& Play version of the BB84 protocol \cite{ribordy98,muller97} (see Section~\ref{sec:two_protocols}). Eve can then compare the phase between the strong pulse and the quantum state sent from Alice. Even in protocols that do not have a strong reference pulse, Eve can learn the phase of the source by using several weak pulses.

The assumption that the phase of each pulse is unknown to Eve is satisfiable, since Alice can apply a random phase to each state she sends into the quantum channel \cite{lo07}. As long as the randomness used to choose this phase is true randomness (see Section~\ref{sec:post-processing}) and the source is isolated from Eve, then Eve cannot get information about the polarization of the states by measuring their relative phase. If the phase is not completely randomized then Eve may get some information about the key \cite{lo05b,sun214,kobayashi14}.

The relative phase between different pulses may also give information to Eve. In the polarization implementation of the BB84 protocol, the relative phase between subsequent states may depend on the polarization of the states prepared. If Eve measures the relative phase, she may learn information about the polarization of the state. The assumption that the relative phase does not give any information about the polarization of the state is a calibrated assumption, since the states are prepared using sources that are not designed to change the phase over subsequent pulses.

\subsection{Multiple Photons}

As mentioned in Sections~\ref{sec:weak_laser} and \ref{sec:PDC}, the states that can be prepared in practice for discrete QKD protocols are coherent states, which are superpositions of photon number states. Sometimes multiple photons can be sent into the channel with the same encoding of the information, such as the polarization or the relative phase. Eve can then do a measurement to determine how many photons are present in the sent signal and then store the extra photons that Alice prepared while sending a single photon on to Bob. This attack is called the \emph{photon-number splitting} attack, and was noticed and analyzed in \cite{lutkenhaus99,lutkenhaus00,gottesman04}. Eve can either get full or partial information about Alice's state depending on the protocol and how many extra photons there are. Also, Alice and Bob will not detect this attack because it does not introduce any errors.

A method that has been developed to compensate for the photon number splitting attack is the \emph{decoy state method} \cite{hwang03,wang05,lo05}. Alice will prepare different states with a different number of average photons. She can choose in advance from a discrete set of possible average photon numbers. From this choice of states, Alice and Bob can estimate the number of errors they have for single photons, for two photons, etc. Formally, Alice and Bob have a set of linear equations for their error rates:
\begin{equation}
Q_{\text{total}} = p_1^iQ_1+p_2^iQ_2 + \cdots,
\end{equation}
for a total error rate $Q_{\text{total}}$, error rates for each photon number $Q_j$ ($j\in\{1,2,\dots\}$), and a set of probability distributions $P^i=(p_1^i,p_2^i,\dots)$ (one for each average photon number setting).

If the decoy state method is not used, then Alice and Bob can assume that all of their errors come from measurement outcomes on single photons. However, this estimation is pessimistic. Using the decoy state method allows for an estimation of the single photon error rate, which is usually lower than for other photon numbers. As an example, in the infinite-key limit, this estimation results in a scaling of the key rate, $r$, so that the probability that a single photon is created ($p_1$) is multiplied with the key rate for single photons, which is a function of the single photon error rate, $r_1(Q_1)$. Also, if the key rate is non-zero for two or more photons then these key rates (which are each functions of the error rate for that many photons) can be taken into account by adding them together, each multiplied by the probability of having that many photons: $r=p_1 r_1(Q_1)+p_2r_2(Q_2)+\cdots$. The details of how to perform the estimation procedure for the decoy state method can be found in \cite{scarani09,hwang03,wang05,lo05,moroder08}.

The assumption that a discrete protocol is secure even when Alice prepares states that can contain multiple photons is satisfiable, since the decoy state method can account for this imperfection in the implementation and analysis. Alternatively, the error rate observed in the protocol without decoy states can still give an estimate of the number of errors for single photons. However, using this kind of bound results in a lower key rate (since the total error rate $Q$ satisfies $Q\geq Q_1$ and the key rate is monotonically decreasing in the error rate) \cite{scarani09}.

The decoy state method can also be done in a passive way, instead of having to actively change the average photon number in the pulses (see \cite{curty09,curty10,krapick14,xu14} and references therein). One way to implement passive decoy states is to use a weak coherent pulse with a beamsplitter followed by a threshold detector \cite{curty09}. Depending on if the threshold detector clicks or not, different superpositions of number states will be prepared, which can be used for a two-state decoy method. The passive decoy state method has been used in security proofs for QKD protocols \cite{zhou14} and implemented in a recent experiment \cite{sun14}.

Another method to counteract the photon-number splitting attack is to do a protocol that is robust against this kind of attack, such as the SARG protocol \cite{scarani04} (see Section~\ref{sec:discrete_protocols}).

\subsection{State Structure and Symmetry}

Many assumptions can be made about the states produced from the source and the states sent into the channel. Usually it is assumed that the states in discrete-variable protocols are independent. This assumption is necessary for analyses that need to treat signals in an independent way. For example, in the polarization implementation of BB84 the states prepared are not independent in practice. As mentioned in Section~\ref{sec:phase_coherence}, photon sources can have coherence in the phase between subsequent pulses, e.g.~the sequence of states have the form $\ket{\alpha}\ket{\alpha}\cdots\ket{\alpha}$. However, this assumption is satisfiable by randomizing the phase. Other degrees of freedom in the state of the photons may also be correlated between subsequent states but usually a calibrated assumption is made that the states do not have such correlations. However, this assumption should be verified by characterizing the source to ensure that this is the case.

Another assumption is that there are an infinite number of signals sent during the protocol. As discussed in Section~\ref{sec:classes}, this is an assumption that is not physically possible and therefore it is a fundamental assumption. However, there are many recent results that take finite-key effects into account, removing the need for this assumption \cite{rosenberg07,scarani08a,curty10,furrer12,tomamichel12a,leverrier13,walenta13,furrer14a,furrer14b,zhou14}.

In addition to correlations between subsequent states, there may be correlations between the degree of freedom used to encode the bits Alice wants to communicate to Bob and other degrees of freedom. For example, the frequency of the photon may be correlated with the polarization of the photon in polarization BB84. In addition, if different sources are used for different states then, as an example, each source might have a different frequency that will tell Eve the polarization of the photons \cite{kurtsiefer01,scarani09b}. Therefore, the assumption that there are no correlations between other degrees of freedom and the intended degree of freedom is a calibrated assumption.

In the phase implementation of the BB84 protocol, Alice prepares states using a phase modulator on one of her pulses. The phase modulator may induce a loss, which lowers the intensity of one pulse compared to the other. This loss creates a different state, which should be taken into account, such as in \cite{ferenczi12}. If this loss is not taken into account in the security proof then it is a calibrated assumption since it is an approximation of the intended state and the state that results from the implementation.

Some security proofs assume that a fixed finite Hilbert space describes the states sent from Alice to Bob. This assumption is fundamental because any particle that Alice sends to Bob has many degrees of freedom that could be correlated with the degree of freedom Alice uses to send information to Bob. If it is assumed that the quantum states have an i.i.d.~structure then this may either be satisfiable by using the post-selection theorem (which requires further assumptions) or it is a fundamental assumption.

If i.i.d.~states are assumed then symmetry may be exploited to prove that instead of having to prove security for all possible i.i.d.~states, only a small class of states need to be considered. For example, the states shared between Alice and Bob in the BB84 and six-state protocols may be completely determined by the number of errors measured in the protocol \cite{rennerphd}. The states are determined because Alice and Bob can rotate their states (as described at the end of Section~\ref{sec:reductions}) and the protocol is identical. This method only requires that the structure of the protocol satisfies this symmetry, which is a calibrated assumption.

\subsection{The Local Oscillator}

The local oscillator is a strong reference pulse that is sent along with a quantum state in some versions of the B92 protocol and continuous-variable protocols. The local oscillator is usually used in the measurement of the quantum state. The fundamental assumption is typically made that Eve does not interfere with the local oscillator. Since the local oscillator is sent through an insecure channel, this is not a justifiable assumption. However, if Bob monitors the intensity and phase of the local oscillator then the assumption can be satisfiable \cite{haseler08}. Alternatively, Bob can do a measurement of the local oscillator followed by a recreation of his own local oscillator with the same phase as the received local oscillator's phase \cite{koashi04}.

\section{Measurement Imperfections and Assumptions} \label{sec:measurements}

Measurements may also deviate from their perfect models in many ways. Measurements may respond to several photons, even when the single photon subspace is used for the encoding of the information Alice wants to send to Bob. The timing of the signals may be changed by Eve, which can influence Bob's measurement outcomes. Measurement devices also have unintended behaviour, such as with threshold detectors, which can have clicks when there is no signal and also have a limited efficiency so even when there is an incoming photon the detector may not click. Measurements can also deviate from their intended model entirely. An example of this kind of deviation can be illustrated with a blinding attack, where Eve completely controls Bob's measurement outcomes by shining bright light into Bob's detector \cite{lydersen10}. Finally, for device-independent protocols, Bell tests, such as the CHSH experiment, need to be performed precisely according to the model, otherwise Alice and Bob may see a CHSH violation but their states may not be quantum, which means that Alice and Bob should not be able to extract a key from these measurement outcomes.

We now investigate the imperfections of measurements in detail.

\subsection{The Squashing Model} \label{sec:squashing}

Since sources used in qubit device-dependent protocols usually produce weak coherent states instead of single photons, Bob may detect multiple photons, even if there is no eavesdropper. In addition, if Eve is present then she is not restricted to only begin able to send single photons to Bob but she can send any state she wants. It used to be a fundamental assumption that Eve does not get any advantage from this deviation from the intended model \cite{gottesman04}. However, now there is a precise way of determining if a given measurement device is equivalent to its perfect model. This technique is called the squashing model.

If the POVM elements of the measurement on the full Hilbert space of all optical modes is known and suitable POVM elements for a perfect model are chosen then it can be determined if these measurement devices are equivalent \cite{beaudry08,fung11,gittsovich14}. Note that the measurement on the full Hilbert space may give outcomes that never occur in the perfect model. For example, if in polarization BB84 multiple photons are input into the measurement device, then both threshold detectors may click. This event is called a \emph{double-click} which does not happen for single photons in the perfect model.

It is not clear what bit Bob should assign to these measurement outcomes. Bob could treat these outcomes as loss so that in the sifting step of the protocol Alice and Bob will ignore these outcomes. However, Eve can attack the protocol in a way that will give her full information about the key if Alice and Bob discard the double-click events \cite{lutkenhaus99}. Instead, Alice and Bob can treat these events as an error and each randomly assign a bit value to these measurement outcomes. This assignment corresponds to assigning one of the perfect model's outcomes (a $0$ or a $1$) randomly to the double-click events.

Given an assignment of detection events in the implementation to detection events in the perfect model, POVM elements are defined that describe the measurement outcomes and any classical assignment of these outcomes to bits. Now we can formally define the squashing model.

Given a POVM on a large Hilbert space $\{F_i\}$ and a POVM on a small Hilbert space $\{F^Q_i\}$ with an association between these two POVM elements (so that $F_i$ represents an outcome in the small Hilbert space described by $F_i^Q$) then there exists a squashing model if there exists a CPTP map $\mathcal{T}$ such that its Choi-Jamio{\l}kowski matrix $T$ (see Theorem~\ref{thm:CJ}) satisfies
\begin{align}
T^R\kket{F^Q_i} = \kket{F_i} \label{eq:squash_1} \;\forall i\\
T^{\dag}=T\geq 0, \label{eq:squash_2}
\end{align}
where $T^R$ is the Normal map for $\mathcal{T}$ (Defn.~\ref{defn:normal_rep}) and $\kket{G}$ is the vector representation of the matrix $G$ (Defn.~\ref{defn:vector}).

If the POVM elements are known, then the linear equations in Eq.~\ref{eq:squash_1} put constraints on the elements of the Hermitian matrix $T$. After these constraints are applied, the matrix $T$ can be checked to see if an assignment of any of the remaining open parameters in $T$ can make it positive semi-definite. If $T$ is positive semi-definite then a squashing model exists. If it is not positive semi-definite then both sets of POVM elements can be mixed with classical noise to form new POVM elements that correspond to adding noise to the outcomes of Alice's and Bob's measurements. If enough noise is added then a squashing model always exists \cite{gittsovich14}. Therefore, a squashing model is only practical if no noise or a low amount of noise is added (e.g.~less than the threshold of the protocol minus any inherent errors in Alice's and Bob's devices and in the quantum channel).

In addition, there are several imperfections in devices that can be taken into account by the squashing model, such as time resolution and inefficiency \cite{gittsovich14,narasimhachar11,fung11}.

As examples of measurements that have squashing models, the BB84 active and passive measurement devices for the polarization implementation, as well as the phase implementation measurement have a squashing model to a single-qubit equivalent measurement \cite{beaudry08,tsurumaru08,narasimhachar11}. Surprisingly, the squashing model for the six-state protocol with an active basis measurement (and without the addition of noise) does not exist \cite{beaudry08}. However, noise can be added in the classical post-processing to make a squashing model possible \cite{gittsovich14}.

The squashing model is a satisfiable assumption that requires the calibrated assumption that the full description of the measurement is known. The squashing model may also require the classical addition of noise in order for a squashing model to be possible, which is a verifiable assumption without the need for further assumptions (since the classical post-processing of the measurement outcomes can be implemented in Alice's and Bob's isolated labs).

\subsection{Measurement Structure}\label{sec:bell_ur}

There are a variety of assumptions made about the structure of measurements in QKD protocols. As with the squashing model above, the calibrated assumption that the measurement POVM elements are completely known is one such assumption. However, this assumption can lead to side-channel-attack strategies for Eve whenever the measurement model deviates from the assumed description. For example, if avalanche photodiodes are used as threshold detectors (see Section~\ref{sec:threshold}), Eve can continuously shine bright light into Bob's detector. This light causes the threshold detector to have an avalanche and it cannot recover from it, since the electrons are constantly excited by the bright light. Eve can then completely control Bob's measurement device \cite{lydersen10}.

For example, in polarization BB84, if Eve stops sending bright light of a certain polarization for a time longer than the recovery time for one of Bob's threshold detectors, Eve can make Bob's threshold detectors click when she wants to. This control allows Eve to measure the states sent by Alice and force Bob's measurement device to have exactly the same outputs, which makes the protocol completely insecure.

This blinding attack has been demonstrated experimentally for BB84 \cite{lydersen10,gerhardt11} and SARG04 \cite{jain14}. This attack has also been examined for superconducting detectors in the DPS protocol \cite{fujiwara13}. Potential ways of avoiding this attack have been discussed in \cite{yuan11,stipcevic14,lim14}.

Eve may do other attacks that work outside of the model for the measurement, such as changing the timing of the signals (Section~\ref{sec:time_resolution}) or by using other degrees of freedom, such as frequency, to change the response of Bob's measurements. For example, the measurement device may be calibrated to measure a certain frequency of light. If the light is outside of a narrow range of frequencies then Bob's detector may have a lower efficiency. Eve could then perform an attack where she changes the frequency of the light depending on its state, so that if Bob gets a measurement outcome, then she has partial knowledge of which outcome for Bob is most likely.

Since the calibrated assumption that Bob's measurement device is completely known can lead to many attacks, it would be ideal to have a weaker assumption that is sufficient to still prove security. If the entropic uncertainty relation (Theorem~\ref{thm:ur}) is used then Bob's measurement device only needs to be characterized by the overlap, $\max_{x,z}\|\sqrt{F_{x}}\sqrt{G_{z}}\|_{\infty}^2$, for the POVM elements $F_x$ and $G_z$. While the POVM elements are needed to determine this overlap, there is a related overlap and corresponding uncertainty relation that can be experimentally verified without knowing the POVM elements \cite{tomamichel13} (see Definition 7.2 in \cite{tomamichelthesis}). This reduces the calibrated assumption that the POVM elements are fully known to a verifiable assumption. The experiment to verify the overlap is a CHSH test, which requires further assumptions (see Sections~\ref{sec:CHSH} and \ref{sec:bell_tests}).

\subsection{Time Resolution}\label{sec:time_resolution}

Measurements have a finite time resolution. For example, threshold detectors are not able to perform a measurement between the time an avalanche has started and the detector has recovered (see Section~\ref{sec:threshold}). This down time is usually not taken into account in the models used for security proofs and therefore Eve may get an advantage from this imperfection.

In addition to the recovery time, threshold detectors can have after pulses, where an electron that is recovering can cause another avalanche. If the time window for detection events is small enough, then after pulses could be registered as a separate detection event, causing an error.

One attack that takes advantage of the dead time is the \emph{time-shift attack}, where Eve changes the timing of the signals so that Bob is more likely to measure one state rather than another (since while one threshold detector is recovering, another detector can still click) \cite{fung07b,weier11}.

Another possible attack is the \emph{phase remapping attack}, where in the phase implementation of the BB84 protocol Eve can change Bob's phase modulation by changing the timing of the signals so that the state reaches the phase modulator right before or right after the phase modulation is applied \cite{fung07,xu10}.

Another way that information can be leaked to Eve is if Alice prepares states using parametric down-conversion (see Section~\ref{sec:PDC}) and uses a threshold detector to measure the idler. During the downtime of the threshold detector, Alice may also produce another state that will be output through the same state preparation. In this case Eve will get multiple states she can use to try to determine Alice's preparation setting. The assumption that Eve does not get an advantage with this preparation, due to the down time in the threshold detector, is satisfiable, since Alice can block any states from being output during the time between sending a state and the recovery of her threshold detector.

Alice and Bob also need to communicate their state preparation and measurement times, so that they can pair each prepared state to a measurement outcome. It turns out that if their timing information is too accurate, the timing communication may give information to Eve \cite{lamas-linares07} and so Alice and Bob should limit the accuracy of their timing information \cite{scarani09b}. Therefore, the assumption that Eve does not get an advantage from the timing information is a calibrated assumption.

\subsection{Loss} \label{sec:losses}

There are two kinds of losses in quantum-cryptography protocols: losses from the quantum channel and losses in Alice's and Bob's devices.

Loss in the quantum channel can usually be taken into account in the security proof, since Eve is allowed to do anything allowed by quantum mechanics. For the key rate, loss usually just scales the key rate by a constant, since the key rate is the number of secure bits that are produced per signal sent. However, loss can be a difficult issue for many device-independent security proofs (see Section~\ref{sec:bell_tests}). Also, the losses in a measurement are usually not taken into account in the security proof and therefore it would be convenient to have a method of relating security proofs that assume lossless measurements with the security of protocols that have lossy measurements.

Typically, security proofs assume that measurements do not have any loss but there is loss in the quantum channel. If we can model loss in a measurement as loss that occurs in the quantum channel followed by a lossless detector then we can apply these security proofs to implementations with lossy measurements.

For example, the loss in a single-photon detector is the probability that it will give a vacuum output given that it receives a single photon. Typically, it is assumed that losses happen with a fixed probability, i.e.~independent of how many photons are input to the detector and independent of the structure of the state. This is a set of calibrated assumptions, since this loss model approximately describes the ways losses occur in practice.

For measurements that have threshold detectors, a lossy threshold detector can be modelled as a beamsplitter followed by a threshold detector with perfect efficiency. One input to the beamsplitter is the input state and the other input is the vacuum; the reflected output goes to the environment and the transmitted output goes to a threshold detector with perfect efficiency.

For example, in the active BB84 measurement using polarization, there are two threshold detectors after a polarizing beamsplitter. If the efficiency of each threshold detector is the same (this is a calibrated assumption) then we can decompose each lossy threshold detector into a beamsplitter and a lossless threshold detector. Since a beamsplitter commutes with a polarizing beamsplitter, the two equivalent beamsplitters after a polarizing beamsplitter are equivalent to one of these beamsplitters followed by the polarizing beamsplitter (see Fig.~\ref{fig:commute_beamsplitter}). This means that loss in the detector can be modelled as loss that occurs in the quantum channel followed by a lossless detector. Now we can apply a security proof that assumes that there are losses in the quantum channel to this implementation.

\begin{figure} \centering
\includegraphics[width=0.8\textwidth]{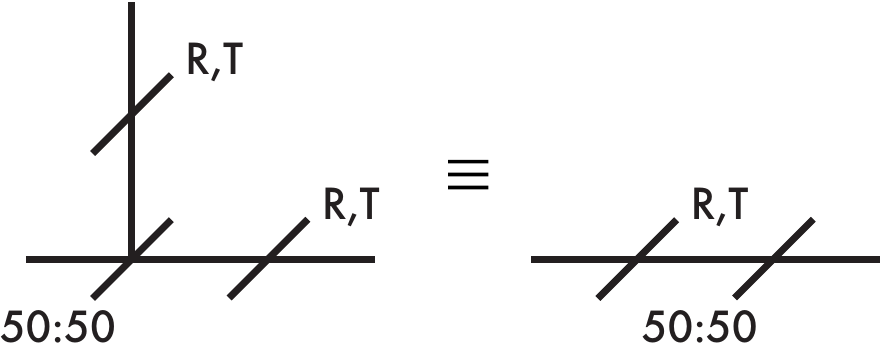}
\caption[Commuting Beamsplitters]{Commuting beamsplitters. The situation on the left is a 50:50 beamsplitter followed by two identical beamsplitters with transmissivity $T$ and reflectivity $R$. This situation is equivalent to one of the $R,T$ beamsplitters followed by a 50:50 beamsplitter.}
\label{fig:commute_beamsplitter}
\end{figure}

However, if the losses for each threshold detector are not the same then this imperfection needs to be taken into account by the security proof (see \cite{fung09} for an example).

As another example, consider the measurement used in the phase implementation of BB84 (see Fig.~\ref{fig:phase}). If there is loss in the threshold detectors then (under the same calibrated assumption that the loss can be modelled as a beamsplitter) the measurement is equivalent to a beamsplitter followed by a lossless measurement. However, the phase modulator may have loss as well. In this case, if the loss is modelled as a beamsplitter then to commute the beamsplitters and make the same argument as with the active polarization BB84 measurement, it changes the ratio of the 50:50 beamsplitter to a new ratio (see \cite{ferenczi12} for the details of this ratio). Then a new security proof is needed for a lossless measurement that does not have a 50:50 beamsplitter \cite{ferenczi12}.

In general, losses in measurements need to be taken into account by either having a calibrated assumption about the model of the loss (to separate a lossy measurement into loss followed by a lossless measurement) or by including lossy measurements directly in the security proof.

\subsection{Bell Tests} \label{sec:bell_tests}

Device-independent QKD protocols require a CHSH experiment or other kind of Bell test (see Section~\ref{sec:DI}). To use the outputs of a Bell test for QKD, some assumptions about the devices and channel may be necessary. If in an implementation there are deviations from these assumptions, such that a Bell inequality violation is observed but there are states that are described by a local hidden variable model that could give the same measurement outcomes, then there is a \emph{loophole} in the implementation.

Two recent reviews discuss many of the loopholes that exist in implementations of Bell tests \cite{larsson14,brunner14}. A major loophole is that detectors are not perfectly efficient so some measurement outcomes are lost \cite{ma08,lim13a}. These lost measurements should be replaced with uniformly random bits to ensure that an eavesdropper did not correlate the losses with Alice's and Bob's measurement outcomes. Another loophole is that Alice's and Bob's devices need to be spatially separated, otherwise there may be signalling between them. For example, it could be that Alice's measurement is chosen at her end and then a message is sent to Bob's measurement to tell it to output a bit that depends on Alice's measurement outcome and basis choice. This situation can be avoided if Alice and Bob ensure the timing of their measurements is close enough that signalling is not possible.

It has been assumed for many device-independent security proofs that each measurement is done independently (see \cite{arnon-friedman12,barrett13a} and references therein). This assumption could only be justified in the device-independent setting if a separate measurement device was used, which is completely impractical. However, now there are security proofs and methods to avoid this assumption \cite{barrett12,reichardt13,vazirani12}.

The assumptions for Bell tests are usually fundamental assumptions because they require a strict adherence to a perfect model that is not possible in practice.

\subsection{Sampling with Measurements} \label{sec:exam_problem}

For P\&M protocols with a basis choice, Alice and Bob need to sift out measurement outcomes that were not prepared and measured in the same basis. Similarly, for entanglement based protocols with a basis choice, Alice and Bob will remove outcomes where they did not measure in the same basis. For a protocol with two bases, such as BB84, the bases can be chosen uniformly at random. However, this means that half of the measurement outcomes will be removed. To increase the key rate, it may be preferable to bias one basis over the other. If one basis is chosen with probability $p<1/2$ and the other with probability $1-p$ then the fraction of measurement outcomes that are sifted will be approximately $2p(1-p)$.

However, the probability $p$ cannot be made too small, or else there will not be enough measurements in one basis for parameter estimation. This probability can then be optimized for the given protocol, such as in \cite{tomamichel12a}.

As we have described parameter estimation in this thesis thus far, it is performed by Alice or Bob picking a random subset of their measurement outcomes to be communicated to the other party. However, for protocols with a biased basis choice it is preferable to use the basis that is chosen with probability $p$ to be used for parameter estimation, while the other basis' measurement outcomes will be used for the key. This assignment would be sufficient for parameter estimation, since by using a tool like the entropic uncertainty relation (Theorem~\ref{thm:ur}), the measurement outcomes of one basis are sufficient to prove security (see Section~\ref{sec:device_dependent}).

However, to perform parameter estimation, it is important to be able to infer the statistics of Alice's and Bob's entire strings from a small sample. Typically a result like Serfling's inequality (Lemma~\ref{lemma:serfling}) is used to do this estimation. These kind of results require that the sample is taken uniformly at random from a classical string. For example, Alice may choose a random subset of her string to communicate to Bob by using some classical randomness. Instead, the sampling can be done through a measurement's basis choice, where Alice and Bob may only communicate measurement results from one basis and use the measurement results from the other basis for the key. In this setting it is not clear if the same sampling statistics like Serfling's inequality still hold and it is currently an open question \cite{beaudry14}. Ideally, Alice and Bob would still only need to communicate a small fraction of their strings in order to get the same bound as in Lemma~\ref{lemma:serfling}. Currently, the assumption that sampling by using a basis choice gives the same estimate as classical estimation is a fundamental assumption.

\section{Classical Post-Processing Assumptions} \label{sec:post-processing}

Not only are there assumptions about the devices and quantum states used in the protocol but there are assumptions made about the classical components of the protocol and the classical post-processing as well.

First, the randomness used in the protocol (such as for basis choices, picking the random sample in parameter estimation, picking a hash function for privacy amplification, etc.) must be truly random (see Randomness Extraction in Section~\ref{sec:qc_protocols}). If the protocol does not use true random numbers then it is not clear if the steps of the protocol that use randomness still produce the same results. Therefore, it is important to investigate if security still holds if this assumption is relaxed. For example, if Eve is able to control some of Alice's and Bob's basis choices then Alice and Bob may still be able to certify that they have violated a Bell inequality \cite{koh12}. However, if true random numbers are generated then the assumption that true randomness is used is a verifiable assumption. If the random numbers are only approximately truly random (for example, only a lower bound on the min-entropy of some classical data is known) then the assumption can be verifiable, since a randomness extraction protocol can extract almost perfect randomness from an imperfect source (see Sections~\ref{sec:qc_protocols} and \ref{sec:privacy_amplification2}). Randomness extraction is universally composable, so randomness that is extracted from an imperfect random source can be combined with a universally composable QKD protocol that uses that randomness. If the QKD protocol requires perfectly uniform randomness then it can use the almost perfect randomness from a randomness extraction protocol and the QKD protocol will still be secure. The security parameter for the composition of randomness extraction and QKD will depend on the specific protocols that are used (see \cite{portmann14} for an example of composition of QKD with another protocol).

Second, Alice and Bob have to estimate how much information Eve gets from the classical communication sent during information reconciliation so that they can remove Eve's knowledge of Alice and Bob's string in privacy amplification. If this estimation is too low, then the protocol may be insecure. A precise analysis of how much information Eve gets from the communication in information reconciliation can be found in \cite{tomamichel14a}.

Lastly, the classical computers that are used to store the measurement outcomes and to perform the classical post-processing should be isolated. For example, if these computers are connected to the internet then Eve may be able to hack into the computer to discover Alice's and Bob's strings. It is therefore a fundamental assumption that Eve does not get access to Alice's and Bob's classical computers.

%% file: Contributions.tex

\chapter{Contributions} \label{chap:contributions}

\section{Introduction}

In this chapter we review two contributions that are relevant to the framework discussed in this thesis.

The first contribution is a security proof of two QKD protocols that use two-way quantum communication. In these protocols Alice sends states to Bob, Bob performs an encoding operation on these states, and then he sends the states back to Alice \cite{beaudry13}. Here we present these idealized protocols and discuss the assumptions necessary for the security proofs to hold.

The second contribution is a proof of one of the most fundamental properties in information theory: the data-processing inequality \cite{beaudry12a}. Informally, this inequality states that if a physical system undergoes a transformation, then the information content of that system cannot increase. If the inequality was not true, then the world we would live in would be very strange! For example, a computer could run an algorithm such that the computer would learn everything about the universe without having to interact with the universe.

The data-processing inequality is used extensively throughout quantum information theory in a variety of contexts. We include the data-processing inequality in this thesis because of its fundamental nature and widespread use, but also because our proof is of a similar spirit to this thesis: this proof separates a fundamental property that is easy to prove from a specialization that is difficult. In Chapters~\ref{chap:security_proofs} and \ref{chap:assumptions} we have separated the `simple' task of proving security for an idealized model of a QKD protocol from the more challenging task of connecting the security statement with implementations.

In addition to these two contributions, this thesis is a contribution, since it discusses several topics in the field of quantum cryptography and quantum key distribution: security, security proof methods and assumptions. Another contribution was outlined in Section~\ref{sec:squashing}, which was work done prior to this PhD \cite{beaudry08,moroder10,gittsovich14}. There is another contribution that deals with generalizing the mutual information but it is unrelated to the contents of this thesis \cite{ciganovic14}. Lastly, there is a contribution in progress that is related to Section~\ref{sec:exam_problem} \cite{beaudry14}.

\section{Two-Way QKD} \label{sec:two_protocols}

Two-way QKD uses two quantum channels: one from Alice to Bob and one from Bob to Alice. In practical implementations, the second channel may be the first channel in the reverse order. These protocols have been introduced as an alternative to one-way QKD, which could have higher key rates compared to one-way equivalents in certain implementation scenarios \cite{cai04a,cai04b,deng04,bostrom02,lucamarini05}.

One inefficiency in BB84-like protocols is the need for basis sifting. Since Alice and Bob need to throw away a fraction of their key, the key rate would be higher if they used a protocol that is deterministic, e.g.~that has a basis choice but does not require basis sifting. Since Alice both prepares and measures the states in the two-way protocols we consider, she can choose the basis of her measurement based on what state she prepared. Therefore, no basis sifting is necessary for these two-way protocols.

Protocols with a basis choice can be made more efficient by choosing the bases with different probabilities (see Section~\ref{sec:exam_problem}). In the infinite-key scenario, this removes the effect of basis sifting, since an arbitrarily strong bias of one basis over the other(s) can be made while still getting perfect statistics about measurement outcomes from each basis. However, in the finite-key regime there is a tradeoff between the basis choice probability and other parameters in the protocol, which results in a limit on the basis bias (see Table II of \cite{tomamichel12a} for an example). Deterministic two-way protocols might have an efficiency advantage over protocols with basis sifting in the finite-key scenario \cite{beaudry13}.

There also exist two-way implementations of one-way protocols. For example, the BB84 protocol can be implemented in a \emph{Plug \& Play} version \cite{ribordy98,muller97}, which is a two-way protocol. In this protocol, one bit is communicated for every quantum state sent forwards and backwards through the same quantum channel. However, it is known that two bits can be communicated by sending one qubit (this task is called super dense coding, see Section~\ref{sec:SDC}) \cite{bennett92d}. We propose a QKD protocol that can send two bits of information per qubit and prove its security.

A severe limitation of two-way QKD is the scaling of the losses with the length of the quantum channel between Alice and Bob. If the loss in each channel is $\eta$ then the total loss is $\eta^2$. This means that two-way QKD is primarily only useful for short-range applications. However, it may still be useful to perform QKD from the ground to a satellite \cite{vallone14}.

Here we describe two two-way protocols in their perfect forms and then discuss how the security proofs of these protocols fit into the framework for assumptions discussed in Chapter~\ref{chap:assumptions}.

\subsection{Modified LM05 QKD Protocol} \label{sec:LM05}

The LM05 protocol described here is a modified version of the original protocol in \cite{lucamarini05}. Alice prepares one of the four states $\{\ket{0},\ket{1},\ket{+},\ket{-}\}$ uniformly at random from the BB84 protocol and sends the state to Bob through one quantum channel (see Fig.~\ref{fig:LM05}).

Bob applies a map with probability $r$ or a measurement followed by a state preparation with probability $1-r$  to the output of the first channel. The map is chosen uniformly at random from one of the four maps $\{\id,\sigma_X,\sigma_Y,\sigma_Z\}$ (which we call an encoding), where $\sigma_i$ are the Pauli operators (see Section~\ref{sec:discrete_protocols}) and $\id$ is the identity map. These maps are applied to the state $\rho$ so that the output is $\sigma_i\rho\sigma_i$.

Bob's measurement and state preparation is one of two possibilities, which defines two versions of the protocol. In Version 1, Bob applies an $X$-basis measurement and in Version 2 he uniformly at random picks between a measurement in the $X$ basis or the $Z$ basis. After his measurement he prepares the state that corresponds to his measurement outcome. For example, if he measured $\ket{0}$ in the $Z$ basis then he would prepare the state $\ket{0}$.

Bob sends his newly prepared state or the outcome of the encoding on his received state into a second quantum channel back to Alice. Alice then measures the state out of the second channel either in the $X$ or $Z$ basis as in the BB84 protocol (see Section~\ref{sec:discrete_protocols}). Alice chooses her basis to match the state she prepared. For example, if she prepared $\ket{-}$ then she will measure in the $X$ basis.

Alice and Bob will use most of the instances where an encoding was performed for the key. They will use the times when Bob did a measurement (and preparation of a state) for parameter estimation only. For direct reconciliation Alice publicly reveals which basis she used for each signal. Bob will reveal which basis he measured in and when he measured instead of applying a map. Alice's raw key is made up of the XOR of her measurement outcomes and her preparation bit (which is $0$ when she prepared $\ket{0}$ or $\ket{+}$ and $1$ when she prepared $\ket{1}$ and $\ket{-}$). Bob's string comes from his encoding operation which correspond the two bits $00, 10, 11, 01$ to $\identity,\sigma_X,\sigma_Y,\sigma_Z$ respectively. When Alice measures in the $Z$ basis Bob keeps his first bit and when Alice measures in the $X$ basis Bob keeps his second bit. The case of reverse reconciliation can be treated similarly.\footnote{In \cite{beaudry13} the roles of Alice and Bob are reversed compared to the way they are presented here. This means that reverse reconciliation for this thesis means direct reconciliation in the paper. We reverse the roles here to be consistent with our definition of direct and reverse reconciliation in this thesis.}

\begin{figure} \centering
\includegraphics[width=\textwidth]{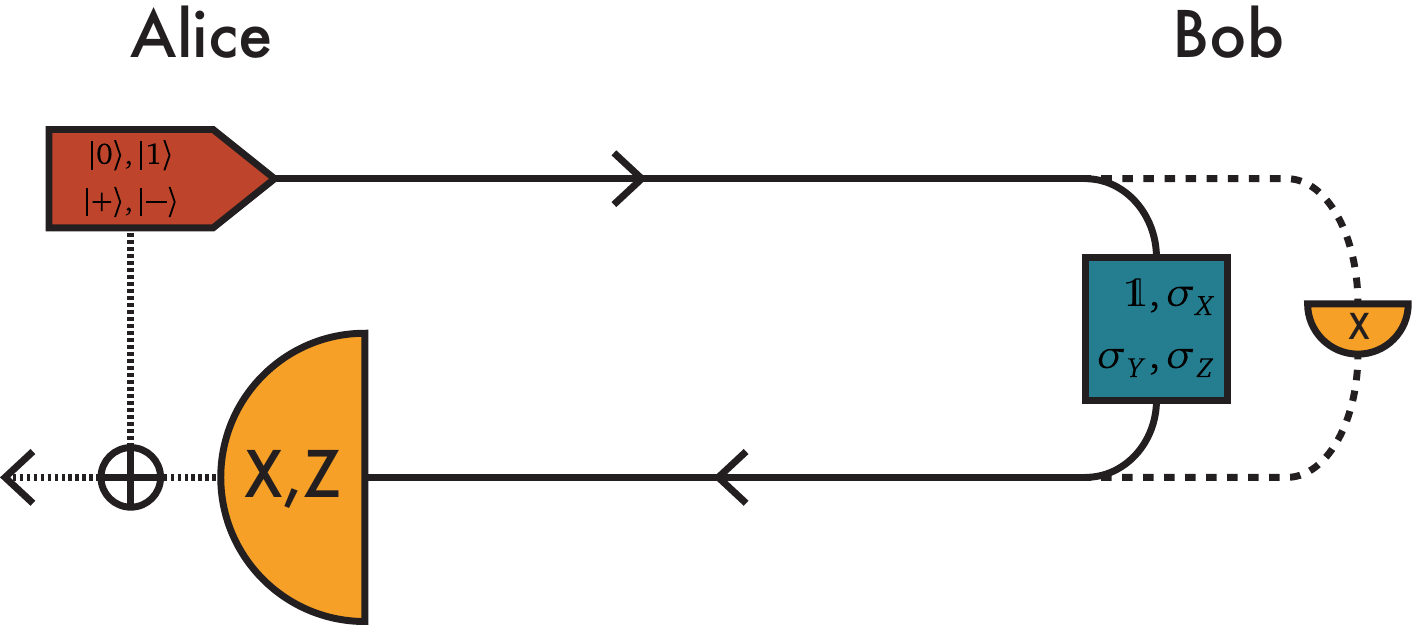}
\caption[The Modified LM05 QKD Protocol]{The modified LM05 protocol (Version 1). Alice prepares one of the four BB84 states and sends it to Bob through one insecure quantum channel. Bob either does one of four possible encoding operations and sends the state back to Alice or a measurement in the $X$ basis. After the measurement he sends the post-measurement state back to Alice. Alice measures in the $X$ or $Z$ basis at random and adds the bit that corresponds to her state preparation and measurement outcome together to form her string. Bob's string is constructed depending on which encoding operations he performs.}
\label{fig:LM05}
\end{figure}

The modification from the original LM05 protocol is the addition of the measurement on Bob's system (see Fig.~\ref{fig:LM05}). This measurement is sufficient to avoid an attack by Eve that gives her full information and does not introduce any errors in Alice's measurements. For example, when Eve receives a state from Alice in the first channel, Eve can store it in a quantum memory and can prepare a maximally entangled Bell state $\ket{\psi^+}$ (see Eq.~\ref{eq:bell_2}). Eve sends half of the Bell state into Bob's encoding. In the unmodified LM05 protocol this state goes directly into Bob's encoding and Eve gets the outcome. Eve can do a Bell measurement on the outcome of the encoding and the other half of her prepared entangled state to determine Bob's encoding with certainty. Eve can apply the encoding to the state she stored from Alice and send it back to Alice in the second channel. By using this attack, Eve learns Bob's encoding and Alice does not get any errors in her measurement. By adding an $X$-basis measurement (in Version 1) or both $Z$- and $X$-basis measurements (Version 2) on Bob's side, Alice and Bob will estimate an error rate of $1/2$ in the first channel (after post-selecting on when Alice's and Bob's bases match) if Eve tries this attack.

We discuss the security proof of the LM05 protocol and the assumptions required in Sections~\ref{sec:2_security} and \ref{sec:2_assume} below.

\subsection{Super Dense Coding QKD Protocol} \label{sec:SDC}

The super dense coding (SDC) QKD protocol is similar to the LM05 protocol and is introduced in \cite{beaudry13}. The SDC QKD protocol is based on the quantum information task of super dense coding (hence the name).

Super dense coding is the task of sending two classical bits from Alice to Bob by Alice sending one qubit to Bob. To do this task, one qubit of a maximally entangled two-qubit state is sent to Alice and the other qubit is sent to Bob. For example, this state may be
\begin{equation} \label{eq:bell_1}
\ket{\psi^+}_{AB} = \frac{\ket{00}_{AB}+\ket{11}_{AB}}{\sqrt{2}}.
\end{equation}
Depending on the two bits Alice wants to communicate $\{00,01,10,11\}$ she will apply $\{\id,\sigma_X,\sigma_Z,\sigma_Y\}$ respectively, where $\sigma_i$ are the Pauli operators (see Section~\ref{sec:discrete_protocols}) and $\id$ is the identity map. As in the LM05 protocol, these are applied as $\sigma_i\rho_A\sigma_i$ to the input $\rho_A$. This map results in either the state $\ket{\psi^+}$ for $00$ or one of the other three states from the Bell basis (Eq.~\ref{eq:bell_2}). Alice now sends her qubit to Bob. Bob can do a measurement in the Bell basis to determine which state he has and learn which two bits Alice wants to communicate.

Note that Alice's sent qubit appears to be the maximally mixed state $\identity_A/2$ to an eavesdropper who does not know Alice's two bits. We show that this kind of protocol can be used for QKD and is provably secure \cite{beaudry13}.

The SDC QKD protocol starts with Alice preparing $\ket{\psi^+}$ and keeping one qubit of the state in a quantum memory (see Fig.~\ref{fig:SDC}). She sends the other qubit to Bob through a quantum channel. Bob will, with probability $r$, apply the same set of maps that are used in the LM05 protocol and superdense coding (uniformly at random one of $\{\id,\sigma_X,\sigma_Z,\sigma_Y\}$); and, with probability $1-r$, he measures the qubit in the $Z$ basis and then uniformly at random prepares a state in the $X$ basis. Bob sends this random $X$-basis state or the outcome of his encoding back to Alice. Alice will, with probability $r$, measure in the Bell basis. This measurement will ideally tell her what map Bob applied and therefore she learns two bits from Bob. With probability $1-r$ Alice measures her stored qubit in the $Z$ basis and her received qubit in the $X$ basis.

The encoding and Bell measurement are used for the key, while the $Z$- and $X$-basis measurements are used for parameter estimation.

\begin{figure} \centering
\includegraphics[width=\textwidth]{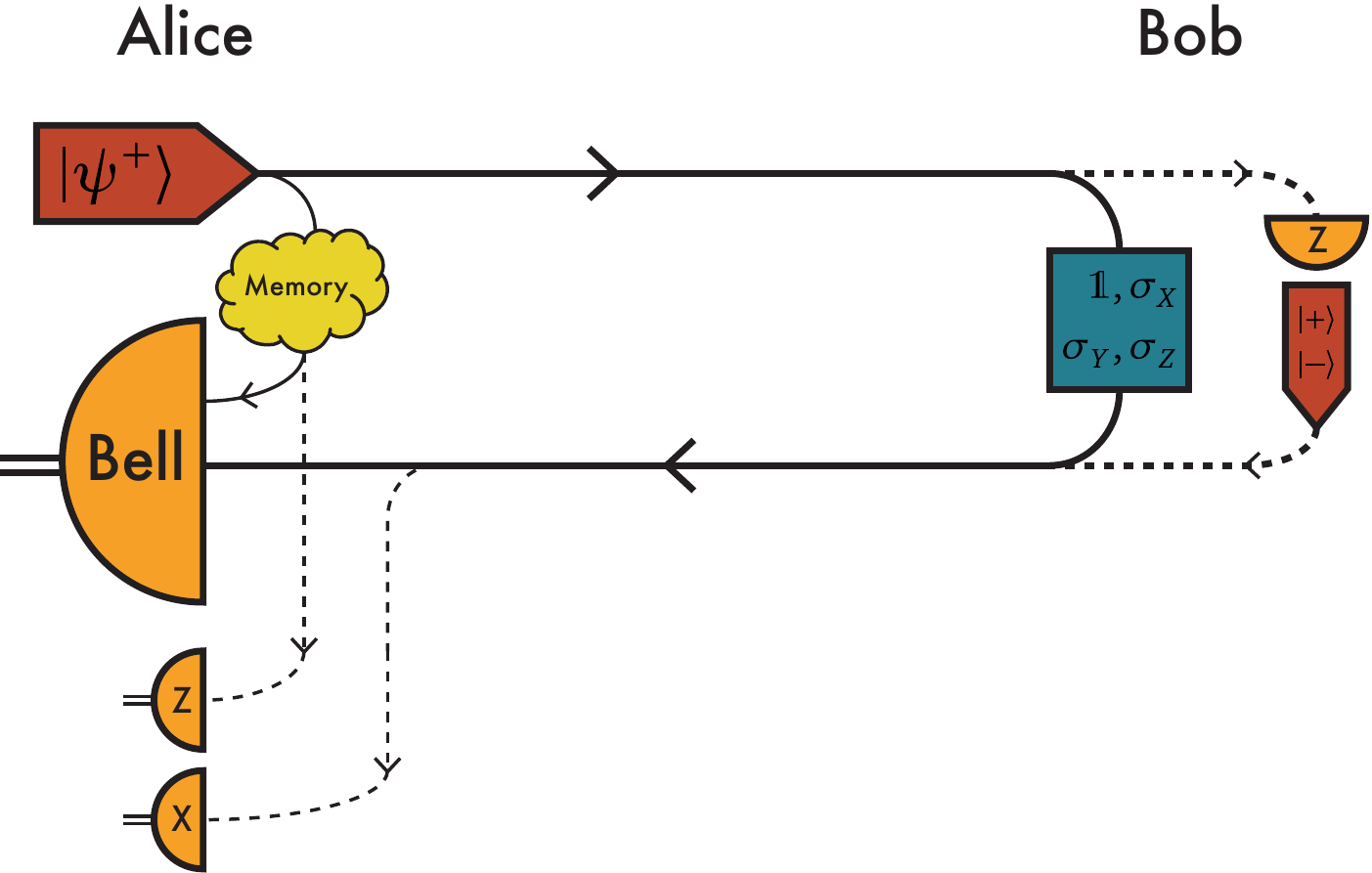}
\caption[The SDC QKD Protocol]{The SDC QKD protocol. Alice prepares the maximally entangled state $\ket{\psi^+}$ and stores half of it in a quantum memory. Bob applies an encoding or a $Z$-basis measurement followed by a random preparation of $\ket{+}$ or $\ket{-}$. Bob sends the resulting state back to Alice who either does a Bell measurement or measures her state in the $Z$ basis and Bob's returned state in the $X$ basis.}
\label{fig:SDC}
\end{figure}

We now discuss the security of the LM05 and SDC protocol, followed by the assumptions under which security holds.

\subsection{Security Proofs of Modified LM05 and SDC QKD} \label{sec:2_security}

The protocol model under which we prove security for both the modified LM05 protocol and the SDC protocol is for the protocol class [partially device-dependent, infinite-key, qubits, basis choice\footnote{While the LM05 protocol has a basis choice, this choice is \emph{deterministic} in that Alice knows which basis to measure in and therefore there is no basis sifting needed.}, coherent attacks] (see Section~\ref{sec:classes}). The security proof applies to both the P\&M protocols described above and equivalent entanglement-based protocols. The proof also applies regardless of whether the measurements are done in an active or passive way.

The proofs of security for both protocols use the entropic uncertainty relation of Theorem~\ref{thm:ur} as well as the Devetak-Winter key rate (see Section~\ref{sec:device_dependent}). To apply these tools, the P\&M protocols are shown to be equivalent to an entanglement-based protocol.

For the LM05 protocol, by using the uncertainty relation and the Devetak-Winter rate, we get a lower bound on the key rate of
\begin{equation} \label{eq:key_LM05}
r \geq 1 - \min_i h(q_{G^i})-h(q_F),
\end{equation}
where $h(\cdot)$ is the binary entropy function (Defn.~\ref{defn:binary_entropy}), $q_F$ is the error rate in Alice's measurement compared to Bob's prepared states, and $q_{G^i}$ is the error rate in Bob's measurements with Alice's preparations, where $i\in\{0,1\}$ denotes whether it is the error rate in the $Z$ basis from the first channel or the $X$ basis in the second channel. Eq.~\ref{eq:key_LM05} is an improvement on the key rate of \cite{lu11}. In addition, we make less assumptions about the states and devices than \cite{lu11} (see Section~\ref{sec:2_assume}).

For the SDC protocol we get a lower bound on the key rate of
\begin{equation}
r \geq 2 -h_4(q_G)-h_4(q_F),
\end{equation}
where $h_4(q)$ is the Shannon entropy of a distribution with four values, $q=\{q_1,q_2,q_3,q_4\}$,
\begin{equation}
h_4(q) := - q_1\log q_1 - q_2 \log q_2 - q_3\log q_3 - q_4\log q_4,
\end{equation}
where $\sum_{i=1}^4 q_i=1$ and $q_i\geq 0, i\in \{1,2,3,4\}$. The error rate $q_G$ is the set of errors between Alice's and Bob's $Z$- and $X$-basis measurements, $q_G:=\{q_G^1,q_G^2,q_G^3,q_G^4\}$: $q_G^1$ corresponds to no error, $q_G^2$ corresponds to an error in the $Z$ basis but not in the $X$ basis, $q_G^3$ to an error in the $X$ basis but not in the $Z$ basis, and $q_G^4$ to an error in both the $Z$ and $X$ basis. The error rate $q_F:=\{q_F^1,q_F^2,q_F^3,q_F^4\}$ is the set of errors between Bob's encoding and Alice's Bell measurement. Each of the error rates correspond to whether there is no error in either bit of the Bell measurement, an error in the first bit only, an error in the second bit only, and whether there is an error in both bits.

\subsection{LM05 and SDC Assumptions} \label{sec:2_assume}

We now consider the assumptions necessary for the security proofs to apply to an implementation of the LM05 or SDC protocol \cite{beaudry13}. We use the classification of Section~\ref{sec:assumption_classification} to denote what kinds of assumptions they are and how they are justified: fundamental, calibrated, verifiable, or satisfiable.

\begin{enumerate}
\item {\bf Qubits are prepared.}

This assumption is fundamental, since qubits cannot be prepared in practice. However, this assumption can be removed: if Alice prepares entangled bipartite states and does a measurement with a basis choice on one half of it, using the other half as her prepared state sent into the first channel, then no assumption is necessary about the preparation of states. 

This assumption is necessary to make the connection between the P\&M protocols and their entanglement-based equivalents.

\item {\bf Bob's output encoded state is a fixed state.}

Formally, this assumption can be stated using the maps that Bob uses in his encoding, $\mathcal{E}_i$, which maps states acting on $\hilbert_A$ to states acting on $\hilbert_D$, as
\begin{equation}\label{eq:encoding_assumption}
\frac{1}{4}\sum_{i=1}^4\mathcal{E}_i (\rho_A) = \sigma_D \quad\forall \rho_A\in S_{=}(\hilbert_A).
\end{equation}

This assumption is calibrated, since Bob can calibrate his encoding device so that on average over the encodings his state is approximately a fixed state, $\sigma_D$. Ideally $\sigma_D$ should be a maximally mixed state, since in the perfect protocol description we have:
\begin{equation}
\frac{1}{4}\left(\rho+\sigma_X\rho\sigma_X+\sigma_Y\rho\sigma_Y+\sigma_Z\rho\sigma_Z\right) = \frac{\identity}{2},
\end{equation}
for any qubit state $\rho$. However, we do not assume that Bob receives a qubit. We only assume that Eq.~\ref{eq:encoding_assumption} holds for any input state $\rho_A$, regardless of the Hilbert space $\hilbert_A$.

This assumption is necessary in order to make the connection between the P\&M protocols and their entanglement based versions.

\item {\bf Measurements detect each signal independently.}

This assumption is calibrated, since the measurement devices are not designed to have a memory. However, if threshold detectors are used then during their dead times they may have some memory effects (see Section~\ref{sec:measurements}).

This assumption is necessary to use the uncertainty relation and the Devetak-Winter rate applied to each signal independently.

\item {\bf Alice or Bob's devices are characterized by a single constant.}

This assumption is verifiable. Depending on whether direct or reverse reconciliation is performed, the uncertainty relation is applied such that Alice's or Bob's devices only need to be characterized by a single constant. This constant can be verified for Alice's measurement by doing a Bell test (see Section~\ref{sec:bell_ur}).

\item {\bf There are no losses.}

This assumption is fundamental, since in practice there will always be some loss. It was made to simplify the analysis.
\end{enumerate}

If Alice does her preparations using entangled states then these security proofs require one unjustified fundamental assumption, two calibrated assumptions, and one verifiable assumption. There are also other fundamental assumptions that are implicitly made in all QKD implementations, such as the isolation of Alice's and Bob's labs and that quantum mechanics is correct (see Section~\ref{sec:universal_assumptions}). It remains as future work to extend this security proof to take losses and finite-key effects into account.

\section{The Data-Processing Inequality} \label{sec:DPI}

The data-processing inequality was first proven in \cite{lieb73,lieb73b}. However, the proof of this inequality has remained quite challenging: it requires many mathematical tools to prove this very fundamental property of information. For example, there have been initial proofs using abstract operator properties \cite{lieb73,lieb73b,simon79} and their simplified versions \cite{nielsen05,petz86,ruskai07}. There are other proofs using the operational meaning of the von Neumann entropy \cite{horodecki06,horodecki05}, Minkoski inequalities \cite{carlen99,carlen08}, or holographic gravity theory \cite{headrick07,hirata07}. All of these different techniques give insight into why the data-processing inequality is true.

Our contribution is a proof that provides intuition as to why the proof is difficult: the property itself is not hard to prove, it is the specialization of this property to the quantity used: the von Neumann entropy \cite{beaudry12a}.

Formally, the data-processing inequality for the von Neumann entropy is stated as
\begin{equation}\label{eq:DPI2}
H(A|BC)_{\rho} \leq H(A|B)_{\rho},
\end{equation}
for a quantum state $\rho_{ABC} \in S_{=}(\hilbert_{ABC})$. As mentioned in Section~\ref{sec:iid_entropy}, this inequality also implies that the entropy cannot decrease under any CPTP map acting on $B$.

We can prove the inequality for the smooth min-entropy, where the proof follows almost directly from the definition and then we specialize the smooth-min-entropy inequality to the von Neumann entropy via the quantum asymptotic equipartition property (QAEP) (Theorem~\ref{thm:qaep}) \cite{tomamichel08}. We also provide an alternative proof to the QAEP than \cite{tomamichel08}, where we do not concern ourselves with the rate at which the smooth min-entropy approaches the von Neumann entropy in the limit of Theorem~\ref{thm:qaep} \cite{beaudry12a}.

The proof methods used in our proof of the data-processing inequality show the power of quantum information theory: it can prove this fundamental property in a simple way for a more general quantity. In addition, our proof highlights that one-shot entropies are more fundamental than the von Neumann entropy, since the proof of the data-processing inequality is easy for the smooth min-entropy.

The proof of the data-processing inequality for the min-entropy is simple, so we reproduce it here. 

\begin{thm} [Smooth min-entropy data-processing inequality \cite{rennerphd,tomamichel10c,koenig08,beaudry12a}] \label{thm:DPI_min} Let $\rho_{ABC} \in S_{=}(\mathcal{H}_{ABC})$. Then
\begin{equation} \label{eq:dpi}
H_{\min}^{\epsilon}(A|BC)_{\rho} \leq H_{\min}^{\epsilon}(A|B)_{\rho}.
\end{equation}
\end{thm}
\begin{proof} The proof works by finding a candidate solution to the maximization inside $H^{\epsilon}_{\min}(A|B)$ by using $H^{\epsilon}_{\min}(A|BC)$. Their definitions (Defns.~\ref{defn:Hmin} and \ref{defn:sHminmax} together) are
\begin{align}
H_{\min}^{\epsilon}(A|B) &= \max_{\rho'_{AB} \in \mathcal{B}^{\epsilon}(\rho_{AB})}  \max_{\sigma'_B} \sup_{\lambda'} \{ \lambda' : \rho'_{AB} \leq 2^{-\lambda'} \identity_A \otimes \sigma'_B\}, \\
H_{\min}^{\epsilon}(A|BC) &= \max_{\tilde{\rho}_{ABC} \in \mathcal{B}^{\epsilon}(\rho_{ABC})} \max_{\sigma_{BC}} \sup_{\lambda} \{ \lambda : \tilde{\rho}_{ABC} \leq 2^{-\lambda} \identity_A \otimes \sigma_{BC}\}.
\end{align}

First, we we find the optimal values for all of the maximizations in the definition of $H_{\min}^{\epsilon}(A|BC)$ so we have $\lambda = H_{\min}^{\epsilon}(A|BC)_{\rho}$, $\tilde{\rho}_{ABC} \in \ball{\rho_{ABC}}$, and $\sigma_{BC}$ fixed. We have the inequality
\begin{equation}
\tilde{\rho}_{ABC} \leq 2^{-\lambda} \mathbbm{1}_{A} \otimes \sigma_{BC}.
\end{equation}
By applying the trace over system $C$ this inequality becomes
\begin{equation}
\tilde{\rho}_{AB} \leq 2^{-\lambda} \mathbbm{1}_{A} \otimes \sigma_{B}.
\end{equation}
We know that $\tilde{\rho}_{ABC} \in \ball{\rho_{ABC}}$ and so $P(\rho_{ABC},\tilde{\rho}_{ABC})\leq \epsilon$. Since the purified distance does not increase under the partial trace (Lemma~\ref{lemma:CPTP_dist}), it follows that $P(\rho_{AB},\tilde{\rho}_{AB})\leq \epsilon$. Therefore we have $\tilde{\rho}_{AB} \in \ball{\rho_{AB}}$, and $\sigma_{B}\in S_{=}(\mathcal{H}_{B})$, which are candidates for maximizing $H_{\min}^{\epsilon}(A|B)_{\rho}$. Since the optimal values for the maximizations in $H_{\min}^{\epsilon}(A|B)$ will result in the largest $\lambda'$, it follows that $\lambda\leq \lambda'$, which is the desired inequality.
\end{proof}

Combining the QAEP with the DPI for the smooth min-entropy immediately implies the data-processing inequality for the von Neumann entropy.

The QAEP can be proved by upper and lower bounding the min-entropy by the von Neumann entropy in the limit as the smoothing parameter goes to zero and the number of systems, $n$, is taken to infinity (called the i.i.d.~limit).

For the upper bound, we upper bound the smooth min-entropy by the von Neumann entropy of a state that is close to the state in the min-entropy. Then in the i.i.d.~limit, since the von Neumann entropy is continuous in its state (via Fannes' inequality \cite{fannes73}), the von Neumann entropy of this close state approaches the i.i.d.~state.

For the lower bound, we use a chain rule to break up the conditional min-entropy into a sum of two non-conditional entropies and these can be lower bounded in the i.i.d.~limit by the non-conditional R{\'e}nyi entropy \cite{renyi61}.

In summary, we have shown a proof of the data-processing inequality of the min-entropy (which is relatively simple) and performed a specialization of this data-processing inequality to the von Neumann entropy (which is much more involved than the proof for the min-entropy DPI). Because of the QAEP, the min-entropy can be thought of as a generalization of the von Neumann entropy to the one-shot scenario. This means that the more fundamental property is the data-processing inequality for the min-entropy, which is easily proved. The difficulty in proving the data-processing inequality for the von Neumann entropy can then be thought of as trying to prove both the fundamental inequality and the specialization at the same time, which has so far proven difficult.

%% file: Conclusion.tex

\chapter{Conclusion and Outlook} \label{chap:conclusion}

This thesis has presented the current understanding of security in quantum key distribution and recent techniques used to prove security for various protocols (Chapter~\ref{chap:security_proofs}). Many common assumptions made in QKD and quantum cryptography have also been discussed (Chapter~\ref{chap:assumptions}). We also outlined two contributions: an intuitive proof of the data-processing inequality and the security proofs of two two-way QKD protocols (Chapter~\ref{chap:contributions}).

In particular, we have presented two frameworks that can be used to classify QKD protocols and their assumptions.

Protocol classes (Section~\ref{sec:classes}) allow for the classification of proof techniques to specify which techniques apply to which kinds of protocols. Several proof techniques and reductions were explained in Section~\ref{sec:methods} and their applicability was specified using the framework of protocol classes.

Assumption classes (Section~\ref{sec:assumption_classification}) classify assumptions into four types: fundamental, calibrated, verifiable, and satisfiable. These clarify the level of justification these assumptions have and whether Eve can exploit them or not. Fundamental assumptions are either dependent on the underlying physical theory (and are therefore justified) or are completely unjustified. Calibrated assumptions are approximately justified by the structure of the devices but Eve can exploit these assumptions to get partial or full information about the key. Verifiable assumptions are completely justified by an experimental test. Satisfiable assumptions are justified by a modification of the protocol, but may require further assumptions that are not justified.

Several examples of the kinds of assumptions made in quantum cryptography and QKD protocols were analyzed in Chapter~\ref{chap:assumptions}. However, this was certainly not an exhaustive list. As several results have shown \cite{lutkenhaus99,dusek99,kurtsiefer01,lo05b,gisin06,lamas-linares07,fung07b,fung07,lydersen10,weier11,barrett13a,sun214,jain14} there are assumptions that have just not been thought of previously and this will continue to be the case for the models of device-dependent protocols and their implementations. In addition, this chapter focused on QKD but many of the assumptions made there also apply to the implementations of other quantum-cryptography protocols (for example, see the recent coin-flipping experiment \cite{berlin11}).

Ideally, a quantum-cryptography protocol should only make verifiable assumptions; fundamental assumptions that rely on the laws of physics or assumptions that are unavoidable (such as the isolation of Alice's and Bob's labs); or satisfiable assumptions that only lead to assumptions that are justified. There is no such proof in QKD that is both physically implementable with current technology and only has assumptions of these types. This is the goal, first motivated in \cite{ekert91,mayers98a} and more recently analyzed in \cite{mckague10,braunstein12,yang14} that QKD strives to achieve. For example, devices could be produced by an eavesdropper that are `self-testing.' Alice and Bob can perform a test of their devices before the protocol so that when they run the protocol they are guaranteed that they get a secure key with high probability or the protocol aborts, regardless of what an eavesdropper does or how the devices behave. This kind of QKD would truly merit the name of ``unconditional'' security, since it would only rely on the most basic of assumptions, such as the laws of physics.

The authors of \cite{scarani09b} proposed that the QKD community would go in two directions: device-dependent proofs with increasingly more realistic models and device-independent proofs that ignore the underlying states and only deal with the conditional probability distributions resulting from inputs and outputs from measurements. It is still not clear if device-independent proofs will be experimentally implementable without making further assumptions, such as the assumption that loophole-free Bell tests are implemented. The protocols must also be able to tolerate realistic errors and loss. However, this direction seems most promising for removing unjustified assumptions from security proofs of QKD protocols.

With the advent of the entropic uncertainty relation (Theorem~\ref{thm:ur}) and measurement-device-independent QKD (Section~\ref{sec:MDI}) it seems a hybrid approach may also be possible, where only some devices need to be characterized. Experimentalists have recently implemented QKD using a security proof that uses the uncertainty relation \cite{bacco13} and MDI QKD has been implemented as well \cite{rubenok13,liu13}. Also, there is other theoretic work that seeks to connect perfect models with realistic implementations. For example, \cite{arrazola14} connects protocols that use qubits, unitaries, and projective measurements to protocols that use coherent states, linear optics, and threshold detectors.

An interesting tool that is currently lacking would be an entropic uncertainty relation that applies to single POVMs, so that a basis choice is not necessary. Perhaps it would also take loss into account, so that security could be proven for the DPS and COW protocols as well as prove security of the B92 protocol in a new way. The current uncertainty relation could be applied to these protocols but they relate the min-entropy of the actual protocol to the max-entropy of a counterfactual protocol. It is then unclear how to infer what the max-entropy is in the counterfactual protocol by only using parameter estimation in the actual protocol.

There are several efforts to make QKD more practical and useful by increasing the maximum distance possible \cite{hughes14} and doing QKD with a satellite \cite{meyer-scott11,wang12,vallone14,qi214}. There are also efforts to extend QKD to new frontiers: to not only use quantum mechanics but to also use relativity for QKD \cite{radchenko14,cotler14} and to perform QKD underwater \cite{shi14}.

This thesis has focused primarily on discrete-variable QKD, but there is much work being done on continuous-variable QKD. Imperfections and side-channel attacks in CV QKD are considered in \cite{jouguet12a,huang14}. Device-independent CV QKD has been proposed \cite{marshall14}. Experiments have also implemented CV QKD, for example, in free space \cite{heim14}. Improvements have also been found to increase the distance and key rate \cite{ma14,jouguet14,furrer14b}.

The fate of two-way QKD is still not clear, i.e.~it is not clear if two-way QKD provides an advantage over one-way protocols. A recent result does an error analysis for the original LM05 protocol \cite{shaari13} in an attempt to see if the protocol is secure without the need for the modifications in \cite{beaudry13} (see Section~\ref{sec:LM05}). It would be interesting to see how the key rate for a secure finite-key two-way QKD protocol would compare to one-way equivalents. While our result \cite{beaudry13} hints that such an advantage may be possible for the protocol based on super dense coding (due to the deterministic basis choices in the measurement) only a finite-key rate would take this advantage into account.

%% file: Appendix.tex

\chapter{Squeezed States and Phase Space} \label{app:squeezed}

This appendix defines squeezed states and how they can be represented in phase space (as well as how to represent coherent states in phase space). These states are used for continuous-variable quantum key distribution (see Section~\ref{sec:continuous_protocols}).

Squeezed states are related to coherent states (Eq.~\ref{eq:coherent_state}), in that they are a superposition of the number states of a particular form. First, note that a coherent state may be written as
\begin{equation} \label{eq:coherent_op}
\ket{\alpha} = e^{\alpha\hat{a}^{\dag} - \alpha^*\hat{a}}\ket{0},
\end{equation}
where $\hat{a}^{\dag}$ is the creation operator, $\hat{a}$ is the annihilation operator, and $\ket{0}$ is the vacuum state \cite{loudon00}. The operator in front of the vacuum state is called the \emph{displacement} operator, $\hat{D}(\alpha):= e^{\alpha\hat{a}^{\dag} - \alpha^*\hat{a}}$. This operator is named this way because it displaces the creation and annihilation operators: $\hat{D}^{\dag}(\alpha)\hat{a}\hat{D}(\alpha)=\hat{a}+\alpha$ and $\hat{D}(\alpha)\hat{a}^{\dag}\hat{D}^{\dag}(\alpha)=\hat{a}^{\dag}+\alpha^*$. Also, the displacement operator is unitary. To show this, we define $\hat{c}:=\alpha\hat{a}^{\dag} - \alpha^*\hat{a}$ and notice that $\hat{c}^{\dag}=-\hat{c}$, then
\begin{equation}
\hat{D}(\alpha)\hat{D}^{\dag}(\alpha) = \hat{D}^{\dag}(\alpha)\hat{D}(\alpha) =  e^{\hat{c}-\hat{c} +\frac{1}{2}[\hat{c},-\hat{c}]} = 1,
\end{equation}
where we use the property $e^{\hat{f}}e^{\hat{g}} = e^{\hat{f}+\hat{g}+\frac{1}{2}[\hat{f},\hat{g}]}$ and $[\hat{f},\hat{g}]:=\hat{f}\hat{g}-\hat{g}\hat{f}$ is the commutator.

A squeezed state may be written in a similar way to Eq.~\ref{eq:coherent_op} as a state being acted upon by the squeezing operator:
\begin{equation}
\hat{S}(\zeta) := e^{\frac{1}{2} \left( \zeta^{*} \hat{a}^2 - \zeta \hat{a}^{\dag 2} \right)}, \quad \zeta\in\mathds{C}.
\end{equation}
The squeezing operator can be shown to be unitary in the same way as the displacement operator. If we write the squeezing operator applied to the vacuum state as $\ket{\zeta}:=\hat{S}(\zeta)\ket{0}$, and write $\zeta$ in its polar form $\zeta=re^{i\theta}$ it can be written in the photon number basis \cite{loudon00} as
\begin{equation}
\ket{\zeta} = \sqrt{\sech r} \sum_{n=0}^{\infty} \frac{\sqrt{(2n)!}}{n!} \left( -\frac{1}{2} e^{i\theta} \tanh r \right)^n \ket{2n},
\end{equation}
where $\sech$ and $\tanh$ are the hyperbolic secant and tangent functions.
This state is called the squeezed vacuum state.

There are also squeezed coherent states, where the squeezed vacuum state is displaced by the displacement operator, $\ket{\alpha,\zeta}:=\hat{D}(\alpha)\ket{\zeta}$.

The idea in continuous-variable QKD is to send either coherent states or squeezed coherent states with different values of $\alpha$ and $\zeta$. Note that coherent states are nonorthogonal:
\begin{equation}
\bk{\beta}{\alpha} = e^{-\frac{|\beta|^2}{2}}e^{-\frac{|\alpha|^2}{2}} \sum_{n=0}^{\infty} \frac{\beta^{* n}\alpha^n}{n!} = e^{-\frac{|\beta|^2}{2}-\frac{|\alpha|^2}{2}+\beta^*\alpha} = e^{-|\alpha-\beta|^2}.
\end{equation}
The overlap is also known for squeezed coherent states \cite{sandoval92}, and they are also non-orthogonal. Therefore, Eve cannot distinguish the states that Alice sends with certainty.

A useful way to depict coherent states and squeezed states is by drawing them pictorially in phase space. If we define the quadrature operators $\hat{X}$ and $\hat{Y}$ as
\begin{equation}
\hat{X} := \frac{1}{2}(\hat{a}^{\dag}+\hat{a}),\quad \hat{Y} := \frac{1}{2} i (\hat{a}^{\dag}-\hat{a}),
\end{equation}
then recalling that coherent states are eigenvectors of coherent states, $\hat{a}\ket{\alpha} = \alpha\ket{\alpha}$, the expectation values of these operators for coherent states and squeezed coherent states are \cite{loudon00}:
\begin{align}
\bra{\alpha}\hat{X}\ket{\alpha} &= \frac{1}{2}\bra{\alpha}(\hat{a}^{\dag}+\hat{a})\ket{\alpha} = \frac{1}{2} (\alpha^*+\alpha) = \text{Re}(\alpha) \\
\bra{\alpha}\hat{Y}\ket{\alpha} &= \frac{i}{2}\bra{\alpha}(\hat{a}^{\dag}-\hat{a})\ket{\alpha} = \frac{i}{2} (\alpha^*-\alpha) = \text{Im}(\alpha) \\
\bra{\alpha,\zeta}\hat{X}\ket{\alpha,\zeta} &= \text{Re}(\alpha) \\
\bra{\alpha,\zeta}\hat{Y}\ket{\alpha,\zeta} &= \text{Im}(\alpha),
\end{align}
where we leave out the calculation of the squeezed coherent state expectation values (see \cite{loudon00}). Also, using the commutation relation $[\hat{a},\hat{a}^{\dag}]=1$ the expectation of the second moments of the coherent state are
\begin{align} 
\bra{\alpha}\hat{X}^2\ket{\alpha} &= \frac{1}{4}\bra{\alpha}(\hat{a}^{\dag}+\hat{a})^2\ket{\alpha} = \frac{1}{4} (\alpha^{* 2} + 2 |\alpha|^2 +\alpha^2 + 1) \\ 
&= \text{Re}(\alpha)^2 +\frac{1}{4}\\
\bra{\alpha}\hat{Y}^2\ket{\alpha} &= \frac{-1}{4}\bra{\alpha}(\hat{a}^{\dag}-\hat{a})^2\ket{\alpha} = \frac{-1}{4} (\alpha^{* 2} - 2 |\alpha|^2 +\alpha^2 - 1) \\
&= \text{Im}(\alpha)^2+\frac{1}{4}.
\end{align}

These expectations can be used to calculate the variances for coherent states and squeezed coherent states, where the squeezing parameter is $\zeta=re^{i\theta}$,
\begin{align}
\langle\Delta X\rangle^2 &= \langle \hat{X}^2\rangle -\langle \hat{X}\rangle^2 = \frac{1}{4} \\
\langle\Delta Y\rangle^2 &= \langle \hat{Y}^2\rangle -\langle \hat{Y}\rangle^2 = \frac{1}{4} \\
\langle\Delta X\rangle^2 &= \frac{1}{4} \left( e^{2r} \sin^2 \left(\frac{\theta}{2}\right) + e^{-2r} \cos^2\left(\frac{\theta}{2}\right) \right)\\
\langle\Delta Y\rangle^2 &= \frac{1}{4} \left( e^{2r} \cos^2 \left(\frac{\theta}{2}\right) + e^{-2r} \sin^2\left(\frac{\theta}{2}\right) \right).
\end{align}
This means that we can represent a coherent state in the $\hat{X}$-$\hat{Y}$ plane as centred on its expected value for $\hat{X}$ and $\hat{Y}$ (i.e.~$(\text{Re}\alpha,\text{Im}\alpha)$) with a region surrounding this point inside the variance. It can be shown that any linear combination of $\hat{X}$ and $\hat{Y}$ also leads to the variance of $1/4$, which means that this region is a circle (see Fig.~\ref{fig:phase_space}).

For a squeezed coherent state, we can similarly centre the state at $(\text{Re}\alpha,\text{Im}\alpha)$, and the variance is now an ellipse with major axis length $e^{2r}$ and minor axis length $e^{-2r}$ at an angle of $\theta/2$ (see Fig.~\ref{fig:phase_space}).

\begin{figure} \centering
\includegraphics[width=0.8\textwidth]{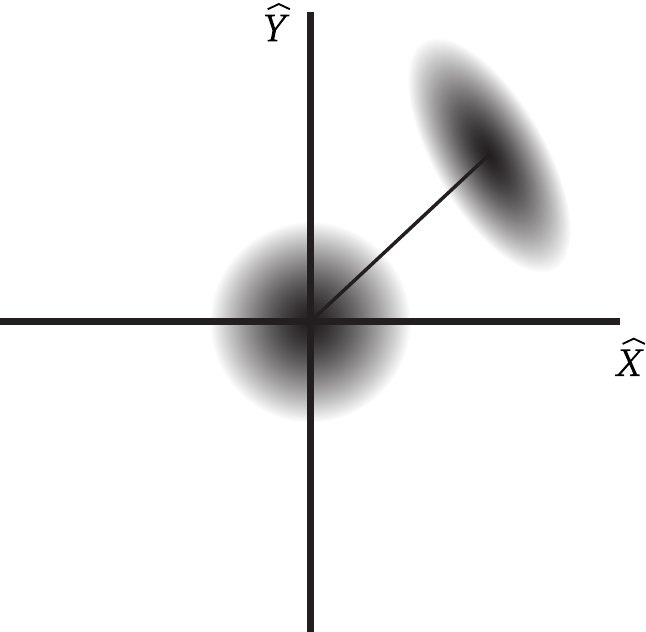}
\caption[A Coherent State and Squeezed State in Phase Space]{A coherent state and a squeezed state represented in phase space. The coherent vacuum state is centred at the origin with a Gaussian distribution of its probability density. The squeezed state is represented as an ellipse centred at $(\text{Re}(\alpha),\text{Im}(\alpha))$ at an angle $\theta$, whose width is given by the squeezing parameter $r$ and with a Gaussian distribution of its density.}
\label{fig:phase_space}
\end{figure}

\chapter{Miscellaneous Math} \label{app:math}

Here we discuss various mathematical properties used in the main text, including fields (Section~\ref{sec:fields}), Big O Notation (Section~\ref{sec:bigO}), and norms (Section~\ref{sec:norms}).

\section{Fields} \label{sec:fields}

In mathematics, a field is a set of elements with two operations, called addition and multiplication, that satisfy a list of properties. If $\mathcal{F}$ is a field, $a,b,c\in \mathcal{F}$ are elements of $\mathcal{F}$, and $+$ and $\times$ are the addition and multiplication operations then the operations $+$ and $\times$ must satisfy
\begin{itemize}
\item Closure: $a+b \in \mathcal{F}$ and $a \times b \in\mathcal{F}$.
\item Associativity: $a+(b+c) = (a+b) +c$ and $a \times (b\times c) = (a\times b)\times c$.
\item Commutativity: $a+b=b+a$ and $a\times b = b\times a$.
\end{itemize}
There must also exist identity elements and inverses in $\mathcal{F}$.
\begin{itemize}
\item Identity elements: There exist elements $0 \in \mathcal{F}$ and $1\in \mathcal{F}$ such that $\forall a \in \mathcal{F}$, $a+0=a$ and $a\times 1 = a$.
\item Inverses: For all $a\in \mathcal{F}$ there exists an element $-a \in \mathcal{F}$ and $a^{-1}\in\mathcal{F}$ such that $a+(-a)=0$ and (except for $a=0$) $a\times a^{-1} =1$.
\end{itemize}
Finally, $\times$ should be distributive over $+$.
\begin{itemize}
\item Distributivity: For all $a,b,c\in\mathcal{F}$ then $a\times (b+c) = (a\times b) + (a\times c)$.
\end{itemize}

One example of a field are the rational numbers with the usual addition and multiplication operations from arithmetic. One kind of field we use explicitly in Section~\ref{sec:post_processing} is a finite field, also called a Galois field, $GF$. They only exist when the number of elements are equal to a prime number to an integer power, such as $2^n$, where $n$ is an integer. An example of a Galois field is the set of integers modulo a prime number.

\section{Big O Notation} \label{sec:bigO}

Given a function of several variables, the limiting behaviour of the function can be characterized by its dominant term in a certain limit. For example, if the number of signals in a QKD protocol approaches infinity then the amount of classical communication needed may scale according to a function of the number of signals. We use the following notation from computer science to denote this scaling behaviour. 

\begin{defn}[Big $O$ notation] \label{defn:bigO}
Given two functions $f$ and $g$ that map from a subset of $\mathds{R}$ to a subset of $\mathds{R}$ then we write $f(x) = O(g(x))$ iff there exists a constant, $c$, and a real number $x_0$ such that 
\begin{equation}
|f(x)|\leq c |g(x)| \;\forall x\geq x_0.
\end{equation}
\end{defn}

As an example, consider the function $f(x) = 2x + \log x$, then $f=O(x)$.

\section{Norms} \label{sec:norms}

Norms are used to measure the size (in an abstract sense) of vectors in a vector space. Formally they are defined as functions from a vector space to a subfield of the complex numbers (i.e.~a field whose elements are complex numbers). A norm is denoted as $\| x \|$ for an element of a vector field $x\in V$. For all $a\in\mathcal{F}\subseteq \mathds{C}$ and $u,v \in V$ the norm $\| \cdot \|$ satisfies the following properties.
\begin{itemize}
\item Absolute linearity: $\|av\| = |a| \|v\|$.
\item Triangle inequality: $\|u+v\| \leq \|u\| + \|v\|$.
\item Zero vector: If $\|v\|=0$ then $v$ is the zero vector.
\end{itemize}

We now consider norms that are used in this thesis. 

\begin{defn}[Operator Norm, Infinity Norm] Let $L$ be a linear operator from $\hilbert_A$ to $\hilbert_B$, then the operator norm is defined as
\begin{equation}
\| L \|_\infty := \sup_{\ket{\psi}\in\hilbert_A} \frac{\left\|L\ket{\psi}\right\|}{\left\|\ket{\psi}\right\|},
\end{equation}
where $\left\| \ket{\psi}\right\|:=\sqrt{\bk{\psi}{\psi}}$. This norm is equivalent to the largest singular value\footnote{Singular values of an operator, $L$, are the eigenvalues of $\sqrt{L^{\dag}L}$.} of $L$. If $L$ is a normal matrix\footnote{A normal matrix, $N$, is one that satisfies $N^{\dag}N = N N^{\dag}$.}, then the singular values of $L$ are the same as the eigenvalues of $L$.
\end{defn}

\begin{defn}[Trace Norm] Let $L$ be a linear operator from $\hilbert_A$ to $\hilbert_B$ with singular values $s_i(L)$ then the trace norm is defined as
\begin{equation}
\| L \|_1 := \sum_{i} s_i (L) = \Tr |L|,
\end{equation}
where $|L|:=\sqrt{L^{\dag}L}$.
\end{defn}

\begin{defn}[Hilbert-Schmidt Norm] \label{defn:HSnorm} Let $L$ be a linear operator from $\hilbert_A$ to $\hilbert_B$ with singular values $s_i(L)$ then the Hilbert-Schmidt norm is defined as
\begin{equation}
\| L \|_2 := \sqrt{\sum_i {s_i(L)}^2} = \sqrt{\Tr (L^{\dag}L)}.
\end{equation}
\end{defn}

Norms can also be used to define a measure of distance, called a \emph{metric}. If $X$ is a set, then metrics are functions from $X\times X$ to the real numbers $\mathds{R}$. A metric, $d$, for all $x,y,z \in X$ has the following defining properties. 
\begin{itemize}
\item Non-negativity: $d(x,y)\geq 0$.
\item Identity of indiscernibles: $d(x,y)=0$ iff $x=y$.
\item Symmetry: $d(x,y)=d(y,x)$.
\item Triangle inequality: $d(x,z) \leq d(x,y) + d(y,z)$.
\end{itemize}

To use our norms defined above to define a metric, we simply take the difference of two vectors under the norm. The trace norm defines the $L_1$-distance between two vectors $x$ and $y$ as
\begin{equation}
\|x-y\|_1 = \sum_i \| x_i-y_i\|,
\end{equation}
where $x_i$ and $y_i$ are the elements of the vectors $x$ and $y$ respectively. As we will see below, this norm can also be used as a metric between quantum states.

There is also a quantum generalization of this distance measure. 

\begin{defn}[Trace Distance] \label{defn:tdist} Let $\rho,\sigma\in S_{\leq}(\hilbert)$ then the trace distance between $\rho$ and $\sigma$ is defined as
\begin{equation}
D(\rho,\sigma) := \frac{1}{2}\| \rho-\sigma \|_1.
\end{equation}
\end{defn}
The trace distance can be interpreted as a distinguishing probability. Given a state that is guaranteed to be either $\rho$ or $\sigma$ then the average success probability of correctly guessing which state it is by performing the optimal measurement whose outcome indicates that the state is $\rho$ or $\sigma$ is given by
\begin{equation}
\Pr[\text{guess}] = \frac{1}{2}+\frac{1}{2}D(\rho,\sigma).
\end{equation}

Another common quantity that characterizes the distance between quantum states is the fidelity. Here we define the fidelity as the generalized fidelity for unnormalized states.

\begin{defn}[Fidelity] \label{defn:fidelity} Let $\rho,\sigma\in S_{\leq}(\hilbert)$ then the generalized fidelity between $\rho$ and $\sigma$ is defined as
\begin{equation}
F(\rho,\sigma) := \|\sqrt{\rho}\sqrt{\sigma}\|_1 + \sqrt{(1-\Tr\rho)(1-\Tr\sigma)}.
\end{equation}
Note that if $\rho$ or $\sigma$ is normalized, then the generalized fidelity reduces to the fidelity $F(\rho,\sigma)=\|\sqrt{\rho}\sqrt{\sigma}\|_1$.
\end{defn}

The fidelity is unfortunately not a proper metric ($F(\rho,\sigma)=0$ iff $\rho=\sigma$ is not true, but instead $F(\rho,\sigma)=1$ iff $\rho=\sigma$). However, the fidelity is useful for several properties. One of these properties is its unitary invariance ($F(\rho,\sigma) = F(U\rho U^{\dag},U\sigma U^{\dag})$). Another useful property is that the states can be purified (see Section~\ref{sec:operators}) and the fidelity remains unchanged. 

\begin{thm}[Uhlmann's Theorem \cite{nielsen00}] Let $\rho,\sigma\in\hilbert$ and let a purification of $\rho$ be $\ket{\phi}$ then
\begin{equation}
F(\rho,\sigma) = \max_{\ket{\psi}}|\bk{\psi}{\phi}|,
\end{equation}
where $\ket{\psi}$ is a purification of $\sigma$.
\end{thm}

One way to turn the fidelity into a metric is by using the purified distance. 

\begin{defn}[Purified Distance \cite{tomamichel10c}] \label{defn:pdist} Let $\rho,\sigma\in S_{\leq}(\hilbert)$ then the purified distance between $\rho$ and $\sigma$ is defined as
\begin{equation}
P(\rho,\sigma) := \sqrt{1-F^2(\rho,\sigma)}.
\end{equation}
\end{defn}

This distance inherits many of the properties of the fidelity, and in addition is now a metric. To see the relationship between these various distance measures, there is the following set of inequalities. 

\begin{lemma}[Relationship of Trace Distance and Fidelity \cite{nielsen00}]
\begin{equation}
1-F(\rho,\sigma) \leq D(\rho,\sigma) \leq \sqrt{1-F^2(\rho,\sigma)}.
\end{equation}
\end{lemma}
This relationship means that the fidelity and trace distance (and also the purified distance) characterize the distance between states in a similar way.

An important property of all of these distances is that they are monotone under CPTP maps. 

\begin{lemma}[Distances under CPTP maps \cite{nielsen00,tomamichelthesis}] \label{lemma:CPTP_dist} Let $\rho,\sigma\in S_{\leq}(\hilbert)$ and given a CPTP map $\mathcal{E}$ from $S_{\leq}(\hilbert)$ to $S_{\leq}(\hilbert')$, then
\begin{align}
D(\rho,\sigma) &\geq D(\mathcal{E}(\rho),\mathcal{E}(\sigma)) \\
F(\rho,\sigma) &\leq F(\mathcal{E}(\rho),\mathcal{E}(\sigma)) \\
P(\rho,\sigma) &\geq P(\mathcal{E}(\rho),\mathcal{E}(\sigma)).
\end{align}
\end{lemma}

In addition, the trace distance is strongly convex. 

\begin{thm}[Strong convexity of the trace distance \cite{nielsen00}] \label{thm:trace_convex} Let $\rho^i,\sigma^i\in S_{\leq}(\hilbert)$ and $P$ and $Q$ be probability distributions with probabilities $p_i$ and $q_i$ for indices $i\in\mathcal{I}$. Then
\begin{equation}
D\left(\sum_i p_i \rho^i ,\sum_i q_i \sigma^i \right) \leq D(P,Q) + \sum_i p_i D(\rho^i,\sigma^i)
\end{equation}
\end{thm}